\documentclass[11pt]{article}
\usepackage{amsmath}
\usepackage{graphicx}
\usepackage{enumerate}
\usepackage[square, numbers]{natbib}
\usepackage{url} 
\usepackage{amssymb}
\usepackage{statmath}
\usepackage{array}
\usepackage{arydshln}
\usepackage{tabularx}

\usepackage{titlesec}

\titleformat{\paragraph}
{\normalfont\normalsize\bfseries}{\theparagraph}{1em}{}
\titlespacing*{\paragraph}
{0pt}{3.25ex plus 1ex minus .2ex}{1.5ex plus .2ex}

\usepackage{times}
\usepackage{bm}
\usepackage{mathtools}

\usepackage{algorithm}
\usepackage{amssymb}
\usepackage{mathrsfs}
\usepackage{graphicx}
\usepackage{rotating}
\usepackage[flushleft]{threeparttable}
\usepackage{multirow}
\usepackage{enumitem} 
\usepackage{tocloft}
\usepackage{star}

\usepackage[usenames,dvipsnames,svgnames,table]{xcolor}
\usepackage[colorlinks,
linkcolor=blue,
anchorcolor=blue,
citecolor=blue
]{hyperref}

\newtheorem{thm}{Theorem}

\def\##1\#{\begin{align}#1\end{align}}
\def\$#1\${\begin{align*}#1\end{align*}}

\newcommand{\Rom}[1]{\text{\uppercase\expandafter{\romannumeral #1\relax}}}

\usepackage{txfonts}

\usepackage[top=1in, bottom=1in, left=1in, right=1in]{geometry}

\theoremstyle{definition}
\newtheorem{defn}{Definition}[section]

\def\beq{\begin{equation}}
\def\eeq{\end{equation}}
\def\beqr{\begin{eqnarray}}
\def\eeqr{\end{eqnarray}}
\def\beqrs{\begin{eqnarray*}}
\def\eeqrs{\end{eqnarray*}}
\def\bet{\begin{theorem}}
\def\eet{\end{theorem}}
\def\bel{\begin{lemma}}
\def\eel{\end{lemma}}
\def\bep{\begin{proposition}}
\def\eep{\end{proposition}}
\def\bg{\begin{figure}[tbph]\begin{center}}
\def\eg{\end{center}\end{figure}}

\def\bc{\begin{center}}
\def\ec{\end{center}}

\def\var{\mbox{var}}

\def\bxzz{\color{black}}

\numberwithin{equation}{section}

\newcommand{\I}{\bm{I}}

\newcommand{\C}{\bm{C}}

\newcommand{\A}{\bm{A}}

\newcommand{\bmSigma}{\bm{\Sigma}}

\newcommand{\B}{\bm{B}}

\newcommand{\Q}{\bm{Q}}
\newcommand{\R}{\bm{R}}
\newcommand{\Z}{\bm{Z}}
\newcommand{\z}{\bm{z}}
\newcommand{\W}{\bm{W}}

\newcommand{\D}{\bm{D}}
\newcommand{\M}{\bm{M}}
\newcommand{\G}{\bm{G}}

\newcommand{\e}{\bm{e}}
\newcommand{\X}{\bm{X}}
\newcommand{\x}{\bm{x}}
\newcommand{\w}{\bm{w}}
\newcommand{\y}{\bm{y}}

\newcommand{\bmeps}{\bm{\epsilon}}
\newcommand{\bmbeta}{\bm{\beta}}

\newcommand{\bbR}{\mathbb{R}}

\newcommand{\Tr}{{\rm Tr}}

\theoremstyle{plain}
\newtheorem{proposition.a}{Proposition A\ignorespaces}
\newtheorem{proposition.s}{Proposition S\ignorespaces}

\newtheorem{thm.s}{Theorem S\ignorespaces}

\newtheorem{cor.s}{Corollary S\ignorespaces}
\newtheorem{lem}{Lemma S\ignorespaces}
\newtheorem{lem.s}{Lemma S\ignorespaces}

\theoremstyle{remark}

\newtheorem{remark.s}{Remark S\ignorespaces}
\newtheorem{ass.s}{Assumption S\ignorespaces}
\usepackage{varioref, hyperref, cleveref}

\usepackage{newfloat}
\DeclareFloatingEnvironment[name={Supplementary Fig.},fileext=lof]{suppfigure}
\DeclareFloatingEnvironment[name={Supplementary Table}]{supptable}
\usepackage{authblk}

\title{Uncertainty of high-dimensional genetic data prediction with
polygenic risk scores}
\author{Haoxuan Fu\footnote{Department of Statistics and Data Science, Cornell University. Email: hf357@cornell.edu}, Jiaoyang Huang\footnote{Department of Statistics and Data Science, University of Pennsylvania. Email: jyhuang@upenn.edu}, Zirui Fan\footnote{Department of Statistics and Data Science, University of Pennsylvania. Email: ziruifan@upenn.edu}, and Bingxin Zhao\footnote{Department of Statistics and Data Science, University of Pennsylvania. Email: bxzhao@upenn.edu}\\
}
\begin{document}
\maketitle
\date{\vspace{-3ex}}
\date{}

%

\begin{abstract}
In many predictive tasks, there are a large number of true predictors with weak signals, leading to substantial uncertainties in prediction outcomes. 
The polygenic risk score (PRS) is an example of such a scenario, where many genetic variants are used as predictors for complex traits, each contributing only a small amount of information.
Although PRS has been a standard tool in genetic predictions, its uncertainty remains largely unexplored. 
In this paper, we aim to establish the asymptotic normality of PRS in high-dimensional predictions without sparsity constraints.  
We investigate the popular marginal and ridge-type estimators in PRS applications, developing central limit theorems for both individual-level predicted values (e.g., genetically predicted human height) and cohort-level prediction accuracy measures (e.g., overall predictive $R$-squared in the testing dataset). 
Our results demonstrate that ignoring the prediction-induced uncertainty can lead to substantial underestimation of the true variance of PRS-based estimators, which in turn may cause overconfidence in the accuracy of confidence intervals and hypothesis testing.
These findings provide key insights omitted by existing first-order asymptotic studies of high-dimensional sparsity-free predictions, which often focus solely on the point limits of predictive risks.
We develop novel and flexible second-order random matrix theory results to assess the asymptotic normality of functionals with a general covariance matrix, without assuming Gaussian distributions for the data.
We evaluate our theoretical results through extensive numerical analyses using real data from the UK Biobank. 
Our analysis underscores the importance of incorporating  uncertainty assessments at both the individual and cohort levels when applying and interpreting PRS.\\
\end{abstract}



\section{Introduction}\label{sec1}

The polygenic risk score (PRS) is one of the most commonly used tools to predict an individual's genetic predisposition for complex traits and diseases \citep{ma2021genetic}. 
PRS combines the predictive power of millions of genetic variants, whose effects/weights are commonly estimated from genome-wide association studies (GWAS) \citep{uffelmann2021genome}. 
With the rapid expansion of large-scale GWAS in global biobanks \citep{zhou2022global}, PRS provides a promising approach for translating GWAS discoveries into medical advancements, such as genomic precision medicine \citep{torkamani2018personal}, disease subgroup stratification and early detection \citep{sinnott2021genetics}, and disease risk and trajectory prediction \citep{khera2019polygenic}. 
A major challenge in PRS is that most complex traits and diseases have a polygenic genetic architecture \citep{timpson2018genetic,boyle2017expanded}, in which a large number of genetic variants contribute to the phenotype. 
For example, genetics can control about $50\%$ of phenotypic variation in human height, however, it is shared by genetic variants spanning about $21\%$ of the whole genome \citep{yengo2022saturated}. 
Over the last two decades, a wide variety of statistical methods have been developed to improve PRS performance by better aggregating prediction power across massive genetic predictors \citep{ma2021genetic}. 
To examine the nonzero genetic effects across the genome, recent studies have used random matrix theory to model genetic data and quantify the asymptotic prediction accuracy of PRS \citep{zhao2022block,zhao2022polygenic}. 
These results relate to the first-order limit analysis of ridge-type estimators in general high-dimensional sparsity-free predictions \citep{dobriban2019one,10.1214/17-AOS1549,10.1214/21-AOS2133}. 

Although PRS applications are extremely popular and their first-order limits have been established, the uncertainty of PRS remains largely unexplored.
Consider a training GWAS dataset $(\X, \y)$ with $n$ subjects and an independent 
testing GWAS dataset $(\Z, \y_z)$ with $n_z$ subjects, PRS aims to develop predictions $\widehat{\y}_z$ for $\y_z$ using the $p$ genetic variants collected in both $\X$ and $\Z$ as predictors. 
In practice, two types of uncertainty measurements are particularly interesting in PRS applications. 
The first concerns the cohort-level prediction accuracy for $\widehat{\y}_z$,  commonly measured by the $R$-squared $A^2$, where $A=(\y_{z}^T\widehat{\y}_z)/(\Vert\y_{z}\Vert\cdot\Vert\widehat{\y}_z\Vert)$. 
The $A^2$ quantifies the overall performance for $n_z$ subjects and is useful for assessing and comparing the predictive abilities of PRS methods  \citep{wang2020theoretical}. However, the quantification of variance of $A^2$ is rarely considered in current PRS applications  \citep{momin2023significance}. 
The second directly relates to the uncertainty of the individual-level predicted values in $\widehat{\y}_z$ \citep{ding2022large}. For example, for $i=1,\dots,n_z$, what is the variance and confidence interval of $\widehat{y}_{z_i}$? 
Understanding the intrinsic uncertainty associated with $\widehat{\y}_{z}$ is essential for downstream clinical applications of PRS. For example, a typical PRS-based stratification analysis assigns $n_z$ subjects to actionable subgroups (e.g., high/low disease risks) based on $\widehat{\y}_{z}$ \citep{sugrue2019polygenic}. Such stratification can be misleading if $\widehat{\y}_{z}$ exhibits substantial noise and uncertainty \citep{wang2024impact}. Examining these two types of second-order results will provide complementary insights into PRS applications.

In random matrix literature, few second-order analyses exist for high-dimensional dense-signal predictions. Consequently, the limited existing studies do not provide direct insights into the two types of uncertainty encountered in PRS applications. For example,  \cite{li2021asymptotic} explores the asymptotic normality and confidence intervals for predictive risks using independent predictors, inspired by the ``more data hurt" phenomenon observed in deep learning. However, our analysis requires assuming a general covariance matrix $\bm{\Sigma}$ among predictors to account for the complex linkage disequilibrium (LD) patterns among genetic variants \citep{pasaniuc2017dissecting}.
In addition, \cite{10.1214/22-AOS2243} recently studies the asymptotic normal approximation of general de-biased estimators, including the ridge-type estimator that is popular in genetic fields, under Gaussian assumptions on $\X$ and $\Z$. Since genetic data in PRS applications are typically discrete values coded as ${0, 1, 2}$, we aim to develop more general results without relying on Gaussian assumptions. Our Gaussian-free analysis is also consistent with the broader random matrix theory literature, which has established first-order limits without assuming Gaussianity \citep{bai2004clt,Bai_Zhou2008,fan2019spectral,10.1214/21-AOS2133}. 

In this paper, we develop central limit theorems (CLTs) for high-dimensional predictions without sparsity constraints, covering the marginal estimator \citep{power2015polygenic} and reference panel-based ridge-type estimator \citep{ge2019polygenic}. Due to privacy restrictions and logistical challenges, methods that do not require access to individual-level data $(\X, \y)$ have become the standard for PRS applications \citep{ma2021genetic}. The marginal and reference panel-based ridge-type estimators only need access to summary statistics $\X^T\y$. Therefore, these two estimators and their extensions have been widely used in genetic data prediction and PRS development.
Given the distinct nature of marginal and ridge estimators, our analyses of these two estimators provide insights into the role of various factors in the second-order behavior of PRS.
For example, we identify the key factors that determine the CLT convergence rate and post-prediction hypothesis testing, including the sample sizes, the dimension of genetic predictors $p$, the overall signal strength (i.e, the heritability $h^2$ in genetics \citep{yang2017concepts}), the underlying signal sparsity $m$, and the LD matrix $\bmSigma$ among the $p$ predictors. 
These results can be viewed as distributional extensions of previously reported first-order limits on the predictive risks \citep{10.1214/17-AOS1549,zhao2022block}.
Our results show the importance of considering uncertainty at both the cohort and individual levels in PRS applications. Specifically, we provide theoretical foundations to address the following two pivotal practical questions

\begin{quote}
{\it \normalsize
\begin{itemize}
\item 
What's the well-calibrated confidence interval for the predicted value $\widehat{y}_{z_i}$?
\item 
How to develop valid statistical inference for prediction accuracy $A^2$? 
\end{itemize}
}
\end{quote}

To examine the second-order behavior of PRS, we introduce novel random matrix theory results and proof workflows to overcome three major technical challenges. 
The first difficulty lies in managing the intricate randomness from training data, testing data, and genetic reference panels \citep{pasaniuc2017dissecting,zhao2022block}. 
To tackle this, we break down our estimators into quadratic and residual functionals. 
We then derive the limiting distributions of these components through a sequential conditioning procedure. This proof strategy introduces a novel framework for dealing with multi-source randomness and has the potential for broad applicability in other problems. 
The second challenge arises from the quadratic forms present in $A^2$, which cannot be straightforwardly decomposed into a sum of independent entries, as is commonly done in the Lindeberg-Feller CLT. To handle this, we use martingale CLT methodologies \citep{10.1214/aoms/1177693494} to derive their limiting distributions. 
The third challenge involves addressing the resolvents in reference panel ridge-type estimators. We use anisotropic local laws \citep{anisotropic_local_law} to quantify their uncertainty. 
We provide an overview of our major contributions in Table~\ref{tab:summary}. 
Given the widespread applications of high-dimensional sparsity-free prediction models in genetic prediction and many other fields, our results and approach provide useful insights into the second-order fluctuations for a variety of practical problems.

\subsection{Notation}
We make use of the following notations frequently. $\Tr(\A)$ is the trace of matrix $\A$, $\mbox{diag}(\A)$ is a diagonal matrix obtained by dropping all off-diagonal elements of $\A$, $\to$ donates the convergence of a series of real numbers, $\overset{p}{\to}$ represents the in probability convergence of a series of random variables, and $\overset{a.s.}{\to}$ is the almost surely convergence of a series of random variables. We use $\|\cdot\|_2$ as the $L_2$ norm and $\|\cdot\|_F$ as the Frobenius norm. In addition, $o_p(1)$, $\Theta_p(1)$, and $O_p(1)$ define the small $o$, Theta $O$ and big $O$ depends on $p$, respectively. Moreover, we use $h\asymp g$ to denote that $\lim_{p\to\infty}h_p/g_p = \Theta_p(1)$. We also use notation $h \prec g$ when $h = o_p(g)$, $h\preceq g$ when $h = O_p(g)$, and the reverse direction $\succ, \succeq$, respectively. $c,C \in \mathbb R$ are some generic large enough constant numbers. We define ${\mathbb\I}_p$ as the $p\times p$ identity matrix, whereas $\mathbb\I_m$ stands for the $p\times p$ diagonal matrix with only first $m$ diagonal entries equal to $1$ and other diagonal entries equal to $0$. We also use $Im(x)$ to stand for the imaginary part of some $x\in \mathbb C$.
$\mathbb E(\cdot)$ is used to denote expectation. Throughout the paper, we use $\Phi(x)$ to denote the cumulative distribution function of standard Gaussian distribution, that is $\Phi(t) = \int_{-\infty}^t(2\pi)^{-1/2}e^{-x^2/2}dx, 
t\in \mathbb R$. 
Moreover, we use $\delta_x$ to denote an indicator with delta mass at $x$.
\subsection{Paper overview}\label{sec1.1}
The rest of the paper proceeds as follows.
In Section~\ref{sec:linear_model}, we introduce model setups, estimators, and basic assumptions. Section~\ref{sec3} presents the results of the marginal estimator and Section~\ref{sec:ref} shows the results of the reference panel-based ridge estimator.
Numerical results based on the real genetic data in the UK Biobank study \citep{bycroft2018uk} are provided in Section~\ref{numerical_sec}. Potential future research directions are discussed in Section~\ref{sec:discussion}. The supplementary file collects additional results and proofs of the main results.

\begin{table}
\caption{ 
{\bf An overview of our major contributions and connections with existing literature.} 
In this paper, we quantify the uncertainty of PRS in high-dimensional genetic data prediction. Specifically, we provide second-order results for popular marginal and reference panel-based ridge estimators. We consider a general $\bmSigma$ and a generic distribution of predictor data $\X$ without assuming Gaussianity. Due to the distinct nature of these estimators, we have different conditions for the $n/p$ ratio.
}  
\label{tab:summary}
{%
\begin{center}
\begin{tabular}{ c c c }
 & Second-order fluctuations & First-order limits\\
 \hline
 Marginal & \begin{tabular}{@{}c@{}}Theorem~\ref{thm: CLT for marginal new} and \ref{thm: CLT for A^2 Marginal} \\ (Flexible $n/p$ ratio) \end{tabular}  &  \begin{tabular}{@{}c@{}}\cite{zhao2022polygenic} \\ (Constant $n/p$ ratio) \end{tabular} \\
 \hline
Reference-panel ridge & \begin{tabular}{@{}c@{}}Theorems~\ref{thm: CLT for reference new} and \ref{thm: CLT for reference A^2} \\ (Constant $n/p$ ratio) \end{tabular}  &  \begin{tabular}{@{}c@{}}\cite{zhao2022block} \\ (Constant $n/p$ ratio) \end{tabular}
\end{tabular}
\end{center}}
\end{table}

\section{Model setups}
\label{sec:linear_model}
In this section, we introduce the model, basic assumptions, estimators, definitions, and the outline of the proof. 
\subsection{The model}\label{subsec:model_constr}
Consider one continuous phenotype that is studied in two independent GWAS cohorts
\begin{itemize}
\item Training dataset: $(\X,\y)$, with $\X\in \bbR^{n\times p}$ and $\y \in \bbR^{n}$.
\item Testing dataset: $(\Z,\y_{z})$, with $\Z\in \bbR^{n_z\times p}$ and $\y_{z} \in \bbR^{n_{z}}$. 
\end{itemize}
Here $\y$ and $\y_{z}$ represent the complex trait in the two GWAS with sample sizes $n$ and $n_z$, respectively. 
We assume that both GWAS have the same number of genetic variants $p$, most of which are single nucleotide polymorphisms (SNPs). 
A linear additive polygenic model is assumed \citep{10.1214/15-AOS1421}
\begin{flalign}
\y= \X\bmbeta+\bmeps=\X_{(1)}\bmbeta_{(1)}+\bmeps  \quad \text{and}  \quad
\y_{z}=\Z\bmbeta+\bmeps_{z}= \Z_{(1)}\bmbeta_{(1)}+\bmeps_{z}, 
\end{flalign}
where $\bmbeta_{(1)}$ is a $m \times 1$ vector of nonzero causal genetic effects, $\bmbeta$ is a $p \times 1$ vector consisting of $\bmbeta_{(1)}$ and $p-m$ zeros, 
and $\bmeps$ and $\bmeps_{z}$ represent independent random error vectors.
Due to data privacy and logistical challenges, individual-level genotypes $\X$ are typically not publicly accessible \citep{pasaniuc2017dissecting}. Therefore, an external genetic reference panel dataset \citep{10002015global} is commonly used as a substitute for $\X$ to estimate the LD pattern among genetic variants
\begin{itemize}
\item Reference panel dataset: $\W\in \bbR^{n_w \times p}$.
\end{itemize}

To establish the random matrix theory framework, we adopt the following assumptions, which are frequently used in the literature \citep{10.1214/aop/1078415845, Yao_Zheng_Bai_2015, Bai_Zhou2008,10.1214/17-AOS1549}. 
These conditions (or more restricted versions) have also been applied in previous first-order analyses of PRS \citep{zhao2022block,zhao2022polygenic}, allowing us to perform comparisons. 
Our first assumption below summarizes the basic conditions of the genetic datasets, such as moments of entries and boundedness of the covariance matrix.
\begin{assumption}
\label{a:Sigmabound}
For $p\geq 1$, we take a sequence of $ \bmSigma^{(p)}=\bmSigma \in \mathbbR^{p\times p}$. 
Denote $n=n(p)\to \infty$, $n_z=n_z(p)\to \infty$, and $n_w=n_w(p)\to \infty$, as $p \to \infty$.
Let $\mathcal{F}$ be a generic distribution with mean zero, variance one, and bounded eighth moment. 
We take sequences of data matrix $\X^{(p)}=\X=\X_0 \bmSigma^{1/2}\in \mathbbR^{n\times p}$, $\Z^{(p)}=\Z=\Z_0\bmSigma^{1/2}\in \mathbbR^{n_z\times p}$, and $\W^{(p)}=\W=\W_0\bmSigma^{1/2}\in \mathbbR^{n_w\times p}$, where the entries of $\X_0$, $\Z_0$, and $\W_0$ are i.i.d. random variables following $\mathcal{F}$. 
We further denote $\x\in \bbR^{p}$, $\z\in \bbR^{p}$, and $\w \in \bbR^{p}$ as single data entries/rows of data matrix $\X$, $\Z$, $\W$, respectively; whereas $\x_{i}, \z_{i}$, and $\w_{i}$ are used to denote the $i_{th}$ column of data matrix $\X$, $\Z$, and $\W$, respectively, $i=1,\dots,p$.
Moreover, we assume there exist constants $0<c<C$ such that the following holds
\begin{enumerate}
    \item[(a)] The norm of $\z$ is bounded: $\|\z\|_2\leq C$.
    \item[(b)] The covariance matrix $\bmSigma$ is lower and upper bounded: $c\mathbb \I_p\prec \bmSigma\prec C\mathbb \I_p$.
\end{enumerate}
\end{assumption}

The second assumption introduces the conditions of genetic effects and noise vectors.

\begin{assumption}\label{a:Sparsity}
    We assume $\bm\beta$ follows from a  generic distribution $\mathcal{L}$ with mean $\mathbf{0}$ and covariance {$\sigma_{\bm\beta}^2\mathbb\I_m$}, where $\mathbb\I_m$ as a $p\times p$ identity matrix with a sparsity of $m$. That is, without loss of generality, the first $m$ diagonal terms to be one and the rest $p-m$ diagonal terms to be equal to zero. We assume the norm of $\bm\beta$ is bounded with $\|\bm\beta\|_2\leq C$. We assume that the entries of $\bm\beta$ are i.i.d. random variables, with  $\mathbb E(\bm\beta_i^4) \asymp O_p(p^{-1})$ for $1\leq i\leq p$, and for simplicity, we write the fourth moment as $\mathbb E(\bm\beta^4)$. We also assume that the entries of $\epsilon$ and $\epsilon_{z}$  are i.i.d. random variables follow distributions with mean $0$ and variances $\sigma_{\epsilon}^2$ and $\sigma_{\epsilon_z}^2$, respectively. Moreover, we assume the fourth moment of entries of $\epsilon$ and  $\epsilon_z$ are bounded, denoted as $\mathbb E(\epsilon^4)$ and $\mathbb E(\epsilon_z^4)$, respectively. 
\end{assumption}
In these two assumptions, we do not require the typical $m\asymp p$ and/or $n\asymp p$ assumptions frequently used in previous studies. This is because they are not needed in our analyses of the marginal estimator. 
As shown in Section~\ref{sec3}, this flexibility allows us to provide deeper insights into the CLT convergence rates. 
We will introduce proportional assumptions in Section~\ref{sec:ref} when analyzing the resolvents in ridge-type estimators. 
For simplicity, we assume the $\bm\beta$ and $\z$ have been normalized, making the limits involved in these quantities more traceable to derive insights. These technical conditions are natural and have also been adopted in previous literature \citep{10.1214/15-AOS1421,zhao2022block,10.1214/17-AOS1549}, such as the ``random regression coefficients" assumption in \cite{10.1214/17-AOS1549}. 


As a key quantity to measure genetic signal strength, we define the heritability below \citep{yang2017concepts}. 
\begin{defn}[Heritability]
\label{def:heritability} 
    Conditional on $\bm\beta$, the heritability of $\y$ and $\y_z$ is defined as 
    \begin{equation}
        \begin{aligned}
            h_{\bm\beta}^2 \coloneqq \frac{\|\bm\Sigma^{1/2}\bm\beta\|_2^2}{\|\bm\Sigma^{1/2}\bm\beta\|_2^2 + \sigma_\epsilon^2}\quad \mbox{and} \quad h^2_{{\bm\beta}_z} \coloneqq \frac{\|\bm\Sigma^{1/2}\bm\beta\|_2^2}{\|\bm\Sigma^{1/2}\bm\beta\|_2^2 + \sigma_{\epsilon_z}^2}, 
        \end{aligned}
    \end{equation}
    respectively. Given Assumption~\ref{a:Sparsity}, we further have 
    $h_{\bm\beta}^2=\{\sigma_{\bm\beta}^2\Tr(\bm\Sigma\mathbb I_m)/p\}/\{\sigma_{\bm\beta}^2\Tr(\bm\Sigma\mathbb I_m)/p+\sigma_\epsilon^2\}$ and $h_{{\bm\beta}_z}^2=\{\sigma_{\bm\beta}^2\Tr(\bm\Sigma\mathbb I_m)/p\}/\{\sigma_{\bm\beta}^2\Tr(\bm\Sigma\mathbb \I_m)/p + \sigma_{\epsilon_z}^2\}$.
\end{defn}
Heritability represents the proportion of phenotypic variance attributed to genetic data, which closely relates to the signal-to-noise ratio in the literature.  
For example, the signal-to-noise ratio defined in  \cite{10.1214/17-AOS1549} equals to $h_{\bm\beta}^2/(1-h_{\bm\beta}^2)$, the ratio between $h_{\bm\beta}^2$ and $1-h_{\bm\beta}^2$ \citep{zhao2022block}. 

\subsection{Estimators}
\label{subsec:estimators}
In this section, we define the two estimators studied in this paper. The marginal estimator can be formulated as $\hat{\bm\beta}_{\text{M}} = n^{-1}\X^\top \y$, corresponding to the popular GWAS marginal summary association statistics available in the public domain \citep{hayhurst2022community}. These marginal genetic effects have been directly used in genetic data prediction \citep{power2015polygenic} and have widespread applications \citep{choi2020tutorial}.
The reference panel-based ridge-type estimator \citep{zhao2022block} is given by $\hat{\bm\beta}_{\text{W}}(\lambda) = (\W^\top \W+n_w\lambda\mathbb\I_p)^{-1}\X^\top \y$, which aims to account for the LD pattern using $\W^\top \W$ estimated from a reference panel, such as the 1000 Genomes \citep{10002015global}. Various reference panel-based estimators have been developed in PRS applications \citep{ma2021genetic}, many of which adopt the $L_2$ regularized estimator and its variants to fit the polygenic genetic architecture of complex traits \citep{boyle2017expanded}.
It is also worth mentioning that both $\hat{\bm\beta}_{\text{M}}$ and $\hat{\bm\beta}_{\text{W}}(\lambda)$ are based on the marginal genetic effect estimates $\X^\top \y$ and do not need access to individual-level data $(\X, \y)$, and 
$\hat{\bm\beta}_{\text{M}}$ can be viewed as a limiting case of $\hat{\bm\beta}_{\text{R}}(\lambda)$ as $\lambda \to \infty$.

We study the uncertainty associated with these two estimators at both individual and cohort levels in high-dimensional predictions. Let $\hat{\bm\beta} \in \mathbb{R}^p$ be a generic estimator of $\bm\beta$. First, consider a single data point $\z$ from one subject in the testing dataset; the genetically predicted value for this individual can be represented as $\z^\top \hat{\bm\beta}$, which is the target of our analysis. The uncertainty of $\z^\top \hat{\bm\beta}$ remains largely unexplored but has substantial impacts on the performance of downstream PRS analyses \citep{ding2022large,wang2024impact}. Second, the predictor of $\y_z$ for all subjects in the testing dataset is given by $\hat{\by} = \Z \hat{\bm\beta}$.
We use the out-of-sample $R$-squared, denoted as $A^2 (\hat{\bm\beta})$ \citep{daetwyler2008accuracy,dudbridge2013power}, to measure the cohort-level prediction accuracy for $\hat{\by}$, where $ A (\hat{\bm\beta})= (\by_z^\top \hat{\by})/(\| \by_z\|_2\|\widehat \by\|_2)$. 
The out-of-sample $R$-squared $A (\hat{\bm\beta})$ is widely reported in PRS applications; however, the quantification of the sampling variance of $A (\hat{\bm\beta})$ is rarely considered in current PRS applications \citep{momin2023significance}. 

\subsection{Outline of the proof} 
We aim to establish CLTs of $\z^\top\hat{\bm\beta}$ and $A (\hat{\bm\beta})$ for both $\hat{\bm\beta}_{\text{M}}$ and $\hat{\bm\beta}_{\text{W}}(\lambda)$. 
Considering the different characteristics of these quantities and estimators, we use appropriate techniques for each estimator accordingly.
For $\z^\top\hat{\bm\beta}_{\text{M}}$, our general strategy is to follow the leave-one-out technique to obtain the limiting distributions of functionals involved in $\z^\top\hat{\bm\beta}_{\text{M}}$. To handle the multi-source randomness from training and testing datasets, we sequentially establish the Berry-Esseen bound \citep{ash2000probability} using isolated randomness while fixing other random variables. Following this strategy, we establish quantitative CLTs for two types of functionals of interest, including quadratic functionals given by $\langle \bm a, \X_0^\top \X_0 \bm b\rangle$ and residual functionals $\langle \bm a, \bmeps\rangle$, under some moment conditions, for some deterministic vectors $\bm a$ and $\bm b$. We then obtain the CLT of $\z^\top\hat{\bm\beta}_{\text{M}}$ by decomposing it into the summation of these functionals.
Furthermore, while the leave-one-out technique is sufficient for showing the normality of $\z^\top\hat{\bm\beta}_{\text{M}}$, it is not straightforward to decorrelate the symmetry structures involved in $ A (\hat{\bm\beta}_{\text{M}})$, such as $\bmbeta^\top\bmSigma\bmbeta$. To address this, we adopt the martingale CLT \citep{10.1214/aoms/1177693494} techniques with additional assumptions on the bounded fourth moment of $\bmbeta$.

Next, for the reference panel-based ridge estimator $\hat{\bm\beta}_{\text{W}}(\lambda)$, the resolvent part $(\W^\top \W+n_w\lambda\mathbb\I_p)^{-1}$ cannot be approached with traditional CLT techniques. Given the independence between the reference panel $\W$ and training data $(\X,\y)$, we isolate the terms involving $\W$ and organize these quantities as products with the structure $\langle \bm a, (\W^\top \W+n_w\lambda\mathbb\I_p)^{-1}\bm b\rangle$. Using anisotropic local laws \citep{anisotropic_local_law}, we obtain limiting deterministic metrics for these quantities involved in $\z^\top\hat{\bm\beta}_{\text{W}}$ and $A (\hat{\bm\beta}_{\text{W}})$. It is important to note that the limiting behavior of quantities involving two resolvents also depends on the variation in the eigenvalue distribution of $\bmSigma$ in response to perturbations, which can be analyzed through eigenvalue perturbation methods in \cite{hogben2013handbook}. We perform a careful examination to establish a union bound that helps exclude extreme cases.



\section{Marginal estimator}\label{sec3}
We start with the marginal estimator $\hat{\bm\beta}_{\text{M}}$ and provide the individual-level uncertainty 
in Section~\ref{subsec:marg_new_pred}. In Section~\ref{subsec:marg_A2}, we present general CLT results for quadratic forms, which serve as the foundation of our cohort-level uncertainty analysis in Section~\ref{subsubsec:marg_A_iso}. In both Sections~\ref{subsec:marg_new_pred} and \ref{subsubsec:marg_A_iso}, we first present the results for the special case $\bm\Sigma = \mathbb{I}_p$ to offer intuitive results and insights. We then provide results for the general $\bmSigma$, highlighting its complex role in the asymptotic results of $\hat{\bm\beta}_{\text{M}}$.

\subsection{Individual-level uncertainty}
\label{subsec:marg_new_pred}
Quantifying individual-level predictive uncertainty in PRS is important for understanding the reliability of genetic prediction results. For example, \cite{ding2022large} illustrates that substantial, non-ignorable uncertainty may be present in genetically predicted values, negatively impacting their applications. However, few studies have rigorously quantified this uncertainty.
In Corollary~\ref{cor: CLT_marg_new_iso}, we provide the asymptotic normality for $\z^\top\hat{\bm\beta}_{\text{M}}$  when $\bm\Sigma = \mathbb\I_p$. 
\begin{coro}
\label{cor: CLT_marg_new_iso}
Under Assumptions~\ref{a:Sigmabound}-\ref{a:Sparsity} and $\bm\Sigma=\mathbb\I_p$, as $\min(n,m,p)\to \infty$, conditioning on the testing data point $\z$ and genetic effect $\bmbeta$, 
let
     $$\bm F_{\text{M}_0}(t) = \mathbb P\left(\sigma_{\text{M}_0}^{-1}\sqrt{n}(\z^\top\hat{\bm\beta}_{\text{M}}-\z^\top\bm\beta) < t\right),$$
    where $$\sigma_{\text{M}_0}^2 = \mathbb E\left(x_0^4-3\right)\sum_{i=1}^p(\z)_i^2(\bmbeta)_i^2 + h_{\bm\beta}^{-2}\|\z\|_2^2\|\bm\beta\|_2^2 + 2(\z^\top \bm \beta )^2$$
    and $x_0$ is an entry of the data matrix $\X_0$.   
    Then the following Berry-Esseen bound holds
    \begin{align*}
        \sup_{t\in\mathbb R}\left|\bm F_{\text{M}}(t) - \Phi(t)\right| \leq Cn^{-1/2}.
    \end{align*} 
\end{coro}
For a given data point $\z$ and underlying genetic effect $\bm{\beta}$, our results suggest that the training data sample size $n$ determines the pointwise CLT convergence rate and the uniform convergence rate of the Berry-Esseen bound. Furthermore, the mean of the genetically predicted value is $\z^\top\bm\beta$, and its variance is
influenced by the two vectors $\z$ and $\bmbeta$ and their interactions, as well as the fourth moment of the training data genetic variants and the overall heritability level $h_{\bm{\beta}}^2$. 
In practice, traits with higher heritability $h_{\bm{\beta}}^2$ tend to have a narrower confidence interval, and an increase in the training data sample size $n$ will result in a faster convergence rate toward normality. 
Importantly, these results do not require $m \asymp p$ or $n \asymp p$. Therefore, they are applicable to datasets with varying sample sizes and numbers of genetic variants, as well as phenotypes with different genetic architectures. In addition to the PRS applications that use genome-wide genetic variants, our results may also apply to omics data predictions that typically focus on a smaller number of {\it cis} genetic variants in a specific genomic region, such as the PrediXcan \citep{gamazon2015gene}.

Theorem~\ref{thm: CLT for marginal new} presents the asymptotic normality for $\z^\top\hat{\bm\beta}_{\text{M}}$ when general covariance structure $\bmSigma$ is considered.
\begin{thm}
\label{thm: CLT for marginal new}
Under Assumptions~\ref{a:Sigmabound}-\ref{a:Sparsity}, as $\min(n,m,p)\to \infty$, conditioning on the testing data point $\z$ and genetic effect $\bmbeta$, 
let
     $$\bm F_{\text{M}}(t) = \mathbb P\left(\sigma_{\text{M}}^{-1}\sqrt{n}(\z^\top\hat{\bm\beta}_{\text{M}}-{\z^\top\bmSigma\bm\beta}) < t\right),$$ 
    where $$\sigma_{\text{M}}^2 = \mathbb E\left(x_0^4-3\right)\sum_{i=1}^p(\bm \Sigma^{1/2}\z)_i^2(\bm\Sigma^{1/2}\bmbeta)_i^2 + h_{\bm\beta}^{-2}\|\bm\Sigma^{1/2}\z\|_2^2\|\bm\Sigma^{1/2}\bm\beta\|_2^2 + 2(\z^\top\bm\Sigma \bm\beta)^2$$
    and $x_0$ is an entry of the data matrix $\X_0$.   
    Then the following Berry-Esseen bound holds
    \begin{align*}
        \sup_{t\in\mathbb R}\left|\bm F_{\text{M}}(t) - \Phi(t)\right| \leq Cn^{-1/2}.
    \end{align*} 
\end{thm}
Theorem~\ref{thm: CLT for marginal new} suggests that the asymptotic Gaussian distribution still holds under a general correlation structure among the genetic predictors. 
Since the marginal estimator $\hat{\bm\beta}_{\text{M}}$ does not account for $\bmSigma$, the mean of $\z^\top\hat{\bm\beta}_{\text{M}}$ is $\z^\top\bmSigma\bm\beta$, where $\bmSigma\bm\beta$ represents the marginal genetic effects. 
These results indicate that the genetic predicted values may not reflect the underlying true phenotype, as the marginal genetic effects estimator $\hat{\bm\beta}_{\text{M}}$ is biased. 
Depending on the specific application, bias correction may be needed. For example, if we aim to recover the exact human height in genetic prediction, the values themselves matter, and we need to correct the bias induced by $\bmSigma$. In other cases, such as stratifying the top $10\%$ of subjects with the highest genetically predicted disease risk, we may not need to correct the predicted values, as the bias is similar for all subjects and the relative ranks may not change. 
It is worth mentioning that such bias is common in high-dimensional sparsity-free prediction, because generally we may not have a consistent estimator of $\bm\beta$ due to the high dimensionality \citep{celentano2023challenges}.

In addition, our results show how the variance of the predicted value is influenced by the LD structure $\bm{\Sigma}$. Specifically, the variance is closely associated with $\bm{\Sigma}^{1/2}\z$ and $\bm{\Sigma}^{1/2}\bm{\beta}$, whose product represents the mean $\z^\top\bmSigma\bm\beta$. 
One of the most popular PRS methods is ``clumping and thresholding" \citep{choi2020tutorial}, which directly uses $\hat{\bm\beta}_{\text{M}}$ as input and removes genetic variants that have very high correlations with other predictors or very large $P$-values. Our results provide insights into genetically predicted values and the development of their confidence intervals.


\subsection{Quadratic form with general covariance matrix}
\label{subsec:marg_A2}
Before moving to the cohort-level prediction accuracy analysis, we will establish the general CLT results for the quadratic form $\bm{\beta}^\top\bm{\Sigma}\bm{\beta}$ involved in $A (\hat{\bm\beta})$. Although the first-order limit of $\bm{\beta}^\top\bm{\Sigma}\bm{\beta}$ has been extensively studied \citep{10.1214/aop/1022855421,Yao_Zheng_Bai_2015}, there is little literature discussing the asymptotic distribution of quadratic forms with a general $\bmSigma$. Previous studies have developed CLTs for quadratic forms under various additional assumptions on $\bm{\Sigma}$, primarily to avoid merging the Gaussian distribution generated by $\bm{\Sigma}$'s diagonal terms with that from the off-diagonal terms. For example, \cite{dejong1987central} proposes a CLT for a quadratic form under the assumption that $\Tr(\bm{\Sigma}) = 0$. 
In addition, \cite{10.1214/aos/1031833676} assumes all off-diagonal terms decay at a rate faster than $m^{-1}$, whereas \cite{gotze2002asymptotic} studies the situation where the off-diagonal terms decay at a rate slower than $m^{-1/2}$. 
These CLTs may not be directly used in our PRS applications due to the complex LD patterns we have in GWAS data.

To develop second-order results in more general situations, we present a CLT for $\bm{\beta}^\top\bm{\Sigma}\bm{\beta}$ that does not require additional assumptions on the entries of $\bmSigma$. Our approach primarily depends on the spectral properties of $\bmSigma$, assuming the eigenvalues are bounded as specified in Assumption~\ref{a:Sigmabound}.
\begin{thm}
\label{thm:quad_form}
    Under Assumptions~\ref{a:Sigmabound}-\ref{a:Sparsity}, as $m\to \infty$, let 
    \begin{equation*}
        \begin{aligned}
            \bm G_{\text{Q}}(t) = \mathbb P\left(\sigma_{\text{Q}}^{-1}\sqrt{p}\left(\bm\beta^\top\bm\Sigma\bm\beta - \sigma_{\bm\beta}^2\Tr(\bm\Sigma\mathbb\I_m)/p\right) < t\right),
        \end{aligned}
    \end{equation*}
    where
    \begin{equation*}
        \begin{aligned}
            \sigma_{\text{Q}}^2 = p\left\{\mathbb E\left({\bm\beta}^4\right)-3\sigma_{\bm\beta}^4/p^2\right\}\sum_{i=1}^m{{(\bm\Sigma\mathbb\I_m)}_{i,i}}^2 + 2\sigma_{\bm\beta}^4\Tr\left\{(\bm\Sigma\mathbb\I_m)^2\right\}/p.
        \end{aligned}
    \end{equation*}
    Then the following Berry-Esseen bound holds 
    \begin{equation*}
        \begin{aligned}
            \sup_{t\in\mathbb R}\left|\bm G_{\text{Q}p}(t) - \Phi(t)\right| &\leq Cm^{-1/5}.
        \end{aligned}
    \end{equation*}
\end{thm}

\begin{rmk}\label{remark1}
The major difficulty in developing the CLT for the quadratic form with general $\bm{\Sigma}$ is the symmetry structure regarding $\bm{\beta}$ in $\bm{\beta}^\top\bm{\Sigma}\bm{\beta}$, as it is challenging to split $\bm{\beta}^\top\bm{\Sigma}\bm{\beta}$ into i.i.d. copies satisfying Lindeberg's condition. Alternatively, we use the martingale CLT \citep{10.1214/aoms/1177693494} to consider the fluctuation of correlated terms. Nevertheless, our result comes at the cost of the rate of the Berry-Esseen uniform convergence bound, achieving a rate of $m^{-1/5}$ instead of $m^{-1/2}$. 
Existing studies indicate that for the martingale CLT to achieve a better rate of the Berry-Esseen bound, some extra stringent assumptions are required. For example, \cite{10.1214/aop/1065725195} shows that for the martingale filtration sequence $\{\mathfrak{F}_t\}_{{1\leq t \leq m}}$ and martingale difference sequence $\{\xi_t\}_{1\leq t \leq m}$ defined in our proof in the supplementary material (Section~\ref{subsec:proof_quad_form}), if there exists $\rho > 0$ such that the following holds almost surely
$$ \left.\mathbb E\left(\xi_t^2\right|\mathfrak{F}_{t-1}\right) = 1/m,\quad \left.\mathbb E\left(\xi_t^3\right|\mathfrak{F}_{t-1}\right) = 0,\quad \mbox{and} \quad \left.\mathbb E\left(|\xi_t|^{3+\rho}\right|\mathfrak{F}_{t-1}\right)\leq cm^{-(3+\rho)/2},$$ 
then the martingale CLT would have a rate of $m^{-1/2}$ for the Berry-Esseen bound. Similarly, \cite{10.1214/aop/1176993776} suggests that we may obtain a rate of $\log(m)/m^{1/2}$ for the Berry-Esseen bound if we require the conditional variance divided by the general variance in our settings to converge to one almost surely. However, it is worth noting that these conditions may be restrictive and difficult to verify in our PRS applications. Therefore, we adopt the upper bound proposed by \cite{10.1214/aop/1176991901} in our proof without any such restrictions, which coincides with the results in \cite{10.3150/12-BEJ417} and yields an $m^{-1/5}$ rate for the Berry-Esseen bound in our setting.
\end{rmk}

The asymptotic normality of the quadratic form $\bmbeta^\top\bmSigma\bmbeta$ may be of broad interest and is a crucial intermediate result for analyzing cohort-level prediction accuracy, as presented in the next section.

\subsection{Cohort-level uncertainty}
\label{subsubsec:marg_A_iso}
First-order limits of the out-of-sample prediction accuracy measure $A(\hat{\bm\beta})$ have been studied in PRS literature \citep{daetwyler2008accuracy, dudbridge2013power, zhao2022polygenic}. These measures quantify the overall performance of the genetic prediction for the $n_z$ testing subjects and are key metrics often used for assessing and comparing the predictive abilities across different PRS methods. However, the quantification of the sampling variance of $A^2$ is rarely considered in current PRS applications \citep{momin2023significance}. Naive approaches often ignore the variations induced by the uncertainty in the estimation of $\bmbeta$ from the training data, which may lead to invalid statistical inference and misleading conclusions.
Corollary~\ref{cor: CLT for marg A^2 iso} characterizes the asymptotic distribution of $A(\hat{\bm\beta}_{\text{M}})$ when $\bm\Sigma=\mathbb\I_p$.
\begin{coro}
\label{cor: CLT for marg A^2 iso}
    Under Assumptions~\ref{a:Sigmabound}-\ref{a:Sparsity} and $\bm\Sigma=\mathbb\I_p$, as $\min{(n,n_z,p,m)} \to \infty$, let 
    \begin{equation*}
        \begin{aligned}
            \bm G_{\text{M}_0}(t) = \mathbb P\left(\sqrt{\eta_{\text{M}_0} n_z}\left(A (\hat{\bm\beta}_{\text{M}}) - \tilde{A}_{\text{M}_0}\right) < t\right),
        \end{aligned}
    \end{equation*}
    where 
    $$\tilde{A}_{\text{M}_0}=h_{\bm\beta_z}\left(\frac{p}{nh_{\bm\beta}^2}+1\right)^{-1/2}$$
    and 
        \begin{equation*}
              \begin{aligned}
                \eta_{\text{M}_0} =\frac{n h_{\bm\beta}^2 +p}{n_z h_{\bm\beta_z}^2 +
                n h_{\bm\beta}^2 +p + 2(n_z + n)h_{\bm\beta}^2h_{\bm\beta_z}^2 +nn_z\left(
                p^2\mathbb E(\bm\beta^4)/\sigma_{\bm\beta}^4-1\right)h_{\bm\beta}^2h_{\bm\beta_z}^2/m}.
            \end{aligned}
        \end{equation*} 
    Then the following Berry-Esseen bound holds
    \begin{equation*}
            \begin{aligned}
            &\sup_{t\in\mathbb R}  \left|\bm G_{\text{M}}(t) - \Phi(t)\right| \leq C\max\left\{n_z^{-1/2}, n^{-1/2},m^{-2\delta}\right\}
            \end{aligned}
        \end{equation*}
    for $\forall \delta \in (0,1/2)$.
\end{coro}
The mean $\tilde{A}_{\text{M}_0}$ is consistent with the first-order results in \cite{zhao2022polygenic}, which studied the special case $n \asymp p$ and points out that the first-order limit of $A^2 (\hat{\bm\beta}_{\text{M}})$ is smaller than the genetic signal strength $h^2_{\bm\beta_z}$.  
Therefore, the prediction accuracy $A(\hat{\bm\beta}_{\text{M}})$ is a shrinkage estimator of $h_{\bm\beta_z}$ in such high-dimensional sparsity-free prediction, and the gap is determined by the ratio between the training data sample size $n$ and the number of genetic variants $p$, as well as the overall heritability $h_{\bm\beta}^2$. 
Our results significantly extend previous findings by providing the asymptotic distribution of $A (\hat{\bm\beta}_{\text{M}})$ without proportional assumptions, yielding several new insights. 
For example, if $p \gg n$, then $\tilde{A}_{\text{M}_0}$ will converge to zero, suggesting that there will be no prediction power if $n$ is much smaller than $p$. In addition, consistent with previous results, the signal sparsity $m$ is not involved in $\tilde{A}_{\text{M}}$, at a fixed heritability level. However, our CLT results reveal that $m$ is indeed involved in the variance of $A (\hat{\bm\beta}_{\text{M}})$, with a sparser genetic signal leading to a larger variance.
Furthermore, the sparsity $m$ determines the rate of the Berry-Esseen bound, jointly with the two sample sizes $n_z$ and $n$. These results emphasize the critical role of $m$ in the uncertainty of out-of-sample prediction accuracy measures, which is not captured in conventional first-order analyses.

Importantly, we establish explicit strong connections between $\eta_{\text{M}0}$ and the training data, primarily reflected by factors $n$, $m$, $p$, and $h_{\bm\beta}^2$.
Therefore, ignoring the variations originating from the training data, such as by unintentionally conditioning on the genetically predicted values in downstream PRS applications, may lead to invalid inference in the testing data. 
To see it, note that the variance of $A (\hat{\bm\beta}_{\text{M}})$
is given by 
$\Var(A (\hat{\bm\beta}_{\text{M}}))=(n_z \eta_{\text{M}_0})^{-1}.$
As $\min{(n,n_z,p,m)} \to \infty$, Corollary~\ref{cor: CLT for marg A^2 iso} indicates that  
 $$\eta_{\text{M}_0} \asymp \frac{n+p}{n+p+n_z+nn_z/m} \preceq 1.$$

If $n_z$ is substantially smaller than $n$ and $p$, in other words, if we have a relatively large training data sample size, $\eta_{\text{M}}$ will be a constant, and thus the CLT convergence rate of $A (\hat{\bm\beta}_{\text{M}})$ may be close to $n_z^{1/2}$. 
In this case, even ignoring the prediction-induced variation, which naively assumes a  $n_z^{1/2}$ standard CLT convergence rate by working on $\Var(A (\hat{\bm\beta}_{\text{M}})|\hat{\bm\beta}_{\text{M}}) \approx  n_z^{-1}$, will not cause serious problems in the inference of $A (\hat{\bm\beta}_{\text{M}})$. 
However, as $n_z$ becomes larger and is more comparable to $n$ and $p$, the variance of $A (\hat{\bm\beta}_{\text{M}})$ may be substantially underestimated if we still ignore the training data's impact on $\eta_{\text{M}_0}$. 
Furthermore, if $n_z$ is larger than $n$ and/or the genetic signal is extremely sparse with a small $m$, the convergence rate of $A (\hat{\bm\beta}_{\text{M}})$ may be much slower than the typical rate  $n_z^{1/2}$ due to $\eta_{\text{M}_0} \to 0$. 
In such situations, naive inference of the prediction accuracy $A (\hat{\bm\beta}_{\text{M}})$, assuming the  $n_z^{1/2}$ convergence rate, will yield unreliable results. 
Further details on the naive CLT results, when ignoring prediction-induced variation, are provided in Section~\ref{sec.s1} of the supplementary material.
In summary, both the mean and the CLT convergence rate of $A (\hat{\bm\beta}_{\text{M}})$ depend on training data. 
These results provide novel insights into how to accurately assess the predictive ability of PRS and test the statistical significance of the differences among various PRS methods \citep{momin2023significance}.



Now we present the asymptotic distribution of $A(\hat{\bm\beta}_{\text{M}})$ with general $\bm\Sigma$. To efficiently state our results, we first introduce the following notations.
\begin{defn}\label{def:kwgamma}
We define three parameters of interest that are useful for managing the flexible relationships among $n$, $n_z$, $p$, and $m$, given that we do not assume they are proportional to each other in the analysis of marginal estimators.
\begin{align*}
    \kappa_1 = \max
    \left\{\frac{m}{p},\frac{p}{n}\right\}, \quad \kappa_2 = \max\left\{\kappa_1,\frac{n_zm}{np}\right\}, \quad \mbox{and} \quad\kappa_3 = \max\left\{\kappa_2, \frac{n_zm}{p^2}\right\}.
\end{align*}
In addition, for integers $1\leq i\leq 3$ and $1\leq j \leq 3$, we define a sequence of parameters 
\begin{align*}
    \omega_i = \frac{\Tr(\bm\Sigma^i)}{p}\quad \mbox{and} \quad  \gamma_j = \frac{\Tr(\bm\Sigma^j\mathbb\I_m)}{p}.
\end{align*}
\end{defn}

The parameters $\kappa_1$, $\kappa_2$, and $\kappa_3$ do not directly appear in our main theorem below; however, they play a crucial role in controlling the convergence rate in CLT and thus are useful in our interpretations. 
Given that we impose no restrictions on the relative sizes of $n$, $n_z$, $p$, and $m$, these results shed light on their distinct roles in the convergence behaviors of $A (\hat{\bm\beta}_{\text{M}})$.
In addition, the definitions of $\omega_i$ and $\gamma_j$ relate to the moments of the eigenvalues of $\bm{\Sigma}$ and $\bm{\Sigma}\mathbb{I}_m$, respectively, with the latter corresponding to the nonzero genetic effects in $\bm{\beta}$. They are present in the mean and variance of the asymptotic distribution of $A (\hat{\bm\beta}_{\text{M}})$. For simplicity, previous studies typically assume certain degrees of homogeneity across the $\bm{\Sigma}$ \citep{zhao2022block}, which can be summarized by the following additional assumption. 
\begin{assumption}
\label{a:omega_gamma_structure}For integers $1\leq i\leq 3$, 
    we assume that $\gamma_i = m/p \cdot\omega_i$.
\end{assumption}
Intuitively, Assumption~\ref{a:omega_gamma_structure} provides the connection between $\gamma_i$ and $\omega_i$ by
implying a largely homogeneous LD structure $\bm{\Sigma}$ across the genetic effects of the $m$ variants with nonzero entries in $\bmbeta$ and the $p - m$ null variants. 
Below, we first state the general results without this LD structural assumption in Assumption~\ref{a:omega_gamma_structure}, and then show how this assumption may simplify our results.
\begin{thm}
\label{thm: CLT for A^2 Marginal}
Under Assumptions~\ref{a:Sigmabound}-\ref{a:Sparsity}, as $\min{(n,n_z,p,m)} \to \infty$, let
\begin{equation*}
    \begin{aligned}
        \bm G_{\text{M}}(t) = \mathbb P\left(\sqrt{\eta_{\text{M}} n_z}\left(A (\hat{\bm\beta}_{\text{M}}) - \tilde{A}_{\text{M}}\right) < t\right)
    \end{aligned}
\end{equation*}
with 
$$\tilde{A}_{\text{M}}=h_{\bm\beta_z}\gamma_2 \left\{\left(\frac{\gamma_1}{h_{\bm\beta}^2}\frac{p}{n}\omega_2+\gamma_3\right)\gamma_1 \right\}^{-1/2} \quad \mbox{and}  \quad \eta_{\text{M}} = \frac{q_{M_1}}{q_{M_1}+q_{M_2}+q_{M_3}},$$
where 
$$q_{M_1}=\frac{\gamma_1}{h_{\bm\beta_z}^2}\left(\frac{\gamma_1}{h_{\bm\beta}^2}\frac{p}{n}\omega_2+\gamma_3 \right), \quad 
q_{M_2}=\frac{\gamma_1}{h_{\bm\beta}^2}\frac{n_z}{n}\gamma_3+2\gamma_2^2\left(\frac{n_z}{n} + 1\right),$$
and 
$$q_{M_3}=n_z\left[\left(\frac{\mathbb E\left(\bm\beta^4\right)}{\sigma_{\bm\beta}^4} - \frac{3}{p^2}\right)\sum_{i=1}^p(\bm\Sigma^2\mathbb\I_m)_{i,i}^2+ \frac{2}{p^2}\Tr\left\{(\bm\Sigma^2\mathbb\I_m)^2\right\}\right].$$
Then the following Berry-Esseen bound holds
    \begin{equation*}
    \label{e:A2clt}
        \begin{aligned}
        &\sup_{t\in\mathbb R}  \left|\bm G_{\text{M}}(t) - \Phi(t)\right| \leq C\max\left\{n_z^{-1/2}, n^{-1/2}, m^{-1/5}\right\}.
        \end{aligned}
    \end{equation*}
\end{thm}
Compared to Corollary~\ref{cor: CLT for marg A^2 iso}, Theorem~\ref{thm: CLT for A^2 Marginal} provides new insights into the influences of $\bmSigma$ on the CLT results of $A (\hat{\bm\beta}_{\text{M}})$. First, different from the special case $\bm\Sigma=\mathbb{I}_p$, the mean $\tilde{A}_{\text{M}}$ is now related to the genetic signal sparsity $m$ through the terms $\gamma_1$, $\gamma_2$, and $\gamma_3$. Therefore, complex traits with the same overall heritability may have different prediction accuracy due to the different 
genetic architecture (e.g., distribution of these genetic signals) across the genome \citep{timpson2018genetic}. 
Second, the influences of $\bmSigma$ on the variance of $A (\hat{\bm\beta}_{\text{M}})$ are explicitly displayed in $q_{M_1}$, $q_{M_2}$, and $q_{M_3}$. It is worth noting that
$$\eta_M \asymp \frac{\kappa_1}{\kappa_3} \preceq 1.$$
Therefore, the CLT convergence rate for $A (\hat{\bm\beta}_{\text{M}})$ is determined by $\kappa_1$ and $\kappa_3$ and, similar to the special case $\bm\Sigma=\mathbb{I}_p$, it cannot exceed $n_z^{1/2}$. 
There are several interesting observations.
For example, when $n_z$ is very small compared to $n$ and $p$, i.e., $n_z \prec n$ and $n_z \prec p$, we would have a CLT convergence rate of $n_z^{1/2}$. This stands for many PRS applications that work with large training datasets and small testing datasets. In such situations, the uncertainty originating from the training process will not have a large impact on the CLT convergence rate of $A (\hat{\bm\beta}_{\text{M}})$, as the training data is large enough compared to the testing data.
However, as $n_z$ increases with applying PRS on more testing data coming from biobanks \citep{zhou2022global}, we would have $n_z \succ p$, and the CLT convergence rate would be $p^{1/2}\kappa_1^{1/2}\max\{m/n, m/p\}^{-1/2} \preceq p^{1/2} \prec n_z^{1/2}$, because $\kappa_1 \succeq \max\{m/n, m/p\}$. Therefore, the CLT convergence rate cannot achieve $n_z^{1/2}$ on a super large testing dataset due to the relatively small training dataset. In such cases, the influences of prediction uncertainty cannot be ignored.
Furthermore, when $p \succ m$, $p \succ n_z$, and $n \succ n_z$, if $p \succeq mn_z$, then we have $n_z^{1/2}\eta^{1/2} \prec 1$. This means that when the genetic signals are extremely sparse, we will only have a trivial CLT convergence rate of order one, implying that we may not have useful asymptotic normality. In summary, the flexibility of the connections among $n$, $n_z$, $p$, and $m$ allows us to analyze a wide range of current and future PRS applications.



Additionally, the rate of the Berry-Esseen bound is determined by the two sample sizes $n_z$ and $n$, as well as the sparsity $m$. As discussed in Remark~\ref{remark1}, the rate $m^{-1/5}$ is due to the martingale CLT used in our proof for handling the asymptotic behavior of the quadratic form with a general $\bm\Sigma$. In the special case $\bm\Sigma=\mathbb{I}_p$, the rate becomes $m^{-1/2}$ in Corollary~\ref{cor: CLT for marg A^2 iso}, as obtaining the asymptotic distribution of the quadratic form does not require the martingale CLT when $\bm\Sigma=\mathbb{I}_p$. More generally, for special cases of $\bm\Sigma$ such as when off-diagonal terms are non-vanishing \citep{gotze2002asymptotic} or dominate the diagonal terms \citep{10.1214/aos/1031833676}, we may obtain a faster rate, as we may not necessarily approach asymptotic normality through the martingale CLT. However, these conditions may be restrictive to model the real genetic data and challenging to verify in practice. Therefore, we develop our results separately for the general case $\bm\Sigma$ and the special case $\bm\Sigma=\mathbb{I}_p$ to provide more comprehensive insights.


Next, with the structural linear relationship in Assumption~\ref{a:omega_gamma_structure}, a corollary from Theorem~\ref{thm: CLT for A^2 Marginal} is as follows. 
\begin{coro}
\label{cor: CLT for A^2 Marginal_omega_gamma_structure}
Under Assumptions~\ref{a:Sigmabound}-\ref{a:Sparsity}, as $\min{(n,n_z,p,m)} \to \infty$, if we further have Assumption~\ref{a:omega_gamma_structure}, we have the following updated values for $\tilde{A}_{\text{M}}$, $q_{M_1}$, and $q_{M_2}$
\begin{align*}
     &\tilde{A}_{\text{M}}=h_{\bm\beta_z}\omega_2 \left\{\left(\frac{\omega_1}{h_{\bm\beta}^2}\frac{p}{n}\omega_2+\omega_3\right)\omega_1 \right\}^{-1/2},\quad q_{M_1}=\frac{\omega_1}{h_{\bm\beta_z}^2}\frac{m^2}{p^2} \left(\frac{\omega_1}{h_{\bm\beta}^2}\frac{p}{n}\omega_2+\omega_3 \right),\\
     &\mbox{and} \quad
    q_{M_2}=
    \frac{n_z}{n}\frac{m^2}{p^2}\omega_1\left(\omega_1\frac{1-h_{\bm\beta}^2}{h_{\bm\beta}^2}+\omega_3\right)+\frac{2m^2}{p^2}\omega_2^2\left(\frac{n_z}{n} + 1\right).
\end{align*}
\end{coro}
An interesting observation in Corollary~\ref{cor: CLT for A^2 Marginal_omega_gamma_structure} is that, under Assumption~\ref{a:omega_gamma_structure}, $\tilde{A}_{\text{M}}$ is no longer related to the sparsity $m$ at a fixed level of heritability. This is due to the uniformity in the LD pattern between the genetic effects of the $m$ causal variants and the $p-m$ null ones. Nevertheless, similar to the special case $\bm\Sigma=\mathbb{I}_p$, sparsity $m$ still plays an important role in the variance and CLT convergence rate of $A (\hat{\bm\beta}_{\text{M}})$. 


\section{Reference panel-based ridge estimator} \label{sec:ref}

In this section, we establish second-order results for $\hat{\bm\beta}_{\text{W}}(\lambda)$, which is related to many popular reference panel-based ridge-type estimators in PRS applications \citep{ma2021genetic, ge2019polygenic, vilhjalmsson2015modeling}. Compared to $\hat{\bm\beta}_{\text{M}}(\lambda)$ studied in Section~\ref{sec3}, $\hat{\bm\beta}_{\text{W}}(\lambda)$ aims to account for the LD patterns by estimating $\bmSigma$ from an external reference panel dataset $\W$ \citep{10002015global}. 
Therefore, our challenge is to understand the role of the ridge-type resolvent $(\W^\top \W + n_w \lambda \mathbb{I}_p)^{-1}$ in second-order fluctuations, which cannot be addressed using traditional CLT techniques. To address this, we focus on the proportional regime \citep{bai2010spectral, 10.1214/21-AOS2133, ledoit2009eigenvectors, 10.1214/17-AOS1549} and quantify the limiting behaviors of resolvent-related quantities using anisotropic local laws \citep{anisotropic_local_law}. 

We present the individual-level and cohort-level results in Sections~\ref{subsubsec:ref_new} and~\ref{subsubsec:ref_A}, respectively. Similar to Section~\ref{sec3}, we first present the results for the special case $\bm\Sigma = \mathbb{I}_p$ and then for the general $\bm\Sigma$. We begin by introducing some additional notations and assumptions.

\begin{assumption}\label{a:anisotropic regularity}
We consider a high-dimensional proportional regime where $n \asymp n_z \asymp n_w \asymp p \asymp m$. Additional, we assume $n/n_w \to \phi_{d} \in \mathbbR$ and $p/n_w\rightarrow \phi_w\in \mathbbR$ as $p\rightarrow \infty$. 
\end{assumption}

Following the random matrix literature, we denote the eigenvalues of $\bm\Sigma$ by $\lambda_1 \geq \lambda_2 \geq \cdots \geq \lambda_p \geq 0$. Let $\pi\coloneqq p^{-1}\sum_{i=1}^p\delta_{\lambda_i}$ denote the empirical spectral density  of $\bm\Sigma$, which can be uniquely characterized by its Stieltjes transformation \citep{bai2010spectral,Yao_Zheng_Bai_2015}. 
We denote $m_q(\lambda)$ as the Stieltjes transform of the asymptotic eigenvalue density of $\Q_0\bm\Sigma \Q_0^\top$ at any $\lambda \in \mathbb{C} \setminus \mathbb{R}^{+}$. Here, $\Q_0$ is a generic matrix with sample size $n_q$ and can be freely replaced by $\W_0$, $\Z_0$, or $\X_0$ with their corresponding sample sizes, depending on the dataset we aim to study. It is well known that $m_q(\lambda)$ is explicitly determined by a fixed-point equation \citep{bai2010spectral}
\begin{equation}
\label{equ:stiej_eigen}
    \begin{aligned}
        \frac{1}{m_q(\lambda)} = -\lambda + \phi_q\int \frac{x}{1+m_q(\lambda)x}\pi(dx), \quad\text{where} \quad \phi_q \coloneqq \lim_{p\to\infty} p/n_q.
    \end{aligned}
\end{equation}
As the Stieltjes transform contains all the information of a distribution, it is extremely useful in random matrix studies. For $\forall \lambda > 0$ and the reference panel $\W$, we define $\mathfrak{m}_w \coloneqq m_w(-\lambda)$ to be the Stieltjes transform of the reference panel. To conduct the perturbation analysis of the spectrum of $\bm\Sigma$, we need to work with the tilting factor derived from $\mathfrak{m}_w$, which is given by
    \begin{equation*}
        \begin{aligned}
            \mathfrak{r}_w = \mathfrak{m}_w^2/\mathfrak{m}_w',
        \end{aligned}
    \end{equation*}
      where $\mathfrak{m}_w' = d(m_w(-\lambda))/d(-\lambda)$ and $d(\cdot)$ denotes the derivative operation. 
      For the special case $\bm\Sigma=\mathbb{I}_p$, the closed-form expressions of $\mathfrak{m}_w$ and $\mathfrak{r}_w$ defined in \cref{equ:stiej_eigen} are 
\begin{equation*}
    \begin{aligned}
    &\mathfrak{m}_w(-\lambda, \bm\Sigma = \mathbb\I_p) = \frac{\sqrt{(\lambda + \phi_w - 1)^2 + 4\lambda} - (\lambda +\phi_w - 1)}{2\lambda} \quad\mbox{and}\\ 
    &\mathfrak{r}_w(-\lambda, \bm\Sigma = \mathbb\I_p) = 1-4\phi_w\left\{\sqrt{(\lambda + \phi_w - 1)^2 + 4\lambda} + (\lambda + \phi_w + 1)\right\}^{-2}.
    \end{aligned}
\end{equation*}
Detailed calculations are provided in Section~\ref{subsec:add_con_spec_ref} of the supplementary material. For general $\bm\Sigma$, the closed-form expressions typically do not exist, and $\mathfrak{m}_w$ and $\mathfrak{r}_w$ are implicitly determined by \cref{equ:stiej_eigen}.

\subsection{Individual-level uncertainty}
\label{subsubsec:ref_new}
In this section, we quantify the individual-level predictive uncertainty of the reference panel-based ridge estimator. Compared to the results of the marginal estimator in Section~\ref{subsec:marg_new_pred}, this analysis allows us to understand the influences of the ridge-type resolvent estimated by a reference panel.
Corollary~\ref{cor:CLT_ref_new_iso} below present the CLT results for $\z^\top\hat{\bm\beta}_{\text{W}}(\lambda)$ in the special case $\bm\Sigma=\mathbb{I}_p$.
\begin{coro}
\label{cor:CLT_ref_new_iso}
Under Assumptions~\ref{a:Sigmabound}-\ref{a:Sparsity}, \ref{a:anisotropic regularity} and $\bmSigma=\mathbb\I_p$, as $\min(n,n_w,m,p) \to \infty$, conditioning on the testing data point $\z$ and genetic effect $\bm\beta$, let
\begin{equation*}
    \begin{aligned}
        \bm F_{\text{W}_0}(t)=\mathbb P\left(\sigma_{\text{W}_0}^{-1}\sqrt{n}\left\{\z^\top\hat{\bm\beta}_{\text{W}}(\lambda) - \frac{\phi_d}{\lambda(1+\mathfrak{m}_w)}\z^\top\bm\beta\right\} < t\right),
    \end{aligned}
\end{equation*}
where $$\sigma_{\text{W}_0}^2 = \frac{\phi_d^2}{\lambda^2(1+\mathfrak{m}_w)^2}\left\{2(\z^\top\bm\beta)^2 + \frac{\|\z\|_2^2\|\bm\beta\|_2^2}{\mathfrak{r}_w h_{\bm\beta}^2}\right\}.$$
Then with probability of at least $1-O_p(p^{-D})$ for some large $D\in\mathbb R$ over the randomness of $\W_0$,
the following Berry-Esseen bound holds
\begin{equation*}
    \begin{aligned}
        &\sup_{t\in\mathbb R}\left|\bm F_{\text{W}_0}(t) - \Phi(t)\right| \leq C{n}^{-1/2}.
    \end{aligned}
\end{equation*}
\end{coro}
Corollary~\ref{cor:CLT_ref_new_iso} shows that the mean and variance of the genetically predicted value $\z^\top\hat{\bm\beta}_{\text{W}}(\lambda)$ are determined by Stieltjes
transform of the asymptotic eigenvalue density through $\mathfrak{m}_w$ and $\mathfrak{r}_w$. 
Particularly, the mean of $\z^\top\hat{\bm\beta}_{\text{W}}(\lambda)$ will be generally different from the underlying true value $\z^\top\bm\beta$, and the shrinkage term is $\phi_d/\{\lambda(1+\mathfrak{m}_w)\}$. This shrinkage term is also involved in the variance. In addition to $\mathfrak{m}_w$, the variance is determined by $\mathfrak{r}_w$. 
These asymptotic normality results still hold for general $\bm\Sigma$, which are presented in Theorem~\ref{thm: CLT for reference new}. 

\begin{thm}
\label{thm: CLT for reference new}
Under Assumptions~\ref{a:Sigmabound}-\ref{a:Sparsity}, \ref{a:anisotropic regularity}, as $\min(n,n_w, m,p) \to \infty$, conditioning on the testing data point $\z$ and genetic effect $\bm\beta$, let
\begin{equation*}
    \begin{aligned}
    \bm F_{\text{W}}(t)= \mathbb P\left(\sigma_{\text{W}}^{-1}\sqrt{n}\left\{\z^\top\hat{\bm\beta}_{\text{W}}(\lambda) - \frac{\phi_d}{\lambda} \z^\top(\mathbb \I_p + {\mathfrak{m}_w}\bm\Sigma)^{-1}\bm\Sigma\bm\beta \right\} < t\right),
    \end{aligned}
\end{equation*}
where $$\sigma_{\text{W}}^2 = \frac{\phi_d^2}{\lambda^2}\left[2\left\{\z^\top(\mathbb \I_p + {\mathfrak{m}_w}\bm\Sigma)^{-1}\bm\Sigma\bm\beta\right\}^2 + \frac{\z^\top(\mathbb \I_p + {\mathfrak{m}_w}\bm \Sigma)^{-2}\bm\Sigma \z\|\bmSigma^{1/2}\bmbeta\|_2^2}{\mathfrak{r}_wh_{\bm\beta}^2}\right].$$
Then with probability of at least $1-O_p(p^{-D})$ for some large $D\in\mathbb R$ over the randomness of $\W_0$, the following Berry-Esseen bound holds
\begin{equation*}
    \begin{aligned}
        &\sup_{t\in\mathbb R}\left|\bm F_{\text{W}}(t) - \Phi(t)\right| \leq C{n}^{-1/2}.
    \end{aligned}
\end{equation*}
\end{thm}

The asymptotic distribution of $\z^\top\hat{\bm\beta}_{\text{W}}(\lambda)$ involves $\bm\Sigma$ mainly through $(\mathbb \I_p+\mathfrak{m}_w\bm\Sigma)^{-1}\bm\Sigma$, $(\mathbb \I_p+\mathfrak{m}_w\bm\Sigma)^{-2}\bm\Sigma$, and $\mathfrak{r}_w$, all of which are functions of the Stieltjes transform $\mathfrak{m}_w$. 
Here, the mean of $\z^\top\hat{\bm\beta}_{\text{W}}(\lambda)$ is determined solely by $(\mathbb{I}_p + \mathfrak{m}_w\bm\Sigma)^{-1}\bm\Sigma$. In contrast, the variance of $\z^\top\hat{\bm\beta}_{\text{W}}(\lambda)$ is influenced by all three factors, suggesting the complex nature of the second-order behavior. 
Particularly, the tilting factor of the reference panel $\mathfrak{r}_w = \mathfrak{m}_w^2/\mathfrak{m}_w'$ appears frequently when the quadratic form of the resolvent is involved, and we will see it again when analyzing the cohort-level uncertainty in the next section. Intuitively, $\mathfrak{r}_w$ rises to accommodate the necessary perturbation towards $\bm\Sigma$ when applying the anisotropic local law \citep{anisotropic_local_law}, which can be reflected by the fact that $d(m_w^{-1}(-\lambda))/d(-\lambda) = -\mathfrak{m}_w'/\mathfrak{m}_w^2$. It is worth noting that $\mathfrak{r}_w$ is positive and always {smaller} than or equal to one. We provide further details in Section~\ref{subsec:add_con_spec_ref}.

\subsection{Cohort-level uncertainty}
\label{subsubsec:ref_A}
First-order limits of the out-of-sample prediction accuracy measure $A(\hat{\bm\beta}_{\text{W}}(\lambda))$ have been quantified in a previous study \citep{zhao2022block}. In this section, we provide its asymptotic distribution. We begin with the special case $\bm\Sigma = \mathbb{I}_p$ to provide intuitive results.
\begin{coro}
\label{cor: CLT for ref A^2 iso}
   Under Assumptions~\ref{a:Sigmabound}-\ref{a:Sparsity}, \ref{a:anisotropic regularity} and $\bm\Sigma = \mathbb\I_p$, as $\min(n,n_,n_w,m,p) \to \infty$, let 
\begin{equation*}
    \begin{aligned}
        \bm G_{\text{W}_0}(t) = \mathbb P\left(\sqrt{\eta_{\text{W}_0}n_z}\left(A(\hat{\bm\beta}_{\text{W}}(\lambda)) - \tilde{A}_{\text{W}_0}\right) < t\right),
    \end{aligned}
\end{equation*}
where 
$$\tilde{A}_{\text{W}_0}=\sqrt{\mathfrak{r}_w}h_{\bm\beta_z}\left(\frac{p}{nh_{\bm\beta}^2}+1\right)^{-1/2}$$
and
\begin{equation*}
    \begin{aligned}
        &\eta_{\text{W}_0} = \frac{nh_{\bm\beta}^2 + p}
        {n_zh_{\bm\beta_z}^2+nh_{\bm\beta}^2 + p +2(n+n_z)h_{\bm\beta}^2h_{\bm\beta_z}^2\mathfrak{r}_w +nn_z\left\{\left(p^2\mathbb E\left(\bm\beta^4\right)/\sigma_{\bm\beta}^4-3\right)\mathfrak{r}_w + 2\right\}h_{\bm\beta}^2h_{\bm\beta_z}^2/m }.
    \end{aligned}
\end{equation*}
Then with probability of at least $1-O_p(p^{-D})$ for some large $D\in\mathbb R$ over the randomness of $\W_0$, the following Berry-Esseen bound holds
\begin{equation*}
    \begin{aligned}
        \sup_{t\in\mathbb R}\left|\bm G_{\text{W}_0}(t) - \Phi(t)\right| \leq C{n^{-1/5}}.
    \end{aligned}
\end{equation*}
\end{coro}

Under Assumption~\ref{a:anisotropic regularity}, we have {$0 < \eta_{\text{W}_0} < 1$}, indicating that the CLT has a $n_z^{1/2}$ standard pointwise convergence rate.
It is clear that $\mathfrak{r}_w$ plays an important role in the asymptotic distribution of $A(\hat{\bm\beta}_{\text{W}}(\lambda))$. 
Interestingly, the mean and variance of $A(\hat{\bm\beta}_{\text{W}}(\lambda))$ will be the same as those of $A(\hat{\bm\beta}_{\text{M}})$ given in Corollary~\ref{subsubsec:marg_A_iso} when $\mathfrak{r}_w = 1$. This makes sense, as $\hat{\bm\beta}_{\text{M}}$ can be viewed as a special case of $\hat{\bm\beta}_{\text{W}} (\lambda)$ when $\lambda \to \infty$, and we also have $\mathfrak{r}_w \to 1$ as $\lambda \to \infty$. 
Furthermore, since $0<\mathfrak{r}_w \leq 1$ and is monotone increasing with respect to $\lambda$, the typical use of a reference panel 
(i.e., $\lambda$ does not go to $\infty$) will indeed lead to smaller values of both the mean of prediction accuracy and its variance when $\bm\Sigma=\mathbb{I}_p$. This observation partially matches previous findings on the first-order mean \citep{zhao2022block}, suggesting that using a reference panel may worsen the cohort-level prediction performance if the features are largely independent. Our CLT results quantify that the shrinkage is proportional to $\mathfrak{r}_w^{1/2}$, and the variance will become smaller simultaneously.
{
In addition, the rate of the Berry-Esseen bound is $n^{-1/5}$. This is attributed to the presence of the quadratic form in $A(\hat{\bm\beta}_{\text{W}}(\lambda))$, with the martingale CLT resulting in a rate of $n^{-1/5}$.}

In order to efficiently present the results of $A(\hat{\bm\beta}_{\text{W}}(\lambda))$ with general $\bm\Sigma$, we introduce the following definitions.  
\begin{defn}\label{def:pixi}
For any $1\leq i\leq 2$, we introduce two sequence parameters
    \begin{align*}
    \pi_i \coloneqq \frac{\Tr\left\{(\mathbb\I_p+\mathfrak{m}_w\bm\Sigma)^{-i}\right\}}{p} \quad \mbox{and} \quad \xi_i \coloneqq \frac{\Tr\left\{(\mathbb\I_p + \mathfrak{m}_w\bm\Sigma)^{-i}\mathbb\I_m\right\}}{p}.
\end{align*}
Then we define 
    \begin{equation*}
        \begin{aligned}
             \varrho_0 = \frac{\mathfrak{m}_w\gamma_1 + \xi_1 - m/p}{\mathfrak{m}_w^2},\quad \varrho_1 = \frac{1-2\pi_1 + \pi_2}{\mathfrak{m}_w^2},\quad \mbox{and} \quad \varrho_2 = \frac{3\xi_1 -\xi_2 + \mathfrak{m}_w\gamma_1 -2m/p}{\mathfrak{m}_w^3}.
        \end{aligned}
    \end{equation*}
     In addition, for $1\leq i \leq p$, we denote the $i_{th}$ { eigenvector} of $\bm\Sigma$ by $u_i$, which satisfies $\bm\Sigma u_i = {\sigma_iu_i}$. 
\end{defn}


Similar to $\omega_i$ and $\gamma_j$ defined in Definition~\ref{def:kwgamma}, $\pi_i$ and $\xi_i$ relate to the moments of the eigenvalues of $\bm{\Sigma}$ and $\bm{\Sigma}\mathbb{I}_m$, respectively. They are present in the asymptotic distribution of $A (\hat{\bm\beta}_{\text{M}})$ through $\varrho_0$, $\varrho_1$, and $\varrho_2$. We explicitly define the eigenvectors as they are involved in the variance of $A (\hat{\bm\beta}_{\text{M}})$. 
Theorem~\ref{thm: CLT for reference A^2} summarizes the CLT results of $A(\hat{\bm\beta}_{\text{W}}(\lambda))$. 
\begin{thm}
\label{thm: CLT for reference A^2} Under Assumptions~\ref{a:Sigmabound}-\ref{a:Sparsity},\ref{a:anisotropic regularity}, as $\min(n, n_z, n_w, p, m) \to \infty$, let 
\begin{equation*}
    \begin{aligned}
        \bm G_{\text{W}}(t) = \mathbb P\left(\sqrt{\eta_{\text{W}}n_z}\left(A(\hat{\bm\beta}_{\text{W}}(\lambda)) - \tilde{A}_{\text{W}}\right) < t\right),
    \end{aligned}
\end{equation*}
with
$$\tilde{A}_{\text{W}}=\sqrt{\mathfrak{r}_w}\varrho_0 h_{\bm\beta_z}\left(\frac{p}{nh_{\bm\beta}^2}\gamma_1^2\varrho_1 +\gamma_1\varrho_2  \right)^{-1/2} \quad \mbox{and} \quad 
\eta_{\text{W}} = \frac{q_{W_1}}{q_{W_1}+q_{W_2}+q_{W_3}},$$
where 
$$q_{W_1}=\frac{1}{\mathfrak{r}_w}\left(p\varrho_1\frac{\gamma_1}{h_{\beta}^2}+n\varrho_2\right)\frac{\gamma_1}{h_{\beta_z}^2}, \quad
q_{W_2}=\frac{1}{\mathfrak{r}_w}n_z\frac{\gamma_1}{h_{\beta}^2}\varrho_2+2(n+n_z)\varrho_0^2,$$
and
$$q_{W_3}=nn_z\Bigg[\left(\frac{\mathbb E(\bm\beta^4)}{\sigma_{\bm\beta}^4} - \frac{3}{p^2}\right)\sum_{i=1}^p\left(\mathfrak{M}\mathbb\I_m\right)_{i,i}^2 + \frac{2}{p^2}\left\{\Tr\left(\left(\mathfrak{M}\mathbb\I_m\right)^2\right) + \mathfrak{n}(\bm\Sigma,m)\Tr\left(\left(\mathfrak{M}\right)^2\right)\right\}\Bigg].$$
Here $\mathfrak{M} = \left(\mathbb\I_p+\mathfrak{m}_w\bm\Sigma\right)^{-1}\bm\Sigma^2$ and 
\begin{equation*}
    \begin{aligned}
        &\mathfrak{n}(\bm\Sigma, m) = \frac{1}{p}\sum_{i=1}^p\frac{\phi_w\mathfrak{m}_w'\pi_i^3}{(1+\pi_i\mathfrak{m}_w)^2}\langle u_i, \mathbb\I_m u_i\rangle. 
    \end{aligned}
\end{equation*}
Then with probability of at least $1-O_p(p^{-D})$ for some large $D\in\mathbb R$ over the randomness of $\W_0$, the following Berry-Esseen bound holds
\begin{equation*}
    \begin{aligned}
        \sup_{t\in\mathbb R}\left|\bm G_{\text{W}}(t) - \Phi(t)\right| \leq C {n^{-1/5}}.
    \end{aligned}
\end{equation*}
\end{thm}
Similar to Theorem~\ref{thm: CLT for A^2 Marginal}, the sparsity $m$ affects the mean of $A(\hat{\bm\beta}_{\text{W}}(\lambda))$ for a general covariance matrix $\bmSigma$, which contrasts with the special case where $\bmSigma = \mathbb{I}_p$. 
More importantly, our findings reveal that the eigenvectors of $\bmSigma$ influence the variance of $A(\hat{\bm\beta}_{\text{W}}(\lambda))$ through $\mathfrak{n}(\bm\Sigma, m)$.  
{Specifically, the locations of the $m$ non-zero genetic effects among the $p$ genetic variants will determine $\mathbb\I_m u_i$, which in turn influences the value of $\langle u_i, \mathbb\I_m u_i \rangle$.} 
To the best of our knowledge, 
this is one of the first results to demonstrate how eigenvectors play a role in the second-order behavior of ridge-type estimators, unlike the previous first-order results that involve only eigenvalues.
Similar to Theorem~\ref{thm: CLT for A^2 Marginal}, the Berry-Esseen upper bound is of the order $n^{-1/5}$ due to the use of the martingale CLT.
\begin{figure}[!t] 
\includegraphics[page=1,width=1\linewidth]{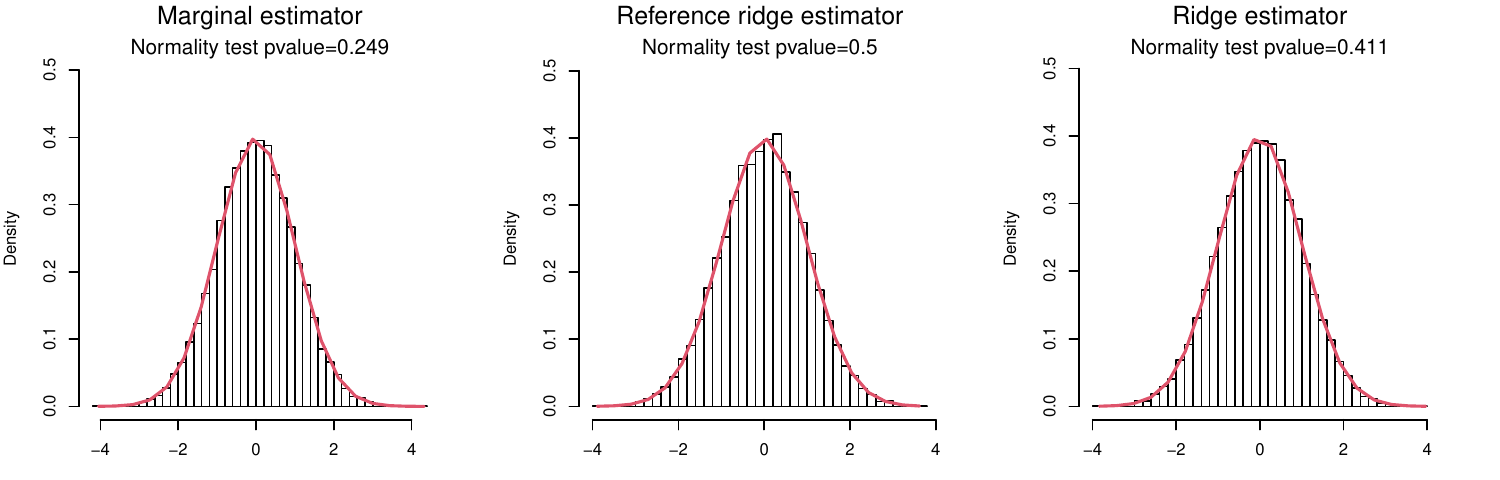}
\centering
\caption{
\textbf{Asymptotic normality of individual-level genetically predicted values.}
Based on the real genetic data from the UK Biobank study, we illustrate the empirical distribution of genetically predicted values $\z^\top\hat{\bm\beta}$ for $\hat{\bm\beta}_{\text{M}}$, $\hat{\bm\beta}_{\text{W}}(\lambda)$, and $\hat{\bm\beta}_{\text{R}}(\lambda)$ (from left to right). 
We assess the asymptotic normality with the Shapiro-Wilk test \citep{shapiro1965analysis}. 
Here we set $p~=~461,488$, heritability $h_{\bm\beta}^2=h^2_{{\bm\beta}_z}=0.3$, sparsity $m/p$ ranging from $0.001$ to $0.5$, and $n~=~50,000$.
}
\label{main_fig_1}
\end{figure}

\subsection{Traditional ridge estimator}\label{sec4.3}
Ridge regression has widespread applications in various high-dimensional production problems without sparsity constraints. In genetics, it has natural connections with the genomic best linear unbiased prediction \citep{lee2008predicting}, which is a fundamental method in genetic prediction with access to individual-level data $(\X, \y)$. 
While the first-order results of its predictive risk have been established for ridge regression with general $\bmSigma$ \citep{10.1214/17-AOS1549,10.1214/21-AOS2133}, the asymptotic distributions remain largely unexplored in the random matrix literature. It has been suggested that CLT results may not be obtainable for ridge regression with general $\bm\Sigma$ and generic distribution $\mathcal{F}$ \citep{li2021asymptotic}. 
Specifically, for the traditional ridge estimator $\hat{\bm\beta}_{\text{R}}(\lambda)= (\X^\top \X+n\lambda\mathbb\I_p)^{-1}\X^\top \y$, the dependency of the resolvent $(\X^\top \X + n\lambda\mathbb\I_p)^{-1}$ on the training data $(\X,\y)$ makes it difficult to completely isolate such weak correlations under a general data distribution $\mathcal{F}$. 
Recently \cite{10.1214/22-AOS2243} developed the asymptotic normal approximation of de-biased estimators under Gaussian assumptions on the data $\X$. In Section~\ref{sec5} of the supplementary material, we adopt these de-biased estimators techniques to present second-order results for $\hat{\bm\beta}_{\text{R}}(\lambda)$ under our model setups, assuming $\mathcal{F}$ is a Gaussian distribution. 
{ 
The phenomena and insights for $\hat{\bm\beta}_{\text{R}}(\lambda)$ in the Gaussian case align with those observed for $\hat{\bm\beta}_{\text{W}}(\lambda)$ under the general distribution $\mathcal{F}$. Briefly, ignoring the prediction-related variation from the training dataset will lead to underestimating individual-level and cohort-level uncertainty. For example, the sparsity of the underlying genetic signals and the LD pattern influence the CLT convergence rate and variance of the estimators. 
These quantities can be expressed using the Stieltjes transform of the eigenvalue distribution but differ in form from those of $\hat{\bm\beta}_{\text{W}}(\lambda)$. 
}

\section{Numerical results}
\label{numerical_sec}
In this section, we perform numerical analyses based on the genotype data from the UK Biobank study \citep{bycroft2018uk}. Specifically, we use data from unrelated white British subjects and apply the following quality controls on the genetic variants: removing SNPs with a minor allele frequency of less than $0.01$, a genotyping rate of less than $90\%$, and a Hardy-Weinberg test $P$-value of less than $1 \times 10^{-7}$. Additionally, we only include subjects with less than $10\%$ missing genotypes. After these procedures, $461,488$ variants remain. 
We numerically evaluate the asymptotic normality for individual-level genetically predicted values and cohort-level prediction accuracy. Genetic variants with non-zero effects are randomly selected, and the genetic effects are independently sampled from $N(0,1/p)$ using GCTA \citep{yang2011gcta}. We consider three heritability levels: $h_{\bm\beta}^2=h^2_{{\bm\beta}_z}=0.1$, $0.3$, and $0.8$. For each level, we consider four levels of sparsity $m/p$, ranging from $0.001$, $0.01$, $0.1$, to $0.5$. We analyze the three estimators: $\hat{\bm\beta}_{\text{M}}$, $\hat{\bm\beta}_{\text{W}}(\lambda)$, and $\hat{\bm\beta}_{\text{R}}(\lambda)$. For the marginal estimator $\hat{\bm\beta}_{\text{M}}$, we perform LD-based pruning (with an LD threshold of $r^2=0.1$ and a window size of $500$ kb) using PLINK \citep{purcell2007plink} to remove one of two genetic variants that are in LD. For the reference panel ridge estimator $\hat{\bm\beta}_{\text{W}}(\lambda)$, we use the LD matrix estimated from European subjects in the 1000 Genomes reference panel \citep{10002015global}. For computational feasibility, the whole genome was divided into $1,703$ approximately independent blocks according to cutoffs provided in \cite{berisa2016approximately}. For the ridge estimator $\hat{\bm\beta}_{\text{R}}(\lambda)$, we use the BLUP estimation implemented in GCTA \citep{yang2011gcta}. The marginal genetic effects used in $\hat{\bm\beta}_{\text{M}}$ and $\hat{\bm\beta}_{\text{W}}(\lambda)$ are estimated from fastGWA \citep{jiang2019resource}, while $\hat{\bm\beta}_{\text{R}}(\lambda)$ requires access to the individual-level training data. We randomly select $50,000$ individuals as training samples to perform the simulation. We evaluate the cohort-level prediction accuracy in $500$ randomly selected subjects that are not in the training dataset. The individual-level genetically predicted values are assessed in $50$ of these testing subjects. Each setup is replicated $100$ times.

\begin{figure}[!t] 
\includegraphics[page=1,width=1\linewidth]{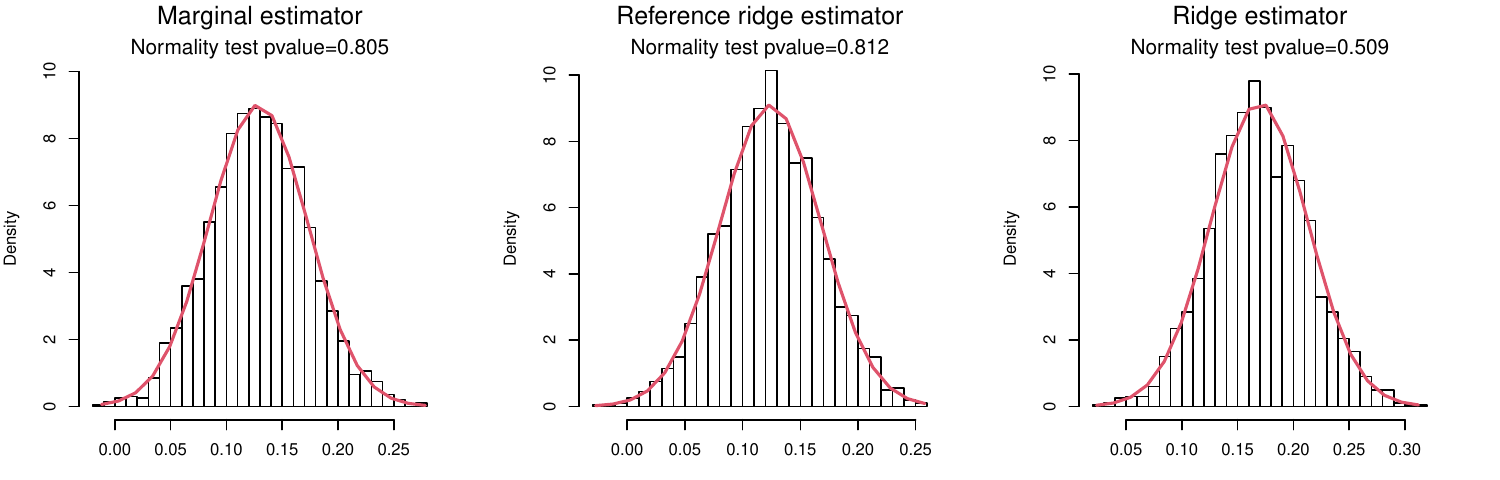}
\centering
\caption{
\textbf{Asymptotic normality of cohort-level prediction accuracy.}
Based on the real genetic data from the UK Biobank study, we illustrate the empirical distribution of cohort-level prediction accuracy $A (\hat{\bm\beta})$ for $\hat{\bm\beta}_{\text{M}}$, $\hat{\bm\beta}_{\text{W}}(\lambda)$, and $\hat{\bm\beta}_{\text{R}}(\lambda)$ (from left to right). 
We assess the asymptotic normality with the Shapiro-Wilk test \citep{shapiro1965analysis}. 
Here we set $p~=~$461,488, heritability $h_{\bm\beta}^2=h^2_{{\bm\beta}_z}=0.3$, sparsity $m/p$ ranging from $0.001$ to $0.5$, $n~=~50,000$, and $n_w~=~500$. 
}
\label{main_fig_2}
\end{figure}

Figure~\ref{main_fig_1} and Supplementary Figure~\ref{main_fig_s1} display the results for individual-level genetically predicted values. As the results across different sparsity levels all show asymptotic normality, we group them together in the figures to save space. These numerical results suggest that the genetically predicted values are normally distributed. In practice, it is possible to develop normality-based confidence intervals after estimating the empirical standard errors with data-driven approaches \citep{ding2022large, wang2024impact}. Additionally, the results of cohort-level prediction accuracy are shown in Figure~\ref{main_fig_2} and Supplementary Figure~\ref{main_fig_s2}. The empirical average of prediction accuracy increases as heritability rises, and all are below the corresponding heritability level. These results suggest that although prediction accuracy underestimates heritability, it exhibits asymptotic normality around its point estimates. By properly accounting for the variances generated through the prediction, we can perform inference on these prediction performance measures \citep{momin2023significance}.

We further evaluate the pattern of variance changes across different model setups. The patterns for cohort-level prediction accuracy have been shown in previous studies \citep{zhao2022polygenic, wang2020theoretical}, so here we focus on the individual-level predicted values. We use a Bayesian-based approach built upon the popular LDpred2 method \citep{ding2022large} to estimate the standard error for each individual-level genetically predicted value. We study the trend across different training data sample sizes and heritability levels. Specifically, for a fixed heritability level $h_{\bm\beta}^2=h^2_{{\bm\beta}_z}=0.25$ and sparsity level $m/p=0.1$, we evaluate the standard error of individual-level predicted values with training data sample sizes of $50,000$, $100,000$, $150,000$, $200,000$, and $250,000$. Similarly, for a fixed training data sample size of $250,000$ and sparsity level $m/p=0.1$, we evaluate across heritability levels at $h_{\bm\beta}^2=h^2_{{\bm\beta}_z}=0.05$, $0.1$, $0.25$, $0.5$, and $0.8$. In each setup, we repeat the analysis for $100$ independent subjects.
The results are summarized in Figure~\ref{main_fig_3}. Consistent with our theoretical findings, the standard error of genetically predicted values is strongly influenced by the training data sample size, with smaller training samples leading to greater uncertainty in the predicted values. Similar patterns are observed across different heritability levels. These results suggest there may be substantial uncertainty when predicting traits with low heritability and limited available training data, and this uncertainty needs to be considered in post-prediction applications. In summary, these numerical results with real genetic data support the asymptotic normality results in our theorems. 


\begin{figure}[!t] 
\includegraphics[page=1,width=0.9\linewidth]{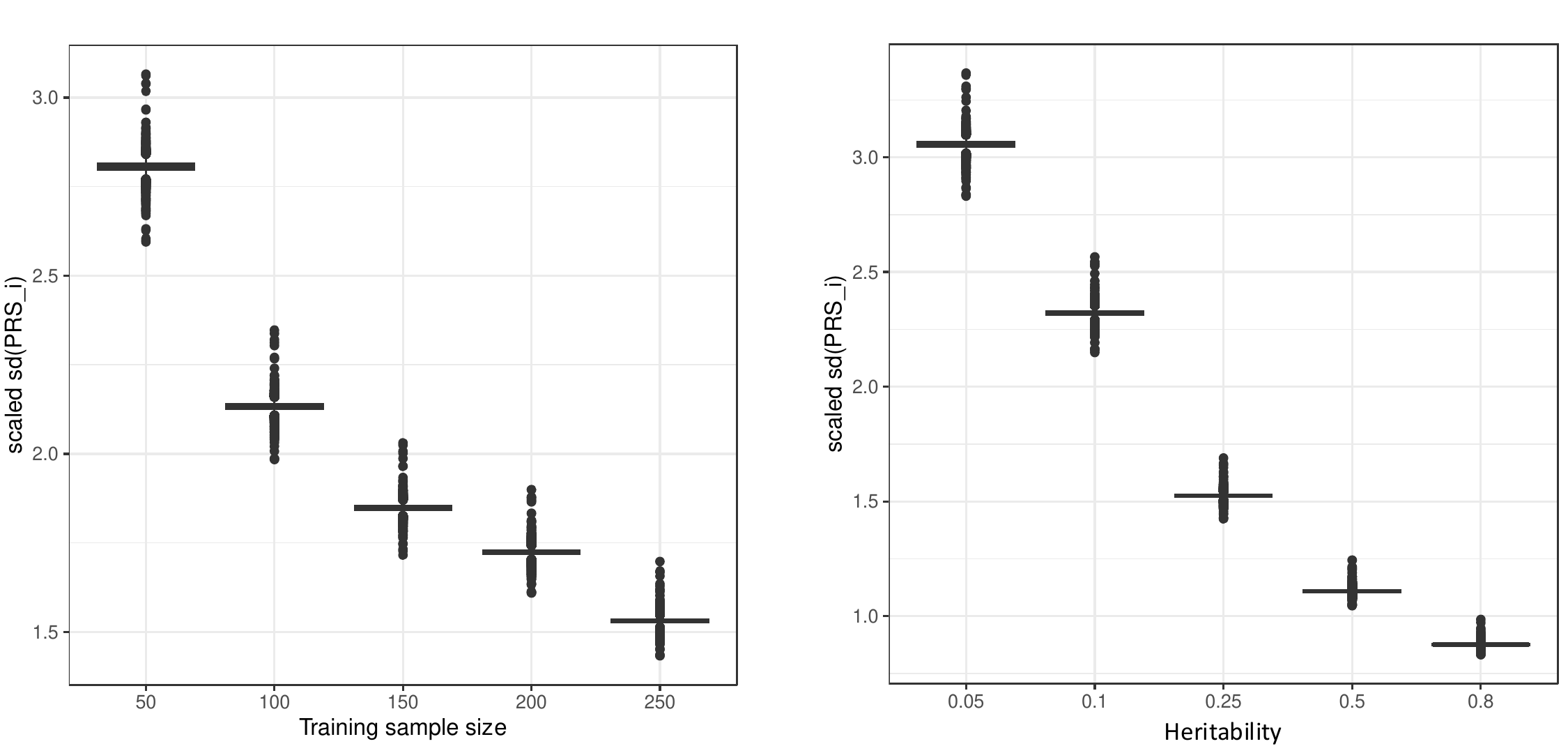}
\centering
\caption{
\textbf{Simulation results estimating the standard error of individual-level genetically predicted values across the training data sample size (left)
and heritability (right).} 
Based on the real genetic data from the UK Biobank study, we estimate the standard error of $\z^\top\hat{\bm\beta}$ with the LDpred2-based method proposed in \cite{ding2022large}. 
Here we set $p~=~$461,488 and sparsity $m/p=0.1$. In the left panel, the training sample size ranges from $50,000$ to $250,000$ with $h_{\bm\beta}^2=h^2_{{\bm\beta}_z}=0.25$.  In the right panel, the heritability $h_{\bm\beta}^2=h^2_{{\bm\beta}_z}$ ranges from $0.05$ to $0.8$ with $n=250,000$.  
}
\label{main_fig_3}
\end{figure}

\section{Discussion}\label{sec:discussion}
In this paper, we study CLTs for PRS applications and provide one of the first second-order results for high-dimensional predictions with general $\bmSigma$ and Gaussian-free data distribution in the random matrix literature. To establish these results, we develop a novel leave-one-out framework to prove quantitative CLTs \citep{ash2000probability} with multi-source randomness from training and testing datasets, provide a martingale CLT \citep{10.1214/aoms/1177693494,10.1214/aop/1176991901} framework for proving CLTs of quadratic forms, and quantify the role of ridge-type resolvents using anisotropic local laws \citep{anisotropic_local_law}. Our analysis offers insights into the non-asymptotic predictive uncertainty of genetically predicted values and the prediction accuracy measured by out-of-sample $R$-squared.
To model complex LD patterns, varying training and testing data resources, and a wide range of phenotypes with different genetic architectures in real data applications, we try to allow flexible assumptions and aim to develop second-order results for general model setups. 
For example, for the marginal estimator $\hat{\bm\beta}_{\text{M}}$, we construct its quantitative CLTs without proportional assumptions among $p$, $m$, $n$, and $n_z$, providing deep insights into their roles in determining the CLT convergence rate. For the reference panel-based ridge estimator $\hat{\bm\beta}_{\text{W}}(\lambda)$, we consider a general $\bmSigma$ without assuming the data to be Gaussian-distributed, imposing only moment assumptions. These results link the Stieltjes transform, the tilting factor, and eigenvectors to the second-order behavior of $\hat{\bm\beta}_{\text{W}}(\lambda)$, illustrating how the ridge-type resolvents contribute to predictive uncertainty. Overall, our thorough analysis of the marginal and reference panel-based ridge estimators provides complementary insights into the roles of different factors in the second-order behaviors of PRS.

Our analysis of individual and cohort-level uncertainty provides comprehensive results in PRS applications and highlights the challenges in high-dimensional sparsity-free prediction. 
First, it is difficult to obtain an unbiased estimator of $\bm\beta$ in such cases, especially under a general $\bm\Sigma$, indicating that the genetically predicted values $\z^\top\hat{\bm\beta}$ may not reflect the true phenotype. This bias is typically ignored in current PRS applications. As discussed in Section~\ref{subsec:marg_new_pred}, this bias may be of interest depending on the specific tasks in some applications. More importantly, the cohort-level prediction accuracy $A^2(\hat{\bm\beta})$ is a shrinkage estimator of the heritability $h^2_{\bm\beta_z}$ towards zero, and thus its variance incorporates both the shrinkage of the mean and the prediction uncertainty. This leads to the complicated second-order behavior of $A^2(\hat{\bm\beta})$. Naively assuming an $n_z^{1/2}$ CLT convergence rate, as is often done in many current PRS applications—such as when reporting whether the prediction accuracy is ``significant" (i.e., there is genetic prediction power) or whether one PRS estimator is better than another—may lead to misleading results. 
The analytical variance derived in our theorems can be approximated by resampling-based methods, which are popular in genetic literature \citep{bulik2015ld,ding2022large, wang2024impact}.
Therefore, our general theoretical results will enable the development of uncertainty measurements and valid statistical inference methods in many future PRS applications. 
For example, based on the established asymptotic normality of $\z^\top\hat{\bm\beta}$, we can obtain empirical confidence intervals in practice once we have an estimate of the variance. Similarly, the asymptotic normality of $A^2(\hat{\bm\beta})$ suggests that $A^2(\hat{\bm\beta})$ follows a normal distribution, but with a larger variance—controlled by various factors in both the training and testing data—than typically assumed when predictive uncertainty is ignored. This larger variance is due to the significant uncertainty inherent in high-dimensional sparsity-free prediction.
In addition, by using our established framework,  we may evaluate other popular PRS methods in future studies, such as reference panel-based $L_1$ regularized estimators \citep{mak2017polygenic}, which may require different technical tools. 

\bibliographystyle{plain}
\bibliography{ridge}
\begin{acks}[Acknowledgments]
{\bxzz The authors would like to thank Edgar Dobriban, Hong Hu, Xiaochen Yang, and Buxin Su for the helpful discussions and feedback.
This research has been conducted using the UK Biobank resources (application number $76139$), subject to a data transfer agreement.
We thank the individuals represented in the UK Biobank for their participation and the research teams for their work in collecting, processing, and disseminating these datasets for analysis. 
We would like to thank the research computing groups at the Wharton School of the University of Pennsylvania and the Rosen Center for Advanced Computing at the Purdue University for providing computational resources and support that have contributed to these research results.} 
\end{acks}

\appendix 
\renewcommand{\theequation}{S.\arabic{equation}}
\renewcommand{\thetable}{S.\arabic{table}}
\renewcommand{\thefigure}{S.\arabic{figure}}
\renewcommand{\thesection}{S.\arabic{section}}
\renewcommand{\thelemma}{S.\arabic{lemma}}
\renewcommand{\thetheorem}{S.\arabic{theorem}}
\vspace{30pt}
\noindent{\bf \LARGE Appendix}
\vspace{20pt}

Due to space constraints, proofs and additional results are deferred to the Appendix.

\section{Results ignoring prediction-induced variance}\label{sec.s1}
In this section, we present the naive CLT results for out-of-sample $R$-squared $A^2 (\hat{\bm\beta})$, assuming the variations originating from the training data are ignored. In other words, we treat the estimator $\hat{\bm\beta}$ (and consequently $\hat{\bm \y} = \Z \hat{\bm\beta}$) as deterministic. This reflects what is typically obtained in practice when the randomness inherent in the estimator is disregarded.
Here the estimator $\hat{\bm\beta}$ is generic and can represent $\hat{\bm\beta}_{\text{M}}$, $\hat{\bm\beta}_{\text{W}}(\lambda)$, or $\hat{\bm\beta}_{\text{R}}(\lambda)$ as defined in Sections~\ref{subsec:estimators} or~\ref{sec4.3} of the main text. Technically, since $\hat{\bm\beta}$ incorporates information from $\X$ and $\y$ (and also $\W$ for $\hat{\bm\beta}_{\text{W}}(\lambda)$), we will treat them as deterministic. Additionally, we will fix the randomness from the testing data $\Z$ as $\hat{\bm\y}$ is fixed. Under these conditions, we have the following lemma.
\begin{lem}
\label{lemma: CLT for A^2 fixed est}
    Consider fixed $\X_0 \in \mathbb R^{n\times p}$, $\Z_0\in \mathbb R^{n_z\times p}$, $\epsilon \in \mathbb R^{n\times 1}$, and $\bm\beta \in \mathbb R^{p\times 1}$, as well as random $\epsilon_z \in \mathbb R^{n_z\times 1}$, where entries are i.i.d random variables following a distribution with mean $0$ and variance $\sigma_{\epsilon_z}^2$. 
    Moreover, we assume that the fourth moment of entries of $\epsilon_z$ is bounded, denoted as $\mathbb E(\epsilon_z^4)$. Let
    \begin{equation*}
        \begin{aligned}
            \bm G_{\text{F}}(t) = \mathbb P\left(\sqrt{\eta_{\text{F}}n_z}\left(A(\hat{\bm\beta}) - \tilde{A}_{\text{F}}\right) < t\right)
        \end{aligned}
    \end{equation*}
    with $$\tilde{A}_{\text{F}} = \frac{\bm\beta^\top\Z^\top\hat{\y}}{\sqrt{{\|\Z\bm\beta\|_2^2 + n_z(1/h_{\bm\beta_z}^2 - 1)\|\bm\Sigma^{1/2}\bm\beta\|_2^2}} \cdot\|\hat{\y}\|_2}\quad \text{and} \quad \eta_{\text{F}} = \frac{\|\Z\bm\beta\|_2^2/n_{z}}{\|\bm\Sigma^{1/2}\bm\beta\|_2^2}\frac{h_{\bm\beta_z}^2}{1-h_{\bm\beta_z}^2} + 1.$$
    Then the following Berry-Esseen bound holds
    \begin{equation*}
        \begin{aligned}
            \sup_{t\in\mathbb R}\left|\bm G_{\text{F}}(t) - \Phi(t)\right| \leq C{n_z}^{-1/2}.
        \end{aligned}
    \end{equation*}
\end{lem}
\begin{proof}
Following similar steps used to prove Theorem~\ref{thm: CLT for A^2 Marginal}, we first decompose the numerator
    \begin{equation*}
        \begin{aligned}
            \y_z^\top\hat{\y} = \bm\beta^\top\Z^\top\hat{\y} + \epsilon_z^\top\hat{\y}.
        \end{aligned}
    \end{equation*}
    Note that the first term is deterministic in our setting, we therefore only need to derive the CLT for the second quantity. By Lemma ~S\ref{lemma:non-asmptotic CLT 3}, we immediately have
    \begin{equation*}
        \sup_{t\in\mathbb R}\left|\mathbb P\left(\frac{\epsilon_z^\top\hat{\y}}{\sigma_{\epsilon_z}\|\hat{\y}\|_2} < t \right) - \Phi(t)\right| \leq c\sqrt{\frac{\tau\mathbb E(\epsilon_{z_i}^4)\sum_{i=1}^{n_z}\hat{\y}_{i}^4}{\sigma_{\epsilon_z}^4\|\hat{\y}\|_2^4}} = O_p(n_z^{-1/2}).
    \end{equation*}
    It follows that 
    \begin{equation*}
        \begin{aligned}
            \sup_{t\in\mathbb R}\left|\mathbb P\left(\frac{\y_z^\top\hat{\y} - \bm\beta^\top\Z^\top\hat{\y}}{\sigma_{\epsilon_z}\|\hat{\y}\|_2} < t \right) - \Phi(t)\right| \leq O_p(n_z^{-1/2}).
        \end{aligned}
    \end{equation*}
    By Chebyshev's inequality, for $\delta \in (0,1/2)$, we have 
    \begin{equation*}
        \begin{aligned}
            \mathbb P\left(\left|\|\y_z\|_2^2 - (\|\Z\bm\beta\|_2^2 + n_z\sigma_{\epsilon_z}^2)\right| < n_z^{1-\delta}\right)  \geq 1- O_p(n_z^{-1+2\delta}).
        \end{aligned}
    \end{equation*}
    This concludes the limit for the denominator. With some rearrangements of the terms, we have 
    \begin{equation*}
        \begin{aligned}
            \sup_{t\in\mathbb R}\left|\mathbb P\left(\sqrt{n_z}\frac{\sqrt{\|\Z\bm\beta\|_2^2/n_z + \sigma_{\epsilon_z}^2}}{\sigma_{\epsilon_z}}\frac{\y_z^\top\hat{\y} - \bm\beta^\top\Z^\top\hat{\y}}{\|\hat{\y}\|_2\|\y_z\|_2} < t \right) - \Phi(t)\right| \leq O_p(n_z^{-1/2}).
        \end{aligned}
    \end{equation*}
    Replacing ${\sigma_{\epsilon_z}^2}$ with $\|\bm\Sigma^{1/2}\bm\beta\|_2^2(1/h_{\bm\beta_z}^2-1)$, we obtain the desired naive CLT.
\end{proof}
Note that we have $\eta_{\text{F}}\ge 1$. Therefore, an immediate result following Lemma~S\ref{lemma: CLT for A^2 fixed est} is that if we treat $\hat{\bm\beta}$ as fixed, we will underestimate the true variance of our estimators, leading to overconfidence in the accuracy of the confidence intervals and hypothesis testing. This conclusion holds for all of $\hat{\bm\beta}_{\text{M}}$, $\hat{\bm\beta}_{\text{W}}(\lambda)$, and $\hat{\bm\beta}_{\text{R}}(\lambda)$. 
The underestimation of variance arises from omitting the fluctuation within the estimator, specifically the fluctuations of $\bm\beta, \X, \bm\epsilon$ for $\hat{\bm\beta}_{\text{M}}$, fluctuations of $\bm\beta, \X, \bm\epsilon, \W$ for $\hat{\bm\beta}_{\text{W}}(\lambda)$, and fluctuations of $\bm\beta, \X, \bm\epsilon$ for $\hat{\bm\beta}_{\text{R}}(\lambda)$.


\section{Results of traditional ridge estimator}\label{sec5}
We start by introducing the additional Gaussian assumptions and notations. 
\begin{ass.s}
\label{a:Gaussian_entries}
    Adopt Assumptions~\ref{a:Sigmabound}-\ref{a:Sparsity}, we further assume $\X_0\in \mathbb R^{n\times p}$ with i.i.d entries follow from $\mathcal{N}(0,1)$ independent from $\boldsymbol\epsilon \sim \mathcal{N}(0,\sigma_{\epsilon}^2\mathbb\I_n)$. {Similar conditions hold for $\Z_0$ and $\sigma_{\epsilon_z}^2$.} 
    Additionally, we assume $n \asymp n_z \asymp p \asymp m$, and $p/n\rightarrow \phi_n\in \mathbb R$ as $p\rightarrow \infty$.
\end{ass.s}
For $\forall \lambda > 0$, we introduce the following definitions of the Stieltjes transform for the training data, including 
\begin{equation*}
    \begin{aligned}
        \mathfrak{m}_n \coloneqq m_n(-\lambda)  \quad\mbox{and} \quad \mathfrak{r}_n  \coloneqq \mathfrak{m}_n^2/\mathfrak{m}_n',
    \end{aligned}
\end{equation*}
where $\mathfrak{m}_n' = d(m_n(-\lambda))/d\lambda$.
For the special case $\bm\Sigma=\mathbb{I}_p$, we further have  
\begin{equation*}
    \begin{aligned}
    &\mathfrak{m}_n(-\lambda, \bm\Sigma = \mathbb\I_p) = \frac{\sqrt{(\lambda + \phi_n - 1)^2 + 4\lambda} - (\lambda +\phi_n - 1)}{2\lambda}
    \end{aligned}
\end{equation*}
and 
\begin{equation*}
    \begin{aligned}
       \mathfrak{r}_n(-\lambda, \bm\Sigma = \mathbb\I_p) = 1-4\phi_n\left\{\frac{1}{\sqrt{(\lambda + \phi_n - 1)^2 + 4\lambda} + (\lambda + \phi_n + 1)}\right\}^2. 
    \end{aligned}
\end{equation*}

Moreover, if we consider the optimal ridge parameter $\lambda^* = \phi_n(1-h_{\bm\beta}^2)/h_{\bm\beta}^2$, we can rewrite $\mathfrak{m}_n$ and $\mathfrak{r}_n$ into some parameter depends only on $h_{\bm\beta}^2$ and $\phi_n$:
$$\mathfrak{m}_n(-\lambda, \bm\Sigma = \mathbb\I_p) = \frac{2h_{\bm\beta}^2}{\sqrt{(\phi_n + h_{\bm\beta}^2)^2 - 4h_{\bm\beta}^4\phi_n} + \phi_n - h_{\bm\beta}^2}$$
and $$\mathfrak{r}_n(-\lambda, \bm\Sigma = \mathbb\I_p) = 1-4\phi_n\left\{\frac{h_{\bm\beta}^2}{\sqrt{(\phi_n + h_{\bm\beta}^2)^2 - 4h_{\bm\beta}^4\phi_n} + \phi_n+h_{\bm\beta}^2}\right\}^2.$$
\subsection{Individual-level uncertainty}
\label{subsubsec:ridge_new}
Corollary~S\ref{cor:CLT_ridge_new_iso} below presents the CLT results for $\z^\top\hat{\bm\beta}_{\text{R}}(\lambda)$ in the special case $\bm\Sigma=\mathbb{I}_p$.

\begin{cor.s}
\label{cor:CLT_ridge_new_iso}
Under Assumptions~\ref{a:Sigmabound}-\ref{a:Sparsity}, S\ref{a:Gaussian_entries} and $\bm\Sigma = \mathbb \I_p$, as $\min(n,p,m) \to \infty$, conditioning on the testing data point $\z$ and genetic effect $\bm\beta$, let
\begin{equation*}
    \begin{aligned}
        \bm F_{\text{R}_0}(t) = \mathbb P\left(\sigma_{\text{R}_0}^{-1}\sqrt{n}\left(\z^\top \hat{\bm\beta}_{\text{R}}(\lambda) -\mu_{\text{R}_0}\right) < t\right),
    \end{aligned}
\end{equation*}
where
\begin{equation*}
    \begin{aligned}
        \mu_{\text{R}_0}=\frac{\mathfrak{m}_n}{1+\mathfrak{m}_n}\z^\top\bm\beta
    \end{aligned}
\end{equation*}
and 
\begin{equation*}
    \begin{aligned}
        \sigma_{\text{R}_0}^2 =& \frac{\phi_n}{\lambda\mathfrak{r}_n(1+\mathfrak{m}_n)^2}\left\{\left(\frac{\lambda}{\phi_n} - \frac{1-h_{\bm\beta}^2}{h_{\bm\beta}^2}\right)\|\z\|_2^2\|\bm\beta\|_2^2 - \frac{1}{(1+\mathfrak{m}_n)^2}(\z^\top\bm\beta)^2\right\} \\
        &+ \frac{1}{\lambda\mathfrak{m}_n}\left\{\frac{1-h_{\bm\beta}^2}{h_{\bm\beta}^2}\|\z\|_2^2\|\bm\beta\|_2^2 + \frac{\phi_n}{\lambda\mathfrak{m}_n(1+\mathfrak{m}_n)^3}(\z^\top\bm\beta)^2\right\}.
    \end{aligned}
\end{equation*}
{Then we have}
    \begin{equation*}
        \begin{aligned}
            \sup_{t\in\mathbb R}\left|\bm F_{\text{R}}(t) - \Phi(t)\right| \to 0.
        \end{aligned}
    \end{equation*}


Moreover, under the optimal ridge parameter $\lambda^*=\phi_n (1-h_{\bm\beta}^2)/h_{\bm\beta}^2$ \citep{zhao2022block}, the mean and variance become
\begin{equation*}
    \begin{aligned}
        &\mu_{\text{R}_0} = \frac{2h_{\bm\beta}^2}{\sqrt{(\phi_n + h_{\bm\beta}^2)^2 - 4h_{\bm\beta}^4\phi_n} + \phi_n + h_{\bm\beta}^2} \quad \text{and}\\
        &\sigma_{\text{R}_0}^2 = \frac{1}{\mathfrak{m}_n\phi_n}\|\z\|_2^2\|\bm\beta\|_2^2 + \frac{h_{\bm\beta}^2}{(1+\mathfrak{m}_n)^3(1-h_{\bm\beta}^2)}\left\{\frac{1}{\phi_n\mathfrak{m}_n^2}\frac{h_{\bm\beta}^2}{1-h_{\bm\beta}^2} - \frac{1}{\mathfrak{r}_n(1+\mathfrak{m}_n)}\right\}(\z^\top\bm\beta)^2. 
    \end{aligned}
\end{equation*}
\end{cor.s}
Corollary~S\ref{cor:CLT_ridge_new_iso} quantifies the mean and variance of $\z^\top \hat{\bm\beta}_{\text{R}}(\lambda)$. The influence of the resolvent $(\X^\top \X + n\lambda\mathbb\I_p)^{-1}$ is explicitly characterized by $\mathfrak{m}_n$ and $\mathfrak{r}_n$, both of which have closed-form expressions when $\bm\Sigma = \mathbb{I}_p$. The bias of $\z^\top \hat{\bm\beta}_{\text{R}}(\lambda)$ is linearly quantified as $(1+\mathfrak{m}_n)^{-1}\z^\top\bm\beta$, which suggests that genetically predicted values are shrunk towards zero by the degrees controlled by $\mathfrak{m}_n$. 
It is worth mentioning that our analysis does not exclude the case when $\lambda \to 0$. Therefore, our results may also be applied to ridge-less estimators \citep{10.1214/21-AOS2133} to understand model overfitting behaviors.
In addition, because we use techniques for de-biased estimators with an additional Gaussian assumption, we will not have an explicit Berry-Esseen type upper bound. 
{
Briefly, in the de-biased estimator framework \citep{10.1214/22-AOS2243}, the Wasserstein distance between the de-biased estimator and a standard Gaussian random variable is upper bounded by a quantity determined by the gradient of the function related to the training data. This quantity is of order $o_p(1)$, and therefore, no explicit convergence rate is provided.} 
Next, we provide the results for general $\bm\Sigma$.  

\begin{definition}\label{def:elleth}
For any $1\leq i\leq 2$, we introduce two sequence parameters 
\begin{align*}
     &\ell_i \coloneqq \frac{\Tr\left\{(\mathbb\I_p+\mathfrak{m}_n\bm\Sigma)^{-i}\right\}}{p} \quad \mbox{and}\quad \eth_i \coloneqq   \frac{\Tr\left\{(\mathbb\I_p + \mathfrak{m}_n\bm\Sigma)^{-i}\mathbb\I_m\right\}}{p},
\end{align*}
and two quantities derived from $\ell_1$ and $\ell_2$
\begin{align*}
    \mathfrak{g} = 1-(1-\ell_1)\phi_n \quad \mbox{and} \quad \mathfrak{h} = \phi_n\left\{\ell_1 - \frac{\lambda\mathfrak{m}_n'}{\mathfrak{m}_n}\left(\ell_1 - \ell_2\right)\right\}.
\end{align*}
\end{definition}

\begin{thm.s}
\label{thm: CLT for ridge new}
Under Assumption~\ref{a:Sigmabound}-\ref{a:Sparsity}, S\ref{a:Gaussian_entries}, as $\min(n,p,m) \to \infty$, conditioning on the testing data point $\z$ and genetic effect $\bm\beta$, let
 \begin{equation*}
        \begin{aligned}
            \bm F_{\text{R}}(t) = \mathbb P\left(\sigma_{\text{R}}^{-1}\sqrt{n}\left\{\z^\top\hat{\bm\beta}_{\text{R}}(\lambda) - \mu_{\text{R}}\right\} < t\right),
        \end{aligned}
    \end{equation*} 
where 
    \begin{equation*}
        \begin{aligned}
           \mu_{\text{R}} = \z^\top\left\{\mathbb\I_p -\frac{\lambda}{\mathfrak{g}}\bm\Sigma^{-1}\left(\mathbb\I_p - (\mathbb\I_p + \mathfrak{m}_n\bm\Sigma)^{-1}\right)\right\}\bm\beta
        \end{aligned}
    \end{equation*}
    and 
    \begin{equation*}
        \begin{aligned}
            \sigma_{\text{R}}^2 &= \mathfrak{h}\mathfrak{s}_1^2/\mathfrak{g}^2+\left[\mathfrak{g}\left(\frac{1-h_{\bm\beta}^2}{h_{\bm\beta}^2}{ \|\bmSigma^{1/2}\bmbeta\|_2^2}\right) + \frac{\lambda\mathfrak{m}_n'}{\mathfrak{m}_n}\left\{\lambda\mathfrak{m}_n\mathfrak{s}_2 -\left(\frac{1-h_{\bm\beta}^2}{h_{\bm\beta}^2}\right)\phi_n\left(\ell_1 - \ell_2\right){\|\bmSigma^{1/2}\bmbeta\|_2^2}\right\}\right]\\
            &\cdot\|\bm\Sigma^{-1/2}\z\|_2^2
            /\mathfrak{g}^2.
        \end{aligned}
    \end{equation*}
Here 
    \begin{equation*}
        \begin{aligned}
            \mathfrak{s}_1 = \z^\top(\mathbb\I_p + \mathfrak{m}_n\bm\Sigma)^{-1}\bm\beta \quad \mbox{and} \quad \mathfrak{s}_2 = \bm\beta^\top(\mathbb\I_p + \mathfrak{m}_n\bm\Sigma)^{-2}\bm\Sigma\bm\beta.
        \end{aligned}
    \end{equation*}
    {Then we have}
    \begin{equation*}
        \begin{aligned}
            \sup_{t\in\mathbb R}\left|\bm F_{\text{R}}(t) - \Phi(t)\right| \to 0.
        \end{aligned}
    \end{equation*}
Moreover, under the optimal ridge parameter $\lambda^*=\phi_n (1-h_{\bm\beta}^2)/h_{\bm\beta}^2$, {the mean and variance become}
\begin{equation*}
    \begin{aligned}
         \mu_{\text{R}} &= \z^\top\left\{ \mathbb\I_p -\frac{\phi_n(1-h_{\bm\beta}^2)}{h_{\bm\beta}^2\mathfrak{g}}\bm\Sigma^{-1}\left(\mathbb\I_p - (\mathbb\I_p + \mathfrak{m}_n\bm\Sigma)^{-1}\right)\right\}\bm\beta \quad \text{and}\\
        \sigma_{\text{R}}^2 &= \mathfrak{h}\mathfrak{s}_1^2/\mathfrak{g}^2+\frac{\phi_n(1-h_{\bm\beta}^2)}{h_{\bm\beta}^2}\left[\frac{\mathfrak{g}}{\phi_n} \|\bmSigma^{1/2}\bmbeta\|_2^2 + \frac{\mathfrak{m}_n'}{\mathfrak{m}_n}\frac{\phi_n(1-h_{\bm\beta}^2)}{h_{\bm\beta}^2}\left\{\mathfrak{m}_n\mathfrak{s}_2 -\left(\ell_1 - \ell_2\right)\|\bmSigma^{1/2}\bmbeta\|_2^2\right\}\right]\\
        &\cdot\|\bm\Sigma^{-1/2}\z\|_2^2/\mathfrak{g}^2.
    \end{aligned}
\end{equation*}
\end{thm.s}


\subsection{Cohort-level uncertainty}
\label{subsubsec:ridge_A}
In the special case $\bm\Sigma=\mathbb{I}_p$, we have the following results for prediction accuracy of ridge estimator.  
\begin{cor.s}
\label{cor: CLT for ridge A^2 iso}
Under Assumption~\ref{a:Sigmabound}-\ref{a:Sparsity}, S\ref{a:Gaussian_entries} and $\bm\Sigma = \mathbb\I_p$, as $\min(n, n_z, p, m)\to \infty$, let 
\begin{equation*}
    \begin{aligned}
        \bm G_{\text{R}_0}(t) = \mathbb P\left(\sqrt{\eta_{\text{R}_0} n_z}\left(A(\hat{\bm\beta}_{\text{R}_0}(\lambda)) - \tilde{A}_{\text{R}_0}\right) < t\right),
    \end{aligned}
\end{equation*}
with
\begin{equation*}
    \begin{aligned}
        &\tilde{A}_{\text{R}_0} = \mathfrak{m}_n h_{\bm\beta_z}\left\{\mathfrak{m}_n^2 + \frac{1-\mathfrak{r}_n}{\mathfrak{r}_n} + \left(\mathfrak{m}_n + \frac{\mathfrak{r}_n-1}{\mathfrak{r}_n}\right)\frac{\phi_n(1-h_{\bm\beta}^2)}{\lambda h_{\bm\beta}^2}\right\}^{-1/2}\quad \mbox{and} \quad\eta_{\text{R}_0} = \frac{\varsigma_1}{\varsigma_1 + \varsigma_2 + \varsigma_3+\varsigma_4},\\
    \end{aligned}
\end{equation*}
where 
\begin{align*}
        &\varsigma_1 = \frac{1}{\lambda^2h_{\bm\beta_z}^2}\left\{\left(\mathfrak{m}_n^2 + \frac{1-\mathfrak{r}_n}{\mathfrak{r}_n}\right) + \left(\mathfrak{m}_n + \frac{\mathfrak{r}_n-1}{\mathfrak{r}_n}\right)\frac{\phi_n(1-h_{\bm\beta}^2)}{h_{\bm\beta}^2}\lambda^{-1}\right\}, 
        \quad \varsigma_2 = 2\mathfrak{m}_n^2,\\
        &\varsigma_3 = \frac{n_z}{n}\left[\frac{\phi_n}{\lambda^2(1+\mathfrak{m}_n)}\left\{\frac{1}{\mathfrak{m}_n^2} - \frac{\lambda}{(1+\mathfrak{m}_n)\mathfrak{r}_n}\right\} + \frac{(1-h_{\bm\beta}^2)(1+\mathfrak{m}_n)^2}{\lambda\mathfrak{m}_n} + \frac{h_{\bm\beta}^2}{\mathfrak{r}_n} - \frac{\phi_n(1-h_{\bm\beta}^2)}{\lambda\mathfrak{r}_n}\right],\quad \mbox{and} \\
        &\varsigma_4 = \frac{n_z\mathfrak{m}_n^2}{m}\left\{\frac{\mathbb E\left(\bm\beta^4\right)p^2}{\sigma_{\bm\beta}^4} - 1\right\}.
\end{align*}
Then we have 
    \begin{equation*}
        \begin{aligned}
            \sup_{t\in \mathbb R}\left|\bm G_{\text{R}_0}(t) - \Phi(t)\right| \to 0.
        \end{aligned}
    \end{equation*}
    Moreover, under the optimal ridge parameter $\lambda^*=\phi_n (1-h_{\bm\beta}^2)/h_{\bm\beta}^2$ \citep{zhao2022block}, we can update the results as follows 
    \begin{equation*}
        \begin{aligned}
            \tilde{A}_{\text{R}_0} = \sqrt{2}h_{\bm\beta_z}h_{\bm\beta}\left\{\left\{(\phi_n + h_{\bm\beta}^2)^2 - 4h_{\bm\beta}^4\phi_n\right\}^{1/2} + \phi_n + h_{\bm\beta}^2\right\}^{-1/2},
        \end{aligned}
    \end{equation*}
    \begin{equation*}
    \begin{aligned}
        &\varsigma_1 = \frac{h_{\bm\beta}^4}{(1-h_{\bm\beta}^2)^2h_{\bm\beta_z}^2\phi_n^2\mathfrak{m}_n(1+\mathfrak{m}_n)}, 
        \quad \varsigma_2 = 2\mathfrak{m}_n^2,\\
        &\varsigma_3 = \frac{n_z}{n}\left[\frac{h_{\bm\beta}^4}{(1+\mathfrak{m}_n)\phi_n(1-h_{\bm\beta}^2)^2}\left\{\frac{1}{\mathfrak{m}_n^2} - \frac{\phi_n(1-h_{\bm\beta}^2)}{(1+\mathfrak{m}_n)\mathfrak{r}_nh_{\bm\beta}^2}\right\} + \frac{h_{\bm\beta}^2(1+\mathfrak{m}_n)^2}{\phi_n\mathfrak{m}_n}\right],\quad \mbox{and} \\
        &\varsigma_4 = \frac{n_z\mathfrak{m}_n^2}{m}\left\{\frac{\mathbb E\left(\bm\beta^4\right)p^2}{\sigma_{\bm\beta}^4} - 1\right\}.
    \end{aligned}
\end{equation*}
\end{cor.s}
Next, we state our CLT results for $A(\hat{\bm\beta}_{\text{R}}(\lambda))$ with general $\bm\Sigma$.
\begin{thm.s} 
\label{thm: CLT for ridge A^2}
Under Assumptions~\ref{a:Sigmabound}-\ref{a:Sparsity}, S\ref{a:Gaussian_entries}, as $\min(n, n_z, p, m) \to \infty$, let 
\begin{equation*}
    \begin{aligned}
        \bm G_{\text{R}}(t) = \mathbb P\left(\sqrt{\eta_{\text{R}} n_z}\left(A(\hat{\bm\beta}_{\text{R}}(\lambda)) - \tilde{A}_{\text{R}}\right) < t\right),
    \end{aligned}
\end{equation*}
with 
\begin{align*}
    &\tilde{A}_{\text{R}} = \tau_0\tau_1^{-1/2}\quad \mbox{and} \quad \eta_{\text{R}} = \frac{\tau_1}{\tau_1 + 2\tau_2^2 + \tau_3+\tau_4^2},
\end{align*}
where 
\begin{equation*}
    \begin{aligned}
        &\tau_0 = \sigma_{\bm\beta}^2\left\{\gamma_1 - \frac{\lambda}{\mathfrak{g}}\left(\frac{m}{p} - \eth_1\right)\right\}, \quad \tau_1 = \zeta_1\zeta_2, \quad \tau_2 = \sigma_{\bm\beta}^2\left\{\gamma_1 - \frac{1}{\mathfrak{m}_n}\left(\frac{m}{p} - \eth_1\right)\right\},\\
        &\tau_3 = \frac{n_z\sigma_{\bm\beta}^4}{n\mathfrak{g}^2}\left[\mathfrak{h}\left\{\frac{1}{\mathfrak{m}_n}\left(\frac{m}{p} - \eth_1\right)\right\}^2 + \gamma_1\left\{\frac{1-h_{\bm\beta}^2}{h_{\bm\beta}^2}\gamma_1\mathfrak{g} + \frac{\lambda\mathfrak{m}_n'}{\mathfrak{m}_n}\left\{\lambda\left(\eth_1 - \eth_2\right) - \frac{1-h_{\bm\beta}^2}{h_{\bm\beta}^2}\gamma_1\phi_n(\ell_1 - \ell_2)\right\}\right\}\right], \\
        &\mbox{and}\quad\tau_4^2 = \frac{n_z\lambda^2}{\mathfrak{g}^2}\left\{\left(\mathbb E\left(\bm\beta^4\right) - \frac{3\sigma_{\bm\beta}^4}{p^2}\right)\sum_{i=1}^m\mathfrak{N}_{i,i}^2 + \frac{2\sigma_{\bm\beta}^4}{p^2}\Tr(\mathfrak{N}^2)\right\}.
    \end{aligned}
\end{equation*}
Here 
\begin{align*}
    &\zeta_1 = \sigma_{\bm\beta}^2\left\{\gamma_1 - \frac{1}{\mathfrak{m}_n}\left\{2\left(\frac{m}{p} - \eth_1\right) - \frac{1}{\mathfrak{r}_n}\left(\eth_1 - \eth_2\right)\right\}\right\} + \frac{1-h_{\bm\beta}^2}{h_{\bm\beta}^2}\frac{\sigma_{\bm\beta}^2\gamma_1\phi_n}{\lambda\mathfrak{m}_n}\left\{\left(1-\ell_1\right) -\frac{1}{\mathfrak{r}_n}\left(\ell_1 - \ell_2\right)\right\}, \\
    &\zeta_2  = \frac{\sigma_{\bm\beta}^2}{h_{\bm\beta_z}^2}\gamma_1, \quad \mbox{and} 
    \quad\mathfrak{N} = \left\{\frac{\mathfrak{g}}{\lambda}\bm\Sigma - \mathbb\I_p + (\mathbb\I_p+\mathfrak{m}_n\bm\Sigma)^{-1}\right\}\mathbb\I_m.
\end{align*}
Then the following quantitative CLT holds 
    \begin{equation*}
        \begin{aligned}
            \sup_{t\in \mathbb R}\left|\bm G_{\text{R}}(t) - \Phi(t)\right| \to 0.
        \end{aligned}
    \end{equation*}
{Note that $\eth_1$ and $\eth_2$ are defined in Definition~\ref{def:elleth}.}
Moreover, under the optimal ridge parameter $\lambda^*=\phi_n (1-h_{\bm\beta}^2)/h_{\bm\beta}^2$ \citep{zhao2022block}, we have the following updates 
    \begin{equation*}
        \begin{aligned}
            &\tilde{A}_{\text{R}} = h_{\bm\beta_z}\gamma_1^{-1/2}\left\{\gamma_1 - \frac{\phi_n(1-h_{\bm\beta}^2)}{h_{\bm\beta}^2\mathfrak{g}}\left(\frac{m}{p} - \eth_1\right)\right\}\left[\gamma_1 - \frac{1}{\mathfrak{m}_n}\left\{2\left(\frac{m}{p} - \eth_1\right) + \frac{1}{\mathfrak{r}_n}\left(\eth_1 - \eth_2\right)\right\} +\right.\\
            &\left.\frac{\gamma_1}{\mathfrak{m}_n}\left\{\left(1-\ell_1\right) -\frac{1}{\mathfrak{r}_n}\left(\ell_1 - \ell_2\right)\right\}\right]^{-1/2},\\
            &\tau_3 = \frac{n_z\sigma_{\bm\beta}^4}{n\mathfrak{g}^2}\left[\mathfrak{h}\left\{\frac{1}{\mathfrak{m}_n}\left(\frac{m}{p} - \eth_1\right)\right\}^2 + \frac{\phi_n\gamma_1(1-h_{\bm\beta}^2)}{h_{\bm\beta}^2}\left\{\frac{\gamma_1\mathfrak{g}}{\phi_n} + \frac{\mathfrak{m}_n'}{\mathfrak{m}_n}\frac{\phi_n(1-h_{\bm\beta}^2)}{h_{\bm\beta}^2}\left\{\left(\eth_1 - \eth_2\right) -\gamma_1(\ell_1 - \ell_2)\right\}\right\}\right],\\
            &\mbox{and} \\
            &\zeta_1 = \sigma_{\bm\beta}^2\left[\gamma_1 - \frac{1}{\mathfrak{m}_n}\left\{2\left(\frac{m}{p} - \eth_1\right) - \frac{1}{\mathfrak{r}_n}\left(\eth_1 - \eth_2\right)\right\} + \frac{\gamma_1}{\mathfrak{m}_n}\left\{\left(1-\ell_1\right) -\frac{1}{\mathfrak{r}_n}\left(\ell_1 - \ell_2\right)\right\}\right].
        \end{aligned}
    \end{equation*}
\end{thm.s}
Here $\tau_1$ emerges naturally from the Gaussian distribution originated from the randomness of $\epsilon_z$ and $\tau_2$ follows from the Gaussian distribution generated by the randomness of $\Z_0$. Moreover, $\tau_3$ captures the Gaussian generated by the randomness of $\X_0$ and $\epsilon$, which is obtained using the asymptotic normal approximation of de-biased estimators \citep{10.1214/22-AOS2243}. In addition, $\tau_4$ follows from the Gaussian distribution generated by the randomness of $\bm\beta$, which is obtained using Theorem~\ref{thm:quad_form}. Similar to Theorems~\ref{thm: CLT for A^2 Marginal} and~\ref{thm: CLT for reference A^2}, we have $0< \eta_{\text{R}} < 1$. 
\section{Key concentration lemmas}
In this section, we present several key concentration lemmas useful in proving our quantitative CLTs. 

\begin{lem.s}[CLT for quadratic functional]
\label{lemma:asmptotic CLT 1}
        Under Assumption~\ref{a:Sigmabound}, assume $\X_0$ is a $n\times p$ matrix and some deterministic vector $\bm a$ and $\bm b$, then $$\frac{\left \langle \bm a, \X_0^\top \X_0 \bm b\right \rangle - n \bm a^\top \bm b}{\sigma\sqrt{n}} \indist \mathcal{N}(0,1),$$ where $$\sigma^2 = \mathbb E\left(x_0^4-3\right)\sum_{i=1}^p\bm a_i^2\bm b_i^2 + 2\left(\bm a^\top \bm b\right)^2 + \|\bm a\|_2^2\|\bm b\|_2^2.$$
        Similar results also hold for $\Z_0$ with scale $n_z\times p$. 
\end{lem.s}
\begin{proof}
    Notice that $$\left \langle \bm a, \X_0^\top \X_0 \bm b \right\rangle = \sum_{i=1}^n \bm a^\top \x_{0_i} \bm b^\top \bm x_{0_i},$$ where $\bm x_{0_i}$ is $i_{th}$ i.i.d. row vectors of $\X_0$.
    We have $\mathbb E\left(\bm a^\top \bm x_0 \bm b^\top \bm x_0\right) = \bm a^\top \bm b$ and $\mathbb E\{(\bm a^\top \bm x_0)^2(\bm b^\top \bm x_0)^2\} = \mathbb E(x_0^4-3)\sum_{i=1}^p\bm a_i^2 \bm b_i^2 + 2(\bm a^\top \bm b)^2 + \|\bm a\|_2^2\|\bm b\|_2^2$.\\
    Denote 
    $$\X_{n,i} = \frac{\bm a^\top \bm x_{0_i}\bm b^\top \bm x_{0_i}- \bm a^\top \bm b}{\bm \sigma\sqrt{n}},$$ 
    we will check Lyapunov’s conditions for the CLT result
    \begin{enumerate}
     \item $\mathbb E(\X_{n,i}) = 0$;
     \item $\sum_{i=1}^n\mathbb E(\X_{n,i}^2) \to 1$ when $n \to \infty$; and 
     \item $\sum_{i=1}^n\mathbb E(\X_{n,i}^4) \to 0$ when $n \to \infty$.
    \end{enumerate}
    The first two conditions follow trivially from our definition of $\X_{n,i}$. Now we examine the fourth moment of $\X_{n,i}$. Notice that
    $$\sum_{i=1}^n\mathbb E(\X_{n,i}^4) = \frac{\mathbb E\left\{(\bm a^\top \bm x_{0}\bm b^\top \bm x_{0} - \bm a^\top \bm b)^4\right\}}{\sigma^4 n} \leq \frac{M\left\{\mathbb E\left\{(\bm a^\top \bm x_0)^4 (\bm b^\top \bm x_0)^4\right\} + (\bm a^\top \bm b)^4\right\}}{\sigma^4 n}$$ for some large constant $M$. Moreover, we have 
    $$\mathbb E\left\{(\bm a^\top \bm x_0)^4(\bm b^\top \bm x_0)^4\right\} \leq \lambda\mathbb E\left\{(\bm a^\top \bm x)^8\right\} + \frac{1}{\lambda}\mathbb E\left\{(\bm b^\top \bm x)^8\right\}$$ 
    for any $\lambda \in \mathbb{R}$.  
    With Assumption~\ref{a:Sigmabound}, the order of 
    $\mathbb E\{(\bm a^\top \bm x)^8\}$ and 
    $\mathbb E\{(\bm b^\top \bm x)^8\}$ only depends on terms with the most summation, 
    which are $\|\bm a\|_2^8$ and $\|\bm b\|_2^8$, respectively. 
    Therefore, we can further bound the right-hand side with the dominant terms
    $$\mathbb E\left\{(\bm a^\top \bm x_0)^4(\bm b^\top \bm x_0)^4\right\} \leq \lambda N \|\bm a\|_2^8 + \frac{1}{\lambda}N \|\bm b\|_2^8$$ for some large constant $N$.
    Since $\sigma^2 \geq O_p(\|\bm a\|_2^2\|\bm b\|_2^2) \geq r\|\bm a\|_2^2\|\bm b\|_2^2$ for some small $r$, we have 
    $$\frac{M\mathbb E\left\{(\bm a^\top \bm x_0)^4(\bm b^\top \bm x_0)^4\right\}}{\sigma^4n} \leq \frac{M}{n}\frac{\lambda\left(N\|\bm a\|_2^8\right) + \left(N\|\bm b\|_2^8\right)/\lambda}{r^2(\|\bm a\|_2^4\|\bm b\|_2^4)} = \frac{1}{n}\frac{MN}{r^2}\left(\lambda \frac{\|\bm a\|_2^4}{\|\bm b\|_2^4} + \frac{1}{\lambda}\frac{\|\bm b\|_2^4}{\|\bm a\|_2^4}\right).$$
    By taking $\lambda = \|\bm b\|_2^4/\|\bm a\|_2^4$, we have 
    $$\frac{M\mathbb E\left\{(\bm a^\top \bm x_0)^4(\bm b^\top \bm x_0)^4\right\}}{\sigma^4n} \leq \frac{1}{n}\frac{MN}{r^2} \to 0,$$ as $n \to \infty$.
    Since $\sigma^2 \geq O_p\{(\bm a^\top \bm b)^2\} \geq s(\bm a^\top \bm b)^2$ for some small $s$, we have
    $$\frac{(\bm a^\top \bm b)^4}{\sigma^4n} \leq \frac{(\bm a^\top \bm b)^4}{s^2(\bm a^\top \bm b)^4n} = \frac{1}{s^2n} \to 0$$ as $n \to \infty$.
    Combining two results above, we conclude the third Lyapunov’s condition.
\end{proof}

\begin{lem.s}[Quantitative CLT for quadratic functional]
\label{lemma:non-asmptotic CLT 1}
Under Assumption~\ref{a:Sigmabound}, assume $\X_0$ is a $n\times p$ matrix and some deterministic vector $\bm a$ and $\bm b$. 
Let $$\bm F_n(t) = \mathbb P\left(\frac{\left \langle \bm a, \X_0^\top \X_0 \bm b\right \rangle - n \bm a^\top \bm  b}{\sigma\sqrt{n}} < t\right),$$ 
we have $$\sup_{t\in\mathbb R}|\bm F_n(t) - \Phi(t)| \leq c\sqrt{\frac{\tau}{n}},$$ 
where 
$$\sigma^2 = \mathbb E\left(\bm x_0^4-3\right)\sum_{i=1}^p \bm a_i^2 \bm b_i^2 + 2(\bm a^\top \bm b)^2 + \|\bm a\|_2^2\|\bm b\|_2^2$$ for $\tau\in\mathbb R$ and $c$ is some constant. 
Similar results also hold for $\Z_0$ with scale $n_z\times p$. 
\end{lem.s}
\begin{proof}
As in Lemma~S\ref{lemma:asmptotic CLT 1}, denote 
$$\X_{n,i} = \frac{\bm a^\top \bm x_{0_i}\bm b^\top \bm x_{0_i}- \bm a^\top \bm b}{\sigma\sqrt{n}}.$$
We have already shown that
\begin{enumerate}
    \item $\mathbb E\left(\X_{n,i}\right) = 0$;
    \item $\sum_{i=1}^n\mathbb E\left(\X_{n,i}^2\right) \to 1$ when $n \to \infty$; and 
    \item $\sum_{i=1}^n\mathbb E\left(\X_{n,i}^4\right) \leq O_p(n^{-1})$.
\end{enumerate}
Define Lyapunov coefficients as 
$$L_t = \sum_{i=1}^n\mathbb E\left|\X_{n,i}\right|^t,$$ 
and by using the Berry-Esseen Theorem with Lyapunov coefficient $L_4$, we have the following result
\begin{equation*}
    \begin{aligned}
        \sup_{t\in\mathbb R}|\bm F_n(t)-\Phi(t)| \leq cL_3 \leq c\sqrt{L_4}\\
        \implies\sup_{t\in\mathbb R}|\bm F_n(t)-\Phi(t)| \leq c\sqrt{\frac{\tau}{n}}
    \end{aligned}
\end{equation*}
for some constant $\tau$ decided by the exact value of $L_4$.
\end{proof}

\begin{lem.s}[CLT for linear functional]
\label{lemma:asmptotic CLT 2}
    Under Assumption~\ref{a:Sigmabound}, assume $\X_0$ is a $n\times p$ matrix and some non-deterministic vector $\bm a$ independent with $\X_0$  and deterministic vector $\bm b$ where $\sum_{i=1}^n \mathbb E(\bm a_i^4)/\mathbb E(\|\bm a\|_2^2)^2 \to 0$ as $n \to \infty$,  we have 
    $$\frac{\left \langle \bm a, \X_0 \bm b \right \rangle}{\sqrt{\mathbb E\left(\|\bm a\|_2^2\right)}\|\bm b\|_2} \indist \mathcal{N}(0,1).$$
    Similar results also hold for $\Z_0$ with scale $n_z\times p$. 
\end{lem.s}
\begin{proof}
    Similarly, it is possible to decompose the desired quantity into i.i.d copies $$\left \langle \bm a, \X_0 \bm b \right\rangle = \sum_{i=1}^n \bm a_i \bm x_i^\top \bm b.$$
    We can show that $\mathbb E(\bm a_i \bm x_i^\top \bm b) = 0$ and 
    $\mathbb E(\bm a_i^2 \bm x_i^\top \bm b \bm x_i^\top \bm b) = \mathbb E(\bm a_i^2) \|\bm b\|_2^2$.
    Denote $$\X_{n,i} = \frac{\bm a_i \bm x_i^\top \bm b}{\sqrt{\mathbb E\left(\|\bm a\|_2^2\right)}\|\bm b\|_2},$$ we proceed by checking Lyapunov's conditions
    \begin{enumerate}
     \item $\mathbb E\left(\X_{n,i}\right) = 0;$
     \item $\sum_{i=1}^n\mathbb E\left(\X_{n,i}^2\right) = 
     \left\{\sum_{i=1}^n \mathbb E\left(\bm a_i^2\right)\|\bm b\|_2^2\right\}\ /\left\{\mathbb E\left(\|\bm a\|_2^2\right)\|\bm b\|_2^2\right\} 
     = 1;$ and 
     \item $\sum_{i=1}^n\mathbb E\left(\X_{n,i}^4\right) = 
     \left[\sum_{i=1}^n \mathbb E\left(\bm a_i^4\right)\left\{\mathbb E\left(\bm x_0^4-3\right)\sum_{j=1}^p \bm b_j^4 + 3\|\bm b\|_2^4\right\} \right]/ \left\{\mathbb E\left(\|\bm a\|_2^2\right)^2\|\bm b\|_2^4 \right\}\\ \qquad=
     O_p\left(\sum_{i=1}^n \mathbb E\left(\bm a_i^4\right)\|\bm b\|_2^4\right)/\left\{ \mathbb E\left(\|\bm a\|_2^2\right)^2\|\bm b\|_2^4 \right\} 
     \to 0$ as $n \to \infty.$
    \end{enumerate}
\end{proof}

\begin{lem.s}[Quantitative CLT for linear functional]
\label{lemma:non-asmptotic CLT 2}
    Under Assumption~\ref{a:Sigmabound}, assume $\X_0$ is a $n\times p$ matrix and some non-deterministic vector $\bm a$ independent with $\X_0$  and deterministic vector $\bm  b$. 
    Let $$\bm F_n(t) = \mathbb P\left(\frac{\left\langle \bm a, \X_0 \bm b \right\rangle}{\sqrt{\mathbb E\left(\|\bm a\|_2^2\right)}\|\bm b\|_2} < t\right),$$ 
    we have 
    $$\sup_{t\in\mathbb R}\left|\bm F_n(t) - \Phi(t)\right| \leq c\sqrt{\frac{\tau\sum_{i=1}^n\mathbb E\left[\bm a_i^4\right]}{\mathbb E\left[\|\bm a\|_2^2\right]^2}},$$ 
    where $\tau\in\mathbb R$ and $c$ is some constant.  
    Similar results also hold for $\Z_0$ with scale $n_z\times p$. 
\end{lem.s}
\begin{proof}
    Denote $$\X_{n,i} = 
    \frac{\bm a_i \bm x_i^\top \bm b}{\sqrt{\mathbb E\left(\|\bm a\|_2^2\right)}\|\bm b\|_2},$$ from our proof in Lemma~S\ref{lemma:asmptotic CLT 2}, we have
    \begin{enumerate}
     \item $\mathbb E\left(\X_{n,i}\right) = 0$;
     \item $\sum_{i=1}^n\mathbb E\left(\X_{n,i}^2\right) = 
     1$; and 
     \item $\sum_{i=1}^n\mathbb E\left(\X_{n,i}^4\right) = 
      O_p\left(\sum_{i=1}^n \mathbb E\left(\bm a_i^4\right)\right)/\mathbb E\left(\|\bm a\|_2^2\right)^2 
      $.
    \end{enumerate}
    Define Lyapunov coefficientss as $$L_t = \sum_{i=1}^n\mathbb E\left|\X_{n,i}\right|^t,$$ using Berry-Esseen Theorem with the Lyapunov coefficient $L_4$, we have the following result
\begin{equation*}
    \begin{aligned}
        \sup_{t\in\mathbb R}|\bm F_n(t)-\Phi(t)| \leq cL_3 \leq c\sqrt{L_4}\\
        \implies\sup_{t\in\mathbb R}|\bm F_n(t)-\Phi(t)| \leq c\sqrt{\frac{\tau\sum_{i=1}^n \mathbb E\left(\bm a_i^4\right)}{\mathbb E\left(\|\bm a\|_2^2\right)^2}}
    \end{aligned}
\end{equation*}
for some constant $\tau$ decided by the exact value of $L_4$.
\end{proof}

\begin{lem.s}[CLT for residue functional]
\label{lemma:asmptotic CLT 3}
    Under Assumptions~\ref{a:Sigmabound}-\ref{a:Sparsity}, assume that $\bm \epsilon$ is a $n\times 1$ random vector where each entry are i.i.d with bounded second and fourth moment, i.e. $\mathbb E(\bm \epsilon_i^2) = \sigma_\epsilon^2 \leq c_1$, $\mathbb E(\bm \epsilon_i^4) \leq c_2$ for some $c_1, c_2 \in \mathbb R$. 
    Assume a deterministic vector $\bm a$ where $\mathbb E(\bm \epsilon_i^4)\sum_{i=1}^n \bm a_i^4/\{\sigma_\epsilon^4\|\bm a\|^4\}\to 0$ as $n\to \infty$, we have  
    $$\frac{\left \langle \bm a, \bm \epsilon \right \rangle}{\sigma_\epsilon\|\bm a\|_2} \to \mathcal{N}(0,1).$$
\end{lem.s}
\begin{proof}
    Denote $$\X_{n,i} = \frac{\bm a_i \bm \epsilon_i}{\sigma_\epsilon\|\bm a\|_2},$$ 
    observe that $\X_{n,i}$s are independent for different $i$,  thus we only need to check Lindeberg's conditions before concluding our CLT results 
    \begin{enumerate}
     \item $\mathbb E\left(\X_{n,i}\right) = 0$;
     \item $\sum_{i=1}^n\mathbb E\left(\X_{n,i}^2\right) = \left\{\sum_{i=1}^n \mathbb E\left(\bm \epsilon_i^2\right)\bm a_i^2 \right\}/\left(\sigma_\epsilon^2\|\bm a\|_2^2\right) = 1$; and 
     \item $\sum_{i=1}^n\mathbb E\left(\X_{n,i}^4\right) = \left\{ \mathbb E\left(\bm \epsilon_i^4\right)\sum_{i=1}^n \bm a_i^4 \right\}/\left(\sigma_\epsilon^4\|\bm a\|^4\right) \to 0$ as $n \to \infty$.
    \end{enumerate}
\end{proof}

\begin{lem.s}[Quantiative CLT for residue functional]
\label{lemma:non-asmptotic CLT 3}
    Under Assumption~\ref{a:Sigmabound}, assume $\bm a$ to be a deterministic vector, and assume that $\bm \epsilon$ is a $n\times 1$ random vector where each entry are i.i.d with bounded second and fourth moment, i.e. $\mathbb E(\bm \epsilon_i^2) = \sigma_\epsilon^2 \leq c_1$, $\mathbb E(\bm \epsilon_i^4) \leq c_2$ for some $c_1, c_2 \in \mathbb R$. 
    Let $$\bm F_n(t) = \mathbb P\left(\frac{\left\langle \bm a, \bm \epsilon \right\rangle}{\sigma_\epsilon\|\bm a\|_2} < t\right),$$ 
    we have 
    $$\sup_{t\in\mathbb R}|\bm F_n(t) - \Phi(t)| \leq c\sqrt{\frac{\tau\mathbb E\left(\bm \epsilon_i^4\right)\sum_{i=1}^n \bm a_i^4}{\sigma_\epsilon^4\|\bm a\|_2^4}}$$
    for some quantities $c$ and $\tau \in \mathbb{R}$.
\end{lem.s}
\begin{proof}
    Denote $$\X_{n,i} = \frac{\bm a_i \bm \epsilon_i}{\sigma_\epsilon\|\bm a\|_2},$$ 
    from our proof in Lemma~S\ref{lemma:asmptotic CLT 3}, we have 
    \begin{enumerate}
     \item $\mathbb E\left(\X_{n,i}\right) = 0$;
     \item $\sum_{i=1}^n\mathbb E\left(\X_{n,i}^2\right)= \left\{ \sum_{i=1}^n \mathbb E\left(\bm \epsilon_i^2\right) \bm a_i^2 \right\}/\left\{\sigma_\epsilon^2 \bm a^\top \bm a \right\} = 1$; and 
     \item $\sum_{i=1}^n\mathbb E\left(\X_{n,i}^4\right) = 
     \left\{ \mathbb E\left(\bm \epsilon_i^4\right)\sum_{i=1}^n \bm a_i^4 \right\}/\left\{\sigma_\epsilon^4\|\bm a\|_2^4 \right\}$.
    \end{enumerate}
    Define Lyapunov coefficients as $$L_t = \sum_{i=1}^n\mathbb E\left|\X_{n,i}\right|^t,$$ using Berry-Esseen Theorem with the Lyapunov coefficient $L_4$, we have the following result
\begin{equation*}
    \begin{aligned}
        \sup_{t\in \mathbb R}|\bm F_n(t)-\Phi(t)| \leq cL_3 \leq c\sqrt{L_4}\\
        \implies\sup_{t\in\mathbb R}|\bm F_n(t) - \Phi(t)| \leq c\sqrt{\frac{\tau\mathbb E\left(\bm \epsilon_i^4\right)\sum_{i=1}^n \bm a_i^4}{\sigma_\epsilon^2\|\bm a\|_2^4}}
    \end{aligned}
\end{equation*}
for some constant $\tau$ decided by the exact value of $L_4$.
\end{proof}

\begin{lem.s}[Von Bahr-Esseen inequality]
\label{lemma:von bahr-Esseen bound}
Consider $\X_i$s to be independent, centered random variables, i.e. $\mathbb E\left(\X_i\right) = 0$ for all $i$, then for $\alpha \in \left[1,2\right]$, denote $\bm S_n = \sum_{i=1}^n \X_i, $ we have
\begin{equation*}
    \begin{aligned}
        \mathbb P\left(\left|\bm S_n\right|\geq t\right) \leq nC_\alpha\mathbb E\left|\X_i\right|^\alpha t^{-\alpha},
    \end{aligned}
\end{equation*}
or equivalently
\begin{equation*}
    \begin{aligned}
        \mathbb P\left(|\bm S_n|< t\right) \geq 1-nC_\alpha\mathbb E\left|\X_i\right|^\alpha t^{-\alpha}.
    \end{aligned}
\end{equation*}
\end{lem.s}

\begin{lem.s}(Law of large number version of Lemma~S\ref{lemma:asmptotic CLT 1})
\label{lemma:LLN_X^TX}
Under Assumption~\ref{a:Sigmabound}, assume $\X_0$ is a $n\times p$ matrix and some deterministic vector $\bm a$ and $\bm b$, 
$$\mathbb P\left(\left|\frac{\left \langle \bm a, \X_0^\top \X_0 \bm b \right\rangle}{n} - \bm a^\top \bm b\right| \geq \epsilon\right) \leq O_p\left(\frac{\|\bm a\|^2_2\|\bm b\|_2^2 + 2\|\bm a^\top \bm b\|_2^2}{n\epsilon^2}\right).$$
Similar results also hold for $\Z_0$ with scale $n_z\times p$. 
\end{lem.s}
\begin{proof}
Notice that $$\left \langle \bm a, \X_0^\top \X_0 \bm b \right \rangle = \sum_{i=1}^n \bm a^\top \bm x_{0_i}\bm b^\top \bm x_{0_i},$$
where $\bm x_{0_i}$ is $i_{th}$ i.i.d. copies of row of $\X_0$. We can conclude our lemma by applying Lemma~S\ref{lemma:von bahr-Esseen bound}.
\end{proof}

Throughout the paper, quantities related to matrices of the form $\bm P\X_0^\top \X_0\bm Q\X_0^\top \X_0 \bm R$ will be frequently mentioned and computed. For the clarity of our proof, we discuss its entries here in advance.

\begin{lem.s}
\label{lemma: matrix_exp}
    Under Assumption~\ref{a:Sigmabound}, for any deterministic symmetric matrix $\bm P,\bm Q,\bm R\in \mathbb{R}^{p\times p}$, we claim $\mathbb E\left(\bm P\X_0^\top \X_0\bm Q\X_0^\top \X_0\bm R\right)$ has the $(i,j)_{th}$ element as 
    $$n\left\{\mathbb E\left(x_0^4-3\right)(\bm P \mbox{diag}(\bm Q)\bm R)_{i,j}+(n+1)(\bm P\bm Q\bm R)_{i,j} + \Tr(\bm Q)(\bm P\bm R)_{i,j}\right\}.$$
    Recall that $\mbox{diag}(\bm Q)$ is a diagonal matrix obtained by dropping all off-diagonal elements of $\bm Q$. 
    Therefore,  we have 
    \begin{equation*}
        \begin{aligned}
            \Tr\left\{\mathbb E\left(\bm P\X_0^\top \X_0\bm Q \X_0^\top \X_0 \bm R\right)\right\} & = n\left\{\mathbb E\left(x_0^4-3\right)\Tr(\bm P\mbox{diag}(\bm Q)\bm R) + (n+1)\Tr(\bm P\bm Q\bm R)\right.\\
            &\left. + \Tr(\bm Q)\Tr(\bm P\bm R)\right\}.
        \end{aligned}
    \end{equation*}
    Under Assumption~\ref{a:Sigmabound}, if the $\Tr(\bm P\mbox{diag}(\bm Q)\bm R)$ is of the same order as $\Tr(\bm P\bm Q\bm R)$, then the first term is negligible. We claim that in our discussion, the neglibility of such first term in $\Tr\left\{\mathbb E\left(\bm P\X_0^\top \X_0\bm Q \X_0^\top \X_0 \bm R\right)\right\}$ applies to all of our cases.
\end{lem.s}

\section{Proof for Section~\ref{subsec:marg_new_pred}}\label{sec_proof_31}
In this section, we consider the genetically predicted value  $\z^\top\hat{\bm\beta}_{\text{M}}$ of the marginal estimator given by $\hat{\bm\beta}_{\text{M}} = n^{-1}\X^\top \y$. Using the leave-one-out strategy, we prove Theorem~\ref{thm: CLT for marginal new} in three steps. We first decompose  $\z^\top\hat{\bm\beta}_{\text{M}}$ into two parts, then provide quantitative CLT for the randomness of $\bm \epsilon$, treating $\X_0$ as fixed, and finally, we provide quantitative CLT considering the randomness of $\X_0$, which landing in the desired second-order limit.
\subsection{Decomposition of new prediction of marginal estimator}
The general strategy for proving the CLT for $\z^\top\hat{\bm\beta}_{\text{M}}$ is first to decompose it into two independent parts, conditional on $\X_0$, and recover the Gaussian distribution using Berry-Esseen theorem \citep{ash2000probability}. Under Assumption~\ref{a:Sigmabound}, we have
\begin{align*}
\z^{\top}\hat{\bm\beta}_{\text{M}}=\frac{1}{n}\left(\z^\top\bm\Sigma^{1/2}\X_0^\top \X_0\bm\Sigma^{1/2}\bm\beta + \z^\top \bm\Sigma^{1/2}\X_0^\top\bm \epsilon\right).
\end{align*}

\subsection{Berry-Esseen bounds with the randomness of training error}
Considering fixed $\X_0$, the only source of randomness is from the error $\bm \epsilon$, and we can provide the corresponding quantitative CLT using Lemma~S\ref{lemma:non-asmptotic CLT 3}. 
Denote $$\bm I_{\text{M}p}(t) = \mathbb P\left(\frac{\z^\top\bm\Sigma^{1/2}\X_0^\top\bm \epsilon}{\sqrt{\sigma_\epsilon^2 \z^\top\bm\Sigma^{1/2}\X_0^\top \X_0\bm\Sigma^{1/2}\z}} < t\right),$$
considering $\X_0$ fixed, we can show that
$$\sup_{t\in\mathbb R}\left|\bm I_{\text{M}p}(t) - \Phi_{\epsilon}(t)\right| \leq O_p\left(\frac{\mathbb E\left(\bm \epsilon^4\right)\sum_{i=1}^n\left(\x_{0_{i}}^\top\bm\Sigma^{1/2} \z\right)^4}{\sigma_\epsilon^4\left(\z^\top\bm\Sigma^{1/2}\X_0^\top \X_0\bm\Sigma^{1/2}\z\right)^2}\right),$$
where 
$\x_{0_i}$ represents $i_{th}$ row of data matrix $\X_0$.
Recall that
\begin{align*}
    \bm F_{\text{M}}(t) = &\mathbb P\left(\z^\top\bm\Sigma^{1/2}\X_0^\top \X_0\bm\Sigma^{1/2}\bm\beta + \z^\top\bm\Sigma^{1/2}\X_0^\top\bm \epsilon - n\bm\beta^\top\bm\Sigma \z < \sigma_{\text{M}p}\sqrt{n}t\right)\\
    \Leftrightarrow = &\mathbb P\left(\frac{\z^\top\bm\Sigma^{1/2}\X_0^\top\bm \epsilon}{\sqrt{\sigma_\epsilon^2\z^\top\bm\Sigma^{1/2}\X_0^\top \X_0\bm\Sigma^{1/2}\z}}\frac{\sqrt{\sigma_\epsilon^2\z^\top\bm\Sigma^{1/2}\X_0^\top \X_0\bm\Sigma^{1/2}\z}}{\sigma_{\text{M}p}} \right.\\ 
    & \quad \quad< \sqrt{n}t - 
    \left.\frac{\z^\top\bm\Sigma^{1/2}\X_0^\top \X_0\bm\Sigma^{1/2}\bm\beta - n\bm\beta^\top\bm\Sigma \z}{\sigma_{\text{M}p}}\right).
\end{align*}
Let $\Lambda_{\epsilon}$ be some standard Gaussian random variable depends only on the randomness of $\bm \epsilon$. 
It follows that 
\begin{align*}
\label{ineq:marg_new_eps}
    &\sup_{t\in\mathbb R}\left|\mathbb P\left(\frac{\z^\top\bm\Sigma^{1/2}\X_0^\top\bm \epsilon}{\sqrt{\sigma_\epsilon^2\z^\top\bm\Sigma^{1/2}\X_0^\top \X_0\bm\Sigma^{1/2}\z}}\frac{\sqrt{\sigma_\epsilon^2\z^\top\bm\Sigma^{1/2}\X_0^\top \X_0\bm\Sigma^{1/2}\z}}{\sigma_{\text{M}p}} < \sqrt{n}t - \right.\right.\\
    &\left.\left.\frac{\z^\top\bm\Sigma^{1/2}\X_0^\top \X_0\bm\Sigma^{1/2}\bm\beta - n\bm\beta^\top\bm\Sigma \z}{\sigma_{\text{M}p}}\right)-\mathbb P\left(\frac{\sqrt{\sigma_\epsilon^2\z^\top\bm\Sigma^{1/2}\X_0^\top \X_0\bm\Sigma^{1/2}\z}}{\sigma_{\text{M}p}}\Lambda_\epsilon < \sqrt{n}t -\right.\right.\\
    &\left.\left.\frac{\z^\top \bm\Sigma^{1/2}\X_0^\top \X_0\bm\Sigma^{1/2}\bm\beta - n\bm\beta^\top\bm\Sigma \z}{\sigma_{\text{M}p}}\right)\right| \leq O_p\left(\frac{\mathbb E\left(\bm \epsilon^4\right)\sum_{i=1}^n\left(\x_{0_{i}}^\top\bm\Sigma^{1/2} \z\right)^4}{\sigma_\epsilon^4\left(\z^\top\bm\Sigma^{1/2}\X_0^\top \X_0\bm\Sigma^{1/2}\z\right)^2}\right).
\end{align*}

\subsection{Berry-Esseen bounds with the randomness of training data matrix}
Now consider the randomness of $\X_0$, denote
\begin{align}
    \bm J_{\text{M}p}(t) = \mathbb P\left(\frac{\z^\top\bm\Sigma^{1/2}\X_0^\top \X_0\bm\Sigma^{1/2}\bm\beta - n\z^\top\bm\Sigma\bm\beta}{\sqrt{n\left\{\mathbb E\left(x_0^4-3\right)\sum_{i=1}^p(\bm\Sigma^{1/2}\bm\beta)_i^2(\bm\Sigma^{1/2}\z)_i^2 + 2(\z^\top\bm\Sigma\bm\beta)^2 + \|\bm\Sigma^{1/2}\bm\beta\|_2^2\|\bm\Sigma^{1/2}\z\|_2^2\right\}}} \leq t\right).
\end{align}
Using Lemma~S\ref{lemma:non-asmptotic CLT 1}, we have following Berry-Esseen inequality
\begin{equation*}
    \begin{aligned}
        \sup_{t\in\mathbb R}\left|\bm J_{\text{M}p}(t) - \Phi_{\X_0}(t)\right| \leq O_p(n^{-1/2}),
    \end{aligned}
\end{equation*}
where $\Phi_{\X_0}(t)$ represents the CDF of standard Gaussian depending on the randomness of $\X_0$. Moreover, with the randomness of $\X_0$, for $\epsilon_1, \epsilon_2 > 0$, $0 < \delta < 1/2$, consider the subset where

\begin{align*}
    \Omega_0(\epsilon_1, \epsilon_2) \coloneqq &\left\{\X_0:\left\{\left|\z^\top\bm\Sigma^{1/2}\X_0^\top \X_0\bm\Sigma^{1/2}\z - n\z^\top\bm\Sigma \z\right| < n^{1/2 + \delta}\epsilon_1\right\}\cap\right.\\
    &\left.\left\{\left|\sum_{i=1}^n(\x_{0_i}^\top\bm\Sigma^{1/2}\z)^4 - n\left(\sum_{i=j}^{p}\mathbb E\left(x_0^4 - 3\right)(\bm\Sigma^{1/2}\z)_{j}^4 + 3(\z^\top\bm\Sigma \z)^2\right)\right| \leq n^{1/2 + \delta}\epsilon_2\right\}\right\}.
\end{align*}
Using the Von Bahr-Esseen bound in Lemma~S\ref{lemma:von bahr-Esseen bound} with $\alpha = 2$, we have
$$\mathbb P\left(\left|\z^\top\bm\Sigma^{1/2}\X_0^\top \X_0\bm\Sigma^{1/2}\z - n\z^\top\bm\Sigma \z\right| < n^{1/2 + \delta}\epsilon_1\right) \geq 1-\frac{C}{n^{2\delta}\epsilon_1^2}$$
for any $\delta\text{, } \epsilon_1 > 0$.
This is equivalent to claim that by choosing $0<\delta<1/2$, $\z^\top\bm\Sigma^{1/2}\X_0^\top \X_0\bm\Sigma^{1/2}\z$ would be of order $O_p(n)$ with high probability.
Similarly, we can also bound the order of $\sum_{i=1}^n(\x_{0_i}^\top\bm\Sigma^{1/2}\z)^4$ using Lemma~S\ref{lemma:von bahr-Esseen bound}
$$\mathbb P\left(\left|\sum_{i=1}^n(\x_{0_i}^\top\bm\Sigma^{1/2}\z)^4 - n\left(\sum_{j=1}^{p}\mathbb E[x_0^4 - 3](\bm\Sigma^{1/2}\z)_{j}^4 + 3(\z^\top\bm\Sigma \z)^2\right)\right| \leq n^{1/2 + \delta}\epsilon_2\right) \geq 1-\frac{C}{n^{2\delta}\epsilon_2^2}.$$
Therefore, we have shown that $$\mathbb P(\X_0 \in \Omega_0) \geq 1-O_p(n^{-2\delta}),$$
and for $\X_0\in\Omega_0$, we have the following Berry-Esseen bound
\begin{equation*}
    \begin{aligned}
        &\sup_{t\in\mathbb R}\left|\mathbb P\left(\frac{\z^\top \bm\Sigma^{1/2}X_0^\top\bm \epsilon}{\sqrt{\sigma_\epsilon^2\left(n\z^\top\bm\Sigma \z + O_p(n^{1/2 + \delta})\right)}} < t\right) - \Phi_\epsilon(t)\right| \leq \\&O_p\left(\sqrt{\frac{n\left(\sum_{j=1}^{p}\mathbb E\left(x_0^4 - 3\right)(\bm\Sigma^{1/2}\z)_{j}^4 + 3(\z^\top\bm\Sigma \z)^2\right) + O_p(n^{1/2 + \delta}p^2)}{n^2(\z^\top\bm\Sigma \z)^2 + O_p(n^{3/2 + \delta})}}\right) = O_p(n^{-1/2}).
    \end{aligned}
\end{equation*}
Now we denote
$$\sigma_{\text{M}_{\X_0}} = \sqrt{n\left(\mathbb E\left(x_0^4-3\right)\sum_{i=1}^p(\bm\Sigma^{1/2}\bm\beta)_i^2(\bm\Sigma^{1/2}\z)_i^2 + 2(\z^\top\bm\Sigma\bm\beta)^2 + \|\bm\Sigma^{1/2}\bm\beta\|_2^2\|\bm\Sigma^{1/2}\z\|_2^2\right)}.$$
Moreover, we denote $\Lambda_{\X_0}$ the second standard Gaussian which only depends on the randomness of $\X_0$, then we have the following Berry-Esseen bound
\begin{align*}
    \sup_{t\in \mathbb R}\left|\mathbb P\left(\frac{\z^\top\bm\Sigma^{1/2}\X_0^\top \X_0\bm\Sigma^{1/2}\bm\beta - n\bm\beta^\top\bm\Sigma \z}{\sigma_{\text{M}_{\X_0}}}\frac{\sigma_{\text{M}_{\X_0}}}{\sigma_{\text{M}p}} < t\right) - \mathbb P\left(\frac{\sigma_{\text{M}_{\X_0}}}{\sigma_{\text{M}p}}\Lambda_{\X_0} < t\right)\right|\leq O_p(n^{-1/2}).
\end{align*}
It follows that 
\begin{equation}
\label{ineq:marg_new_2}
    \begin{aligned}
        \sup_{t\in\mathbb R}&\Bigg|\mathbb P\Bigg(\frac{\z^\top\bm\Sigma^{1/2}\X_0^\top \X_0\bm\Sigma^{1/2}\bm\beta - n\bm\beta^\top\bm\Sigma \z}{\sigma_{\text{M}p}}
         < \sqrt{n}t-\frac{\sqrt{\sigma_\epsilon^2\z^\top\bm\Sigma^{1/2}\X_0^\top \X_0\bm\Sigma^{1/2}\z}}{\sigma_{\text{M}p}}\Lambda_\epsilon\Bigg)\Bigg.\\
        &\left. \qquad- \mathbb P\left(\frac{\sigma_{\text{M}_{\X_0}}}{\sigma_{\text{M}p}}\Lambda_{\X_0} < \sqrt{n}t - \frac{\sqrt{n\sigma_\epsilon^2\z^\top\bm\Sigma \z}}{\sigma_{\text{M}p}}\Lambda_\epsilon\right)\right| \leq O_p(n^{-1/2}).
    \end{aligned}
\end{equation}
Combining \cref{ineq:marg_new_eps} and \cref{ineq:marg_new_2} above, we have
\begin{align*}
    \sup_{t\in\mathbb R}\left|\bm F_{\text{M}}(t) - \mathbb P\left(\frac{\sigma_{\text{M}_{\X_0}}}{\sigma_{\text{M}p}}\Lambda_{\X_0} + \frac{\sqrt{n\sigma_\epsilon^2\z^\top\bm\Sigma \z}}{\sigma_{\text{M}p}}\Lambda_\epsilon < \sqrt{n}t\right)\right| \leq O_p(n^{-1/2}).
\end{align*}
For $\X_0 \in \Omega_0$, we have shown that $$O_p\left(\frac{\mathbb E\left(\bm \epsilon^4\right)\sum_{i=1}^n\left(\x_{0_{i}}^\top\bm\Sigma^{1/2} \z\right)^4}{\sigma_\epsilon^4\left(\z^\top\bm\Sigma^{1/2}\X_0^\top \X_0\bm\Sigma^{1/2}\z\right)^2}\right) = O_p(n^{-1/2}).$$
Using the convolution formula for independent Gaussian variables, we have
$$\mathbb P\left(\frac{\sigma_{\text{M}_{\X_0}}}{\sigma_{\text{M}p}}\Lambda_{\X_0} + \frac{\sqrt{n\sigma_\epsilon^2 \z^\top\bm\Sigma \z}}{\sigma_{\text{M}p}}\Lambda_\epsilon < \sqrt{n}t\right) = \Phi(t).$$
It follows that
\begin{align*}
    \sup_{t\in\mathbb R}\left|\bm F_{\text{M}}(t) - \Phi(t)\right| \leq O_p(n^{-1/2}).
\end{align*}

\section{Proof for Section~\ref{subsec:marg_A2}}\label{subsec:proof_quad_form}
To obtain quantitative CLTs for $A (\hat{\bm\beta})$,
we will need first to establish the quantitative CLTs of the quadratic form $\bm\beta^\top\bm\Sigma\bm\beta$. 
In this section, we prove Theorem~\ref{thm:quad_form} using martingale CLT. 
Due to the sparsity structure of $\bm\beta$, we only consider the first $m$ non-zero entries of $\bm\beta$ and treat them as i.i.d. random variables. 
That is, we have $\mathbb E(\bm\beta_i) = 0$, $\mathbb E(\bm\beta_i^2) = \sigma_{\bm\beta}^2/p$ for $1\leq i\leq m$ and $\bm\beta_1, \bm\beta_2, \cdots\bm\beta_m$ are i.i.d. random variables.  
We proceed by applying the martingale CLT to show that $\bm\beta^\top\bm\Sigma\bm\beta$ converges to Gaussian as $m \to \infty$, and use Berry-Esseen type bound to quantify the Kolmorogrov distance between CDF of the quadratic form $\bm\beta^\top\bm\Sigma\bm\beta$ and that of a standard Gaussian random variable.

Let $\bm\Sigma'$ be the leading principle minor of $\bm\Sigma$ by removing the last $p-m$ rows and columns of matrix $\bm\Sigma$.  Let $\tilde{\bm\Sigma}$ be the strictly upper-triangular nilpotent matrix by setting all diagonal elements and lower-triangular elements of $\bm\Sigma'$ to be zero. 
By applying Proposition~S\ref{prop:eigen_interlace_thm}, all eigenvalues of $\bm\Sigma'$ would be of order $\Theta_p(1)$. Moreover, since the last $p-m$ entries of $\bm\beta$ consist of $0$, we drop those terms by only considering the truncated $\tilde{\bm\beta}$, where we drop zero entries of $\bm\beta$.
{We aim to prove the following theorem for general symmetric statistics.}

\begin{thm.s}[Berry-Esseen bound for quadratic form]
\label{thm:BE_quad_form}
    Adopt Assumption~\ref{a:Sparsity}, and consider a general symmetric, positive semi-definite matrix $\bm\Sigma$ with $c\mathbb\I_p \prec \bm\Sigma \prec C\mathbb \I_p$ for some $0 < c < C$. We use $\bm\Sigma'$ to denote the truncated matrix produced by removing the last $p-m$ rows and columns of $\bm\Sigma$, we have the following quantitative CLT for symmetric quantity $\bm\beta^\top\bm\Sigma\bm\beta$
    \begin{equation*}
        \begin{aligned}
            \sup_{t\in\mathbb R}\left|\mathbb P\left(\frac{p\bm\beta^\top\bm\Sigma\bm\beta - \sigma_{\bm\beta}^2\Tr(\bm\Sigma\mathbb\I_m)}{\sqrt{p^2\left\{\mathbb E({\bm\beta}^4)-3\sigma_{\bm\beta}^4/p^2\right\}\sum_{i=1}^m{{(\bm\Sigma\mathbb\I_m)}_{i,i}}^2 + 2\sigma_{\bm\beta}^4\Tr\left\{(\bm\Sigma\mathbb\I_m)^2\right\}}} < t\right) - \Phi(t)\right| &\leq O_p(m^{-1/5}).
        \end{aligned}
    \end{equation*}
\end{thm.s}
\begin{remark.s}
    We may drop the assumption that $\bm\Sigma$ is positive definite. Notice that $\tilde{\bm\Sigma}^\top\tilde{\bm\Sigma}$ is always positive semi-definite in Lemma~S\ref{lemma:tilde_Sigmabound}.
\end{remark.s}

Before proving Theorem~S\ref{thm:BE_quad_form}, we use Propositions~S\ref{prop:eigen_interlace_thm}-S\ref{prop:weyl_ineq} to bound the change of eigenvalues concerning matrices minor and matrices summation.
\begin{proposition.s}[Eigenvalue interlacing theorem]
\label{prop:eigen_interlace_thm}
    Suppose $\bm A\in\mathbb R^{n\times n}$ a symmetric matrix and denote $\bm B\in \mathbb R^{m\times m}$ with $m<n$ as the principal minor of matrix $\bm A$. If $A$ has eigenvalues $\lambda_1\geq\cdots\geq\lambda_n$ and $\bm B$ has eigenvalues $\theta_1 \geq \cdots \geq\theta_m$, then we have 
    $$\lambda_{k+n-m} \leq \theta_k \leq \lambda_k$$ for $k \in [1,m]$.
\end{proposition.s}

\begin{proposition.s}[Weyl's inequality]
\label{prop:weyl_ineq}
    For $m\times m$ symmetric matrices $\bm A$ and $\bm B$, and use $\lambda_i(\cdot)$ to represent the $i_{th}$ largest eigenvalue of the matrix, we have 
    \begin{equation*}
        \begin{aligned}
            \lambda_i(\bm A) + \lambda_m(\bm B) \leq \lambda_i(\bm A+\bm B) \leq \lambda_i(\bm A) + \lambda_1(\bm B).
        \end{aligned}
    \end{equation*}
\end{proposition.s}

\begin{remark.s}
    The eigenvalue interlacing theorem will be applied to show that truncating the covariance matrix $\bm\Sigma$ to its leading principles would not change the order of eigenvalues, whereas Weyl's inequality helps provide an upper bound for a key quantity in determining the convergence rate.
\end{remark.s}
By Proposition~S\ref{prop:eigen_interlace_thm}, we have 
\begin{equation*}
    \begin{aligned}
         &\lambda_{n-m+1}(\bm\Sigma) \leq \lambda_1(\bm\Sigma') \leq \lambda_1(\bm\Sigma) \quad \mbox{and} \quad
         &\lambda_{n}(\bm\Sigma) \leq \lambda_m(\bm\Sigma') \leq \lambda_m(\bm\Sigma).
    \end{aligned}
\end{equation*}
This means that every eigenvalue of $\bm\Sigma'$ is of order $\Theta_p(1)$. Now we quantify the eigenvalues of the key quantities. 
\begin{lem.s}
\label{lemma:tilde_Sigmabound}
    Adopt Assumption~\ref{a:Sigmabound}, let $\tilde{\bm\Sigma}$ be the strictly upper-triangular nilpotent matrix by setting all diagonal elements and lower-triangular elements of $\bm\Sigma'$ to be zero, then we can conclude that there exists some $ \psi_1\geq 0, \psi_2 \geq 0$, such that
    $$\psi_1\mathbb \I_m \preceq \tilde{\bm\Sigma}^\top\tilde{\bm\Sigma}\preceq \psi_2\mathbb \I_m.$$
\end{lem.s}
\begin{proof}
    Notice $\bm\Sigma' = diag(\bm\Sigma') + \tilde{\bm\Sigma}^\top + \tilde{\bm\Sigma}$. From the fact that $c\mathbb\I_p\prec\bm\Sigma\prec C\mathbb\I_p$, every diagonal element of matrix $\bm\Sigma'$ is of order one. By applying Lemma~S\ref{prop:weyl_ineq}, we have $\lambda_i(\bm\Sigma) = \lambda_i(\tilde{\bm\Sigma}^\top + \tilde{\bm\Sigma}) + \Theta_p(1)$. This means that the eigenvalues of $\tilde{\bm\Sigma}^\top + \tilde{\bm\Sigma}$ are also of order one. 
    By applying eigenvalue decomposition, it follows that all eigenvalues of $(\tilde{\bm\Sigma}^\top + \tilde{\bm\Sigma})^2$ are bounded. Expanding this term, we have
    \begin{align*}
        (\tilde{\bm\Sigma}^\top+\tilde{\bm\Sigma})^2 = \tilde{\bm\Sigma}^\top\tilde{\bm\Sigma} + \tilde{\bm\Sigma}\tilde{\bm\Sigma}^\top + (\tilde{\bm\Sigma}^2)^\top + \tilde{\bm\Sigma}^2.
    \end{align*}
    Therefore, for any constant $i \in [1,m]$, by Lemma~S\ref{prop:weyl_ineq}, we have
    \begin{align}
        &\lambda_i\left\{\left(\tilde{\bm\Sigma}^\top+\tilde{\bm\Sigma}\right)^2\right\} = \lambda_i\left\{\tilde{\bm\Sigma}^\top\tilde{\bm\Sigma} + \tilde{\bm\Sigma}\tilde{\bm\Sigma}^\top + \left(\tilde{\bm\Sigma}^2\right)^\top + \tilde{\bm\Sigma}^2\right\}\implies \\
        \label{ineq:half_matrices}
        &\lambda_i\left(\tilde{\bm\Sigma}^\top\tilde{\bm\Sigma} + \tilde{\bm\Sigma}\tilde{\bm\Sigma}^\top\right) + \lambda_m\left\{\left(\tilde{\bm\Sigma}^2\right)^\top + \tilde{\bm\Sigma}^2\right\}\leq\lambda_i\left(\left(\tilde{\bm\Sigma}^\top+\tilde{\bm\Sigma}\right)^2\right) \leq \lambda_i\left(\tilde{\bm\Sigma}^\top\tilde{\bm\Sigma} + \tilde{\bm\Sigma}\tilde{\bm\Sigma}^\top\right) + \lambda_1\left\{(\tilde{\bm\Sigma}^2)^\top + \tilde{\bm\Sigma}^2\right\}.
    \end{align}
    Note that we only need to work on the upper bound for $\tilde{\bm\Sigma}^\top\tilde{\bm\Sigma} + \tilde{\bm\Sigma}\tilde{\bm\Sigma}^\top$, because it is the summation of two positive semi-definite matrices and we have $\tilde{\bm\Sigma}^\top\tilde{\bm\Sigma} + \tilde{\bm\Sigma}\tilde{\bm\Sigma}^\top\succeq 0$. Since $\tilde{\bm\Sigma}$ is a strictly upper-triangular matrix, $\tilde{\bm\Sigma}^2$ is also strictly upper-triangular. 
    Therefore, all its eigenvalues would be zero. This means that $\tilde{\bm\Sigma}^2$ is a positive semi-definite matrix, and so is $(\tilde{\bm\Sigma}^2)^\top$. Therefore, we conclude that $\tilde{\bm\Sigma}^2 + (\tilde{\bm\Sigma}^2)^\top \succeq 0$. By plugging this back to \cref{ineq:half_matrices}, we have 
    $$\lambda_i\left(\tilde{\bm\Sigma}^\top\tilde{\bm\Sigma} + \tilde{\bm\Sigma}\tilde{\bm\Sigma}^\top\right) \leq \lambda_i\left\{\left(\tilde{\bm\Sigma}^\top+\tilde{\bm\Sigma}\right)^2\right\} = O_p(1).$$
    Similarly, $\tilde{\bm\Sigma}^\top\tilde{\bm\Sigma}$ and $\tilde{\bm\Sigma}\tilde{\bm\Sigma}^\top$ are both symmetric and positive semi-definite. By Weyl's inequality, for $\forall i \in [1,m]$, we have 
    \begin{equation*}
        \begin{aligned}
            \lambda_i\left(\tilde{\bm\Sigma}^\top\tilde{\bm\Sigma}\right) + 
            \lambda_m\left(\tilde{\bm\Sigma}^\top\tilde{\bm\Sigma}\right) &\leq \lambda_i\left(\tilde{\bm\Sigma}^\top\tilde{\bm\Sigma}+\tilde{\bm\Sigma}\tilde{\bm\Sigma}^\top\right)\leq \lambda_i\left(\tilde{\bm\Sigma}^\top\tilde{\bm\Sigma}\right) + \lambda_1\left(\tilde{\bm\Sigma}^\top\tilde{\bm\Sigma}\right)\\
            \implies 0\leq\lambda_i\left(\tilde{\bm\Sigma}^\top\tilde{\bm\Sigma}\right) &\leq \lambda_i\left(\tilde{\bm\Sigma}^\top\tilde{\bm\Sigma}+\tilde{\bm\Sigma}\tilde{\bm\Sigma}^\top\right) = O_p(1).
        \end{aligned}
    \end{equation*}
\end{proof}
\begin{remark.s}
    The reason why we construct and examine the eigenvalues of $(\tilde{\bm\Sigma}^\top + \tilde{\bm\Sigma})^2$ is that the Weyl's inequality holds only when the matrix is symmetric. 
    Our overall strategy is to extract information of eigenvalues from ${\bm\Sigma'}^2$, link it to $\tilde{\bm\Sigma}^\top\tilde{\bm\Sigma}$, and then guarantee the positive semi-definiteness of each component to obtain an upper bound.
\end{remark.s}
Now we are ready to prove Theorem~S\ref{thm:BE_quad_form}.

\begin{proof}[Proof of Theorem~S\ref{thm:BE_quad_form}]
Considering the following filtration
$$\mathfrak{F}_t = \left.\sigma\left(\tilde{\bm\beta_i}\right|1\leq i\leq t\right).$$
We can construct corresponding martingale difference sequence
\begin{equation*}
    \begin{aligned}
        \xi_t = 2p\sum_{i=1}^{t-1}\bm\Sigma'_{i,t}\tilde{\bm\beta}_i\tilde{\bm\beta}_t + p\bm\Sigma'_{t,t}\left(\tilde{\bm\beta}_t^2-\frac{\sigma_{\bm\beta}^2}{p}\right)
    \end{aligned}
\end{equation*}
for $1\leq t \leq m$. It is trivial to check that
\begin{enumerate}
    \item $\xi_t$ is adapted to the filtration $\mathfrak{F_t}$;
    \item $\mathbb E\left(\xi_t|\mathfrak{F}_{t-1}\right) = 0$;
    \item $T_m = \sum_{t=1}^m \xi_t = p\tilde{\bm\beta}^\top\bm\Sigma'\tilde{\bm\beta} - \sigma_{\bm\beta}^2\Tr(\bm\Sigma')$; and
    \item $\var\left(T_m\right) = p^2\left\{\mathbb E(\tilde{\bm\beta}^4)-3\sigma_{\bm\beta}^4/p^2\right\}\sum_{i=1}^m{\bm\Sigma'_{i,i}}^2 + 2\sigma_{\bm\beta}^4\Tr({\bm\Sigma'}^2) = \Theta_p(m)$.
\end{enumerate}
Denote the conditional variance $S_m = \var(T_m)^{-1}\sum_{i=1}^m\mathbb E\left(\xi_t^2|\mathfrak{F}_{t-1}\right)$, we aim to check $\mathbb E\left|S_m-1\right|^2 \to 0$ as $m\to\infty$, which leads to $S_m \overset{p}{\to} 1$ in martingale CLT. 
We also aim to quantify the convergence behavior in $L_2$ to provide the Berry-Esseen bound for the convergence rate. 

We start by decomposing the $\mathbb E\left(\xi_t^2|\mathfrak{F}_{t-1}\right)$ into three parts: an uncentered part, a centered part, and a constant. Notice that
\begin{equation*}
    \begin{aligned}
        \left.S_m = \frac{p^2}{\var(T_m)}\sum_{t=1}^m\mathbb E\left\{4\tilde{\bm\beta}_t^2\left(\sum_{i=1}^{t-1}\bm\Sigma'_{i,t}\tilde{\bm\beta}_i\right)^2 + 4\tilde{\bm\beta}_t\sum_{i=1}^{t-1}\bm\Sigma_{i,t}\tilde{\bm\beta}_i\bm\Sigma'_{t,t}\left(\tilde{\bm\beta}_t^2 - \frac{\sigma_{\bm\beta}^2}{p}\right)  + {\bm\Sigma'}_{t,t}^2\left({\tilde{\bm\beta}}_t^2 - \frac{\sigma_{\bm\beta}^2}{p}\right)^2\right| \mathfrak{F}_{t-1}\right\}.
    \end{aligned}
\end{equation*}
 By conditioning on $\mathfrak{F}_{t-1}$, the first part is an uncentered random variable, the second part is a centered random variable, and the third part is a constant. After the conditional expectation, the only term with randomness is of the form 
 $$\upsilon_m(t) = 2\sqrt{p}\sum_{i=1}^{t-1}\bm\Sigma'_{i,t}\bm\beta_i.$$
Denote the following two constants $$D_1 = \mathbb E\left\{p^{3/2}\tilde{\bm\beta}_t\left(\tilde{\bm\beta}_t^2 - \frac{\sigma_{\bm\beta}^2}{p}\right)\right\}\quad \mbox{and} \quad D_2 = \mathbb E\left\{p^2\left(\tilde{\bm\beta}_t^2 - \frac{\sigma_{\bm\beta}^2}{p}\right)^2\right\}.$$
Note that $D_1$ and $D_2$ are constants of order $\Theta_p(1)$. Rearrange these terms, we have 
\begin{equation*}
    \begin{aligned}
        S_m = \frac{p^2}{\var(T_m)}\sum_{t=1}^m\left(\sigma_{\bm\beta}^2\upsilon_m(t)^2 + 2D_1\bm\Sigma'_{t,t}\upsilon_m(t) + D_2{\bm\Sigma'}_{t,t}^2\right).
    \end{aligned}
\end{equation*}
It is easy to check that $\mathbb E(S_m) = 1$. Again, notice that the second term has a mean of zero and the third term is a constant, we have the following inequality
\begin{equation*}
    \begin{aligned}
        \left|S_m - 1\right| \leq C\var(T_m)^{-1}\left\{\left|\sum_{t=1}^m\sigma_{\bm\beta}^2\upsilon_m(t)^2 - \mathbb E\left(\sigma_{\bm\beta}^2\upsilon_m(t)^2\right)\right| + \left|\sum_{t=1}^m\bm\Sigma'_{t,t}\upsilon_{m}(t)\right|\right\}.
    \end{aligned}
\end{equation*}
For the first term, we further split it into ``off-diagonal"  and ``on-diagonal" parts, subtracted by their corresponding means
\begin{equation*}
    \begin{aligned}
        |S_m - 1| \leq C\var(T_m)^{-1}\left\{p\left|\sum_{t=1}^m\sum_{1\leq i < j \leq t-1}\tilde{\bm\beta}_i\tilde{\bm\beta}_j\right| + p\left|\sum_{t=1}^m\sum_{1\leq k\leq t-1}{\bm\Sigma'_{k,t}}^2\left(\tilde{\bm\beta}_k^2 - \frac{\sigma_{\bm\beta}^2}{p}\right)\right| + \left|\sum_{t=1}^m\bm\Sigma'_{t,t}\upsilon_m(t)\right|\right\}.
    \end{aligned}
\end{equation*}
For clarity, we denote 
$$S_{m,1} = p\left|\sum_{t=1}^m\sum_{1\leq i < j \leq t-1}\tilde{\bm\beta}_i\tilde{\bm\beta}_j\right|, \quad S_{m,2} = p\left|\sum_{t=1}^m\sum_{1\leq k\leq t-1}{\bm\Sigma'_{k,t}}^2\left(\tilde{\bm\beta}_k^2 - \frac{\sigma_{\bm\beta}^2}{p}\right)\right|,$$
and 
$$
\quad S_{m,3} = \left|\sum_{t=1}^m\bm\Sigma'_{t,t}\upsilon_m(t)\right|.$$
Intuitively, we try to isolate the contribution from the diagonal elements of $\bm\Sigma'$ from the contribution of the off-diagonal terms. 
Such inequality immediately leads to the following inequality
$$\mathbb E\left|S_m - 1\right|^2 \leq C\var(T_m)^{-2}\left\{\mathbb E\left|S_{m,1}\right|^2 + \mathbb E\left|S_{m,2}\right|^2 + \mathbb E\left|S_{m,3}\right|^2\right\}.$$
Now we will carefully examine the upper bounds of $\mathbb E\left|S_{m,1}\right|^2$, $\mathbb E\left|S_{m,2}\right|^2$, and $\mathbb E\left|S_{m,3}\right|^2$. For $\mathbb E\left|S_{m,1}\right|^2$, we have 
\begin{align*}
    \mathbb E\left|S_{m,1}\right|^2 &= p^2\sum_{t_1 = 1}^m\sum_{t_2 = 1}^m\sum_{1\leq s_1 < s_2 \leq t_1 - 1}\sum_{1\leq s_3 < s_4\leq t_2 -1}\mathbb E\left(\bm\Sigma'_{t_1,s_1}\bm\Sigma'_{t_1,s_2}\bm\Sigma'_{t_2,s_3}\bm\Sigma'_{t_2,s_4}\bm\beta_{s_1}\bm\beta_{s_2}\bm\beta_{s_3}\bm\beta_{s_4}\right)\\
    &\leq C\sum_{t_1,t_2 = 1}^m\sum_{1\leq s_1 < s_2 < \min(t_1,t_2)}\bm\Sigma'_{t_1, s_1}\bm\Sigma'_{t_1,s_2}\bm\Sigma'_{t_2, s_1}\bm\Sigma'_{t_2,s_2}.
\end{align*}
Here $\min(t_1, t_2)$ is inspired by the fact that we need to choose $s_1 = s_3, s_2 = s_4$ to have a nonzero expectation. Note that $s_1 \leq t_1 -1$ and $s_3 \leq t_2 - 1$, so we require $s_1, s_2 < \min(t_1, t_2)$. The same idea applies wherever $\min(t_1, t_2)$ appears. Moreover, for $\mathbb E\left|S_{m,2}\right|^2$, we have 
\begin{align*}
    \mathbb E\left|S_{m,2}^2\right| &= p^2\mathbb E\left\{\sum_{t_1, t_2 = 1}^m\sum_{1\leq k < \min(t_1, t_2)}{\bm\Sigma'_{t,t_1}}^2{\bm\Sigma'_{k,t_2}}^2\left(\tilde{\bm\beta}_k^2 - \frac{\sigma_{\bm\beta}^2}{p}\right)^2\right\}\\
    & \leq C\sum_{t_1, t_2 = 1}^m\sum_{1\leq s_1, s_2 < \min(t_1, t_2)}{\bm\Sigma'_{k, t_1}}^2{\bm\Sigma'_{k,t_2}}^2.
\end{align*}
It follows that 
$$\mathbb E\left(\left|S_{m,1}^2\right| + \left|S_{m,2}^2\right|\right) \leq C\sum_{t_1 = 1}^m\sum_{t_2 = 1}^m \sum_{1\leq s_1,s_2< \min(t_1, t_2)}\bm\Sigma'_{t_1,s_1}\bm\Sigma'_{t_1, s_2}\bm\Sigma'_{t_2, s_1}\bm\Sigma'_{t_2, s_2} = C\|\tilde{\bm\Sigma}^\top\tilde{\bm\Sigma}\|_F^2,$$ 
where $\|\cdot\|_F$ is the Frobenius norm. 
In addition, for $\mathbb E\left|S_{m,3}\right|^2$, recall that $\tilde{\bm\Sigma}$ is the matrix obtained from $\bm\Sigma'$ by dropping the lower-triangular part and setting all diagonal terms equal to $0$. Then we have 
\begin{align*}
    \mathbb E\left|S_{m,3}^2\right| &\leq Cp\mathbb E\left(\sum_{t_1, t_2}^m\bm\Sigma'_{t_1,t_1}\bm\Sigma'_{t_2,t_2}\sum_{i=1}^{t_1-1}\bm\Sigma'_{i,t_1}\tilde{\bm\beta_i}\sum_{j=1}^{t_2 - 1}\bm\Sigma'_{j,t_2}\tilde{\bm\beta}_j\right) \\
    &\leq C\left(\sum_{t_1 = 1}^m\sum_{t_2 = 1}^m\bm\Sigma'_{t_1,t_1}\bm\Sigma'_{t_2,t_2}\sum_{1\leq i < \min(t_1,t_2)}\bm\Sigma'_{i,t_1}\bm\Sigma'_{i,t_2}\right)\\
    & \leq C\left\{\sum_{t_1 = 1}^m\sum_{t_2 = 1}^m\bm\Sigma'_{t_1,t_1}\bm\Sigma'_{t_2,t_2}\left(\tilde{\bm\Sigma}^\top\tilde{\bm\Sigma}\right)_{t_1, t_2}\right\}\\
    &\leq C \left\langle\bm  w, \tilde{\bm\Sigma}^\top\tilde{\bm\Sigma}\bm w \right\rangle\\
    & \leq C'\|\tilde{\bm\Sigma}\|_2\|\bm w\|_2^2 = \|\tilde{\bm\Sigma}\|_2O_p(m),
\end{align*}
where we use vector $\bm w$ to denote the spectrum of $\bm\Sigma'$.
Now we show that $\psi_1\mathbb \I_m \preceq \tilde{\bm\Sigma}^\top\tilde{\bm\Sigma}\preceq \psi_2\mathbb \I_m$ for some constant $\psi_1, \psi_2 \geq 0$. This will immediately implies that $\|\tilde{\bm\Sigma}\|_2 = \Theta_p(1)$. 
To do this, we first show that the principle minor is bounded using the eigenvalue interlacing theorem in Proposition~S\ref{prop:eigen_interlace_thm}, then the boundness of $\tilde{\bm\Sigma}^\top\tilde{\bm\Sigma}$ follows from Weyl's inequality.
Adopting Assumption~\ref{a:Sigmabound} and by Proposition~S\ref{prop:eigen_interlace_thm}, we have that every eigenvalue of $\tilde{\bm\Sigma}$ is of order $\Theta_p(1)$.
Notice that $\|\tilde{\bm\Sigma}\|_2 = \lambda_1(\tilde{\bm\Sigma}^\top\tilde{\bm\Sigma}) = O_p(1)$ and $\|\tilde{\bm\Sigma}^\top\tilde{\bm\Sigma}\|_F = \Tr\left(\tilde{\bm\Sigma}^\top\tilde{\bm\Sigma}\tilde{\bm\Sigma}^\top\tilde{\bm\Sigma}\right)^{1/2} \leq O_p(m^{1/2})$. Therefore, we have
\begin{equation*}
    \begin{aligned}
      &  \mathbb E\left|S_{m,1}\right|^2 + \mathbb E\left|S_{m,2}\right|^2 + \mathbb E\left|S_{m,3}\right|^2 \leq O_p(m) \implies\\
     &\mathbb E\left|S_{m} - 1\right|^2 \leq C\var(T_m)^{-2}\left\{\mathbb E\left|S_{m,1}\right|^2 + \mathbb E\left|S_{m,2}\right|^2 + \mathbb E\left|S_{m,3}\right|^2\right\}\leq O_p(m^{-1}).
    \end{aligned}
\end{equation*}
By the Berry-Esseen bound proven in \cite{10.1214/aop/1176991901}, the following inequality holds
\begin{equation*}
    \begin{aligned}
        \sup_{t\in\mathbb R}\left|\frac{p\tilde{\bm\beta}^\top\bm\Sigma'\tilde{\bm\beta} - \sigma_{\bm\beta}^2\Tr(\bm\Sigma')}{\var(T_m)^{1/2}}\right| \leq C_p\left(\mathbb E\left|S_m-1\right|^2 + \frac{\sum_{i=1}^m\mathbb E\left|\xi_i\right|^4}{\var(T_m)^2}\right)^{-1/5}.
    \end{aligned}
\end{equation*}
Notice that $\mathbb E\left|\xi_i\right|^4 = O_p(1)$ for $\forall i \in [1,m]$. Moreover, we have 
$$\var(T_m) = p^2\left\{\mathbb E\left(\tilde{\bm\beta}^4\right)-3\frac{\sigma_{\bm\beta}^4}{p^2}\right\}\sum_{i=1}^m{\bm\Sigma'_{i,i}}^2 + 2\sigma_{\bm\beta}^4\Tr\left({\bm\Sigma'}^2\right) = \Theta_p(m).$$
This leads to the following quantitative CLT
\begin{equation*}
    \begin{aligned}
        \sup_{t\in\mathbb R}\left|\mathbb P\left(\frac{p\tilde{\bm\beta}^\top\bm\Sigma'\tilde{\bm\beta} - \sigma_{\bm\beta}^2\Tr\left(\bm\Sigma'\right)}{\sqrt{p^2\left\{\mathbb E\left(\tilde{\bm\beta}^4\right)-3\sigma_{\bm\beta}^4/p^2\right\}\sum_{i=1}^m{\bm\Sigma'_{i,i}}^2 + 2\sigma_{\bm\beta}^4\Tr\left({\bm\Sigma'}^2\right)}} < t\right) - \Phi(t)\right| \leq O_p(m^{-1/5}).
    \end{aligned}
\end{equation*}
We can then translate it to the following notation
\begin{equation*}
    \begin{aligned}
        \Leftrightarrow \sup_{t\in\mathbb R}\left|\mathbb P\left(\frac{p\bm\beta^\top\bm\Sigma\bm\beta - \sigma_{\bm\beta}^2\Tr\left(\bm\Sigma\mathbb\I_m\right)}{\sqrt{p^2\left\{\mathbb E\left({\bm\beta}^4\right)-3\sigma_{\bm\beta}^4/p^2\right\}\sum_{i=1}^m{{\left(\bm\Sigma\mathbb\I_m\right)}_{i,i}}^2 + 2\sigma_{\bm\beta}^4\Tr\left\{\left(\bm\Sigma\mathbb\I_m \right)^2\right\}}} < t\right) - \Phi(t)\right| &\leq O_p(m^{-1/5}).
    \end{aligned}
\end{equation*}
\end{proof}

The next lemma shows that if features are isotropic, we do not need to use the martingale CLT. Instead, we can approach the second-order limit using the Berry-Esseen theorem, which provides a better Berry-Esseen bound.

\begin{lem.s}[Berry-Esseen bounds for quadratic grom with isotropic features]
\label{lemma:BE_quad_form_iso}
    Adopt Assumption~\ref{a:Sparsity} and when $\bm\Sigma = \mathbb\I_p$, consider quadratic quantity $\bm\beta^\top\bm\beta$, we have the following quantitative CLT
    \begin{equation*}
        \begin{aligned}
            \sup_{t\in\mathbb R}\left|\mathbb P\left(\frac{p\bm\beta^\top\bm\beta - m\sigma_{\bm\beta}^2}{\sqrt{mp^2\left\{\mathbb E\left(\bm\beta^4\right) - 3\sigma_{\bm\beta}^4/p^2\right\} + 2m\sigma_{\bm\beta}^4}} < t\right) - \Phi(t)\right| \leq O(m^{-1/2}).
        \end{aligned}
    \end{equation*}
\end{lem.s}
\begin{proof}
    Notice that we can decompose $\bm\beta^\top\bm\beta = \sum_{i=1}^m\bm\beta_i^2$, where $\bm\beta_i^2$ are i.i.d. items. 
    Consider
    $$\xi_{m,i} = \frac{\bm\beta_i^2}{\sqrt{m\left\{\mathbb E\left(\bm\beta^4\right) - 3\sigma_{\bm\beta}^4/p^2\right\} + 2m\sigma_{\bm\beta}^4/p^2}}.$$
    It is easy to check
    \begin{enumerate}
        \item $\mathbb E\left(\xi_{m,i} \right) = 0$;
        \item $\sum_{i=1}^m\mathbb E\left(\xi_{m,i}^2\right) = 1$; and 
        \item $\sum_{i=1}^m\mathbb E\left(\xi_{m,i}^4\right) \leq O_p(m^{-1})$.
    \end{enumerate}
    Recall Lyapunov coefficients are defined as $$L_t = \sum_{i=1}^n\mathbb E\left|X_{n,i}\right|^t.$$
    Using the Berry-Esseen Theorem \citep{ash2000probability} with the Lyapunov coefficient $L_4$, we have the following result
    \begin{equation*}
        \begin{aligned}
            &\sup_{t\in\mathbb R}\left|\mathbb P\left(\sum_{i=1}^m\xi_i < t\right)-\Phi(t)\right| \leq CL_3 \leq C\sqrt{L_4}\\
            &\implies\sup_{t\in\mathbb R}\left|\mathbb P\left(\sum_{i=1}^m\xi_i < t\right)-\Phi(t)\right| \leq O_p(m^{-1/2})
        \end{aligned}
    \end{equation*}
for some constant $\tau$ decided by the exact value of $L_4$. With proper rescaling, we have our desired Berry-Esseen inequality.
\end{proof}

\section{Proof for Section~\ref{subsubsec:marg_A_iso}}
\label{subsec:proof_marg_A2}

To prove the quantitative CLT of $A(\hat{\bm\beta}_{\text{M}})$ in Theorem~\ref{thm: CLT for A^2 Marginal}, we use eight steps with a leave-one-out technique. Briefly, we first decompose the numerator of $A(\hat{\bm\beta}_{\text{M}})$ into two parts and then provide the quantitative CLT regarding the randomness of $\bm \epsilon_z$. Next, we provide the quantitative CLT regarding the randomness of $\Z_0$, $\bm \epsilon$, $\X_0$, and $\bm\beta$, along with first-order concentrations. This concludes the CLT for the numerator part in $A(\hat{\bm\beta}_{\text{M}})$. Similarly, we obtain the first-order limit of the denominator part in $A(\hat{\bm\beta}_{\text{M}})$, and then we conclude the CLT of $A(\hat{\bm\beta}_{\text{M}})$ by applying Slutsky's theorem. 

\subsection{Numerator decomposition}
Notice that we cannot decompose the numerator into the summation of i.i.d. random variables. Therefore, we need to consider multiple quantitative CLTs to establish our Theorem~\ref{thm: CLT for A^2 Marginal}. The numerator can be rewritten as
\begin{equation*}
    \begin{aligned}
        (\bm\beta^\top \Z^\top +\bm \epsilon_z^\top)\Z\X^\top(\X\bm\beta + \bm \epsilon) = \bm\beta^\top \Z^\top \Z\X^\top (\X\bm\beta + \bm \epsilon) +\bm  \epsilon_z^\top \Z\X^\top (\X\bm\beta + \bm \epsilon) .
    \end{aligned}
\end{equation*}
The advantage of such decomposition is that the first quantity does not involve $\bm \epsilon_z$. Therefore, we can consider the first quantity fixed, while we can use Lemma~S\ref{lemma:non-asmptotic CLT 3} to provide quantitative CLT for the second quantity.

\subsection{Berry-Esseen bounds with the randomness of testing error}
Conditional on $\Z_0, \X_0, \bm \epsilon, \bm\beta$, and considering the randomness of $\bm \epsilon_z$, then by using Lemma~S\ref{lemma:asmptotic CLT 3}, we obtain the following Berry-Esseen inequality
\begin{equation}
\label{ineq:BE_A2_marg_eps_z}
    \begin{aligned}
        \sup_{t\in\mathbb R}\left|\mathbb P\left(\frac{\bm \epsilon_z^\top \Z\X^\top \left(\X\bm\beta + \bm \epsilon \right)}{\sqrt{\sigma_{\epsilon_z}^2\y^\top \X\Z^\top \Z\X^\top \y}} < t\right) - 
        \Phi_{\epsilon_z}(t)\right| \leq O_p\left(\sqrt{\frac{\sum_{i=1}^{n_z}\left\{\Z\X^\top\left(\X\bm\beta+\bm \epsilon\right)_i^4\right\}}{\|\Z\X^\top\left(\X\bm\beta+\bm \epsilon\right)\|_2^4}}\right).
    \end{aligned}
\end{equation}
Denote
\begin{equation*}
    \begin{aligned}
        \sigma^2 = &n_zn^2\left[\frac{\sigma_{\bm\beta}^4n_z}{n}\left(2\gamma_2^2 + \gamma_1\gamma_3\right) + \frac{n_z}{n}\sigma_\epsilon^2\sigma_{\bm\beta}^2\gamma_1 + 2\sigma_{\bm\beta}^4\gamma_2^2 + \left(\sigma_{\epsilon_z}^2 + \sigma_{\bm\beta}^2\gamma_1\right)\left\{\sigma_{\bm\beta}^2\left(\gamma_3 + \omega_2\gamma_1\frac{p}{n}\right) + \sigma_\epsilon^2\omega_2\frac{p}{n}\right\}\right.\\
        &\left. + n_z\left\{\left(\mathbb E\left(\bm\beta^4\right) - 3\frac{\sigma_{\bm\beta}^4}{p^2}\right)\sum_{i=1}^m\left(\bm\Sigma^2\mathbb\I_m\right)_{i,i}^2 + \frac{2\sigma_{\bm\beta}^4}{p^2}\Tr\left\{\left(\bm\Sigma^2\mathbb\I_m\right)^2\right\}\right\}\right].
    \end{aligned}
\end{equation*}
Consider the following CDF
\begin{equation*}
    \begin{aligned}
        \bm H_{\text{M}}(t) = &\mathbb P\left(\frac{\bm\beta^\top \Z^\top \Z\X^\top(\X\bm\beta + \bm \epsilon) + \bm \epsilon_z^\top \Z\X^\top(\X\bm\beta + \bm \epsilon) - nn_z\sigma_{\bm\beta}^2\gamma_2}{\sigma} < t\right)\\
       = &\mathbb P\left(\frac{\bm \epsilon_z^\top \Z\X^\top(\X\bm\beta + \bm \epsilon)}{\sqrt{\sigma_{\epsilon_z}^2\y^\top \X\Z^\top \Z\X^\top \y}}\frac{\sqrt{\sigma_{\epsilon_z}^2\y^\top \X\Z^\top \Z\X^\top \y}}{\sigma} < t - \frac{\bm\beta^\top \Z^\top \Z\X^\top(\X\bm\beta + \bm \epsilon) - nn_z\sigma_{\bm\beta}^2\gamma_2}{\sigma}\right).
    \end{aligned}
\end{equation*}
regarding the randomness of $\bm \epsilon_z$, and by the Berry-Esseen bound we have in \cref{ineq:BE_A2_marg_eps_z}, denoted by $\Lambda_{\epsilon_z}$, the standard Gaussian variable depending only on the randomness of $\bm \epsilon_z$, we have
\begin{equation}
\label{ineq:BE_marg_A2_eps_z_app}
    \begin{aligned}
        &\sup_{t\in\mathbb R}\left|\bm H_{\text{M}}(t) - \mathbb P\left(\frac{\sqrt{\sigma_{\epsilon_z}^2\y^\top \X\Z^\top \Z\X^\top \y}}{\sigma}\Lambda_{\epsilon_z} < t - \frac{\bm\beta^\top \Z^\top \Z\X^\top(\X\bm\beta + \bm \epsilon) - nn_z\sigma_{\bm\beta}^2\gamma_2}{\sigma}\right)\right| \\
        &\leq O_p\left(\sqrt{\frac{\sum_{i=1}^{n_z}\left\{\Z\X^\top\left(\X\bm\beta+\bm \epsilon\right)_i^4\right\}}{\|\Z\X^\top\left(\X\bm\beta+\bm \epsilon\right)\|_2^4}}\right).
    \end{aligned}
\end{equation}

\subsection{Berry-Esseen bounds with the randomness of testing data matrix}
Conditional on $\bm\beta$, $\X_0$ and $\epsilon$, we now consider the randomness of $\Z_0$. By applying Lemma~S\ref{lemma:asmptotic CLT 1}, we have 
\begin{equation}
\label{ineq:BE_A2_marg_Z}
    \begin{aligned}
        \sup_{t\in\mathbb R}\left|\mathbb P\left(\frac{\bm\beta^\top \Z^\top \Z\X^\top(\X\bm\beta +\bm  \epsilon)-n_z\bm\beta^\top\bm\Sigma \X^\top (\X\bm\beta + \bm  \epsilon)}{\sigma_1} \leq t\right) - \Phi_{\Z_0}(t)\right| \leq O_p(n_z^{-1/2}),
    \end{aligned}
\end{equation}
where
\begin{equation*}
    \begin{aligned}
        &\sigma_1^2 = n_z\left\{\mathbb E\left(z_0^4-3\right)\sum_{i=1}^p\left(\bm\Sigma^{1/2}\bm\beta\right)_i^2\left\{\bm\Sigma^{1/2}\X^\top\left(\X\bm\beta+\epsilon\right)\right\}_i^2 + 2\left\{\bm\beta^\top\bm\Sigma \X^\top\left(\X\bm\beta + \bm \epsilon \right)\right\}^2 +\right.\\
        &\left.\qquad\|\bm\Sigma^{1/2}\bm\beta\|_2^2\|\bm\Sigma^{1/2} \X^\top\left(\X\bm\beta +\bm \epsilon\right)\|_2^2\right\}.
    \end{aligned}
\end{equation*}
Note that $\y^\top \X\Z^\top \Z\X^\top \y = \sum_i^{n_z}\|\z_{0_i}^\top\bm\Sigma^{1/2} \X^\top \y\|_2^2$.
Using Lemma~S\ref{lemma:von bahr-Esseen bound}, for $\forall \delta \in \left(0, 1/2\right)$ and $\forall \epsilon_1 > 0$, we have 
\begin{equation*}
    \begin{aligned}
        \mathbb P\left(\left|\y^\top \X\Z^\top \Z \X^\top \y - n_z\y^\top \X\bm \Sigma \X^\top \y\right| < n_z^{1/2+\delta}n^2\kappa_1\epsilon_1\right) \geq 1-\frac{C}{\epsilon_1^2 n_z^{2\delta}}.
    \end{aligned}
\end{equation*}
 Moreover, we have $\mathbb E_{\Z_0}\left(\y^\top \X\Z^\top \Z\X^\top \y\right) = O_p(n_zn^2\kappa_1)$ 
 and 
 $$\mathbb E_{\X_0}\left(\|\z_{0_i}^\top\bm\Sigma^{1/2} \X^\top \y\|_2^4\right) = \mathbb E\left(z_0^4-3\right)\sum_{j=1}^p\left(\bm\Sigma \X^\top \y\right)_j^4 + 3\|\bm\Sigma^{1/2}\X^\top \y\|_2^4 = O_p(n^4\kappa_1^2).$$
Similarly, note that $\mathbb E\left\{\left(\z_0^\top\bm\Sigma \X^\top \y\right)^8\right\} = O_p(n^8\kappa_1^4)$. For $\forall \delta \in \left(0,1/2\right)$ and $\forall \epsilon_2 > 0$, we have 
\begin{equation*}
    \begin{aligned}
        \mathbb P\left(\left|\sum_{i=1}^{n_z} \left(\z_{0_i}^\top\bm\Sigma^{1/2} \X^\top \y \right)^4 - n_z\mathbb E\left\{(\z_0^\top\bm\Sigma^{1/2} \X^\top \y)^4\right\}\right| < n_z^{\frac{1}{2} + \delta}n^4\kappa_1^2\epsilon_2\right) \geq 1-\frac{C}{\epsilon_2^2n_z^{2\delta}}.
    \end{aligned}
\end{equation*}
Consider the randomness of $\Z_0$, we define the subset $\Omega_1$ as following:
\begin{align*}
    \Omega_1(\epsilon_1, \epsilon_2) \coloneqq &\Bigg\{\Z_0:
    \left\{\left|\y^\top \X\Z^\top \Z \X^\top \y - n_z \y^\top \X\bm \Sigma \X^\top \y\right| < n_z^{\frac{1}{2}+\delta}n^2\kappa_1\epsilon_1\right\} \cap\\
    &\qquad\left\{\left|\sum_{i=1}^{n_z} (\z_{0_i}^\top\bm\Sigma \X^\top \y)^4 - n_z\mathbb E\left\{(\z_0^\top\bm\Sigma \X^\top \y)^4\right\}\right| < n_z^{\frac{1}{2} + \delta}n^4\kappa_1^2\epsilon_2\right\}\Bigg\}
\end{align*}
for some $\epsilon_1, \epsilon_2 > 0$. Notice that we have already shown that
\begin{equation*}
    \begin{aligned}
        \mathbb P(\Z_0 \in \Omega_1) \geq 1-O_p(n_z^{-2\delta})
    \end{aligned}
\end{equation*}
for $\forall \epsilon_1, \epsilon_2 > 0$.

Notice that if we rescale $\bm\beta$ to have a length of one, the higher moments of its entries cannot guarantee a smaller order. This hinders us from directly deriving its concentration using Markov's inequality. Alternatively, some random matrix tricks have been used to prove the desired concentrations. 
Proposition~S\ref{prop:quad_general_limit} is originally proposed in \cite{10.1214/aop/1022855421} and is refined by \cite{zhao2022estimating}. Our Proposition~S\ref{prop:quad_first_limit} adopts the refined version. Lemma~S\ref{lemma:neg_high_order} is designed to obtain the negligibility of higher moment terms in $\sigma_1^2$. It helps us eliminate all higher moment terms in our later analysis.

\begin{proposition.s} (Originally proposed in \citep{10.1214/aop/1022855421})
\label{prop:quad_general_limit}
    Let $\x = (x_1, \cdots, x_n)^\top$ be i.i.d. standardized entries and $\A$ be an $n\times n$ matrix. For any $p\geq 2$, we have 
    \begin{align*}
        \mathbb E\left|\x^\top \A \x - \Tr(\A)\right|^p \leq K_p\left[\left\{\mathbb E|x_1|^4\Tr(\A\A^\top)\right\}^{p/2} + \mathbb E|x_1|^{2p}\Tr(\A\A^\top)^{p/2}\right]. 
    \end{align*}
\end{proposition.s}
\begin{remark.s}
    The original version proposed in Lemma 2.7 of \cite{10.1214/aop/1022855421} is more general, as it also includes the matrices with complex entries. We only consider the real matrix in our setting for the simplicity of the usage.
\end{remark.s}
\begin{proposition.s}
(Originally proposed in \citep{zhao2022estimating}) 
\label{prop:quad_first_limit}
    For a random vector $\bm a \in \mathbb{R}^p$ with standardized entries, assume that $\bm a \sim \bm{F}(\bm 0, \bm\Sigma_a)$, where $\bm\Sigma_a$ is positive and semi-definite, and a deterministic positive semi-definite matrix $\A\in \mathbb{R}^{p\times p}$ where $\Tr(\A\bm\Sigma_a) \neq 0$. 
    Moreover, if $\mathbb E\left(|a_i|^4\right) \leq \upsilon_4$, then we have
    $$\mathbb P\left(\left|\frac{\bm a^\top \A\bm  a - \Tr(\A\bm\Sigma_a)}{\Tr(\A\bm\Sigma_a)}\right| \geq \epsilon\right) \leq \var\left(\frac{\bm a^\top \A\bm a}{\Tr(\A\bm\Sigma_a)}\right)/\epsilon^2,$$
    where $$\var\left(\frac{\bm a^\top \A\bm a}{\Tr(\A\bm\Sigma_a)}\right) \leq 2\cdot C_2\cdot\upsilon_4\cdot\frac{\Tr(\A\bm\Sigma_a \A^\top \bm\Sigma_a)}{\Tr(\A\bm\Sigma_a)^2}.$$
\end{proposition.s}

\begin{lem.s}
\label{lemma:neg_high_order} Under Assumption~\ref{a:Sigmabound}, for $\forall \delta \in (0, 1/2)$, we have
    $$\mathbb P\left(\left|(\bm\Sigma^{1/2}\bm\beta)_i^2 - \sum_{j=1}^p(\bm\Sigma_{i,j}^{1/2})^2\frac{\sigma_{\bm\beta}^2}{p}\mathbb\I_{m_{j,j}}\right| \geq \epsilon p^{-1/2+\delta}\right)\leq O_p(p^{-2\delta-1}).$$
    Notice that $\sigma_{\bm\beta}^2/p\bm\Sigma_{i,i} = O_p(p^{-1})$.
\end{lem.s}
\begin{proof}
    Using Chebyshev's inequality, we have 
    \begin{equation*}
        \begin{aligned}
            \mathbb P\left(\left|\left(\bm\Sigma^{1/2}\bm\beta\right)_i^2 - \sum_{j=1}^p\left(\bm\Sigma_{i,j}^{1/2}\right)^2\frac{\sigma_{\bm\beta}^2}{p}\mathbb\I_{m_{j,j}}\right| > \epsilon p^{-1/2+\delta}\right) &\leq \frac{\mathbb E\left\{\left|\left(\bm\Sigma^{1/2}\bm\beta \right)_i^2 - \sum_{j=1}^p \left(\bm\Sigma_{i,j}^{1/2}\right)^2\bm\beta_j^2\mathbb\I_{m_{j,j}}\right|^2\right\}}{\epsilon^2p^{-1+2\delta}}\\ &\leq \frac{M}{\epsilon^2p^{-1+2\delta}}\left[\mathbb E\left\{\left(\bm\Sigma^{1/2}\bm\beta\right)_i^4\right\} + \frac{\sigma_{\bm\beta}^4}{p^2}\bm\Sigma_{i,i}^2\right].
        \end{aligned}
    \end{equation*}
    Notice that
    \begin{equation*}
        \begin{aligned}
            \mathbb E\left\{\left(\bm\Sigma^{1/2}\bm\beta\right)_i^4\right\}& = \sum_{j=1}^p\mathbb E\left(\bm\beta^4 - \frac{3\sigma_{\bm\beta}^4}{p^2}\right)\left(\bm\Sigma^{1/2}_{i,j}\right)^4\mathbb\I_{m_{j,j}} + \frac{3\sigma_{\bm\beta}^4}{p^2}\sum_{j=1}^p\left\{\left(\bm\Sigma_{i,j}^{1/2}\right)^2\mathbb\I_{m_{j,j}}\right\}\sum_{k=1}^p\left\{\left(\bm\Sigma_{i,k}^{1/2}\right)^2\mathbb\I_{m_{k,k}}\right\}\\
            &\leq \sum_{j=1}^p\mathbb\mathbb E\left(\bm\beta^4 - \frac{3\sigma_{\bm\beta}^4}{p^2}\right)\left(\bm\Sigma^{1/2}_{i,j}\right)^4\mathbb\I_{m_{j,j}} + \frac{3\sigma_{\bm\beta}^4}{p^2}\bm\Sigma_{i,i}^2.
        \end{aligned}
    \end{equation*}
    Therefore, we have 
    $\mathbb E\left\{\left(\bm\Sigma^{1/2}\bm\beta\right)_i^4\right\}= O_p(p^{-2})$, 
    which implies 
    $$\mathbb P\left(\left|\left(\bm\Sigma^{1/2}\bm\beta\right)_i^2 - \sum_{j=1}^p\left(\bm\Sigma_{i,j}^{1/2}\right)^2\sigma_{\bm\beta}^2\mathbb\I_{m_{j,j}}\right| > \epsilon p^{-1/2+\delta}\right) \leq O_p(p^{-1-2\delta}).$$
\end{proof}

\begin{lem.s}
    Following Proposition~S\ref{prop:quad_first_limit}, for $\forall \delta \in (0,1/2)$, we have
    \begin{equation*}
        \begin{aligned}
            \mathbb P\left(\left|\frac{p}{\sigma_{\bm\beta}^2}\bm\beta^\top\bm \Psi\bm\beta - \Tr\left(\bm \Psi\mathbb I_m\right)\right| < p^{\delta + 1/2}\epsilon_1\right) \geq 1-O_p\left(p^{-1-2\delta}\Tr\left(\bm \Psi\mathbb\I_m\bm \Psi\mathbb\I_m\right)\right) \geq 1- O_p(p^{-2\delta}). 
        \end{aligned}
    \end{equation*}
    for any deterministic positive semi-definite matrix $\bm \Psi$ with bounded eigenvalues.
\end{lem.s}
{By Boole's inequality, for a countable set of events $A_1, A_2, \cdots $, we have 
$$\mathbb P\left(\bigcup_{i=1}^{\infty}A_i\right) \leq \sum_{i=1}^{\infty}\mathbb P(A_i).$$
}
Then using Lemma~S\ref{lemma:neg_high_order} {and this union bound},
we have 
    \begin{equation*}
        \begin{aligned}
           \mathbb E\left(z_0^4-3\right)\sum_{i=1}^p\left(\bm\Sigma^{1/2}\bm\beta \right)_i^2\left\{\bm\Sigma^{1/2}\X^\top\left(\X\bm\beta+\epsilon\right)\right\}_i^2 &= O_p(p^{-1/2+\delta})\|\bm\Sigma^{1/2}\X^\top\left(\X\bm\beta + \bm \epsilon \right)\|_2^2 \\
           &\prec \|\bm\Sigma^{1/2}\X^\top\left(\X\bm\beta + \bm \epsilon \right)\|_2^2.
        \end{aligned}
    \end{equation*}
Therefore, with probability of at least $1-O_p(p^{-2\delta})$, the first term of $\sigma_1^2$ in \cref{ineq:BE_A2_marg_Z} is negligible. As we rescale the variance of equation~\cref{ineq:BE_A2_marg_Z} to order one, this term can be omitted. We will not need to further discuss this term in our proof.
Furthermore, for $\forall \Z_0 \in \Omega_1$, we have
\begin{equation}
\label{concent:BE_bound_A2_marg_Z}
    \begin{aligned}
        &O_p\left(\sqrt{\frac{\sum_{i=1}^{n_z}\left\{\Z\X^\top\left(\X\bm\beta+\bm \epsilon \right)\right\}_i^4}{\|\Z\X^\top\left(\X\bm\beta+\bm  \epsilon\right)\|_2^4}}\right) \\
        &= O_p\left(\sqrt{\frac{n_z \left\{\mathbb E\left(z_0^4-3 \right)\sum_{j=1}^p\left(\bm\Sigma \X^\top \y \right)_j^4 + 3\|\bm\Sigma^{1/2}\X^\top \y\|_2^4\right\} + O_p\left(n_z^{1/2+\delta}n^4\kappa_1^2\right)}{n_z^2\|\bm\Sigma^{1/2}\X^\top \y\|_2^4 + O_p\left(n^4n_z^{3/2+\delta} \kappa_1^2\right)}}\right).
    \end{aligned}
\end{equation}
Recall that we have shown $\mathbb P(\Z_0 \in \Omega_1) \geq 1-O_p(n_z^{-2\delta})$ for $\forall \delta \in (0,1/2)$.
Therefore, consider $\Z_0 \in \Omega_1$, the second probability in \cref{ineq:BE_marg_A2_eps_z_app} can be further replaced by
\begin{equation*}
    \begin{aligned}
        &\mathbb P\left(\frac{\bm\beta^\top \Z^\top \Z\X^\top(\X\bm\beta + \epsilon) - n_z\bm\beta^\top\bm\Sigma \X^\top(\X\bm\beta + \epsilon)}{\sigma_1}\frac{\sigma_1}{\sigma}  \right.\\
        & \qquad < t- \frac{\sqrt{\sigma_{\epsilon_z}^2n_z\y^\top \X\bm\Sigma \X^\top \y + O_p(n_z^{1/2 + \delta}n^2)}}{\sigma}\Lambda_{\epsilon_z} -
        \left.\frac{n_z\bm\beta^\top\bm\Sigma \X^\top (\X\bm\beta + \epsilon) - nn_z\sigma_{\bm\beta}^2\gamma_2}{\sigma}\right),
    \end{aligned}
\end{equation*}
where
\begin{equation*}
    \begin{aligned}
        \sigma_1^2 = &n_z\Bigg[\mathbb E\left(z_0^4-3\right)\sum_{i=1}^p\left(\bm\Sigma^{1/2}\bm\beta\right)_i^2\left\{\bm\Sigma^{1/2}\X^\top\left(\X\bm\beta+\bm \epsilon\right)\right\}_i^2 \\
        &\qquad+ 2\left\{\bm\beta^\top\bm\Sigma \X^\top\left(\X\bm\beta +\bm \epsilon\right)\right\}^2 +\|\bm\Sigma^{1/2}\bm\beta\|_2^2\|\bm\Sigma^{1/2} \X^\top\left(\X\bm\beta + \bm \epsilon\right)\|_2^2\Bigg].
    \end{aligned}
\end{equation*}
Using \cref{ineq:BE_A2_marg_Z}, we have 
\begin{equation}
\label{ineq:BE_A2_marg_Z_app}
    \begin{aligned}
        &\sup_{t\in\mathbb R} \left|\mathbb P\left(\frac{\bm\beta^\top \Z^\top \Z\X^\top\left(\X\bm\beta + \bm \epsilon\right) - n_z\bm\beta^\top\bm\Sigma \X^\top\left(\X\bm\beta +  \bm \epsilon\right)}{\sigma} 
        < t - \frac{\sqrt{\sigma_{\epsilon_z}^2n_z\y^\top \X\Z^\top \Z \X^\top \y}}{\sigma}\Lambda_{\epsilon_z} -\right.\right.\\
        &\left.\left.\frac{n_z\bm\beta^\top\bm\Sigma \X^\top \left(\X\bm\beta + \bm  \epsilon\right) - nn_z\sigma_{\bm\beta}^2\gamma_2}{\sigma}\right) - \mathbb P\left(\frac{\sigma_1}{\sigma}\Lambda_{\Z_0} < t - \frac{\sqrt{\sigma_{\epsilon_z}^2n_z\y^\top \X\bm\Sigma \X^\top \y + O_p\left(n_z^{1/2 + \delta}n^2\right)}}{\sigma}\Lambda_{\epsilon_z}\right.\right.\\
        &\left.\left.-\frac{n_z\bm\beta^\top\bm\Sigma \X^\top \left(\X\bm\beta +  \bm \epsilon \right) - nn_z\sigma_{\bm\beta}^2\gamma_2}{\sigma}\right)\right| \leq O_p\left(\max\left(n_z^{-1/2}, n_z^{-2\delta}\right)\right).
    \end{aligned}
\end{equation}

\subsection{Berry-Esseen bounds with the randomness of training error}
Conditional on $\X_0$ and $\bm\beta$, we now focus on the concentration over the randomness of $\bm \epsilon$ and the corresponding Gaussian generated by this randomness. 
By Markov's inequality, for $\forall \epsilon_1 > 0$ 
and  
$0 < \delta < 1/2$, we have 
\begin{equation*}
    \begin{aligned}
\mathbb P\left(\left|\bm \epsilon^\top \X\bm\Sigma \X^\top \X\bm\beta\right| < n^{3/2}m^\delta\epsilon_1\right) &\geq 1-\frac{\sigma_\epsilon^2\bm\beta^\top \X^\top \X\bm\Sigma \X^\top \X\bm\Sigma \X^\top \X\bm\beta}{n^{3}m^{2\delta}\epsilon_1^2} \\
&= 1-O_p\left(\frac{m^{1-2\delta}}{p}\right) \geq 1-O_p(m^{-2\delta}).
    \end{aligned}
\end{equation*}
Therefore, consider the subset
\begin{equation*}
    \begin{aligned}
        \Omega_2(\epsilon_1) \coloneqq \left\{\bm \epsilon: \bigl|\|\bm\Sigma^{1/2}\X^\top \y\|_2^2 - \|\bm\Sigma^{1/2} \X^\top \X\bm\beta\|_2^2 - \sigma_\epsilon^2\Tr(\X\bm\Sigma \X^\top)\bigr| \leq n^{3/2}m^\delta\epsilon_1\right\}.
    \end{aligned}
\end{equation*}
Then for $\forall \delta \in (0, 1/2), \epsilon_1 > 0$, we have $\mathbb P(\epsilon \in \Omega_2) \geq 1-O_p(m^{-2\delta})$.
Therefore, for $\forall \bm \epsilon \in \Omega_1$, we can simplify the Berry-Esseen bound in \cref{concent:BE_bound_A2_marg_Z} as follows 
\begin{equation}
\label{concent:BE_bound_A2_marg_eps}
    \begin{aligned}
    &O_p\left(\sqrt{\frac{n_z \left\{\mathbb E\left(z_0^4-3 \right)\sum_{j=1}^p\left(\bm\Sigma \X^\top \y \right)_j^4 + 3\|\bm\Sigma^{1/2}\X^\top \y\|_2^4\right\} + O_p\left(n_z^{1/2+\delta}n^4\kappa_1^2\right)}{n_z^2\|\bm\Sigma^{1/2}\X^\top \y\|_2^4 + O_p\left(n^4n_z^{3/2+\delta} \kappa_1^2\right)}}\right)\\
    \Leftrightarrow &O_p\left(\sqrt{\frac{n_zn^4\kappa_1^2}{n_z^2\left\{\|\bm\Sigma^{1/2}\X^\top \X\bm\beta\|_2^2 + \sigma_\epsilon^2\Tr\left(\X\bm\Sigma \X^\top\right)^2\right\}^2 + O_p\left(n^4n_z^{3/2+\delta}\kappa_1^2 + n_z^2n^{7/2}m^{\delta}\kappa_1^2\right)}}\right).
    \end{aligned}
\end{equation}
Now with the randomness of $\bm \epsilon$, by applying Lemma~S\ref{lemma:non-asmptotic CLT 3}, we have the following Berry-Esseen bound
\begin{equation}
\label{ineq:BE_A2_marg_eps}
    \begin{aligned}
        \sup_{t\in\mathbb R}\left|\mathbb P\left(\frac{\bm \epsilon^\top \X_0\bm\Sigma^{3/2}\bm\beta}{\sigma_3} < t\right) -\Phi_{\epsilon}(t)\right| \leq O_p\left(\sqrt{\frac{\mathbb E\left(\epsilon^4\right)\sum_{i=1}^n \left(\x_{0_i}^\top\bm\Sigma^{3/2}\bm\beta \right)^4}{\sigma_\epsilon^4\bm\beta^\top\bm\Sigma^{3/2}\X_0^\top \X_0\bm\Sigma^{3/2}\bm\beta}}\right),
    \end{aligned}
\end{equation}
where $\sigma_3^2 = \sigma_\epsilon^2\|\X_0\bm\Sigma^{3/2}\bm\beta\|_2^2.$

{Next, we analyze the non-asymptotic behavior of $\sigma_1$, which is defined in \cref{ineq:BE_A2_marg_Z}.}   
Notice that the first term in $\sigma_1$ is of lower order than the last one. Thus, we will omit it in the analysis from now on. Using Markov's inequality, we have 
\begin{equation*}
    \begin{aligned}
        \mathbb P&\left(\left|\bm\beta^\top\bm\Sigma \X^\top \X\bm\beta\bm\beta^\top\bm\Sigma \X^\top \bm \epsilon\right| < n^{3/2}m^{1+\delta}p^{-1}\epsilon_3\right) \\
        & \qquad \geq 1-\frac{\sigma_\epsilon^2\bm\beta^\top\bm\Sigma \X^\top \X\bm\beta\bm\beta^\top\bm\Sigma \X^\top \X\bm\Sigma\bm\beta\bm\beta^\top \X^\top \X\bm\Sigma\bm\beta}{n^3m^{2 + 2\delta}p^{-2}\epsilon_3^2}\\
        & \qquad \geq 1-O_p\left(\frac{(nm/p)^3}{n^3m^{2+2\delta}p^{-2}}\right) \geq 1-O_p(m^{-2\delta}).
    \end{aligned}
\end{equation*}
We further restrict our choice of $\bm \epsilon$ to a smaller subset where $$\Omega_3(\epsilon_3) \coloneqq \left\{\bm \epsilon:\left|\bm\beta^\top\bm\Sigma \X^\top \X\bm\beta\bm\beta^\top\bm\Sigma \X^\top\bm  \epsilon\right| < n^{3/2}m^{1+\delta}p^{-1}\epsilon_3\right\}.$$
For $\forall \bm \epsilon \in \Omega_2\cap\Omega_3$, we have
\begin{equation*}
    \begin{aligned}
        \sigma_1^2 = &n_z\left[2\left(\bm\beta^\top\bm\Sigma \X^\top \X\bm\beta\right)^2 + 2\sigma_\epsilon^2\bm\beta^\top\bm\Sigma \X^\top \X\bm\Sigma\bm\beta + \right.\\
        &\qquad\left.\|\bm\Sigma^{1/2}\bm\beta\|_2^2\left\{\|\bm\Sigma^{1/2}\X^\top \X\bm\beta\|_2^2 + \sigma_\epsilon^2\Tr\left(\X\bm\Sigma \X^\top\right)\right\} +
        O_p\left(n^{3/2}m^{1+\delta}p^{-1}\right)\right].
    \end{aligned}
\end{equation*}
Moreover, we have $\mathbb P\left(\epsilon \in \Omega_2\cap\Omega_3\right) \geq 1-O_p\left(m^{-2\delta}\right)$.
For simplicity of the proof in this section, we will denote
\begin{equation*}
    \begin{aligned}
        \tilde{\alpha}_1 = \bm\beta^\top\bm\Sigma\X^\top\X\bm\Sigma\bm\beta,\quad \tilde{\alpha}_2 = \|\bm\Sigma^{1/2}\X^\top\X\bm\beta\|_2^2,\quad \mbox{and}\quad \tilde{\alpha}_3 = \bm\beta^\top\bm\Sigma\X^\top\X\bm\beta.
    \end{aligned}
\end{equation*}
Rearranging the terms in the second probability of \cref{ineq:BE_A2_marg_Z_app} and considering $\bm \epsilon \in \Omega_2\cap\Omega_3$, we have the following simplified form
\begin{equation*}
    \begin{aligned}
        \Leftrightarrow\mathbb P&\left(\frac{n_z\bm\beta^\top\bm\Sigma \X^\top\bm \epsilon}{\sigma_3}\frac{\sigma_3}{\sigma}< t - \frac{\sqrt{\sigma_{\epsilon_z}^2n_z\left\{\tilde{\alpha}_2 + \sigma_\epsilon^2\Tr(\X\bm\Sigma \X^\top) + O_p(n^{3/2}m^\delta)\right\}}}{\sigma}\Lambda_{\epsilon_z} - \right.\\
        &\left.\frac{\sqrt{n_z\left[2\tilde{\alpha}_3^2 + 2\sigma_\epsilon^2\tilde{\alpha}_1 + \|\bm\Sigma^{1/2}\bm\beta\|_2^2\left\{\tilde{\alpha}_2 + \sigma_\epsilon^2\Tr\left(\X\bm\Sigma \X^\top\right)\right\} + O_p(n^{3/2}m^{1+\delta}p^{-1})\right]}}{\sigma}\Lambda_{\Z_{0}}\right.\\
        &\left.- \frac{n_z\bm\beta\bm\Sigma \X^\top \X\bm\beta - nn_z\sigma_{\bm\beta}^2\gamma_2}{\sigma}\right).
    \end{aligned}
\end{equation*}

By \cref{ineq:BE_A2_marg_eps}, we have the following inequality
\begin{equation}
\label{ineq:BE_marg_A2_eps_app}
    \begin{aligned}
        \sup_{t\in\mathbb R}&\left|\mathbb P\left(\frac{n_z\bm\beta^\top\bm\Sigma \X^\top \bm \epsilon}{\sigma_3}\frac{\sigma_3}{\sigma} < t - \frac{\sqrt{\sigma_{\epsilon_z}^2n_z\y^\top \X\bm\Sigma \X^\top \y + O_p(n_z^{1/2 + \delta}n^2)}}{\sigma}\Lambda_{\epsilon_z}- \frac{\sigma_1}{\sigma}\Lambda_{\Z_0} - \right.\right.\\
        &\left.\left.\frac{n_z\tilde{\alpha}_3 - nn_z\sigma_{\bm\beta}^2\gamma_2}{\sigma}\right) - \mathbb P\left(\frac{n_z\sigma_3}{\sigma}\Lambda_{\epsilon} < t - \frac{\sqrt{\sigma_{\epsilon_z}^2n_z\left\{\tilde{\alpha}_2 + \sigma_\epsilon^2\Tr(\X\bm\Sigma \X^\top) + O_p(n^{3/2}m^\delta)\right\}}}{\sigma}\Lambda_{\epsilon_z} \right.\right.\\
        &\left.\left. -\frac{\sqrt{n_z\left[2\tilde{\alpha}_3^2 + 2\sigma_\epsilon^2\tilde{\alpha}_1 + \|\bm\Sigma^{1/2}\bm\beta\|_2^2\left\{\tilde{\alpha}_2 + \sigma_\epsilon^2\Tr(\X\bm\Sigma \X^\top)\right\} + O_p(n^{3/2}m^{1+\delta}p^{-1})\right]}}{\sigma}\Lambda_{\Z_{0}}\right.\right.\\
        &\left.\left.- \frac{n_z\tilde{\alpha}_3 - nn_z\sigma_{\bm\beta}^2\gamma_2}{\sigma}\right)\right| \leq O_p\left(\max\left(\sqrt{\frac{\sum_{i=1}^n \left(\x_{0_i}^\top\bm\Sigma^{3/2}\bm\beta \right)^4}{\bm\beta^\top\bm\Sigma^{3/2}\X_0^\top \X_0\bm\Sigma^{3/2}\bm\beta}}, m^{-2\delta}\right)\right).
    \end{aligned}
\end{equation}

\subsection{Berry-Esseen bounds with the randomness of training data matrix}
Conditional on $\bm\beta$, we now focus on the concentration over the randomness of $\X_0$ and the corresponding Gaussian generated by its randomness. By Lemma~S\ref{lemma:non-asmptotic CLT 1}, we have 
\begin{equation*}
    \begin{aligned}
    \sup_{t\in\mathbb R}\left|\mathbb P\left(\frac{\bm\beta^\top\bm\Sigma^{3/2}\X_0^\top \X_0\bm\Sigma^{1/2}\bm\beta - n\bm\beta^\top\bm\Sigma^2\bm\beta}{\sigma_2} < t\right) - \Phi_{\X_0}(t)\right| \leq O_p(n^{-1/2}),
    \end{aligned}
\end{equation*}
where $\sigma_2^2 = n\left\{\mathbb E\left(x_0^4-3\right)\sum_{i=1}^p\left(\bm\Sigma^{3/2}\bm\beta\right)_i^2\left(\bm\Sigma^{1/2}\bm\beta\right)_i^2 + 2\left(\bm\beta^\top\bm\Sigma^2\bm\beta\right)^2 + \|\bm\Sigma^{3/2}\bm\beta\|_2^2\|\bm\Sigma^{1/2}\bm\beta\|_2^2\right\}$.
By Lemma~S\ref{lemma:neg_high_order}, the first term of $\sigma_2^2$ is negligible. 
Moreover, by Lemma~S\ref{lemma:von bahr-Esseen bound} and considering the randomness of $\X_0$, for $\forall \epsilon_1, \epsilon_2 > 0$ and $0<\delta<1/2$, we have
\begin{equation*}
    \begin{aligned}
        &\mathbb P\left(\left|\sum_{i=1}^n\left(x_{0_i}^\top\bm\Sigma^{1/2}\bm\beta\right)^4 - n\left\{\mathbb E\left(x_0^4-3\right)\sum_{j=1}^p\left(\bm\Sigma^{1/2}\bm\beta\right)_j^4 + 3\left(\bm\beta^\top\bm\Sigma\bm\beta\right)^2\right\}\right| \leq \left(\frac{m}{p}\right)^2n^{1/2 + \delta}\epsilon\right)\\ &\geq 1-O_p(n^{-2\delta})\\
        &\mbox{and} \quad \mathbb P\left(\left|\bm\beta^\top\bm\Sigma^{1/2}\X_0^\top \X_0\bm\Sigma^{1/2}\bm\beta - n\bm\beta^\top\bm\Sigma\bm\beta\right| < \frac{m}{p}n^{1/2 + \delta}\epsilon\right) \geq 1-O_p(n^{-2\delta}).
    \end{aligned}
\end{equation*}
For $\forall \epsilon_1,\epsilon_2,\cdots,\epsilon_6 >0$, we further define the subset  
\begin{equation*}
    \begin{aligned}
        &\Omega_4(\epsilon_1, \epsilon_2,\cdots,\epsilon_6) \coloneqq \\
        &\left\{\X_0: \left\{\left|\sum_{i=1}^n\left(x_{0_i}^\top\bm\Sigma^{1/2}\bm\beta\right)^4 - n\left\{\mathbb E\left(x_0^4-3\right)\sum_{j=1}^p\left(\bm\Sigma^{1/2}\bm\beta\right)_j^4 + 3\left(\bm\beta^\top\bm\Sigma\bm\beta\right)^2\right\}\right| < \left(\frac{m}{p}\right)^2n^{1/2 + \delta}\epsilon_1\right\}\cap\right.\\   
        &\left.\left\{\left|\bm\beta^\top\bm\Sigma^{1/2}\X_0^\top \X_0\bm\Sigma^{1/2}\bm\beta - n\bm\beta^\top\bm\Sigma\bm\beta\right| < \frac{m}{p}n^{1/2 + \delta}\epsilon_2\right\}\cap\left\{\left|\Tr\left(\X\bm\Sigma \X^\top\right) - n\Tr\left(\bm\Sigma^2\right)\right| < n^{1/2+\delta}p\epsilon_3 \right\}\right.\\
        &\left. \cap\left\{\left|\tilde{\alpha}_1 - n\bm\beta^\top\bm\Sigma^3\bm\beta\right| < n^{1/2+\delta}\frac{m}{p}\epsilon_4\right\}\cap\left\{\left|\tilde{\alpha}_2 - n\left(n\bm\beta^\top\bm\Sigma^3\bm\beta + \Tr(\bm\Sigma^2)\bm\beta^\top\bm\Sigma\bm\beta\right)\right| < n^{3/2 + \delta}\frac{m}{p}\epsilon_5\right\}\cap\right.\\
        &\left.\left\{\left|\tilde{\alpha}_3 - n\bm\beta^\top\bm\Sigma^2\bm\beta\right| < n^{1/2+\delta}\frac{m}{p}\epsilon_6\right\}\right\}.
    \end{aligned}
\end{equation*}
Note that for $\mathbb E_{\X_0}\|\bm\Sigma^{1/2}\X^\top \X\bm\beta\|_2^4$, the term with the highest order is
\begin{equation*}
    \begin{aligned}
        \left\{n\left(n\bm\beta\bm\Sigma^3\bm\beta + \Tr(\bm\Sigma^2)\bm\beta^\top\bm\Sigma\bm\beta\right)\right\}^2 = O_p\left\{\frac{n^2m}{p}\max(n,p)\right\}.
    \end{aligned}
\end{equation*}
After careful calculation using Markov's inequality and Lemma~S\ref{lemma:von bahr-Esseen bound}, we have 
$$\mathbb P\left(\X_0\in\Omega_4\right) \geq 1-O_p\left(n^{-2\delta}\right).$$
Furthermore, for $\X_0 \in \Omega_4$, now we consider the concentration of our previous Berry-Esseen bound in \cref{ineq:BE_marg_A2_eps_app} 
\begin{equation*}
    \begin{aligned}
         &O_p\left(\max\left(\sqrt{\frac{\sum_{i=1}^n \left(\x_{0_i}^\top\bm\Sigma^{3/2}\bm\beta\right)^4}{\bm\beta^\top\bm\Sigma^{3/2}\X_0^\top \X_0\bm\Sigma^{3/2}\bm\beta}}, m^{-2\delta}\right)\right)\\
         &= O_p\left(\max\left(\sqrt{\frac{\mathbb E\left(\epsilon^4\right)\left[n\left\{\mathbb E\left(x_0^4-3\right)\sum_{j=1}^p\left(\bm\Sigma^{3/2}\bm\beta\right)_j^4 + 3\bm\beta^\top\bm\Sigma\bm\beta\right\} + O_p\left((m/p)^2n^{1/2 + \delta}\right)\right]}{\sigma_\epsilon^4n^2\left(\bm\beta^\top\bm\Sigma^3\bm\beta\right)^2 + O_p\left((m/p)^2n^{3/2 + \delta}\right)}}, m^{-2\delta}\right)\right).
    \end{aligned}
\end{equation*}
Moreover, our previous Berry-Esseen upper bound in \cref{concent:BE_bound_A2_marg_eps} can also be simplified as 
\begin{equation*}
    \begin{aligned}
    &O_p\left(\sqrt{\frac{n_zn^4\kappa_1^2}{n_z^2\left(\tilde{\alpha}_2 + \sigma_\epsilon^2\Tr\left(\X\bm\Sigma \X^\top\right)^2\right)^2 + O_p\left(n^4n_z^{3/2+\delta}\kappa_1^2 + n_z^2n^{7/2}m^{\delta}\kappa_1^2\right)}}\right)\\
    &=O_p\left(\left[\frac{n_zn^4\kappa_1^2}{\splitfrac{n_z^2\left\{n\left(n\bm\beta^\top\bm\Sigma^3\bm\beta + \Tr(\bm\Sigma^2)\bm\beta^\top\bm\Sigma\bm\beta\right)\right\}^2 + \sigma_\epsilon^4n^2n_z^2\Tr(\bm\Sigma)^2}
    {+ O_p\left((n^4n_z^{3/2+\delta} + n_z^2n^{7/2}m^{\delta})\kappa_1^2 + \left\{n^{3/2+\delta}p^2 + (m/p)^2n^{7/2+\delta}(1+p/n)\right\}n_z^2\right)}}\right]^{1/2}\right).
    \end{aligned}
\end{equation*}
Now we rearrange the second probability in \cref{ineq:BE_marg_A2_eps_app} to have the last Gaussian generated from the randomness of $\X_0$
\begin{equation*}
    \begin{aligned}
        &\mathbb P\left(\frac{n_z\tilde{\alpha}_3 - nn_z\sigma_{\bm\beta}^2\gamma_2}{\sigma_2}\frac{\sigma_2}{\sigma} < t - \frac{n_z\sigma_3}{\sigma}\Lambda_\epsilon -\frac{\sqrt{\sigma_{\epsilon_z}^2n_z\left(\tilde{\alpha}_2 + \sigma_\epsilon^2\Tr(\X\bm\Sigma \X^\top) + O_p(n^{3/2}m^\delta)\right)}}{\sigma}\Lambda_{\epsilon_z}-\right.\\
        &\left.\frac{\sqrt{n_z\left\{2\tilde{\alpha}_3^2 + 2\sigma_\epsilon^2\tilde{\alpha}_1 + \|\bm\Sigma^{1/2}\bm\beta\|_2^2\left(\tilde{\alpha}_2 + \sigma_\epsilon^2\Tr(\X\bm\Sigma \X^\top)\right) + O_p(n^{3/2}m^{1+\delta}p^{-1})\right\}}}{\sigma}\Lambda_{\Z_{0}}\right).
    \end{aligned}
\end{equation*}
Using the Berry-Esseen bound in \cref{ineq:BE_A2_marg_eps}, we have 
\begin{equation}
\label{ineq:BE_marg_A2_X_app}
    \begin{aligned}
        &\sup_{t\in\mathbb R}\left|\mathbb P\left(\frac{n_z\tilde{\alpha}_3 - nn_z\sigma_{\bm\beta}^2\gamma_2}{\sigma_2}\frac{\sigma_2}{\sigma} < t - \frac{n_z\sigma_3}{\sigma}\Lambda_\epsilon -\frac{\sqrt{\sigma_{\epsilon_z}^2n_z\left(\tilde{\alpha}_2 + \sigma_\epsilon^2\Tr(\X\bm\Sigma \X^\top) + O_p(n^{3/2}m^\delta)\right)}}{\sigma}\Lambda_{\epsilon_z}\right.\right.\\
        &\left.\left.-\frac{\sqrt{n_z\left\{2\tilde{\alpha}_3^2 + 2\sigma_\epsilon^2\tilde{\alpha}_1 + \|\bm\Sigma^{1/2}\bm\beta\|_2^2\left(\tilde{\alpha}_2 + \sigma_\epsilon^2\Tr(\X\bm\Sigma \X^\top)\right) + O_p(n^{3/2}m^{1+\delta}p^{-1})\right\}}}{\sigma}\Lambda_{\Z_{0}}\right)\right.\\
        &\left.-\mathbb P\left(n_z\frac{\sigma_2}{\sigma}\Lambda_{\X_0} + n_z\frac{n\bm\beta^\top\bm\Sigma^{2}\bm\beta-n\sigma_{\bm\beta}^2\gamma_2}{\sigma} < t - \frac{n_z\sqrt{n\sigma_\epsilon^2\bm\beta^\top\bm\Sigma\bm\beta}}{\sigma}\Lambda_\epsilon - \right.\right.\\
        &\left.\left.\frac{\sqrt{\sigma_{\epsilon_z}^2n_z\left\{n\left\{n\bm\beta^\top\bm\Sigma^3\bm\beta + \Tr(\bm\Sigma^2)\bm\beta^\top\bm\Sigma\bm\beta\right\} + n\sigma_\epsilon^2\Tr(\bm\Sigma^2) + O_p(n^{3/2+\delta}(m/p) + n^{3/2}m^{\delta})\right\}}}{\sigma}\Lambda_{\epsilon_z}-\right.\right.\\
        &\left.\left.\frac{\sqrt{\splitfrac{n_z\Big\{2(n\bm\beta^\top\bm\Sigma^2\bm\beta)^2 + 2n\sigma_\epsilon^2\bm\beta^\top\bm\Sigma^3\bm\beta + \|\bm\Sigma^{1/2}\bm\beta\|_2^2\left\{n\left\{n\bm\beta^\top\bm\Sigma^3\bm\beta + \Tr(\bm\Sigma^2)\bm\beta^\top\bm\Sigma\bm\beta\right\}\right\} }{
        + n\sigma_\epsilon^2\Tr(\bm\Sigma^2)+ O_p\left(n^{3/2}mp^{-1}(m^\delta + n^\delta\kappa_1)\right) } \Big\}}}{\sigma}\Lambda_{\Z_{0}}\right)\right|\\
        &\leq O_p\left(\max(n^{-1/2}, n^{-2\delta})\right).
    \end{aligned}
\end{equation}

\subsection{Berry-Esseen bounds with the randomness of genetic effects}
In this section, we take the randomness of $\bm\beta$ into account. We first define the subset where the desired concentration holds. 
For $\forall\epsilon_1, \epsilon_2,\epsilon_3 > 0$ and $0<\delta<1/2$, consider the subset
\begin{equation*}
    \begin{aligned}
        \Omega_5(\epsilon_1, \epsilon_2, \epsilon_3) \coloneqq \Bigg\{\bm\beta:\left\{\left|\bm\beta^\top\bm\Sigma^3\bm\beta - \sigma_{\bm\beta}^2\gamma_3\right| < m^{1/2+\delta}p^{-1}\epsilon_1\right\}
        &\cap\left\{\left|\bm\beta^\top\bm\Sigma^2\bm\beta-\sigma_{\bm\beta}^2\gamma_2\right| < m^{1/2+\delta}p^{-1}\epsilon_2\right\}\\
        &\cap\left\{\left|\bm\beta^\top\bm\Sigma\bm\beta - \sigma_{\bm\beta}^2\gamma_1\right| < m^{1/2+\delta}p^{-1}\epsilon_3\right\}\Bigg\}.
    \end{aligned}
\end{equation*}
For the simplicity of our proof, we introduce some temporary quantities related to $\bm\beta$ in this section
\begin{equation*}
    \begin{aligned}
        \tilde{\theta}_i \coloneqq \bm\beta^\top\bm\Sigma^{i}\bm\beta \quad \text{for } i = 1,2,3.
    \end{aligned}
\end{equation*}
By Lemma~S\ref{prop:quad_first_limit}, we have 
\begin{equation*}
    \begin{aligned}
        \mathbb P(\Omega_5) \geq 1-O_p(m^{-2\delta}). 
    \end{aligned}
\end{equation*}
Therefore, for $\forall \bm\beta \in \Omega_5$, we have
\begin{equation*}
    \begin{aligned}
        O_p\left(\sqrt{\frac{n\left\{\mathbb E\left(x_0^4-3\right)\sum_{j=1}^p(\bm\Sigma^{3/2}\bm\beta)_j^4 + 3\tilde{\theta}_1\right\} + O_p\left((m/p)^2n^{1/2 + \delta}\right)}{\sigma_\epsilon^4n^2\tilde{\theta}_3^2+ O_p\left((m/p)^2n^{3/2 + \delta}\right)}}\right) = O_p(n^{-1/2})
    \end{aligned}
\end{equation*}
and 
\begin{equation*}
    \begin{aligned}
        &O_p\left(\left[\frac{n_zn^4\kappa_1^2}{\splitfrac{n_z^2\left\{n\left(n\tilde{\theta}_3 + \Tr(\bm\Sigma^2)\tilde{\theta}_1\right)\right\}^2 + \sigma_\epsilon^4n^2n_z^2\Tr(\bm\Sigma)^2}{
        + O_p\left((n^4n_z^{3/2+\delta} + n_z^2n^{7/2}m^{\delta})\kappa_1^2 + (n^{3/2+\delta}p^2 + (m/p)^2n^{7/2+\delta}(1+p/n))n_z^2\right)}}\right]^{1/2}\right)\\
        &= O_p(n_z^{-1/2}).
    \end{aligned}
\end{equation*}
The randomness of $\bm\beta$ generates another independent Gaussian random variable, and by applying Theorem~S\ref{thm:BE_quad_form}, we have 
\begin{equation}
\label{ineq:BE_marg_A2_beta_app}
    \begin{aligned}
        \sup_{t\in\mathbb R}&\left|\mathbb P\left(n_zn\frac{\tilde{\theta}_2-\sigma_{\bm\beta}^2\gamma_2}{\sigma_4}\frac{\sigma_4}{\sigma} < t - \frac{n_z\sigma_2}{\sigma}\Lambda_{\X_0} - \frac{n_z\sqrt{n\sigma_\epsilon^2\tilde{\theta}_3}}{\sigma}\Lambda_{\epsilon} - \right.\right.\\
        &\left.\left.\frac{\sqrt{\sigma_{\epsilon_z}^2n_z\left\{n\left\{n\tilde{\theta}_3 + \Tr(\bm\Sigma^2)\tilde{\theta}_1\right\} + n\sigma_\epsilon^2\Tr(\bm\Sigma^2) + O_p(n^{3/2+\delta} (m/p) + n^{3/2}m^{\delta})\right\}}}{\sigma}\Lambda_{\epsilon_z} -\right.\right.\\ &\left.\left.\frac{\sqrt{\splitfrac{n_z\Big[2(n\tilde{\theta}_2)^2 + 2n\tilde{\theta}_3 + \tilde{\theta}_1\left\{n\left\{n\tilde{\theta}_3 + \Tr(\bm\Sigma^2)\tilde{\theta}_1\right\} + n\sigma_\epsilon^2\Tr(\bm\Sigma^2)\right\}}{
        + O_p\left((n^{3/2}m/p)(m^\delta + n^\delta\kappa_1)\right)}\Big]}}{\sigma}\Lambda_{\Z_{0}}\right) - \right.\\
        &\left.\mathbb P\left(n_zn\frac{\sigma_4}{\sigma}\Lambda_{\bm\beta} < t - \frac{n_z\sqrt{n\sigma_{\bm\beta}^4(2\gamma_2^2+\gamma_1\gamma_3)}}{\sigma}\Lambda_{\X_0} - \right.\right.\\
        &\left.\left.\frac{n_z\sqrt{n\sigma_\epsilon^2\sigma_{\bm\beta}^2\gamma_3}}{\sigma}\Lambda_\epsilon - \frac{\sqrt{\sigma_{\epsilon_z}^2n_z\left\{n\sigma_{\bm\beta}^2(n\gamma_3+\omega_2\gamma_1p) + n\sigma_\epsilon^2\omega_2p\right\}}}{\sigma}\Lambda_{\epsilon_z} -\right.\right.\\ &\left.\left.
        \frac{\sqrt{n_z\left[2(n\sigma_{\bm\beta}^2\gamma_2)^2 + 2n\sigma_\epsilon^2\sigma_{\bm\beta}^2\gamma_3 + \sigma_{\bm\beta}^2\gamma_1\left\{n(n\sigma_{\bm\beta}^2\gamma_3 + \sigma_{\bm\beta}^2\omega_2\gamma_1p) + np\sigma_\epsilon^2\omega_2\right\}\right]}}{\sigma}\Lambda_{\Z_{0}}\right)\right| \\
        &\leq O_p(m^{-1/5}, m^{-2\delta}),
    \end{aligned}
\end{equation} 
where 
$$\sigma_4^2 = \left(\mathbb E(\bm\beta^4) - 3\frac{\sigma_{\bm\beta}^4}{p^2}\right)\sum_{i=1}^m\left(\bm\Sigma^2\mathbb\I_m\right)_{i,i}^2 + \frac{2\sigma_{\bm\beta}^4}{p^2}\Tr\left\{\left(\bm\Sigma^2\mathbb\I_m\right)^2\right\}.$$
Moreover, by using the convolution formula for independent Gaussian random variables, we have 
\begin{equation*}
    \begin{aligned}
        &\mathbb P\left(\frac{n_zn\sigma_4}{\sigma}\Lambda_{\bm\beta} + \frac{n_z\sqrt{n\sigma_{\bm\beta}^4\left(2\gamma_2^2+\gamma_1\gamma_3\right)}}{\sigma}\Lambda_{\X_0} + \frac{n_z\sqrt{n\sigma_\epsilon^2\sigma_{\bm\beta}^2\gamma_3}}{\sigma}\Lambda_\epsilon + \right.\\ 
        &\left.\frac{\sqrt{\sigma_{\epsilon_z}^2n_z\left\{n\sigma_{\bm\beta}^2\left(n\gamma_3+\omega_2\gamma_1p\right) + n\sigma_\epsilon^2\omega_2p\right\}}}{\sigma}\Lambda_{\epsilon_z} +\right.\\ &\left.\frac{\sqrt{n_z\left[2\left(n\sigma_{\bm\beta}^2\gamma_2\right)^2 + 2n\sigma_\epsilon^2\sigma_{\bm\beta}^2\gamma_3 + \sigma_{\bm\beta}^2\gamma_1\left\{n\left(n\sigma_{\bm\beta}^2\gamma_3 + \sigma_{\bm\beta}^2\omega_2\gamma_1p\right) + np\sigma_\epsilon^2\omega_2\right\}\right]}}{\sigma}\Lambda_{\Z_{0}} < t\right) = \Phi(t).
    \end{aligned}
\end{equation*}
Combine the Berry-Esseen inequalities in \cref{ineq:BE_marg_A2_eps_z_app}, \cref{ineq:BE_A2_marg_Z_app}, \cref{ineq:BE_marg_A2_eps_app}, \cref{ineq:BE_marg_A2_X_app}, and \cref{ineq:BE_marg_A2_beta_app}, we have
\begin{equation}
\label{ineq:BE_marg_A2_num}
    \begin{aligned}
        \sup_{t\in\mathbb R}  \left|\bm H_{\text{M}}(t)- \Phi(t)\right| \leq O_p\left(\max\left(m^{-1/5}, m^{-2\delta}, n_z^{-1/2}, n_z^{-2\delta}, n^{-1/2}, n^{-2\delta}\right)\right)
    \end{aligned}
\end{equation}
for $\forall \delta \in (0,1/2)$. Since $\delta$ is arbitrary, the optimal choice of $\delta$ will be $1/4$, which gives us following Berry-Esseen inequality
\begin{equation*}
    \begin{aligned}
        \sup_{t\in\mathbb R}  \left|\bm H_{\text{M}}(t) - \Phi(t)\right| \leq O_p\left(\max\left(m^{-1/5}, n_z^{-1/2}, n^{-1/2}\right)\right).
    \end{aligned}
\end{equation*}

\subsection{Limits of the denominator}
In this section, we provide the limits of quantities in the denominator, which can be used for constructing the limiting distribution of $A(\hat{\bm\beta}_{\text{M}})$ with Slutsky's Theorem.
We have shown that 
for $\forall \delta\in (0,1/2)$, we have 
\begin{equation}
\label{concent:marg_A2_denom_1}
    \begin{aligned}
        &\mathbb P\left(\left|\y^\top \X\Z^\top \Z\X^\top \y - nn_z\left\{n\sigma_{\bm\beta}^2\gamma_3 + \left(\gamma_1\sigma_{\bm\beta}^2 + \sigma_\epsilon^2\right)\omega_2p\right\}\right| < O_p\left(n_zn^2\left(m^{\delta}n^{-1/2} + \kappa_1n_z^{\delta-1/2} + \right.\right.\right.\\ 
        &\left.\left.\left.m/p n^{\delta-1/2} +pn^{\delta-3/2} + m^{1/2 + \delta}p^{-1} + m^{1/2 + \delta}p^{-1}\right)\right)\right)\geq 
        1-O_p\left(\max\left(m^{-2\delta},n_z^{-2\delta}, n^{-2\delta},p^{-2\delta}\right)\right).
    \end{aligned}
\end{equation}
Similarly, for $\forall \epsilon_1, \epsilon_2, \epsilon_3 > 0$, consider the subset 
\begin{equation*}
    \begin{aligned}
        &\Omega_6(\epsilon_1, \epsilon_2, \epsilon_3) \coloneqq \left\{\Z_0, \epsilon_z, \bm\beta: \left\{\bigl|\bm\beta^\top \Z^\top \Z\bm\beta - n_z\sigma_{\bm\beta}^2\gamma_1\bigr| < n_z^{1/2 + \delta}mp^{-1}\epsilon_1\right\}\cap \left\{\bigl|\epsilon_z^\top\epsilon_z - n_z\sigma_{\epsilon_z}^2\bigr| < n_z^{1/2+\delta}\epsilon_2\right\}\right.\\
        &\left.\cap\left\{\bigl|\bm\beta^\top Z^\top\epsilon_z\bigr| > n_z^{1/2+\delta}\epsilon_3\right\}\right\},
    \end{aligned}
\end{equation*}
then we have 
\begin{equation}
\label{concent:marg_A2_denom_2}
    \begin{aligned}
        \mathbb P\left(\left\{\Z_0, \epsilon_z, \bm\beta\right\} \in \Omega_6\right) \geq 1-O_p(n_z^{-2\delta}).
    \end{aligned}
\end{equation}
Moreover, for $\forall \left\{\Z_0, \epsilon_z, \bm\beta\right\} \in \Omega_6$, we have
$$\left|\|\Z\bm\beta + \epsilon_z\|_2^2 - n_z\left(\sigma_{\bm\beta}^2\gamma_1 + \sigma_{\epsilon_z}^2\right)\right| \leq O_p(n_z^{1/2+\delta}).$$

\subsection{Major quantitative CLT}
Combining~\cref{ineq:BE_marg_A2_num}, \cref{concent:marg_A2_denom_1} and \cref{concent:marg_A2_denom_2}, by  Slutsky's Theorem, we have the following major quantitative CLT for \texorpdfstring{$A(\hat{\bm\beta}_{\text{M}})$}{TEXT}. 
\begin{thm.s}
\label{thm:CLT_marg_A2_raw}
Under Assumptions~\ref{a:Sigmabound}-\ref{a:Sparsity}, consider the marginal estimator $\hat{\bm\beta}_{\text{M}} = n^{-1}\X^\top y$, and its corresponding prediction accuracy $A(\hat{\bm\beta}_{\text{M}})$. Let 
    \begin{equation*}
        \begin{aligned}
            \eta = \frac{(\sigma_{\bm\beta}^2\gamma_1)/h_{\bm\beta_z}^2\left\{\sigma_{\bm\beta}^2\gamma_3 + (p/n)(\sigma_{\bm\beta}^2\gamma_1)/h_{\bm\beta}^2\omega_2\right\}}{\splitfrac{2\gamma_2^2\sigma_{\bm\beta}^4(n_z/n + 1) + (n_z/n)\gamma_3\sigma_{\bm\beta}^2(\gamma_1\sigma_{\bm\beta}^2 + \sigma_\epsilon^2) + (\sigma_{\bm\beta}^2\gamma_1)/h_{\bm\beta_z}^2\left\{\sigma_{\bm\beta}^2\gamma_3 + (p/n)(\sigma_{\bm\beta}^2\gamma_1)/h_{\beta}^2\omega_2\right\} }{\qquad +n_z\Big\{\big(\mathbb E\left(\bm\beta^4\right) - 3\sigma_{\bm\beta}^4/p^2\big)\sum_{i=1}^m(\bm\Sigma^2\mathbb\I_m)_{i,i}^2+ 2\sigma_{\bm\beta}^4/p^2\Tr((\bm\Sigma^2\mathbb\I_m)^2)\Big\}}},
        \end{aligned}
    \end{equation*}
    As $\min(n,n_z,p,m)\to\infty$, the following Berry-Esseen inequality holds
    \begin{equation*}
        \begin{aligned}
       \sup_{t\in\mathbb R} & \left|\mathbb P\left(\sqrt{n_z\eta}\left(A(\hat{\bm\beta}_{M})- h_{\bm\beta_z}\gamma_2 \left\{\left(\frac{\gamma_1}{h_{\bm\beta}^2}\frac{p}{n}\omega_2+\gamma_3\right)\gamma_1 \right\}^{-1/2}\right) < t\right)- \Phi(t)\right|\\
        &\leq O_p\left(\max(n_z^{-1/2}, n^{-1/2},m^{-1/5})\right).
        \end{aligned}
    \end{equation*}
\end{thm.s}
{Theorem~S\ref{thm:CLT_marg_A2_raw} leads to the Theorem~\ref{thm: CLT for A^2 Marginal} in Section~\ref{subsubsec:marg_A_iso}. 
For the special case $\bm\Sigma=\mathbb{I}_p$, we do not need to use the martingale CLT. Therefore, we can apply the results in Lemma~S\ref{lemma:BE_quad_form_iso} to obtain Corollary~\ref{cor: CLT for marg A^2 iso}.} 

\section{Proof for Section~\ref{subsubsec:ref_new}}
For this section, we consider the genetically predicted value $\z^\top\hat{\bm\beta}_{\text{W}}(\lambda)$ of the reference panel-based estimator given by $\hat{\bm\beta}_{\text{W}}(\lambda) = (\W^\top\W + n\lambda\mathbb{I}_p)^{-1}\X^\top \y$. Similar to Section~\ref{sec_proof_31}, we prove Theorem~\ref{thm: CLT for reference new} in three steps using a leave-one-out strategy.
First, we decompose $\z^\top\hat{\bm\beta}_{\text{W}}(\lambda)$ into two parts: one related to the randomness of $\bm \epsilon$ and one without. Next, we provide a quantitative CLT regarding the randomness of $\bm \epsilon$, considering $\X_0$ as fixed. Following the leave-one-out technique, we then consider the randomness of $\X_0$ to retrieve the corresponding CLT results. Finally, we provide the concentration for the randomness of $\W_0$ using the anisotropic local law \citep{anisotropic_local_law}.
We will frequently use Schur's complement formula when computing limits using the anisotropic local law. 
\begin{proposition.s}[Schur's complement formula, \cite{Schur1917}] 
\label{prop:Schur_comp}
Suppose $p$ and $q$ are non-negative integers, and suppose $\A,\B,\C,$ and $\D$ are $p\times p$, $p\times q$, $q\times p$, and $q\times q$ matrices of real numbers, respectively. Moreover, we assume $\D$ is invertible, then we have
\begin{flalign*}
        \left[\begin{array}{ll}{\A} & {\B} \\ {\C} & {\D}\end{array} \right]^{-1} = \left[\begin{array}{ll}
            {(\M/\D)^{-1}} & {-(\M/\D)^{-1}(\B\D^{-1})} \\
            {-\D^{-1}C(\M/\D)^{-1}} & {\D^{-1}+\D^{-1}\C(\M/\D)^{-1}\B\D^{-1}}
        \end{array} \right],
\end{flalign*}
where $$\M/\D = \A-\B\D^{-1}\C.$$
\end{proposition.s}

\subsection{Decomposition of new prediction of reference panel-based ridge estimator}
Notice that
\begin{equation}
\label{decomp:ref_new}
    \begin{aligned}
        \z^\top(\W^\top \W + n_w\lambda \mathbb\I_p)^{-1}\X^\top \y = \z^\top (\W^\top \W + n_w\lambda\mathbb\I_p)^{-1}\X^\top \X\bm\beta + \z^\top (\W^\top \W+n_w\lambda\mathbb\I_p)^{-1}\X^\top \bm \epsilon,
    \end{aligned}
\end{equation}
By fixing $\X_0$ and $\W_0$, we can treat the first quantity in \cref{decomp:ref_new} as deterministic. This gives us a good starting point to apply our Lemma~S\ref{lemma:non-asmptotic CLT 3} to quantify the randomness of $\bm \epsilon$.

\subsection{Berry-Esseen bounds with the randomness of training error}
Considering fixed $\X_0$ and $\W_0$, we consider the randomness of $\bm \epsilon$. By fixing $\W_0$, we regard {$\R=(\W^\top\W + n_w\lambda\mathbb\I_p)^{-1}$} 
as deterministic. Since $\W^\top \W + n_w\lambda\mathbb\I_p$ is symmetric, $\R$ is also symmetric.
Considering the randomness of $\bm \epsilon$, we have the following inequality
\begin{equation}
\label{ineq:BE_new_ref_eps}
    \begin{aligned}
        \sup_{t\in\mathbb R}\left|\mathbb P\left(\frac{\z^\top(\W^\top \W + n_w\lambda\mathbb \I_p)^{-1}\X^\top\bm \epsilon}{\sqrt{\sigma_\epsilon^2\z^\top \R\bm\Sigma^{1/2}\X_0^\top \X_0\bm\Sigma^{1/2}\R\z}} < t\right) - \Phi_{\epsilon}(t)\right| \leq O_p\left(\sqrt{\frac{\mathbb E(\epsilon^4)\sum_{i=1}^n(\x_{0_i}^\top\bm\Sigma^{1/2}\R\z)^4}{\sigma_\epsilon^4\|\X_0\bm\Sigma^{1/2}\R\z\|_2^2}}\right).
    \end{aligned}
\end{equation}
For simplicity, we denote
\begin{equation*}
    \begin{aligned}
        \sigma_2^2 = \sigma_\epsilon^2\z^\top \R\bm\Sigma^{1/2}\X_0^\top \X_0\bm\Sigma^{1/2}\R\z,
    \end{aligned}
\end{equation*}
and the quantity we are interested in is
\begin{equation*}
    \begin{aligned}
        \sup_{t\in\mathbb R}\left|\mathbb P\left(\frac{\z^\top(\W^\top \W+n_w\lambda\mathbb I_p)^{-1}\X^\top \y - \phi_d/\lambda\z^\top(\mathbb \I_p + \mathfrak{m}_w\bm\Sigma)^{-1}\bm\Sigma\bm\beta}{\sigma_{\text{W}}} < t\right) - \Phi(t)\right|.
    \end{aligned}
\end{equation*}
Notice that
\begin{equation*}
    \begin{aligned}
        &\mathbb P\left(\frac{\z^\top(\W^\top \W + n_w\lambda \mathbb \I_p)^{-1}\X^\top\epsilon}{\sigma_2}\frac{\sigma_2}{\sigma_{\text{W}}} < \right. \\
        &\left. \qquad t - \frac{\z^\top(\W^\top \W + n_w\lambda\mathbb\I_p)^{-1}\X^\top \X\bm\beta - \phi_d/\lambda\z^\top(\mathbb I_p + \mathfrak{m}_w\bm\Sigma)^{-1}\bm\Sigma\bm\beta}{\sigma_{\text{W}}}\right).
    \end{aligned}
\end{equation*}
Using the Berry-Esseen inequality in \cref{ineq:BE_new_ref_eps}, we have 
\begin{equation}
\label{ineq:BE_new_ref_eps_app}
    \begin{aligned}
        &\sup_{t\in\mathbb R}\left|\mathbb P\left(\frac{\z^\top \R\X^\top \bm \epsilon}{\sigma_{\text{W}}} < t - \frac{\z^\top \R\X^\top \X\bm\beta - \phi_d/\lambda\z^\top(\mathbb \I_p + \mathfrak{m}_w\bm\Sigma)^{-1}\bm\Sigma\bm\beta}{\sigma_{\text{W}}}\right) 
        -\mathbb P\left(\frac{\sigma_2}{\sigma_{\text{W}}}\Lambda_{\epsilon} \right.\right.\\
        &\left.\left. < t-\frac{\z^\top \R\X^\top \X\bm\beta - \phi_d/\lambda\z^\top(\mathbb \I_p + \mathfrak{m}_w\bm\Sigma)^{-1}\bm\Sigma\bm\beta}{\sigma_{\text{W}}}\right)\right| \leq O_p\left(\sqrt{\frac{\sum_{i=1}^n(\x_{0_i}^\top\bm\Sigma^{1/2}\R\z)^4}{\|\X_0\bm\Sigma^{1/2}\R\z\|_2^4}}\right).
    \end{aligned}
\end{equation}

\subsection{Berry-Esseen bounds with the randomness of training data matrix}
Considering fixed $\W_0$, we now consider the randomness of $\X_0$. By Lemma~S\ref{lemma:non-asmptotic CLT 1}, we have 
\begin{equation}
\label{ineq:BE_new_ref_X}
    \begin{aligned}
       \sup_{t\in\mathbb R} &\left|\mathbb P\left(\frac{\z^\top(\W^\top \W+n_w\lambda \mathbb\I_p)^{-1}\X^\top \X\bm\beta - n\z^\top(\W^\top \W+n_w\lambda\mathbb\I_p)^{-1}\bm\Sigma\bm\beta}{\sigma_1} < t\right) 
 - \Phi_{\X_0}(t)\right| \\
  &\leq O_p(n^{-1/2}),
    \end{aligned}
\end{equation}
where $$\sigma_1^2 = n\left\{\mathbb E\left(x_0^4-3\right)\sum_{i=1}^p\left(\bm\Sigma^{1/2}R\z\right)_i^2\left(\bm\Sigma^{1/2}\bm\beta\right)_i^2 + 2\left(\z^\top \R\bm\Sigma\bm\beta\right)^2 + \|\bm\Sigma^{1/2}\bm\beta\|_2^2\|\bm\Sigma^{1/2}\R\z\|_2^2\right\}.$$
We now restrict our analysis to a subset where some quantities involving $\X_0$ are properly concentrated, allowing us to refine the Berry-Esseen inequality in \cref{ineq:BE_new_ref_eps_app}. By Lemma~S\ref{lemma:von bahr-Esseen bound}, for $\forall \epsilon_1, \epsilon_2 > 0$ and  $0<\delta<1/2$, we have 
\begin{equation*}
    \begin{aligned}
        &\mathbb P\left(\left|\z^\top \R\bm\Sigma^{1/2}\X_0^\top \X_0\bm\Sigma^{1/2}\R\z - n\z^\top \R\bm\Sigma \R\z\right| < \frac{n^{1/2+\delta}}{n_w^2}\epsilon_1\right) \geq 1-\frac{nC\mathbb E\left|\x_{0_i}^\top\bm\Sigma^{1/2}\R\z\right|^4}{n^{1+2\delta}/n_w^4\epsilon_1^2} \quad\mbox{and}\\
        &\mathbb P\left(\left|\sum_{i=1}^n\left(\x_{0_i}^\top\bm\Sigma^{1/2}\R\z\right)^4 - n\mathbb E\left|\x_{0_i}^\top\bm\Sigma^{1/2}\R\z\right|^4\right| < \frac{n^{1/2+\delta}}{n_w^4}\epsilon_2\right) \geq 1-\frac{nC\mathbb E\left|\x_{0_i}^\top\bm\Sigma^{1/2}\R\z\right|^8}{n^{1+2\delta}/n_w^8\epsilon_2^2}.
    \end{aligned}
\end{equation*}
Note that
$$\mathbb E\left|\x_{0_i}^\top\bm\Sigma^{1/2}\R\z\right|^4 = \mathbb E\left(x_0^4-3\right)\sum_{i=1}^p\left(\bm\Sigma^{1/2} \R\z\right)_i^4 + 3\left(\z^\top \R\bm\Sigma \R\z\right)^2 = O_p(n_w^{-4}),$$
and 
$$\mathbb E\left|\x_{0_i}^\top\bm\Sigma^{1/2}\R\z\right|^8 = O_p(n_w^{-8}).$$
Therefore, for $\forall\delta \in \left(0,1/2\right)$, we have 
\begin{equation*}
    \begin{aligned}
        &\mathbb P\left(\left|\z^\top \R\bm\Sigma^{1/2}\X_0^\top \X_0\bm\Sigma^{1/2}\R\z - n\z^\top \R\bm\Sigma \R\z\right| < \frac{n^{1/2+\delta}}{n_w^2}\epsilon_1\right) \geq 1-O_p(n^{-2\delta}) \quad \mbox{and}\\
        &\mathbb P\left(\left|\sum_{i=1}^n(\x_{0_i}^\top\bm\Sigma^{1/2}\R\z)^4 - n\left\{\mathbb E\left(\x_0^4-3\right)\sum_{i=1}^p(\bm\Sigma^{1/2} \R\z)_i^4 + 3\left(\z^\top \R\bm\Sigma \R\z\right)^2\right\}\right| < \frac{n^{1/2+\delta}}{n_w^4}\epsilon_2\right)\\
        &\qquad\geq 1-O_p(n^{-2\delta}).
    \end{aligned}
\end{equation*}
Denote the subset 
\begin{equation*}
    \begin{aligned}
        &\Xi_1(\epsilon_1,\epsilon_2) \coloneqq \left\{\X_0:\left\{\left|\z^\top \R\bm\Sigma^{1/2}\X_0^\top \X_0\bm\Sigma^{1/2}\R\z - n\z^\top \R\bm\Sigma \R\z\right| < \frac{n^{1/2+\delta}}{n_w^2}\epsilon_1\right\}\cap\right.\\
        &\left.\left\{\left|\sum_{i=1}^n(\x_{0_i}^\top\bm\Sigma^{1/2}\R\z)^4 - n\left\{\mathbb E\left(x_0^4-3\right)\sum_{i=1}^p\left(\bm\Sigma^{1/2} \R\z\right)_i^4 + 3\left(\z^\top \R\bm\Sigma \R\z\right)^2\right\}\right| < \frac{n^{1/2+\delta}}{n_w^4}\epsilon_2\right\}\right\},
    \end{aligned}
\end{equation*}
and we have shown that $$\mathbb P(\X_0\in \Xi_1) \geq 1- O_p(n^{-2\delta}).$$ 
Furthermore, for $\X_0 \in \Xi_1$, we can simplify the Berry-Esseen bound in \cref{ineq:BE_new_ref_eps_app} as follows 
\begin{equation}
\label{concent:BE_bound_new_ref_X}
    \begin{aligned}
        &O_p\left(\sqrt{\frac{\sum_{i=1}^n(\x_{0_i}^\top\bm\Sigma^{1/2}R\z)^4}{\|\X_0\bm\Sigma^{1/2}\R\z\|_2^4}}\right)  \\
        &=O_p\left(\sqrt{\frac{n\left\{\mathbb E\left(x_0^4-3\right)\sum_{i=1}^p(\bm\Sigma^{1/2} \R\z)_i^4 + 3\left(\z^\top \R\bm\Sigma \R\z\right)^2\right\} + O_p(n^{1/2+\delta}p^2n_w^{-2})}{n^2\|\bm\Sigma^{1/2}\R\z\|_2^4 + O_p(n^{3/2 + \delta}p^2n_w^{-4})}}\right).
    \end{aligned}
\end{equation}
Moreover, we rearrange the second term in \cref{ineq:BE_new_ref_eps_app} to obtain the Gaussian variable generated from randomness of $\X_0$
\begin{equation}
\label{ineq:BE_new_ref_X_app}
    \begin{aligned}
        &\sup_{t\in\mathbb R}\left|\mathbb P\left(\frac{\z^\top \R\X^\top \X\bm\beta - n\z^\top \R\bm\Sigma\bm\beta}{\sigma_1}\frac{\sigma_1}{\sigma_{\text{W}}} < t - \frac{\sigma_2}{\sigma_{\text{W}}}\Lambda_{\epsilon}-\frac{nz^\top \R\bm\Sigma\bm\beta - \phi_d/\lambda\z^\top(\mathbb \I_p + \mathfrak{m}_w\bm\Sigma)^{-1}\bm\Sigma\bm\beta}{\sigma_{\text{W}}} \right) - \right.\\
        &\left.\mathbb P\left(\frac{\sigma_1}{\sigma_{\text{W}}}\Lambda_{\X_0} < t- \frac{\sqrt{\sigma_\epsilon^2n\z^\top \R\bm\Sigma \R\z + O_p(n^{1/2+\delta}n_w^{-2})}}{\sigma_{\text{W}}}\Lambda_{\epsilon}-\frac{n\z^\top \R\bm\Sigma\bm\beta - \phi_d/\lambda\z^\top(\mathbb \I_p + \mathfrak{m}_w\bm\Sigma)^{-1}\bm\Sigma\bm\beta}{\sigma_{\text{W}}}\right)\right|\\
        &\leq O_p\left(\max\left(n^{-1/2}, n^{-2\delta}\right)\right)
    \end{aligned}
\end{equation}

\subsection{Concentration using the anisotropic local law}
In this section, we illustrate how we obtain the CLT conditional on $\z$ and $\bm\beta$, considering the randomness of $\W$. In general, the limit concerning the resolvent can be obtained using the isotropic local law \citep{anisotropic_local_law}. Intuitively, as the dimensions increase, the quantity concerning the resolvent becomes increasingly deterministic, leading to an almost sure limit.
\begin{lem.s}
\label{lemma:one_res_limit}
    Under Assumptions~\ref{a:Sigmabound}-\ref{a:anisotropic regularity}, with probability of at least $1-O_p(p^{-D})$ for some large $D\in\mathbb R$, and for all small $\vartheta > 0$, we have 
    $$\left|-\lambda \z^\top(n_w\lambda\mathbb\I_p + \bm\Sigma^{1/2} \W_0^\top \W_0 \bm\Sigma^{1/2})^{-1} \bm\Sigma\bm\beta + \frac{1}{n_w}\z^\top(\mathbb\I_p+\mathfrak{m}_w\bm\Sigma)^{-1}\bm\Sigma\bm\beta\right| \leq \psi(-\lambda)O_p(\frac{p^{\vartheta}}{n_w}),$$
    where  
    $$\psi(\lambda) = \sqrt{\frac{Im(m_w(-\lambda))}{n_w\eta}} + \frac{1}{n_w\eta}$$
    {and $Im(\cdot)$ is the imaginary part.}
\end{lem.s}
\begin{proof}
Let
    \begin{equation*}
        \begin{aligned}
                &\R_p(-\lambda) = \left(\bm\Sigma^{1/2}\frac{\W_0^\top}{\sqrt{n_w}} \frac{\W_0}{\sqrt{n_w}}\bm\Sigma^{1/2} + \lambda\right)^{-1},\quad \G(-\lambda) = \left[\begin{array}{cc}
            -\bm\Sigma^{-1} & \frac{\W_0^\top}{\sqrt{n_w}} \\
            \frac{\W_0}{\sqrt{n_w}} & \lambda \mathbb\I_{n_w}
        \end{array} \right]^{-1},\\
        &\bm \Pi(-\lambda) = \left[\begin{array}{cc}
            -\bm\Sigma(1+m_w(-\lambda)\bm\Sigma)^{-1} &  0\\
            0 & m_w(-\lambda)
        \end{array} \right], \quad \mbox{and} \quad \underline{\bm\Sigma} = \left[\begin{array}{cc}
             \bm\Sigma & 0 \\
                0   & \mathbb\I_{n_w}
        \end{array}\right].
        \end{aligned}
    \end{equation*}
    By Proposition~S\ref{prop:Schur_comp}, let $B$ and $C$ be deterministic matrices with sizes of $p\times p$ and $n_w \times n_w$, respectively. Then we can explicitly compute for $\underline{\bm\Sigma}^{-1}$ and $\G(-\lambda)$  as follows 
    \begin{equation*}
        \begin{aligned}
        &\underline{\bm\Sigma}^{-1} = \left[\begin{array}{cc}
           \bm\Sigma^{-1}  & 0 \\
            0           & \mathbb\I_{n_w}
        \end{array}\right]\quad \mbox{and} \\
        &\G(-\lambda) = \\
        &\left[\begin{array}{cc}
           (-\lambda n_w)\bm\Sigma^{1/2}(\lambda n_w\mathbb\I_p + \bm\Sigma^{1/2}\W_0^\top \W_0\bm\Sigma^{1/2})^{-1}\bm\Sigma^{1/2}  & \B \\
            \C           & \frac{1}{\lambda}\mathbb\I_{n_w} - \frac{1}{n_w\lambda^2}\W_0(\bm\Sigma^{-1} + \frac{1}{n_w}\W_0^\top \W_0)^{-1}\W_0^\top
        \end{array}\right].
        \end{aligned}
    \end{equation*}
    Denote $$\vec{\underline{\bm v}} =  
    \left[
        \begin{array}{c}
            \bm\Sigma^{1/2}\z \\ \hdashline[2pt/2pt]
            0
        \end{array}
    \right] \text{, } \vec{\underline{\bm w}} =
    \left[\begin{array}{c}
         \bm\Sigma^{3/2}\bm\beta  \\ \hdashline
         0 
    \end{array}\right],$$
    where $\vec{\underline{\bm v}}$ and $\Vec{\underline{\bm w}}$ are of shape $(n_w+p) \times 1$.
    From Theorem 3.7 of \cite{anisotropic_local_law}, we have
    \begin{equation*}
        \begin{aligned}
            \left| \left\langle \vec{\underline{\bm v}}, \underline{\bm\Sigma}^{-1}\left(\G(-\lambda) - \bm \Pi(-\lambda)\right)\underline{\bm\Sigma}^{-1}\vec{\underline{\bm w}}\right\rangle\right|  \prec \psi(-\lambda)O_p(1),
        \end{aligned}
    \end{equation*}
    where $\|\vec{\underline{\bm v}}\| \leq 1$ and $\|\vec{\underline{\bm w}}\| \leq 1$. After careful simplification, we have 
    \begin{equation*}
        \begin{aligned}
            &\left|\z^\top\bm\Sigma^{1/2} \left\{\left(-\lambda n_w\right)\bm\Sigma^{-1/2}\left(\lambda n_w\mathbb\I_p + \bm\Sigma^{1/2}\W_0^\top \W_0\bm\Sigma^{1/2}\right)^{-1}\bm\Sigma^{-1/2}+\left(\mathbb\I_p+\mathfrak{m}_w\bm\Sigma\right)^{-1}\bm\Sigma^{-1}\right\}\bm\Sigma^{3/2}\bm\beta\right|\\
            &\qquad \prec \psi(-\lambda)O_p(1).
        \end{aligned}
    \end{equation*}
    It follows that 
    \begin{equation*}
        \begin{aligned}
            \left|-\lambda n_w\z^\top\left(n_w\lambda\mathbb\I_p + \bm\Sigma^{1/2} \W_0^\top \W_0 \bm\Sigma^{1/2}\right)^{-1} \bm\Sigma\bm\beta + \z^\top\left(\mathbb\I_p+\mathfrak{m}_w\bm\Sigma\right)^{-1}\bm\Sigma\bm\beta\right| \leq \psi(-\lambda)O_p(p^{\vartheta}).
        \end{aligned}
    \end{equation*}
    Therefore, we have 
    \begin{equation*}
        \begin{aligned}
            \left|\z^\top\left(n_w\lambda\mathbb\I_p + \W^\top \W\right)^{-1}\bm\Sigma\bm\beta - \frac{1}{n_w\lambda}\z^\top\left(\mathbb\I_p+\mathfrak{m}_w\bm\Sigma\right)^{-1}\bm\Sigma\bm\beta\right| \leq \psi(-\lambda)O_p(p^{\vartheta-1}).
        \end{aligned}
    \end{equation*}
\end{proof}

\begin{lem.s}\label{lemma:inv_Sigmabound}
   Adopt Assumption~\ref{a:Sigmabound}, for $\forall \lambda \in \mathbb R$, we have $\gamma_1\mathbb\I_p \prec (\mathbb\I_p+\lambda\bm\Sigma)^{-1} \prec \gamma_2\mathbb\I_p$ for some constants $\gamma_1$ and $\gamma_2$.
\end{lem.s}
\begin{proof}
    Since $\bm\Sigma$ is symmetry, using eigen-decomposition, we have
    $$\bm\Sigma = \bm{U}\bm{\Lambda}\bm{U}^{-1}$$ and 
    $$\left(\mathbb\I_p + \lambda\bm\Sigma\right)^{-1} = \bm{U}\left(\mathbb\I_p + \lambda\Lambda\right)^{-1}\bm{U}^{-1}.$$
    As we have $c\mathbb\I_p \prec \bm\Sigma \prec C\mathbb\I_p$, it follows that 
    $$\frac{1}{\lambda C + 1}\mathbb\I_p \prec \left(\mathbb\I_p + \lambda\bm\Sigma\right)^{-1} \prec \frac{1}{\lambda c + 1}\mathbb\I_p.$$
\end{proof}
We now provide the first-order limit for the term $\z^\top \R\bm\Sigma \R\z$ in \cref{ineq:BE_new_ref_X_app}. We aim to fetch the desired quantity as the form of the derivative of some other quantity, involving only one resolvent $\R$. After careful calculation, we obtain the following identity
$$\left.\frac{d(\z^\top\bm\Sigma^{-1/2}(\bm\Sigma^{-1/2}\R^{-1}\bm\Sigma^{-1/2} - \tau\mathbb\I_p)^{-1}\bm\Sigma^{-1/2}\z)}{d\tau}\right|_{\tau = 0} = \z^\top \R\bm\Sigma \R\z.$$
We now focus on $\z^\top\bm\Sigma^{-1/2}(\bm\Sigma^{-1/2}\R^{-1}\bm\Sigma^{-1/2} - \tau\mathbb\I_p)^{-1}\bm\Sigma^{-1/2}\z$, which can be reorganized in the following form
\begin{equation*}
    \begin{aligned}
        &\z^\top\left(\mathbb\I_p-\frac{\tau}{n_w\lambda}\bm\Sigma\right)^{-1/2}\left\{\left(\bm\Sigma^{-1}-\frac{\tau}{n_w\lambda}\mathbb\I_p\right)^{-1/2}\W_0^\top \W_0\left(\bm\Sigma^{-1}-\frac{\tau}{n_w\lambda}\mathbb\I_p\right)^{-1/2} \right.\\
        &\left.\qquad \qquad  +n_w\lambda\mathbb\I_p\right\}^{-1}\left(\mathbb\I_p-\frac{\tau}{n_w\lambda}\bm\Sigma\right)^{-1/2}\z.
    \end{aligned}
\end{equation*}
By Lemma~S\ref{lemma:one_res_limit}, and replacing the $\bm\Sigma$ in $\underline{\bm\Sigma}$ and $\G(-\lambda)$ with 
$$\left(\bm\Sigma^{-1} - \frac{\tau}{n_w\lambda}\mathbb\I_p\right)^{-1},$$ 
we now set $$\vec{\underline{\bm v}} = \vec{\underline{\bm w}} = 
    \left[
        \begin{array}{c}
            (\mathbb\I_p -\frac{\tau}{n_w\lambda}\bm\Sigma)^{-1}\bm\Sigma^{1/2}\z \\ \hdashline[2pt/2pt]
            0
        \end{array}
    \right].$$ 
It is easy to check $\|\vec{\underline{\bm v}}\| = \|\vec{\underline{\bm w}}\| \leq O_p(1)$. 
We denote the Stieltjes transform of asymptotic eigenvalue density of $(\bm\Sigma^{-1} - \tau/(n_w\lambda)\mathbb\I_p)^{-1}$ by $m_{w}(\lambda, \tau)$, which satisfies
\begin{equation}
\label{equ:eig_perturb_1}
    \begin{aligned}
        \frac{1}{m_{w}(\lambda,\tau)} = \lambda + \frac{\phi_w}{p}\sum_{i=1}^p\frac{\left\{1/\pi_i - \tau/(n_w\lambda)\right\}^{-1}}{1+m_{w}(\lambda,\tau)\left\{1/\pi_i - \tau/(n_w\lambda)\right\}^{-1}}.
    \end{aligned}
\end{equation}
Let 
$$\bm \Pi(-\lambda) = \left[\begin{array}{cc}
        -\left(\bm\Sigma^{-1} - \frac{\tau}{n_w\lambda}\mathbb \I_p\right)^{-1}\left\{1+m_{w}(\lambda, \tau)\left(\bm\Sigma^{-1} - \frac{\tau}{n_w\lambda}\mathbb\I_p\right)^{-1}\right\}^{-1} &  0\\
        0 & m_w(-\lambda)
    \end{array} \right].$$
Similarly, we have 
\begin{equation*}
    \begin{aligned}
        \left|\left\langle \vec{\underline{\bm v}}, \underline{\bm\Sigma}^{-1}\left(\G(-\lambda) - \bm \Pi(-\lambda)\right)\underline{\bm\Sigma}^{-1}\vec{\underline{\bm w}}\right\rangle\right|  \prec \psi(-\lambda)O_p(1),
    \end{aligned}
\end{equation*}
where 
$$\G(-\lambda) = \left[\begin{array}{cc}
        -(\bm\Sigma^{-1} - \frac{\tau}{n_w\lambda}\mathbb\I_p) & \frac{\W_0^\top}{\sqrt{n_w}} \\
        \frac{\W_0}{\sqrt{n_w}} & \lambda \mathbb\I_{n_w}
    \end{array} \right]^{-1} \text{ and } \underline{\bm\Sigma} = \left[\begin{array}{cc}
         (\bm\Sigma^{-1} - \frac{\tau}{n_w\lambda}\mathbb\I_p)^{-1} & 0 \\
            0   & \mathbb\I_{n_w}
    \end{array}\right].$$
It follows that 
\begin{equation}
\label{concent:ref_new_aniso}
    \begin{aligned}
        &\left|n_w\z^\top\bm\Sigma^{-1/2}\left(\bm\Sigma^{-1/2}R^{-1}\bm\Sigma^{-1/2} - \tau\mathbb\I_p\right)^{-1}\bm\Sigma^{-1/2}\z-\frac{1}{\lambda}\z^\top\bm\Sigma^{-1}\left(\bm\Sigma^{-1} + m_{w}(\lambda, \tau) -\frac{\tau}{n_w\lambda}\right)^{-1}\z\right| \\
        & \qquad\leq \psi(-\lambda)O_p(n_w^{\vartheta})\\ 
    \end{aligned}
\end{equation}
with probability of at least $1-O_p(p^{-D})$.
For $m_w(-\lambda)$, recall that it satisfies
$$\frac{1}{m_w(-\lambda)} = \lambda + \phi_w\int\frac{x}{1+m_w(-\lambda)x}\pi(dx).$$
By taking the derivative with respect to $\lambda$, we have 
\begin{equation}
    \begin{aligned}
        m_w'(-\lambda)\left\{\frac{1}{m_w(-\lambda)^2} - \frac{\phi_w}{p}\sum_{i=1}^p\frac{\pi_i^2}{(1+m_w(-\lambda)\pi_i)^2}\right\} = 1,
    \end{aligned}
\end{equation}
where 
$$m_w'(-\lambda)\coloneqq\frac{dm_w(-\lambda)}{d(-\lambda)}.$$
By taking derivative on both sides of \cref{equ:eig_perturb_1}, we have 
\begin{equation*}
    \begin{aligned}
        &\frac{\phi_w}{p}\sum_{i=1}^p\frac{-\pi_i^2n_w\lambda}{\left(m_{w}(\lambda,\tau)\pi_in_w\lambda + n_w\lambda\right)^2} \\
        &\qquad= \left\{\frac{1}{m_{w}(\lambda,\tau)^2} - \frac{\phi_w}{p}\sum_{i=1}^p\frac{\pi_i^2n_w^2\lambda^2}{(m_{w}(\lambda,\tau)\pi_in_w\lambda + n_w\lambda-\tau\pi_i)^2}\right\}\frac{dm_{w}(\lambda,\tau)}{d\tau}.
    \end{aligned}
\end{equation*}
Note that our target is $\left.dm_{w}(\lambda,\tau)/d\tau\right|_{\tau=0}$, and notice that $\left.m_{w}(\lambda,\tau)\right|_{\tau = 0} = m_w(-\lambda)$, we have the following identity
\begin{equation*}
    \begin{aligned}
        \frac{\phi_w}{n_wp\lambda}\sum_{i=1}^p\frac{-\pi_i^2}{\left(m_w(-\lambda)\pi_i + 1\right)^2} = \left\{\frac{1}{m_w(-\lambda)^2} - \frac{\phi_w}{p}\sum_{i=1}^p\frac{\pi_i^2}{(m_w(-\lambda)\pi_i + 1)^2}\right\}\left.\frac{dm_{w}(\lambda,\tau)}{d\tau}\right|_{\tau=0}.
    \end{aligned}
\end{equation*}
From our previous observation, we derive
\begin{equation*}
    \begin{aligned}
        \left.\frac{dm_{w}(\lambda,\tau)}{d\tau}\right|_{\tau = 0} = \frac{1}{ n_w\lambda}\left\{1 - \frac{m_w'(-\lambda)}{m_w(-\lambda)^2}\right\}.
    \end{aligned}
\end{equation*}
Taking the derivative on both sides of \cref{concent:ref_new_aniso} with respect to $\tau$ and set $\tau = 0$, we have 
\begin{equation}
\label{concent:ASigA}
    \begin{aligned}
        &\left|n_w^2\z^\top \R\bm\Sigma \R\z - \frac{n_w}{\lambda}\z^\top\bm\Sigma^{-1}\left(\bm\Sigma^{-1}+m_w(-\lambda)\right)^{-1}\left(\frac{1}{n_w\lambda}-\left.\frac{dm_{w}(\lambda,\tau)}{d\tau}\right|_{\tau=0}\right)\left(\bm\Sigma^{-1}+m_w(-\lambda)\right)^{-1}\z\right|\\ 
        &\leq \psi(-\lambda)O_p(n_w^{\vartheta}) \implies\\
        &\left|n_w^2\z^\top \R\bm\Sigma \R\z - \frac{m_w'(-\lambda)}{m_w(-\lambda)^2}\frac{1}{\lambda^2}\z^\top(\mathbb\I_p+m_w(-\lambda)\bm\Sigma)^{-2}\bm\Sigma \z\right| \leq \psi(-\lambda)O_p(n_w^{\vartheta}).
    \end{aligned}
\end{equation}
And on the subset where the \cref{concent:ASigA} holds, we can further simplify \cref{concent:BE_bound_new_ref_X} as follows 
\begin{equation*}
    \begin{aligned}
        O_p\left(\sqrt{\frac{n\left\{\mathbb E\left(x_0^4-3\right)\sum_{i=1}^p(\bm\Sigma^{1/2} \R\z)_i^4 + 3\left(\z^\top \R\bm\Sigma \R\z\right)^2\right\} + O_p(n^{1/2+\delta}p^2n_w^{-2})}{n^2\|\bm\Sigma^{1/2}\R\z\|_2^4 + O_p(n^{3/2 + \delta}p^2n_w^{-4})}}\right) = O_p(n^{-1/2}).
    \end{aligned}
\end{equation*}
Concentration inequalities in Lemmas~S\ref{lemma:one_res_limit} and \cref{concent:ASigA} also lead to that with probability of at least $1-O_p(p^{-D})$ over randomness of $\W_0$, the following inequality holds
\begin{equation*}
    \begin{aligned}
        &\sup_{t\in\mathbb R}\left|\mathbb P\left(\frac{\z^\top \R\X^\top \y - \phi_d/\lambda\z^\top(\mathbb I_p + m_w(-\lambda)\bm\Sigma)^{-1}\bm\Sigma\bm\beta}{\sigma_{\text{W}}} < t\right) \right.\\
        &\left. - 
        \mathbb P\left(\frac{\sqrt{n\left\{2\theta_1+\theta_0\theta_2\left(m_w'(-\lambda)/m_w(-\lambda)^2\right)\right\}}}{\lambda n_w\sigma_{\text{W}}}\Lambda_{\X_0}< t- \frac{\sqrt{\sigma_\epsilon^2n\theta_2\left(m_w'(-\lambda)/m_w(-\lambda)^2\right)}}{\lambda n_w\sigma_{\text{W}}}\Lambda_{\epsilon}\right)\right|\\
        &\quad\leq O_p(p^{-1/2}).
    \end{aligned}
\end{equation*}
By the convolution formula for independent Gaussian variables, we have the following theorem
\begin{thm.s}
\label{thm:CLT_A2_ref_raw}
Under Assumptions~\ref{a:Sigmabound}-\ref{a:anisotropic regularity}, recall that $\hat{\bm\beta}_{\text{W}}(\lambda)= (\W^\top \W+n_w\lambda\mathbb \I_p)^{-1}\X^\top \y$, consider the new prediction from the testing data $\z$. With probability of at least $1-O_p(p^{-D})$, we have the following Berry-Esseen bound
\begin{equation*}
    \begin{aligned}
        &\sup_{t\in\mathbb R}\left|\mathbb P\left(\frac{\z^\top (\W^\top \W + n_w\lambda\mathbb I_p)^{-1}\X^\top \y - \phi_d/\lambda\z^\top(\mathbb \I_p + \mathfrak{m}_w\bm\Sigma)^{-1}\bm\Sigma\bm\beta}{\sigma_{\text{M}}} < t\right) - \Phi(t)\right| \leq O_p(p^{-1/2}),
    \end{aligned}
\end{equation*}
where $$\sigma_{\text{W}}^2 = \frac{\phi_d^2}{\lambda^2}\left\{2\theta_1 + \left(\frac{m_w'(-\lambda)}{m_w(-\lambda)^2}\right)\frac{\theta_0\theta_2}{h_{\bm\beta}^2}\right\},$$
with 
\begin{equation*}
\begin{aligned}
        \theta_0= \|\bmSigma^{1/2}\bmbeta\|_2^2, \quad \theta_1 = \left\{\z^\top(\mathbb \I_p + {\mathfrak{m}_w}\bm\Sigma)^{-1}\bm\Sigma\bm\beta\right\}^2, \quad \mbox{and}\quad \theta_2 = \z^\top(\mathbb \I_p + {\mathfrak{m}_w}\bm \Sigma)^{-2}\bm\Sigma \z.
\end{aligned}
\end{equation*}
\end{thm.s}



\section{Proof for Section~\ref{subsubsec:ref_A}}
To prove the quantitative CLT of $A(\hat{\bm\beta}_{\text{W}}(\lambda))$ in Theorem~\ref{thm: CLT for reference A^2}, we use nine steps with a leave-one-out technique. First, we decompose the numerator of $A(\hat{\bm\beta}_{\text{W}}(\lambda))$ into two parts and provide the quantitative CLT regarding the randomness of $\bm \epsilon_z$ by adopting Berry-Esseen bounds along with first-order concentrations. Similarly, we provide the quantitative CLT regarding the randomness of $\bm \epsilon_z$, $\Z_0$, $\bm \epsilon$, $\X_0$, and $\bm\beta$, along with the corresponding first-order concentrations.
Next, by using the CLT of the quadratic form, we consider the randomness of $\bm\beta$ to obtain the desired CLT. We then provide a concentration from the randomness of $\W_0$ by adopting the anisotropic local law. Finally, we obtain the first-order limit of the denominator, and the desired CLT of $A(\hat{\bm\beta}_{\text{W}}(\lambda))$ follows from Slutsky's theorem.

\subsection{Numerator decomposition}
The numerator of $A(\hat{\bm\beta}_{\text{W}}(\lambda))$ can be rewritten as
\begin{equation*}
    \begin{aligned}
        (\Z\bm\beta &+ \epsilon_z)^\top \Z(\W^\top \W + n_w\lambda\mathbb\I_p)^{-1}\X^\top \y \\
        &= \bm\beta^\top \Z^\top \Z(\W^\top \W + n_w\lambda\mathbb\I_p)^{-1}\X^\top \y + \epsilon_z^\top \Z(\W^\top \W + n_w\lambda \mathbb\I_p)^{-1}\X^\top \y\\
        &= \bm\beta^\top\bm\Sigma^{1/2}\Z_0^\top \Z_0\bm\Sigma^{1/2}(\W^\top \W + n_w\lambda\mathbb\I_p)^{-1}\X^\top \y + \bm \epsilon_z^\top \Z_0\bm\Sigma^{1/2}(\W^\top \W + n_w\lambda\mathbb\I_p)^{-1}\X^\top \y.
    \end{aligned}
\end{equation*}
Notice that for fixed $\X_0$, $\Z_0$ $\W_0$, $\bm \epsilon$, and $\bm\beta$, the first quantity can be treated as a constant, while the second quantity relies on the randomness of $\bm \epsilon_z$, where Lemma~S\ref{lemma:non-asmptotic CLT 3} can be applied.
In this section, we denote
\begin{equation*}
    \begin{aligned}
        \hbar^2 &=\frac{\sigma_{\bm\beta}^2nn_z}{\lambda^2n_w^2}\left\{\frac{1}{m_w(-\lambda)^3}\left(3\xi_1-\xi_2+m_w(-\lambda)\gamma_1-2m/p\right)\left\{n(\sigma_{\epsilon_z}^2 + \sigma_{\bm\beta}^2\gamma_1) + n_z(\sigma_{\bm\beta}^2\gamma_1 + \sigma_\epsilon^2)\right\}\right.\\
        &\left.+\frac{p}{m_w(-\lambda)^2}(1-2\pi_1+\pi_2)(\sigma_{\epsilon_z}^2 + \sigma_{\bm\beta}^2\gamma_1)(\gamma_1 + \frac{\sigma_\epsilon^2}{\sigma_{\bm\beta}^2}) + \frac{2(n+n_z)}{m_w(-\lambda)^4}(m_w(-\lambda)\gamma_1+\xi_1-m/p)^2\sigma_{\bm\beta}^2\right.\\
        &+nn_z\left\{\left(\frac{\mathbb E\left(\bm\beta^4\right)}{\sigma_{\bm\beta}^2} - \frac{3\sigma_{\bm\beta}^2}{p^2}\right)\sum_{i=1}^p\left\{\left(\mathbb\I_p+m_w(-\lambda)\bm\Sigma\right)^{-1}\bm\Sigma^2\mathbb\I_m\right\}_{i,i}^2 \right.\\ 
        &\left.\left.+ \frac{2\sigma_{\bm\beta}^2}{p^2}\Tr\left((\bm\Sigma(\mathbb\I_p+m_w(-\lambda)\bm\Sigma)^{-1}\bm\Sigma\mathbb\I_m)^2\right)\right\}\right\}.
    \end{aligned}
\end{equation*}
The CDF we try to analyze is given by 
\begin{equation}
\label{quant:ref_A2}
    \begin{aligned}
        &\bm H_{\text{W}}(t) = \\
        &\mathbb P\left(\frac{(\Z\bm\beta + \bm \epsilon_z)^\top \Z(\W^\top \W + n_w\lambda\mathbb\I_p)^{-1}\X^\top \y - \sigma_{\bm\beta}^2nn_z/(\lambda n_wm_w(-\lambda)^2)\left(m_w(-\lambda)\gamma_1 + \xi_1 - m/p\right)}{\hbar} < t\right),
    \end{aligned}
\end{equation}
We will prove that Kolmogorov distance between this quantity and the CDF of standard Gaussian goes to $0$ as $p\to\infty$.

\subsection{Berry-Esseen bounds with the randomness of testing error}
Conditional on $\Z_0$, $\X_0$, $\W_0$, $\bm \epsilon$ and $\bm\beta$, and consider the randomness of $\bm \epsilon_z$, by Lemma~S\ref{lemma:non-asmptotic CLT 3}, we have 
\begin{equation}
\label{ineq:BE_ref_A2_eps_z}
    \begin{aligned}
        \sup_{t\in\mathbb R}\left|\mathbb P\left(\frac{\bm \epsilon_z^\top \Z(\W^\top \W+n_w\lambda\mathbb\I_p)^{-1}\X^\top \y}{\sqrt{\sigma_{\epsilon_z}^2\|\Z(\W^\top \W+n_w\lambda\mathbb\I_p)^{-1}\X^\top \y\|_2^2}} < t\right) - \Phi_{\epsilon_z}(t)\right| \leq O_p\left(\sqrt{\frac{\sum_{i=1}^{n_z}(\z_{0_{i}}^\top\bm\Sigma^{1/2}\R\X^\top \y)^4}{\|\Z \R\X^\top \y\|_2^2}}\right).
    \end{aligned}
\end{equation}
Now we denote $$\sigma_1^2 = \sigma_{\epsilon_z}^2\|\Z(\W^\top \W+n_w\lambda\mathbb\I_p)^{-1}\X^\top \y\|_2^2.$$
By applying \cref{ineq:BE_ref_A2_eps_z} to \cref{quant:ref_A2}, we have 
\begin{equation}
\label{ineq:BE_ref_A2_eps_z_app}
    \begin{aligned}
        \sup_{t\in\mathbb R}&\left|\bm H_{\text{W}}(t) - \mathbb P\left(\frac{\sigma_1}{\hbar}\Lambda_{\epsilon_z} < t - \left\{\bm\beta^\top \Z^\top \Z \R\X^\top y - \frac{\sigma_{\bm\beta}^2nn_z}{\lambda n_wm_w(-\lambda)^2}\left(m_w(-\lambda)\gamma_1 + \xi_1 - \frac{m}{p}\right)\right\}/\hbar\right)\right|\\
        &\leq O_p\left(\sqrt{\frac{\sum_{i=1}^{n_z}(\z_{0_{i}}^\top\bm\Sigma^{1/2}\R\X^\top \y)^4}{\|\Z \R\X^\top \y\|_2^2}}\right).
    \end{aligned}
\end{equation}

\subsection{Berry-Esseen bounds with the randomness of testing data matrix}
Conditional on $\X_0$, $\W_0$, $\bm \epsilon$, and $\bm\beta$, we consider the randomness of $\Z_0$. By applying Lemma~S\ref{lemma:non-asmptotic CLT 1}, we have
\begin{equation*}
    \begin{aligned}
       \sup_{t\in\mathbb R} &\left|\mathbb P\left(\frac{\bm\beta^\top \Z^\top \Z(\W^\top \W + n_w\lambda\mathbb\I_p)^{-1}\X^\top \y - n_z\bm\beta^\top\bm\Sigma(\W^\top \W+n_w\lambda\mathbb\I_p)^{-1}\X^\top \y}{\sigma_2} < t\right)-\Phi_{\Z_0}(t)\right| \\
        &\leq O_p(n_z^{-1/2}),
    \end{aligned}
\end{equation*}
where 
$$\sigma_2^2 = n_z\left\{\mathbb E\left(z_0^4 -3\right)\sum_{i=1}^p(\bm\Sigma^{1/2}\bm\beta)_i^2\left(\R\X^\top \y\right)_i^2 + 2\left(\bm\beta^\top\bm\Sigma \R\X^\top \y\right)^2 + \|\bm\Sigma^{1/2}\bm\beta\|_2^2\|\bm\Sigma^{1/2}\R\X^\top \y\|_2^2\right\}.$$
Consider the subset
\begin{equation*}
    \begin{aligned}
        \Xi_2(\epsilon_1,\epsilon_2) \coloneqq& \Bigg\{\Z_0:\Bigg\{\bigl|\|\Z \R\X^\top \y\|_2^2 - n_z\|\bm\Sigma^{1/2}R\X^\top \y\|_2^2\bigr| < n_z^{1/2+\delta}n^2n_w^{-2}\epsilon_1\Bigg\} \cap \\
        &\Bigg\{\Big|\sum_{i=1}^{n_z}(z_{0_i}^\top\bm\Sigma^{1/2}\R\X^\top \y)^4 - n_z\mathbb E_{\Z_0}\left|z_{0_i}^\top\bm\Sigma^{1/2}\R\X^\top \y\right|^4\Big| < n_z^{1/2+\delta}n^4n_w^{-4}\epsilon_2\Bigg\}\Bigg\}.
    \end{aligned}
\end{equation*}
By regarding the randomness of $\Z_0$, and applying the concentration inequality in Lemma~S\ref{lemma:von bahr-Esseen bound}, for $\forall \delta \in (0, 1/2)$, we have the following two inequalities
\begin{equation*}
    \begin{aligned}
        &\mathbb P\left(\left|\|\Z \R\X^\top \y\|_2^2 - n_z\|\bm\Sigma^{1/2}\R\X^\top \y\|_2^2\right| < n_z^{1/2+\delta}n^2n_w^{-2}\epsilon_1\right) \geq 1-\frac{Cn_z\mathbb E_{\Z_0}\left|\z_{0_i}^\top\bm\Sigma^{1/2}\R\X^\top \y\right|^4}{n_z^{1+2\delta}n^4n_w^{-4}\epsilon_1^2} \quad \mbox{and} \\
        &\mathbb P\left(\left|\sum_{i=1}^{n_z}(\z_{0_i}^\top\bm\Sigma^{1/2}\R\X^\top \y)^4 - n_z\mathbb E_{\Z_0}\left|\z_{0_i}^\top\bm\Sigma^{1/2}\R\X^\top \y\right|^4\right| < n_z^{1/2+\delta}n^4n_w^{-4}\epsilon_2\right) \\
        &\qquad\geq 1-\frac{Cn_z\mathbb E_{\Z_0}\left|\z_{0_i}^\top\bm\Sigma^{1/2}\R\X^\top \y\right|^8}{n_z^{1+2\delta}n^8n_w^{-8}\epsilon_2^2}.
    \end{aligned}
\end{equation*}
Note that 
$$\mathbb E_{\Z_0}\left|\z_{0_i}^\top\bm\Sigma^{1/2}\R\X^\top \y\right|^4 = \mathbb E\left(z_0^4-3\right)(\bm\Sigma^{1/2}\R\X^\top \y)_i^4 + 3(\y^\top \X R\bm\Sigma \R\X^\top \y)^2 = O_p(n^4n_w^{-4})$$ and 
$$\mathbb E_{\Z_0}\left|z_{0_i}^\top\bm\Sigma^{1/2}R\X^\top y\right|^8 = O_p(n^8n_w^{-8}).$$ 
It follows that 
$$\mathbb P(\Z_0 \in \Xi_2) \geq 1-O_p(n_z^{-2\delta}).$$ 
Intuitively, we restrict our choice of $\Z_0$ to a subset where it concentrates well, and we show that this concentration holds with high probability.
Therefore, for $\Z_0 \in \Xi_2$, our Berry-Esseen bound in \cref{ineq:BE_ref_A2_eps_z} can be further simplified
\begin{equation}
\label{concent:BE_bound_A2_ref_Z}
    \begin{aligned}
        &O_p\left(\sqrt{\frac{\sum_{i=1}^{n_z}(\z_{0_{i}}^\top\bm\Sigma^{1/2}\R\X^\top \y)^4}{\|\Z \R\X^\top \y\|_2^2}}\right)\\
        &= O_p\left(\sqrt{\frac{n_z\left\{\mathbb E\left(z_0^4-3\right)\left(\bm\Sigma^{1/2}\R\X^\top \y\right)_i^4 + 3\left(\y^\top \X \R\bm\Sigma \R\X^\top \y\right)^2\right\} + O_p(n_z^{1/2 + \delta}n^4n_w^{-4})}{n_z^2\|\bm\Sigma^{1/2}\R\X^\top \y\|_2^4 + O_p(n_z^{3/2+\delta}n^4n_w^{-4})}}\right).
    \end{aligned}
\end{equation}
Moreover, consider $\Z_0 \in \Xi_2$, the second probability in \cref{ineq:BE_ref_A2_eps_z_app} can be further replaced by
\begin{equation*}
    \begin{aligned}
        &\mathbb P\left(\frac{\bm\beta^\top \Z^\top \Z \R\X^\top \y - n_z\bm\beta^\top\bm\Sigma \R\X^\top \y}{\hbar} < t-\frac{\sqrt{n_z\sigma_{\epsilon_z}^2\|\bm\Sigma^{1/2}R\X^\top \y\|_2^2 + O_p(n_z^{1/2+\delta}n^2n_w^{-2})}}{\hbar}\Lambda_{\epsilon_z} \right.\\
        &\left. \qquad -\frac{n_z\bm\beta^\top \R\X^\top \y - \sigma_{\bm\beta}^2nn_z/(\lambda n_w)\left(m_w(-\lambda)\gamma_1 + \xi_1 - m/p\right)}{\hbar}\right).
    \end{aligned}
\end{equation*}
It follows that 
\begin{equation}
\label{ineq:BE_ref_A2_Z_app}
    \begin{aligned}
       \sup_{t\in\mathbb R} &\left|\mathbb P\left(\frac{\bm\beta^\top \Z^\top \Z \R\X^\top \y - n_z\bm\beta^\top\bm\Sigma \R\X^\top \y}{\hbar} < t-\frac{\sqrt{\sigma_{\epsilon_z}^2n_z\|\bm\Sigma^{1/2}\R\X^\top \y\|_2^2 + O_p(n_z^{1/2+\delta}n^2n_w^{-2})}}{\hbar}\Lambda_{\epsilon_z}  \right.\right.\\
        &\left.\left. -\frac{n_z\bm\beta^\top \bm\Sigma \R\X^\top \y - \sigma_{\bm\beta}^2nn_z/(\lambda n_wm_w(-\lambda)^2)\left(m_w(-\lambda)\gamma_1 + \xi_1 - m/p\right)}{\hbar}\right) - \right.\\
        &\left. \mathbb P\left(\frac{\sigma_2}{\hbar}\Lambda_{\Z_0} < t - \frac{\sqrt{\sigma_{\epsilon_z}^2n_z\|\bm\Sigma^{1/2}\R\X^\top \y\|_2^2 + O_p(n_z^{1/2+\delta}n^2n_w^{-2})}}{\hbar}\Lambda_{\epsilon_z} \right.\right.\\
        &\left.\left. -\frac{n_z\bm\beta^\top \bm\Sigma \R\X^\top \y - \sigma_{\bm\beta}^2nn_z/(\lambda n_wm_w(-\lambda)^2)\left(m_w(-\lambda)\gamma_1 + \xi_1 - m/p\right)}{\hbar}\right)\right|\\
        &\qquad\leq O_p\left(\max(n_z^{-1/2}, n_z^{-2\delta})\right).
    \end{aligned}
\end{equation}

\subsection{Berry-Esseen bounds with the randomness of training error}
Conditional on $\X_0$, $\W_0$, and $\bm\beta$, we now consider the randomness of $\epsilon$. By applying Lemma~S\ref{lemma:non-asmptotic CLT 3}, we have the following inequality
\begin{equation}
\label{ineq:BE_ref_A2_eps}
    \begin{aligned}
        \sup_{t\in\mathbb R}\left|\mathbb P\left(\frac{\bm\beta^\top\bm\Sigma(\W^\top \W+n_w\lambda\mathbb\I_p)^{-1}\X^\top\bm \epsilon}{\sqrt{\sigma_\epsilon^2\|\X_0\bm\Sigma^{1/2} \R\bm\Sigma\bm\beta\|_2^2}} < t\right) - \Phi_{\epsilon}(t)\right| \leq O_p\left(\sqrt{\frac{\sum_{i=1}^n(\x_{0_i}^\top\bm\Sigma^{1/2} \R\bm\Sigma\bm\beta)^4}{\|\X_0\bm\Sigma^{1/2} \R\bm\Sigma\bm\beta\|_2^4}}\right).
    \end{aligned}
\end{equation}
Now we denote $\sigma_4^2 = \sigma_\epsilon^2\|\X_0\bm\Sigma^{1/2}\R\bm\Sigma\bm\beta\|_2^2$ for simplicity. We will restrict our choice of $\bm \epsilon$ to a subset, which is defined as 
\begin{equation*}
    \begin{aligned}
        \Xi_3(\epsilon_1, \epsilon_2, \epsilon_3)& \coloneqq \left\{\bm\epsilon:\left\{\left|\bm\epsilon^\top \X \R\bm\Sigma \R \X^\top\bm \epsilon - \sigma_\epsilon^2\Tr(\X \R\bm\Sigma \R\X^\top)\right| < p^{\delta-1/2}\epsilon_1\right\}\cap\right.\\
        &\left.\left\{\left|\bm\beta^\top \X^\top \X \R\bm\Sigma \R\X^\top\bm \epsilon\right| < p^{\delta - 1/2}\epsilon_2\right\}\cap\left\{\left|\bm\beta^\top\bm\Sigma \R\X^\top \X\bm\beta\bm\beta^\top\bm\Sigma \R\X^\top\bm \epsilon\right| < p^{\delta - 1/2}\epsilon_3\right\}\right\},
    \end{aligned}
\end{equation*}
for $\forall \epsilon_1, \epsilon_2, \epsilon_3 > 0$ and $0<\delta<1/2$.
By Markov's inequality and Proposition~S\ref{prop:quad_first_limit}, we have 
$\mathbb P(\bm \epsilon \in \Xi_3) \geq 1-O_p(p^{-2\delta})$. It follows that 
\begin{equation*}
\begin{aligned}
       \mathbb P&\left(\left|\bm \epsilon^\top \X \R\bm\Sigma \R\X^\top\bm \epsilon - \sigma_\epsilon^2\Tr(\X \R\bm\Sigma \R\X)\right|  < p^{\delta - 1/2}\epsilon_1\right) \\
        &\qquad\geq 1- p^{1-2\delta}\epsilon_1^{-2}\frac{\Tr\left\{(\X \R\bm\Sigma \R\X^\top)^2\right\}}{\Tr(\X \R\bm\Sigma \R\X^\top)^2} = 1-O_p(p^{-2\delta}),\\
\end{aligned}
\end{equation*}
\begin{equation*}
    \begin{aligned}
        \mathbb P&\left(\left|\bm\beta^\top \X^\top \X \R\bm\Sigma \R \X^\top\bm \epsilon\right| < n^{\delta - 1/2}\epsilon_2\right) \\
        &\qquad\geq 1-\frac{\sigma_\epsilon^2\bm\beta^\top \X^\top \X \R\bm\Sigma \R\X^\top \X \R\bm\Sigma \R\X^\top \X\bm\beta}{p^{2\delta - 1}\epsilon_2^2} = 1-O_p(p^{-2\delta}),\\
    \end{aligned}
\end{equation*}
and 
\begin{equation*}
    \begin{aligned}
        &\mathbb P\left(\left|\bm\beta^\top\bm\Sigma \R\X^\top\bm \epsilon\right| < p^{\delta - 1/2}\epsilon_3\right) \geq 1- \frac{\mathbb E_{\epsilon}\left|\bm\beta^\top\bm\Sigma \R\X^\top\bm \epsilon\right|^2}{p^{2\delta - 1}\epsilon_3^2} = 1-O_p(p^{-2\delta}).
    \end{aligned}
\end{equation*} 
Rearranging the second term in \cref{ineq:BE_ref_A2_Z_app} and applying \cref{ineq:BE_ref_A2_eps}, we have
\begin{equation}
\label{ineq:BE_ref_A2_eps_app}
    \begin{aligned}
        &\sup_{t\in\mathbb R}\left|\mathbb P\left(\frac{n_z\bm\beta^\top \bm\Sigma R\X^\top\bm  \epsilon}{\sigma_4}\frac{\sigma_4}{\hbar}< t-\frac{\sqrt{\sigma_{\epsilon_z}^2n_z\|\bm\Sigma^{1/2}\R\X^\top \y\|_2^2 + O_p(n_z^{1/2+\delta}n^2n_w^{-2})}}{\hbar}\Lambda_{\epsilon_z} - \frac{\sigma_2}{\hbar}\Lambda_{\Z_0}\right.\right.\\
        &\left.\left.- \frac{n_z\bm\beta^\top \bm\Sigma \R\X^\top \X\bm\beta - \sigma_{\bm\beta}^2nn_z/(\lambda n_wm_w(-\lambda)^2)(m_w(-\lambda)\gamma_1 + \xi_1 - m/p)}{\hbar}\right) - 
        \mathbb P\left(\frac{n_z\sigma_4}{\hbar}\Lambda_{\epsilon}< t-\right.\right.\\
        &\left.\left.\frac{\sqrt{\sigma_{\epsilon_z}^2n_z\left(\|\bm\Sigma^{1/2}\R\X^\top \X\bm\beta\|_2^2 + \sigma_\epsilon^2\Tr(\X \R\bm\Sigma \R\X^\top)\right) + O_p(p^{1/2+\delta})}}{\hbar}\Lambda_{\epsilon_z} - \right.\right.\\
        &\left.\left.\frac{\sqrt{2n_z(\bm\beta^\top\bm\Sigma \R\X^\top \X\bm\beta)^2 + n_z\|\bm\Sigma^{1/2}\bm\beta\|_2^2\left(\|\bm\Sigma^{1/2}\R\X^\top \X\bm\beta\|_2^2 + \sigma_\epsilon^2\Tr(\X \R\bm\Sigma \R\X^\top)\right) + O_p(p^{1/2 + \delta})}}{\hbar}\Lambda_{\Z_0}\right.\right.\\
        &\left.\left.- \frac{n_z\bm\beta^\top\bm\Sigma \R\X^\top \X\bm\beta - \sigma_{\bm\beta}^2nn_z/(\lambda n_wm_w(-\lambda)^2)(m_w(-\lambda)\gamma_1 + \xi_1 - m/p)}{\hbar}\right)\right|\\ 
        &\qquad\leq O_p\left(\max\left(\sqrt{\frac{\sum_{i=1}^n(\x_{0_i}^\top\bm\Sigma^{1/2} \R\bm\Sigma\bm\beta)^4}{\|\X_0\bm\Sigma^{1/2} \R\bm\Sigma\bm\beta\|_2^4}}, p^{-2\delta}\right)\right).
    \end{aligned}
\end{equation}
Moreover, for $\bm \epsilon \in \Xi_3$, the previous Berry-Esseen bound in \cref{concent:BE_bound_A2_ref_Z} can also be simplified as follows 
\begin{equation}
\label{concent:BE_bound_A2_ref_eps}
    \begin{aligned}
        &O_p\left(\sqrt{\frac{n_z\left\{\mathbb E\left(z_0^4-3\right)(\bm\Sigma^{1/2}\R\X^\top \y)_i^4 + 3(\y^\top \X \R\bm\Sigma \R\X^\top \y)^2\right\} + O_p(n_z^{1/2 + \delta}n^4n_w^{-4})}{n_z^2\|\bm\Sigma^{1/2}\R\X^\top \y\|_2^4 + O_p(n_z^{3/2+\delta}n^4n_w^{-4})}}\right)\\ 
        &= O_p\left(\sqrt{\frac{n_z\left\{\|\bm\Sigma^{1/2}\R\X^\top \X\bm\beta\|_2^4 + \Tr(\X \R\bm\Sigma \R\X^\top)^2\right\} + O_p(p^{1/2+\delta})}{n_z^2\left\{\|\bm\Sigma^{1/2}\R\X^\top \X\bm\beta\|_2^4 + \Tr(\X \R\bm\Sigma \R\X^\top)^2\right\} + O_p(p^{3/2 + \delta})}}\right).
    \end{aligned}
\end{equation}

\subsection{Berry-Esseen bounds with the randomness of training data matrix}
Conditional on $\W_0$ and $\bm\beta$, we now consider the randomness of $\X_0$. By applying Lemma~S\ref{lemma:non-asmptotic CLT 1}, we have following inequality
\begin{equation}
\label{ineq:BE_ref_A2_X}
    \begin{aligned}
        &\sup_{t\in\mathbb R}\left|\mathbb P\left(\frac{\bm\beta^\top\bm\Sigma(\W^\top \W + n_w\lambda\mathbb\I_p)^{-1}\X^\top \X\bm\beta - n\bm\beta^\top\bm\Sigma(\W^\top \W+n_w\lambda\mathbb\I_p)^{-1}\bm\Sigma\bm\beta}{\sigma_3} < t\right) - \Phi_{\X_0}(t)\right|\\
        &\qquad\leq O_p(p^{-1/2}),
    \end{aligned}
\end{equation}
where $$\sigma_3^2 = n\left\{\mathbb E\left(x_0^4 - 3\right)\sum_{i=1}^p(\bm\Sigma^{1/2}\R\bm\Sigma\bm\beta)_i^2(\bm\Sigma^{1/2}\bm\beta)_i^2 + 2(\bm\beta^\top\bm\Sigma \R\bm\Sigma\bm\beta)^2 + \|\bm\Sigma^{1/2}\R\bm\Sigma\bm\beta\|_2^2\|\bm\Sigma^{1/2}\bm\beta\|_2^2\right\}.$$
For $\forall\epsilon_1,\epsilon_2, \cdots, \epsilon_6 > 0$, we define the subset   
\begin{equation*}
    \begin{aligned}
        &\Xi_4(\epsilon_1, \epsilon_2, \cdots, \epsilon_6) \coloneqq \left\{\X_0:\left\{\bigl|\|\X_0\bm\Sigma^{1/2} \R\bm\Sigma\bm\beta\|_2^2 - n\|\bm\Sigma^{1/2}\R\bm\Sigma\bm\beta\|_2^2\bigr| < n^{1/2+\delta}n_w^{-2}\epsilon_1\right\}\cap\right.\\
        &\left.\left\{\Big|\sum_{i=1}^n(x_{0_i}^\top\bm\Sigma^{1/2}\R\bm\Sigma\bm\beta)^4 - n\mathbb E\left|x_{0_i}^\top\bm\Sigma^{1/2}\R\bm\Sigma\bm\beta\right|^4\Big| < n^{1/2+\delta}n_w^{-4}\epsilon_2\right\}\cap\right.\\
        &\left.\left\{\bigl|\bm\beta^\top\bm\Sigma \R\X^\top \X\bm\beta - n\bm\beta^\top\bm\Sigma \R\bm\Sigma\bm\beta| < p^{\delta - 1/2}\epsilon_3\bigr|\right\}\cap \left\{\bigl|\Tr(\X \R\bm\Sigma \R \X^\top) - n\Tr(\R\bm\Sigma \R)\bigr| < p^{\delta - 1/2}\epsilon_4\right\}\cap\right.\\
        &\left.\left\{\Big|\|\bm\Sigma^{1/2}\R\X^\top \X\bm\beta\|_2^2 - n\left\{n\bm\beta^\top\bm\Sigma \R\bm\Sigma \R\bm\Sigma\bm\beta + \Tr(\bm\Sigma^{1/2}\R\bm\Sigma \R\bm\Sigma^{1/2})\|\bm\Sigma^{1/2}\bm\beta\|_2^2\right\}\Big| < p^{\delta -1/2}\epsilon_5\right\}\cap\right.\\
        &\left.\left\{\bigl|\Tr(\X \R\bm\Sigma \R\X^\top) - n\Tr(\bm\Sigma^{1/2}\R\bm\Sigma \R\bm\Sigma^{1/2})\bigr| < p^{\delta - 1/2}\epsilon_6\right\}\right\}.
    \end{aligned}
\end{equation*}
By Lemma~S\ref{lemma:von bahr-Esseen bound} and Markov's inequality, we have the following concentration inequalities
\begin{equation*}
    \begin{aligned}
        &\mathbb P\left(\left|\|\X_0\bm\Sigma^{1/2} \R\bm\Sigma\bm\beta\|_2^2 - n\|\bm\Sigma^{1/2}\R\bm\Sigma\bm\beta\|_2^2\right| < n^{1/2+\delta}n_w^{-2}\epsilon_1\right) \geq 1-\frac{nC\mathbb E\left|\x_{0_i}^\top\bm\Sigma^{1/2}\R\bm\Sigma\bm\beta\right|^4}{n^{1+2\delta}n_w^{-4}\epsilon_1^2},
    \end{aligned}
\end{equation*}
\begin{equation*}
    \begin{aligned}
        &\mathbb P\left(\left|\sum_{i=1}^n(\x_{0_i}^\top\bm\Sigma^{1/2}\R\bm\Sigma\bm\beta)^4 - n\mathbb E\left|\x_{0_i}^\top\bm\Sigma^{1/2}\R\bm\Sigma\bm\beta\right|^4\right| < n^{1/2+\delta}n_w^{-4}\epsilon_2\right) \geq 1-\frac{nC\mathbb E\left|\x_{0_i}^\top\bm\Sigma^{1/2}\R\bm\Sigma\bm\beta\right|^8}{n^{1+2\delta}n_w^{-8}\epsilon_2^2}, 
    \end{aligned}
\end{equation*}
\begin{equation*}
    \begin{aligned}
        \mathbb P\left(\Big|\Tr(\X \R\bm\Sigma \R\X^\top) - n\Tr(\bm\Sigma^{1/2}\R\bm\Sigma \R\bm\Sigma^{1/2})\Big| < p^{\delta -1/2}\epsilon_3\right) &\geq 1-\frac{nC\mathbb E_{\X_0}\left|\x_0^\top\bm\Sigma^{1/2}\R\bm\Sigma \R\bm\Sigma^{1/2}\x_0\right|^2} {p^{2\delta - 1}\epsilon_3^2} \\
        &= 1-O_p(p^{-2\delta}), 
    \end{aligned}
\end{equation*}
and 
\begin{equation*}
    \begin{aligned}
        \mathbb P\left(\bigl|\bm\beta^\top\bm\Sigma \R\X^\top \X\bm\beta - n\bm\beta^\top\bm\Sigma \R\bm\Sigma\bm\beta\bigr| < p^{\delta - 1/2}\epsilon_4\right)& \geq 1-\frac{nC\mathbb E_{\X_0}\left\{(\bm\beta^\top\bm\Sigma \R\bm\Sigma^{1/2}\x_0)^2(\x_0^\top\bm\Sigma^{1/2}\bm\beta)^2\right\}}{p^{2\delta-1}\epsilon_4^2} \\
        &= 1-O_p(p^{-2\delta}).
    \end{aligned}
\end{equation*}

For $\forall \X_0 \in \Xi_4$, the Berry-Esseen bounds in \cref{ineq:BE_ref_A2_eps_app} and \cref{concent:BE_bound_A2_ref_eps} can be simplified as follows 
{
\allowdisplaybreaks
\begin{equation}
\label{concent:BE_bound_A2_ref_X}
    \begin{aligned}
        &O_p\left(\sqrt{\frac{\sum_{i=1}^n(\x_{0_i}^\top\bm\Sigma^{1/2} \R\bm\Sigma\bm\beta)^4}{\|\X_0\bm\Sigma^{1/2} \R\bm\Sigma\bm\beta\|_2^4}}\right) = \\
        &O_p\left(\sqrt{\frac{ n\left\{\mathbb E\left(x_0^4-3\right)\sum_{i=1}^p(\bm\Sigma^{1/2}\R\bm\Sigma\bm\beta)_i^4 + 3(\bm\beta^\top\bm\Sigma \R\bm\Sigma \R\bm\Sigma\bm\beta)^2\right\} + O_p(n^{1/2+\delta}n_w^{-4})}{n^2\|\bm\Sigma^{1/2} \R\bm\Sigma\bm\beta\|_2^4 + O_p(n^{3/2+\delta}n_w^{-4})}}\right) \quad \mbox{and}\\
        &O_p\left(\sqrt{\frac{n_z(\|\bm\Sigma^{1/2}\R\X^\top \X\bm\beta\|_2^4 + \Tr(\X \R\bm\Sigma \R\X^\top)^2) + O_p(p^{1/2+\delta})}{n_z^2(\|\bm\Sigma^{1/2}\R\X^\top \X\bm\beta\|_2^4 + \Tr(\X \R\bm\Sigma \R\X^\top)^2) + O_p(p^{3/2 + \delta})}}\right) = \\
        &O_p\left(\sqrt{\frac{\splitfrac{n_z\left\{n^2\left(n\bm\beta^\top\bm\Sigma \R\bm\Sigma \R\bm\Sigma\bm\beta + \Tr(\bm\Sigma^{1/2}\R\bm\Sigma \R\bm\Sigma^{1/2})\|\bm\Sigma^{1/2}\bm\beta\|_2^2\right)^2 + n^2\Tr(\bm\Sigma^{1/2}\R\bm\Sigma \R\bm\Sigma^{1/2})^2\right\}}{ + O_p(p^{1/2+\delta})}}{\splitfrac{n_z^2\left\{n^2\left(n\bm\beta^\top\bm\Sigma \R\bm\Sigma \R\bm\Sigma\bm\beta + \Tr(\bm\Sigma^{1/2}\R\bm\Sigma \R\bm\Sigma^{1/2})\|\bm\Sigma^{1/2}\bm\beta\|_2^2\right)^2 + n^2\Tr(\bm\Sigma^{1/2}\R\bm\Sigma \R\bm\Sigma^{1/2})^2\right\}}{+ O_p(p^{3/2 + \delta})}}}\right). 
    \end{aligned}
\end{equation}
}

We have shown that $\mathbb P(\X_0 \in \Xi_4) \geq 1-O_p(p^{-2\delta})$ for $\forall \delta \in \left(0,1/2\right)$. For the simplicity of our proof, we introduce the following temporary notations in this section
\begin{equation*}
    \begin{aligned}
        \tilde{\omega}_1 = \bm\beta^\top\bm\Sigma \R\bm\Sigma \R\bm\Sigma\bm\beta, \quad \tilde{\omega}_2 = \|\bm\Sigma^{1/2}\R\bm\Sigma\bm\beta\|_2^2,\quad \mbox{and} \quad\tilde{\omega}_3 = \bm\beta^\top\bm\Sigma \R\bm\Sigma\bm\beta.
    \end{aligned}
\end{equation*}
Rearranging terms in second probability of \cref{ineq:BE_ref_A2_eps_app} and applying \cref{ineq:BE_ref_A2_X}, we have the following Berry-Esseen bound

{
\allowdisplaybreaks
    \begin{align}\label{ineq:BE_ref_A2_X_app}
        &\sup_{t\in\mathbb R}\Bigg|\mathbb P\left(\frac{n_z\bm\beta^\top \bm\Sigma \R\X^\top \X\bm\beta - nn_z\tilde{\omega}_3}{\sigma_3}\frac{\sigma_3}{\hbar} < t- \right.\bigr.\nonumber\\
        &\bigl.\left.\frac{\sqrt{\sigma_{\epsilon_z}^2n_z\left(\|\bm\Sigma^{1/2}\R\X^\top \X\bm\beta\|_2^2 + \sigma_\epsilon^2\Tr(\X \R\bm\Sigma \R\X^\top)\right) + O_p(p^{1/2+\delta})}}{\hbar}\Lambda_{\epsilon_z} - \right.\bigr.\nonumber\\
        &\bigl.\left.\frac{\sqrt{2n_z(\bm\beta^\top\bm\Sigma \R\X^\top \X\bm\beta)^2 + n_z\|\bm\Sigma^{1/2}\bm\beta\|_2^2\left(\|\bm\Sigma^{1/2}\R\X^\top \X\bm\beta\|_2^2 + \sigma_\epsilon^2\Tr(\X \R\bm\Sigma \R\X^\top)\right) + O_p(p^{1/2 + \delta})}}{\hbar}\Lambda_{\Z_0}\right.\bigr.\nonumber\\
        &\bigl.\left.-\frac{n_z\sigma_4}{\hbar}\Lambda_\epsilon- \frac{nn_z\tilde{\omega}_3 - \sigma_{\bm\beta}^2nn_z/(\lambda n_wm_w(-\lambda)^2)\left(m_w(-\lambda)\gamma_1 + \xi_1 - m/p\right)}{\hbar}\right) - \bigr.\nonumber\\
        &\bigl.\mathbb P\left(\frac{n_z\sigma_3}{\hbar}\Lambda_{\X_0} < t- \frac{\sigma_{\bm\beta}^2nn_z\bm\beta^\top \bm\Sigma R\bm\Sigma\bm\beta - nn_z\sigma_{\bm\beta}^2/(\lambda n_wm_w(-\lambda)^2)\left(m_w(-\lambda)\gamma_1 + \xi_1 - m/p\right)}{\hbar}  \right.\bigr.\nonumber\\
        &\left.\bigl. - \frac{n_z\sqrt{n\sigma_\epsilon^2\tilde{\omega}_2 + O_p(p^{\delta -3/2})}}{\hbar}\Lambda_\epsilon \right.\bigr.\nonumber\\
        &\left.\bigl.-\frac{\sqrt{\sigma_{\epsilon_z}^2n_z\left\{n\left(n\tilde{\omega}_1 + \Tr(\bm\Sigma^{1/2}\R\bm\Sigma \R\bm\Sigma^{1/2})\|\bm\Sigma^{1/2}\bm\beta\|_2^2\right) + n\sigma_\epsilon^2\Tr(\bm\Sigma^{1/2}\R\bm\Sigma \R\bm\Sigma^{1/2})\right\} + O_p(p^{1/2+\delta})}}{\hbar}\Lambda_{\epsilon_z} \right.\bigr.\nonumber\\
        &\bigl.\left.-\frac{\sqrt{2n_zn^2\tilde{\omega}_3^2 + n_z\|\bm\Sigma^{1/2}\bm\beta\|_2^2\left\{n\left\{n\|\bm\Sigma^{1/2}\R\bm\Sigma\bm\beta\|_2^2 + \Tr(\bm\Sigma^{1/2}\R\bm\Sigma \R\bm\Sigma^{1/2})\left(\tilde{\omega}_2 + \sigma_\epsilon^2\right)\right\}\right\}}}{\hbar}\Lambda_{\Z_0}\right)\Bigg|\nonumber\\
        &\qquad\leq O_p\left(\max(p^{-1/2}, p^{-2\delta})\right).
    \end{align}
}

\subsection{Berry-Esseen bounds with the randomness of genetic effects}
In this section, conditional on $\W_0$, we consider the randomness of $\bm\beta$ using Theorem~S\ref{thm:BE_quad_form}. 
Before that, we first quantify the first-order concentration of terms that involve in the variance of \cref{ineq:BE_ref_A2_X_app}. 
Notice the symmetric quantities related to $\bm\beta$ concentrate following Lemma~S\ref{prop:quad_general_limit}. For $\forall \delta \in (0,1/2)$, we have 
\begin{equation*}
    \begin{aligned}
        &\mathbb P\left(\left|\bm\beta^\top\bm\Sigma\bm\beta - \sigma_{\bm\beta}^2\gamma_1\right| < O_p(p^{\delta-1/2})\right) \geq 1-O_p(p^{-1+2\delta}),\\
        &\mathbb P\left(\left|\frac{p}{\sigma_{\bm\beta}^2}\bm\beta^\top\bm\Sigma \R\bm\Sigma\bm\beta - \Tr(\bm\Sigma \R\bm\Sigma\mathbb\I_m)\right| < O_p(p^{-\delta})\right) \geq 1-O_p\left(p^{2\delta}\Tr\left\{(\bm\Sigma \A^{-1}\bm\Sigma\mathbb\I_m)^2\right\}\right)\\
         &\implies\mathbb P\left(\left|\bm\beta^\top\bm\Sigma \R\bm\Sigma\bm\beta - \frac{\sigma_{\bm\beta}^2}{p}\Tr(\bm\Sigma \R\bm\Sigma\mathbb\I_m)\right| < O_p(p^{-1-\delta})\right) \geq 1-O_p(p^{-1+2\delta}),\quad \mbox{and} \\
        &\mathbb P\left(\left|\frac{p}{\sigma_{\bm\beta}^2}\|\bm\Sigma^{1/2}\R\bm\Sigma\bm\beta\|_2^2 - \Tr(\bm\Sigma \R\bm\Sigma \R\bm\Sigma\mathbb\I_m)\right| < O_p(p^{-1-\delta})\right) \geq 1-O_p\left(p^{2+2\delta}\Tr\left\{(\bm\Sigma \R\bm\Sigma \R\bm\Sigma\mathbb\I_m)^2\right\}\right)\\
        &\implies\mathbb P\left(\left|\|\bm\Sigma^{1/2}\R\bm\Sigma\bm\beta\|_2^2 - \frac{\sigma_{\bm\beta}^2}{p}\Tr(\bm\Sigma \R\bm\Sigma \R\bm\Sigma\mathbb\I_m)\right| < O_p(p^{-2-\delta})\right) \geq 1-O_p(p^{-1+2\delta}).
    \end{aligned}
\end{equation*}
For $\forall \epsilon_1, \epsilon_2, \epsilon_3 > 0$ and $\delta \in (0,1/2)$, we define the following subset where quantities related to $\bm\beta$ concentrate properly 
\begin{align*} 
    \Xi_5(\epsilon_1, \epsilon_2, &\epsilon_3) \coloneqq \left\{\bm\beta:\left\{\left|\bm\beta^\top\bm\Sigma\bm\beta - \sigma_{\bm\beta}^2\gamma_1\right| < p^{\delta-1/2}\epsilon_1\right\}\cap\right.\left\{\left|\|\bm\Sigma^{1/2}\R\bm\Sigma\bm\beta\|_2^2 -\right.\right.\\
    &\left.\left.\left. \frac{\sigma_{\bm\beta}^2}{p}\Tr(\bm\Sigma \R\bm\Sigma \R\bm\Sigma\mathbb\I_m)\right| < p^{-2-\delta}\epsilon_2\right\}\cap\left\{\left|\bm\beta^\top\bm\Sigma \R\bm\Sigma\bm\beta - \frac{\sigma_{\bm\beta}^2}{p}\Tr(\bm\Sigma \R\bm\Sigma\mathbb\I_m)\right| < p^{-1-\delta}\epsilon_3\right\}\right\}.
\end{align*}
We have shown that $\mathbb P(\bm\beta \in \Xi_5) \geq 1-O_p(p^{-1+2\delta})$. We introduce the following notations for simplicity
\begin{enumerate}
    \item $\tilde{\kappa}_1 = \Tr(\bm\Sigma \R\bm\Sigma)$;
    \item $\tilde{\kappa}_2 = \Tr(\bm\Sigma^{1/2}\R\bm\Sigma \R\bm\Sigma^{1/2})$;
    \item $\tilde{\kappa}_3 = \Tr(\bm\Sigma \R\bm\Sigma \R\bm\Sigma)$;
    \item $\tilde{\kappa}_4 = \Tr(\bm\Sigma \R\bm\Sigma \R\bm\Sigma\mathbb\I_m)$;
    \item $\tilde{\kappa}_5 = \Tr\left\{(\bm\Sigma \R\bm\Sigma\mathbb \I_m)^2\right\}$; and 
    \item $\tilde{\kappa}_6 = \Tr(\bm\Sigma \R\bm\Sigma\mathbb\I_m)$.
\end{enumerate}
Notice that for $\bm\beta \in \Xi_5$, we can further simplify our Berry-Esseen upper bounds in \cref{concent:BE_bound_A2_ref_X} as follows 
\begin{align}
    \label{concent:BE_bound_A2_ref_beta}
        &O_p\left(\sqrt{\frac{n_z\left\{n^2\left(n\tilde{\omega}_1 + \tilde{\kappa}_2\|\bm\Sigma^{1/2}\bm\beta\|_2^2\right)^2 + n^2\tilde{\kappa}_2^2\right\} + O_p(p^{1/2+\delta})}{n_z^2\left\{n^2\left(n\tilde{\omega}_1 + \tilde{\kappa}_2\|\bm\Sigma^{1/2}\bm\beta\|_2^2\right)^2 + n^2\tilde{\kappa}_2^2\right\} + O_p(p^{3/2 + \delta})}}\right)\nonumber \\
        &= O_p\left(\sqrt{\frac{n_z\left\{n^2\left(n\tilde{\kappa}_4\right)^2 + n^2\tilde{\kappa}_2^2\right\} + O_p(p^{1/2+\delta})}{n_z^2\left\{n^2\left(n\tilde{\kappa}_4\right)^2 + n^2\tilde{\kappa}_2^2\right\} + O_p(p^{3/2 + \delta})}}\right) \quad \mbox{and} \nonumber\\
        &O_p\left(\sqrt{\frac{ n\left\{\mathbb E\left(x_0^4-3\right)\sum_{i=1}^p(\bm\Sigma^{1/2}\R\bm\Sigma\bm\beta)_i^4 + 3\tilde{\omega}_1^2\right\} + O_p(n^{1/2+\delta}n_w^{-4})}{n^2\tilde{\omega}_2^2 + O_p(n^{3/2+\delta}n_w^{-4})}}\right)\nonumber\\
        &= O_p\left(\sqrt{\frac{3n\left(\sigma_{\bm\beta}^2/p\tilde{\kappa}_4\right)^2+ O_p(n^{1/2+\delta}n_w^{-4})}{n^2\left(\sigma_{\bm\beta}^2/p\tilde{\kappa}_4\right)^2 + O_p(n^{3/2+\delta}n_w^{-4})}}\right).
\end{align}

Now we use Theorem~S\ref{thm:BE_quad_form} to obtain the quantitative CLT for the quadratic quantities. Theorem~S\ref{thm:BE_quad_form} holds because we have $\psi_1\mathbb\I_p \preceq n\left(\W^\top \W + n_w\lambda\mathbb\I_p\right)^{-1}\preceq \psi_2\mathbb\I_p$, under the condition $n_w\asymp n$ in Assumption~\ref{a:anisotropic regularity}.
Denote
\begin{align}\label{delta_5}
    \sigma_5^2 = \left(\mathbb E\left({\bm\beta}^4\right)-3\frac{\sigma_{\bm\beta}^4}{p^2}\right)\sum_{i=1}^p{{\left(\bm\Sigma \R\bm\Sigma\mathbb\I_m\right)}_{i,i}}^2 + \frac{2\sigma_{\bm\beta}^4}{p^2}\tilde{\kappa}_5.
\end{align}
By Theorem~S\ref{thm:BE_quad_form}, we have
   \begin{align}\label{ineq:BE_ref_A2_beta_app}
        \sup_{t\in\mathbb R}&\left|\mathbb P\left(nn_z\frac{\tilde{\omega}_3 -\sigma_{\bm\beta}^2/p\tilde{\kappa}_6}{\sigma_5}\frac{\sigma_5}{\hbar} < t- \frac{n_z\sigma_3}{\hbar}\Lambda_{\X_0}- \right.\right.\nonumber \\
        &\left.\left. nn_z\frac{\sigma_{\bm\beta}^2/p\tilde{\kappa}_6 - \sigma_{\bm\beta}^2/(\lambda n_wm_w(-\lambda)^2)\left(m_w(-\lambda)\gamma_1 + \xi_1 - m/p\right)}{\hbar} - \right.\right.\nonumber \\
        &\left.\left.\frac{n_z\sqrt{n\sigma_\epsilon^2\tilde{\omega}_2 + O_p(p^{\delta - 3/2})}}{\hbar}\Lambda_\epsilon-\frac{\sqrt{\sigma_{\epsilon_z}^2n_z\left\{n\left\{n\tilde{\omega}_1 + \tilde{\kappa}_2\|\bm\Sigma^{1/2}\bm\beta\|_2^2\right\} + n\sigma_\epsilon^2\tilde{\kappa}_2\right\} + O_p(p^{1/2+\delta})}}{\hbar}\Lambda_{\epsilon_z} - \right.\bigr.\right.\nonumber \\
        &\left.\bigl.\left.\frac{\sqrt{2n_zn^2\tilde{\omega}_3^2 + n_z\|\bm\Sigma^{1/2}\bm\beta\|_2^2\left\{n\left\{n\tilde{\omega}_2 + \tilde{\kappa}_2\left(\|\bm\Sigma^{1/2}\bm\beta\|_2^2 + \sigma_\epsilon^2\right)\right\}\right\}}}{\hbar}\Lambda_{\Z_0}\right)\right.\nonumber \\
        &\left.-\mathbb P\left(nn_z\frac{\sigma_5}{\hbar}\Lambda_{\bm\beta} < t- \frac{n_z\sqrt{n\sigma_{\bm\beta}^4/p^2\left\{2\tilde{\kappa}_6^2 + \tilde{\kappa}_4\Tr(\bm\Sigma\mathbb\I_m)\right\} + O_p(p^{\delta-3/2})}}{\hbar}\Lambda_{\X_0}-\right.\right.\nonumber \\
        &\left. \left.nn_z\frac{\sigma_{\bm\beta}^2/p\tilde{\kappa}_6 - \sigma_{\bm\beta}^2/(\lambda n_wm_w(-\lambda)^2)\left(m_w(-\lambda)\gamma_1 + \xi_1 - m/p\right)}{\hbar} - \frac{n_z\sqrt{n\sigma_\epsilon^2\sigma_{\bm\beta}^2/p\tilde{\kappa}_4 + O_p(p^{\delta - 3/2})}}{\hbar}\Lambda_\epsilon\right.\bigr.\right.\nonumber \\
        &\left.\left.\bigl.-\frac{\sqrt{\sigma_{\epsilon_z}^2n_z\left\{n\sigma_{\bm\beta}^2/p\left\{n\tilde{\kappa}_4 + \tilde{\kappa}_2\Tr(\bm\Sigma\mathbb\I_m)\right\} + n\sigma_\epsilon^2\tilde{\kappa}_2\right\} + O_p(p^{1/2+\delta})}}{\hbar}\Lambda_{\epsilon_z}\right.\bigr.\right.\nonumber \\
        &\left.\bigl.\left.-\frac{\sqrt{2n_zn^2\sigma_{\bm\beta}^4/p^2\tilde{\kappa}_6^2 + n_z\sigma_{\bm\beta}^2\gamma_1\left[n\left\{n\sigma_{\bm\beta}^2/p\tilde{\kappa}_6 + \tilde{\kappa}_2(\sigma_{\bm\beta}^2\gamma_1 + \sigma_\epsilon^2)\right\}\right]}}{\hbar}\Lambda_{\Z_0}\right)\right| \leq O_p(m^{-1/5}).
    \end{align}

\subsection{Concentration using anisotropic local law with the randomness of reference panel}
In this section, we provide the first-order limit of quantities related to $\W_0$ using the anisotropic local law \citep{anisotropic_local_law}. Let $\e_i$ be a $p\times1$ vector with an entry of $1$ on $i_{th}$ row. 
Notice that $|\e_i|_2 = 1$ and we have $\sum_{i=1}^m \e_i^\top\bm\Sigma \e_i = \Tr(\bm\Sigma\mathbb\I_m)$. We will use this relationship to isolate each diagonal element of the desired matrices and apply the anisotropic local law to obtain their first-order limits. By summation and applying the union bound, we will obtain the concentration over the desired trace. Note that we can drop the matrix $\mathbb \I_m$ by simply summing up to $m$. This property is convenient as it allows us to avoid most of the $\mathbb \I_m$ when applying the anisotropic local law. However, some quantities of $\sigma_5^2$ in \cref{delta_5} still require more careful analysis due to their special structure.
Specifically, there are four quantities in \cref{ineq:BE_ref_A2_beta_app} that require further analysis, including 
\begin{enumerate}
    \item $\Tr(\bm\Sigma \R\bm\Sigma)$;
    \item $\Tr(\bm\Sigma^{1/2}\R\bm\Sigma \R\bm\Sigma^{1/2})$;
    \item $\Tr(\bm\Sigma \R\bm\Sigma \R\bm\Sigma)$; and 
    \item $\Tr\left\{(\bm\Sigma \R\bm\Sigma\mathbb \I_m)^2\right\}$.
\end{enumerate}

The matrix $\mathbb \I_m$ does not cause trouble for the first three quantities as we only sum up to $m$ terms. However, the last quantity requires more non-trivial analysis, as we cannot avoid the matrix $\mathbb \I_m$ in its complicated structure. This makes the analysis much different from that in \cref{concent:ASigA}, and we will illustrate this point later.

For the first quantity $\Tr(\bm\Sigma \R\bm\Sigma)$, similar to our argument in Lemma~S\ref{lemma:one_res_limit}, for some small $\vartheta > 0$ and a sufficiently large $D$, the following inequality holds with probability at least $1- O_p(p^{-D})$

$$\left|-\lambda \e_i^\top \bm\Sigma \R\bm\Sigma \e_i + \frac{1}{n_w}\e_i^\top(\mathbb \I_p + m_w\left(-\lambda)\bm\Sigma\right)^{-1}\bm\Sigma^2\e_i\right| \leq \psi(-\lambda)O_p(\frac{p^{\vartheta}}{n_w}).$$ 
Therefore, by union bound, with probability of at least $1-O_p(p^{-D+1})$, we have
\begin{equation*}
    \begin{aligned}
        &\left|\Tr\left(\bm\Sigma \R\bm\Sigma\right) - \Tr\left\{(\mathbb\I_p + m_w(-\lambda)\bm\Sigma)^{-1}\bm\Sigma^2\right\}\right|\leq\psi(-\lambda)O_p(p^{\vartheta})\quad \mbox{and} \\
        &\left|\Tr\left(\bm\Sigma \R\bm\Sigma\mathbb\I_m\right) - \Tr\left\{(\mathbb\I_p + m_w(-\lambda)\bm\Sigma)^{-1}\bm\Sigma^2\mathbb\I_m \right\}\right|\leq\psi(-\lambda)O_p(p^{\vartheta}).
    \end{aligned}
\end{equation*}

For the second quantity $\Tr(\bm\Sigma^{1/2}\R\bm\Sigma \R\bm\Sigma^{1/2})$ and third quantity $\Tr(\bm\Sigma \R\bm\Sigma \R\bm\Sigma)$, notice that they differ only by two fixed matrices $\bm\Sigma^{1/2}$ with bounded eigenvalues. Then by similar argument in \cref{concent:ASigA}, we can conclude that, with probability of at least $1-O_p(p^{-D})$, we have 
\begin{equation*}
    \begin{aligned}
        &\left|\e_i^\top\bm\Sigma^{1/2}\R\bm\Sigma \R\bm\Sigma^{1/2}\e_i - \frac{m_w'(-\lambda)}{m_w(-\lambda)^2}\frac{1}{n_w^2\lambda^2}e_i^\top(\mathbb\I_p + m_w(-\lambda)\bm\Sigma)^{-2}\bm\Sigma^2 \e_i\right| \leq \psi(-\lambda)O_p(p^{\vartheta - 2}) \quad\mbox{and}\\
        &\left|\e_i^\top\bm\Sigma \R\bm\Sigma \R\bm\Sigma \e_i - \frac{m_w'(-\lambda)}{m_w(-\lambda)^2}\frac{1}{n_w^2\lambda^2}\e_i^\top(\mathbb\I_p + m_w(-\lambda)\bm\Sigma)^{-2}\bm\Sigma^3 \e_i\right| \leq \psi(-\lambda)O_p(p^{\vartheta - 2}).
    \end{aligned}
\end{equation*}
It follows that 
\begin{align*}
        &\left|\Tr\left(\bm\Sigma^{1/2}\R\bm\Sigma \R\bm\Sigma^{1/2}\right) - \frac{m_w'(-\lambda)}{m_w(-\lambda)^2}\frac{1}{n_w^2\lambda^2}\Tr\left\{(\mathbb\I_p + m_w(-\lambda)\bm\Sigma)^{-2}\bm\Sigma^2\right\}\right| \leq \psi(-\lambda)O_p(p^{\vartheta-1}),\\
        &\left|\Tr\left(\bm\Sigma \R\bm\Sigma \R\bm\Sigma\right) - \frac{m_w'(-\lambda)}{m_w(-\lambda)^2}\frac{1}{n_w^2\lambda^2}\Tr\left\{(\mathbb\I_p + m_w(-\lambda)\bm\Sigma)^{-2}\bm\Sigma^3\right\}\right| \leq \psi(-\lambda)O_p(p^{\vartheta-1}),\\
        &\left|\Tr\left(\bm\Sigma^{1/2}\R\bm\Sigma \R\bm\Sigma^{1/2}\mathbb\I_m\right) - \frac{m_w'(-\lambda)}{m_w(-\lambda)^2}\frac{1}{n_w^2\lambda^2}\Tr\left\{(\mathbb\I_p + m_w(-\lambda)\bm\Sigma)^{-2}\bm\Sigma^2\mathbb\I_m\right\}\right| \leq \psi(-\lambda)O_p(p^{\vartheta-1}),
\end{align*}
and 
\begin{align*}
        &\left|\Tr\left(\bm\Sigma \R\bm\Sigma \R\bm\Sigma\mathbb\I_m\right) - \frac{m_w'(-\lambda)}{m_w(-\lambda)^2}\frac{1}{n_w^2\lambda^2}\Tr\left\{(\mathbb\I_p + m_w(-\lambda)\bm\Sigma)^{-2}\bm\Sigma^3\mathbb\I_m\right\}\right| \leq \psi(-\lambda)O_p(p^{\vartheta-1}),
\end{align*}
with probability of at least $1-O_p(p^{-D+1})$.

Now we consider the last quantity $\Tr\left\{(\bm\Sigma \R\bm\Sigma\mathbb \I_m)^2\right\}$. 
Note that 
$$\Tr\left\{(\bm\Sigma \R\bm\Sigma\mathbb \I_m)^2\right\} = \Tr\left(\bm\Sigma \R\bm\Sigma\mathbb\I_m\bm\Sigma \R\bm\Sigma\mathbb\I_m\right) = \sum_{i=1}^m \e_i^\top\bm\Sigma \R\bm\Sigma\mathbb\I_m\bm\Sigma \R\bm\Sigma \e_i.$$
We have 
\begin{align*}
    \left.\e_i^\top \bm\Sigma\left(\frac{\bm\Sigma^{1/2}\W_0^\top \W_0\bm\Sigma^{1/2}}{n_w} + \lambda\mathbb\I_p\right)^{-1}\bm\Sigma\mathbb\I_m\bm\Sigma \left(\frac{\bm\Sigma^{1/2}\W_0^\top \W_0\bm\Sigma^{1/2}}{n_w} + \lambda\mathbb\I_p\right)^{-1}\bm\Sigma \e_i \right.\\
    \left.=\frac{d\left(\e_i^\top\left\{\bm\Sigma^{-1}\left(n_w^{-1}\bm\Sigma^{1/2}\W_0^\top \W_0\bm\Sigma^{1/2} + \lambda\mathbb\I_p\right)\bm\Sigma^{-1} -\tau\mathbb\I_m\right\}^{-1}\e_i\right)}{d\tau}\right|_{\tau = 0}.
\end{align*}
As mentioned previously, we cannot avoid the matrix $\mathbb\I_m$ here. Therefore, we match $\mathbb\I_m$ with the derivative of some quantity. Similar to our proof of Lemma~S\ref{lemma:one_res_limit}, we have 
\begin{align*}
    \vec{\underline{\bm v}} = \vec{\underline{\bm w}} = 
    \left[
        \begin{array}{c}
            (\bm\Sigma^{-1} - \tau/\lambda\bm\Sigma^{1/2}\mathbb\I_m\bm\Sigma^{1/2})^{-1}\bm\Sigma^{1/2}\e_i \\ \hdashline[2pt/2pt]
            0
        \end{array}
    \right].
\end{align*}
We denote the Stieltjes transform of the asymptotic eigenvalue distribution of 
$$\bm\Sigma^{-1} - \tau/\lambda\bm\Sigma^{1/2}\mathbb\I_m\bm\Sigma^{1/2}$$ 
as $m_{w}(\lambda,\tau)$. Here we are not able to establish an explicit linear relationship between eigenvalues of $(\bm\Sigma^{-1} - \tau/\lambda\bm\Sigma^{1/2}\mathbb\I_m\bm\Sigma^{1/2})^{-1}$ and eigenvalues of $\bm\Sigma$. Therefore, we approach 
$$\left.dm_{w}(\lambda,\tau)/d\tau\right|_{\tau = 0}$$ 
through implicit function and eigenvalue perturbation formula. Specifically, $m_{w}(\lambda,\tau)$ satisfies following equality
\begin{equation}
\label{equ:eigen_pert_2}
    \begin{aligned}
        \frac{1}{m_{w}(\lambda,\tau)} = -\lambda + \frac{\phi_w}{p}\sum_{i=1}^p\frac{\pi_i'}{1+\pi_i'm_{w}(\lambda,\tau)},
    \end{aligned}
\end{equation}
where $\pi_i'$ denotes the $i_{th}$ eigenvalue of $(\bm\Sigma^{-1} - \tau/\lambda\bm\Sigma^{1/2}\mathbb\I_m\bm\Sigma^{1/2})^{-1}$, $i=1,\ldots,p$. 
It is worth mentioning that both $m_{w}(\lambda,\tau)$ and $\pi_i'$ depend on $\tau$. 
Now taking the derivative with respect to $\tau$ on both sides of \cref{equ:eigen_pert_2}, we have
   \begin{align*}
        \frac{\phi_w}{p}\sum_{i=1}^p&\frac{d\pi_i'}{d\tau}\left\{\frac{1}{1+\pi_i'm_{w}(\lambda,\tau)} - \frac{m_{w}(\lambda,\tau)\pi_i'}{(1+\pi_i'm_{w}(\lambda,\tau))^2}\right\} \\
        &= -\frac{dm_{w}(\lambda,\tau)}{d\tau}\left\{\frac{1}{m_{w}(\lambda,\tau)^2} - \frac{\phi_w}{p}\sum_{i=1}^p\frac{\pi_i'^2}{(1+\pi_i'm_{w}(\lambda,\tau))^2}\right\}.
    \end{align*}
Note that $\left.m_{w}(\lambda,\tau)\right|_{\tau = 0} = m_w(-\lambda)$ and $\left.\pi_i = \pi_i'\right|_{\tau = 0}$. It follows that 
\begin{equation*}
    \begin{aligned}
        \frac{\phi_w}{p}\sum_{i=1}^p\left.\frac{d\pi_i'}{d\tau}\right|_{\tau = 0}\frac{1}{(1+ m_w(-\lambda)\pi_i)^2} = -\left.\frac{dm_{w}(\lambda,\tau)}{d\tau}\right|_{\tau = 0}\frac{1}{m_w'(-\lambda)}.
    \end{aligned}
\end{equation*}
Therefore, as our target is $\left.d\pi_i'/d\tau\right|_{\tau = 0}$, which represents how the eigenvalues of $\bm\Sigma$ change after a small perturbation (though such perturbation is non-linear due to the inverse), we can quantify this change using the formula for eigenvalue perturbation \citep{hogben2013handbook}
\begin{equation*}
    \begin{aligned}
        \left.\frac{d\pi_i'}{d\tau}\right|_{\tau = 0} &= \bm u_i^\top \left.\frac{d(\bm\Sigma^{-1} - \tau/\lambda\bm\Sigma^{1/2}\mathbb\I_m\bm\Sigma^{1/2})^{-1}}{d\tau}\right|_{\tau = 0}\bm u_i\\
        &=  \left.\bm u_i^\top\left(\bm\Sigma^{-1}-\tau/\lambda\bm\Sigma^{1/2}\mathbb\I_m\bm\Sigma^{1/2}\right)^{-1}\lambda^{-1}\bm\Sigma^{1/2}\mathbb\I_m\bm\Sigma^{1/2}\left(\bm\Sigma^{-1}-\tau/\lambda\bm\Sigma^{1/2}\mathbb\I_m\bm\Sigma^{1/2}\right)^{-1}\bm u_i\right|_{\tau = 0}\\
        & = \lambda^{-1} \bm u_i^\top \bm\Sigma^{3/2}\mathbb\I_m\bm\Sigma^{3/2}\bm u_i = \pi_i^3\lambda^{-1}\left\langle \bm u_i, \mathbb\I_m \bm u_i\right\rangle,
    \end{aligned}
\end{equation*}
where $\bm u_i$ represents the $i_{th}$ eigenvector of $\bm\Sigma$. It follows that 
\begin{equation*}
    \begin{aligned}
        \left.\frac{dm_{w}(\lambda,\tau)}{d\tau}\right|_{\tau = 0} = -\frac{\phi_w}{\lambda p}\sum_{i=1}^p\frac{m_w(-\lambda)'\pi_i^3}{(1+\pi_im_w(-\lambda))^2}\left\langle u_i,\mathbb\I_m u_i\right\rangle.
    \end{aligned}
\end{equation*}
Let
\begin{align*}
   \bm \Pi(-\lambda) = \left[\begin{array}{cc}
        -\left(\bm\Sigma^{-1}-\tau/\lambda\bm\Sigma^{1/2}\mathbb\I_m\bm\Sigma^{1/2} + m_{w}(\lambda, \tau)\right)^{-1} &  0\\
        0 & m_{w}(\lambda,\tau)
    \end{array} \right],
\end{align*}
with probability of at least $1-O_p(p^{-D})$, we have 
\begin{align*}
    \left|\langle \vec{\underline{\bm v}}, \underline{\bm\Sigma}^{-1}(\G(-\lambda) -\bm  \Pi(-\lambda))\underline{\bm\Sigma}^{-1}\vec{\underline{\bm w}}\rangle\right|  \prec \psi(-\lambda)O_p(1),
\end{align*}
where $$\G(-\lambda) = \left[\begin{array}{cc}
        -(\bm\Sigma^{-1} - \tau/\lambda\bm\Sigma^{1/2}\mathbb\I_m\bm\Sigma^{1/2}) & \frac{\W_0^\top}{\sqrt{n_w}} \\
        \frac{\W_0}{\sqrt{n_w}} & \lambda \mathbb\I_{n_w}
    \end{array} \right]^{-1}$$
and 
    $$ \underline{\bm\Sigma} = \left[\begin{array}{cc}
         -(\bm\Sigma^{-1} - \tau/\lambda\bm\Sigma^{1/2}\mathbb\I_m\bm\Sigma^{1/2})^{-1} & 0 \\
            0   & \mathbb\I_{n_w}
    \end{array}\right].$$
Therefore, we have 
\begin{equation*}
    \begin{aligned}
        &\left|\e_i^\top\bm\Sigma^{1/2}\left\{\left(\frac{\W_0^\top \W_0}{\lambda n_w} + \bm\Sigma^{-1} - \frac{\tau}{\lambda}\bm\Sigma^{1/2}\mathbb\I_m\bm\Sigma^{1/2}\right)^{-1} -\left(\bm\Sigma^{-1}-\frac{\tau}{\lambda}\bm\Sigma^{1/2}\mathbb\I_m\bm\Sigma^{1/2} + m_{w}(\lambda,\tau)\right)^{-1}\right\}\bm\Sigma^{1/2}\e_i\right| \\
        &\qquad\prec\psi(-\lambda)O_p(1).
    \end{aligned}
\end{equation*}
Then we have the following inequality
\begin{equation*}
    \begin{aligned}
        &\left|\e_i^\top\bm\Sigma^{1/2}\left(\frac{\W_0^\top \W_0}{n_w} + \lambda\bm\Sigma^{-1} - \tau\bm\Sigma^{1/2}\mathbb\I_m\bm\Sigma^{1/2}\right)^{-1}\bm\Sigma^{1/2}\e_i -\right. \\
        &\left.\qquad\frac{1}{\lambda} \e_i^\top\bm\Sigma^{1/2}\left(\bm\Sigma^{-1}-\frac{\tau}{\lambda}\bm\Sigma^{1/2}\mathbb\I_m\bm\Sigma^{1/2} + m_{w}(\lambda,\tau)\right)^{-1}\bm\Sigma^{1/2}\e_i\right|
        \leq \psi(-\lambda)O_p(p^{\vartheta-1}).
    \end{aligned}
\end{equation*}
Taking the derivative with respect to $\tau$ and set $\tau = 0$, we have 
\begin{equation*}
    \begin{aligned}
        &\Bigg|\e_i^\top \bm\Sigma\left(\frac{\bm\Sigma^{1/2}\W_0^\top \W_0\bm\Sigma^{1/2}}{n_w} + \lambda\mathbb\I_p\right)^{-1}\bm\Sigma\mathbb\I_m\bm\Sigma \left(\frac{\bm\Sigma^{1/2}\W_0^\top \W_0\bm\Sigma^{1/2}}{n_w} + \lambda\mathbb\I_p\right)^{-1}\bm\Sigma \e_i-  \\
        &\qquad\frac{1}{\lambda^2}\e_i^\top\bm\Sigma(\mathbb\I_p+\mathfrak{m}_w\bm\Sigma)^{-1}\left(\bm\Sigma\mathbb\I_m\bm\Sigma - \lambda\frac{dm_{w}(\lambda,\tau)}{d\tau}\right)(\mathbb\I_p+\mathfrak{m}_w\bm\Sigma)^{-1}\bm\Sigma \e_i\Bigg|\leq\psi(-\lambda)O_p(1)\\
       \implies &\Bigg|\e_i^\top \bm\Sigma\left(\frac{\bm\Sigma^{1/2}\W_0^\top \W_0\bm\Sigma^{1/2}}{n_w} + \lambda\mathbb\I_p\right)^{-1}\bm\Sigma\mathbb\I_m\bm\Sigma\left(\frac{\bm\Sigma^{1/2}\W_0^\top \W_0\bm\Sigma^{1/2}}{n_w} + \lambda\mathbb\I_p\right)^{-1}\bm\Sigma \e_i - \\
        &\qquad\frac{1}{\lambda^2}\e_i^\top\bm\Sigma(\mathbb\I_p+\mathfrak{m}_w\bm\Sigma)^{-1}\left(\bm\Sigma\mathbb\I_m\bm\Sigma + \mathfrak{m}_w'\phi_w \mathfrak{n}(\bm\Sigma, m)\right)(\mathbb\I_p+\mathfrak{m}_w\bm\Sigma)^{-1}\bm\Sigma \e_i\Bigg|\leq\psi(-\lambda)O_p(1),
    \end{aligned}
\end{equation*}
where $$\mathfrak{n}(\bm\Sigma, m) = \frac{1}{p}\sum_{i=1}^p\frac{\phi_w\mathfrak{m}_w'\pi_i^3}{(1+\pi_i\mathfrak{m}_w)^2}\left\langle u_i, \mathbb\I_m u_i\right\rangle.$$
Now we denote 
$$\mathcal{M} = \bm\Sigma(\mathbb\I_p + \mathfrak{m}_w\bm\Sigma)^{-1}\left(\bm\Sigma\mathbb\I_m\bm\Sigma + \mathfrak{n}(\bm\Sigma, m)\right)(\mathbb\I_p + \mathfrak{m}_w\bm\Sigma)^{-1}\bm\Sigma.$$
After proper rescaling, we may conclude that, with  probability of at least $1-O_p(p^{-D})$, we have
\begin{equation*}
    \begin{aligned}
        \left|\e_i^\top\bm\Sigma \R\bm\Sigma\mathbb\I_m\bm\Sigma \R\bm\Sigma \e_i - \frac{1}{\lambda^2n_w^2}\e_i^\top\mathcal{M} \e_i\right|\leq \psi(-\lambda)O_p(p^{\vartheta-2})
    \end{aligned}
\end{equation*}
for some small enough $\vartheta > 0$.

Combining with our result on $\e_i^\top\bm\Sigma \R\bm\Sigma \e_i$, using union bound, we can show that, with probability of at least $1-O_p(p^{-D+1})$, we have 
\begin{equation*}
    \begin{aligned}
        &\left|\sigma_5^2 - \left\{\frac{1}{\lambda^2n_w^2}\left(\mathbb E\left(\bm\beta^4\right) - \frac{3\sigma_{\bm\beta}^4}{p^2}\right)\sum_{i=1}^p\left\{\left(\mathbb\I_p+\mathfrak{m}_w\bm\Sigma\right)^{-1}\bm\Sigma^2\mathbb\I_m\right\}_{i,i}^2 + \frac{2\sigma_{\bm\beta}^4}{\lambda^2n_w^2p^2}\Tr\left(\mathcal{M}\right)\right\}\right|\\
        &\qquad\leq \psi(-\lambda)O_p(p^{\vartheta-3}).
    \end{aligned}
\end{equation*}
Our concentration inequalities over the randomness of $\W_0$ also help us simplify the Berry-Esseen bounds in \cref{concent:BE_bound_A2_ref_beta} as follows 
\begin{equation*}
    \begin{aligned}
        &O_p\left(\sqrt{\frac{n_z\left\{n^2\left(n\tilde{\kappa}_4\right)^2 + n^2\tilde{\kappa}_2^2\right\} + O_p(p^{1/2+\delta})}{n_z^2\left\{n^2\left(n\tilde{\kappa}_4\right)^2 + n^2\tilde{\kappa}_2^2\right\} + O_p(p^{3/2 + \delta})}}\right) = O_p(m^{-1/2}) \quad \mbox{and} \\
        &O_p\left(\sqrt{\frac{3n\left(\sigma_{\bm\beta}^2/p\tilde{\kappa}_4\right)^2+ O_p(n^{1/2+\delta}n_w^{-4})}{n^2\left(\sigma_{\bm\beta}^2/p\tilde{\kappa}_4\right)^2 + O_p(n^{3/2+\delta}n_w^{-4})}}\right) = O_p(m^{-1/2})
    \end{aligned}
\end{equation*}

Combining \cref{ineq:BE_ref_A2_eps_z_app}, \cref{ineq:BE_ref_A2_Z_app}, \cref{ineq:BE_ref_A2_eps_app}, \cref{ineq:BE_ref_A2_X_app}, and \cref{ineq:BE_ref_A2_beta_app}, the worst Berry-Esseen upper bound is of order $O_p(m^{-1/5})$. By combining all those independent Gaussian random variables and using the convolution formula for independent Gaussian variables, we obtain the following quantitative CLT for the numerator
\begin{equation}
\label{ineq:BE_ref_A2_num}
    \begin{aligned}
        \sup_{t\in\mathbb R}\left|\bm H_{\text{W}}(t) - \Phi(t)\right| \leq O_p(m^{-1/5}),
    \end{aligned}
\end{equation}
with probability of at least $1-O_p(p^{-D+1})$ over the randomness of $\W_0$.

\subsection{Limits of the denominator}
In this section, we provide the limits of the quantity $\|\hat{\y_z}\|$ in the denominator. The following lemma provides us with the desired concentration. 
\begin{lem.s}
\label{concent:ref_A2_denom_1}
Under Assumptions~\ref{a:Sigmabound}-\ref{a:anisotropic regularity},
with probability of at least $1-O_p(p^{-D})$ over the randomness of $\W_0$, we have
\begin{align*}
        &\mathbb P\left(\left|\y^\top \X \R\Z^\top \Z \R\X^\top \y - \left(\frac{\mathfrak{m}_w'}{\mathfrak{m}_w^2}\right)\frac{\sigma_{\bm\beta}^2nn_z}{\lambda^2n_w^2}\left\{\frac{n}{\mathfrak{m}_w^3}\left(3\xi_1-\xi_2+\mathfrak{m}_w\gamma_1- \frac{2m}{p}\right)\right.\right.\right.\\
        &\left.\left.\left. \qquad+ \frac{1}{\mathfrak{m}_w^2}\left(\gamma_1 + \frac{\sigma_\epsilon^2}{\sigma_{\bm\beta}^2}\right)\left(1-2\pi_1+\pi_2\right)p\right\}\right| < O_p(p^{1/2 + \delta})\right) \geq 1-O_p(p^{-2\delta}).
\end{align*}
\end{lem.s}
This concentration has been proven sequentially as we approach the quantitative CLT of the numerator of $A(\hat{\bm\beta}_{\text{W}}(\lambda))$. Briefly, it involves $\sigma_1^2$ from  \cref{ineq:BE_ref_A2_eps_z}. We can conduct the same sequential conditioning and restrict our random variables to the same subsets where quantities concentrate well, as done for the numerator analysis. In addition, as we have shown in \cref{concent:marg_A2_denom_2}, with even much looser assumptions, we still have
\begin{equation*}
    \begin{aligned}
        \mathbb P\left(\left|\|\Z\bm\beta + \bm \epsilon_z\|_2^2 - n_z(\sigma_{\bm\beta}^2\gamma_1 + \sigma_{\epsilon_z}^2)\right| < O_p(n_z^{1/2 + \delta})\right) \geq 1-O_p(p^{-2\delta}). 
    \end{aligned}
\end{equation*}

\subsection{Major quantitative CLT}
Combining Lemma~S\ref{concent:ref_A2_denom_1} and \cref{ineq:BE_ref_A2_num}, we obtain the following Berry-Esseen inequality as our major conclusion.
\begin{thm.s}
\label{thm:CLT_ref_A2_raw}
Under Assumptions~\ref{a:Sigmabound}-\ref{a:anisotropic regularity}, as $p \to \infty$, with probability of at least $1-O_p(p^{-D})$ for some large enough $D\in \mathbb R$ over the randomness of $\W_0$, the following Berry-Esseen bound holds for $A(\hat{\bm\beta}_{\text{W}}(\beta))$
\begin{equation*}
    \begin{aligned}
        \sup_{t\in\mathbb R}\left|\mathbb P\left(\sqrt{\eta_\text{W}n_z}\left\{\frac{(\Z\bm\beta + \bm \epsilon_z)^\top \Z(\W^\top \W + n_w\lambda\mathbb\I_p)^{-1}\X^\top \y}{\|\Z\bm\beta + \bm \epsilon_z\|_2\|\Z(\W^\top \W + n_w\lambda\mathbb I_p)^{-1}\X^\top \y\|_2} - \tilde{A}_{\text{W}}\right\} < t\right) - \Phi(t)\right| \leq O_p(m^{-1/5}).
    \end{aligned}
\end{equation*}
The $\tilde{A}_{\text{W}}$ and $\eta_\text{W}$ are defined in Theorem~\ref{thm: CLT for reference A^2}. 
We explicitly write out our previous definition on $\eta_\text{W}$ as 
$\eta_\text{W}=\eta_\text{Wn}/(\eta_\text{W{d1}}+\eta_\text{W{d2}}+\eta_\text{W{d3}})$, where 
\begin{equation*}
    \begin{aligned}
        &\eta_\text{Wn} = \left\{n/\mathfrak{m}_w^3\left(3\xi_1-\xi_2+\mathfrak{m}_w\gamma_1-2m/p\right) + p/\mathfrak{m}_w^2(\gamma_1/h_{\bm\beta}^2)\left(1-2\pi_1+\pi_2\right)\right\}\sigma_{\bm\beta}^2\gamma_1/h_{\bm\beta_z}^2,\\
        &\eta_\text{Wd1} =1/\mathfrak{m}_w^3\left(3\xi_1-\xi_2+\mathfrak{m}_w\gamma_1-2m/p\right)\left(n\sigma_{\bm\beta}^2\gamma_1/h_{\bm\beta_z}^2 + n_z\sigma_{\bm\beta}^2\gamma_1/h_{\bm\beta}^2\right),\\
        &\eta_\text{Wd2} =p/(\sigma_{\bm\beta}^2\mathfrak{m}_w^2)(1-2\pi_1+\pi_2)(\sigma_{\bm\beta}^2\gamma_1/h_{\bm\beta_z}^2)\sigma_{\bm\beta}^2\gamma_1/h_{\bm\beta}^2,\\
        &\eta_\text{Wd3} =2(n+n_z)/\mathfrak{m}_w^4(\mathfrak{m}_w\gamma_1+\xi_1-m/p)^2\sigma_{\bm\beta}^2, \quad\mbox{and}\\ 
        &\eta_\text{Wd4} =nn_z\Bigg\{\left(\mathbb E\left(\bm\beta^4\right)/\sigma_{\bm\beta}^2 - 3\sigma_{\bm\beta}^2/p^2\right)\sum_{i=1}^p\left\{\left(\mathbb\I_p+\mathfrak{m}_w\bm\Sigma\right)^{-1}\bm\Sigma^2\mathbb\I_m\right\}_{i,i}^2 \bigr.\\
        &\qquad \bigl. + 2\sigma_{\bm\beta}^2/p^2\Tr\left\{\left((\mathbb\I_p+\mathfrak{m}_w\bm\Sigma)^{-1}\bm\Sigma^2\mathbb\I_m\right)^2\right\}\Bigg\}.
    \end{aligned}
\end{equation*}
\end{thm.s}
Theorem~S\ref{thm:CLT_ref_A2_raw} leads to the Theorem~\ref{thm: CLT for reference A^2} in Section~\ref{subsubsec:ref_A}.

\subsection{Additional results for isotropic features}
\label{subsec:add_con_spec_ref}
In this section, we will present more details for the CLT of $A(\hat{\bm\beta}_{\text{W}}(\lambda))$ when  $\bm\Sigma = \mathbb\I_p$.\\~\\
\noindent\textbf{Fact1:} When $\bm\Sigma = \mathbb\I_p$, \cref{equ:stiej_eigen} has a closed-form solution for $\mathfrak{m}_w$
\begin{equation*}
    \begin{aligned}
        \mathfrak{m}_w = \frac{\sqrt{(\lambda + \phi_w - 1)^2 + 4\lambda} - (\lambda +\phi_w - 1)}{2\lambda}.
    \end{aligned}
\end{equation*}
Moreover, we have 
$$\mathfrak{m}_w' = \left(\frac{1}{\mathfrak{m}_w^2} - \phi_w\frac{1}{\left(1+\mathfrak{m}_w\right)^2}\right)^{-1}$$
and 
\begin{equation*}
    \begin{aligned}
        \mathfrak{r}_w = \frac{\mathfrak{m}_w^2}{\mathfrak{m}_w'} = 1 - \phi_w\left(\frac{\mathfrak{m}_w}{1+\mathfrak{m}_w}\right)^2 = 1-4\phi_w\left\{\frac{1}{\sqrt{(\lambda + \phi_w - 1)^2 + 4\lambda} + (\lambda + \phi_w + 1)}\right\}^2.
    \end{aligned}
\end{equation*}
\noindent\textbf{Fact2:} Another interesting quantity appears in $\eta_{\text{W}}$ is 
$$\mathfrak{n}(\bm\Sigma, m) = \frac{1}{p}\sum_{i=1}^p\frac{\phi_w\mathfrak{m}_w'\pi_i^3}{(1+\pi_i\mathfrak{m}_w)^2}\langle \bm u_i, \mathbb\I_m \bm u_i\rangle.$$
Notice that when $\bm\Sigma = \mathbb\I_p$, all eigenvalues $\pi_i = 1$ and eigenvectors $\bm u_i = \e_i$. Therefore, we have 
\begin{equation*}
    \begin{aligned}
        &\mathfrak{n}(\bm\Sigma = \mathbb\I_p, m) = \frac{m}{p}\frac{\phi_w\mathfrak{m}_w'}{(1+\mathfrak{m}_w)^2}.
    \end{aligned}
\end{equation*}
Note that 
$$\frac{\mathfrak{m}_w'}{(1+\mathfrak{m}_w)^2} = \frac{\mathfrak{m}_w'}{\mathfrak{m}_w^2}\frac{\mathfrak{m}_w^2}{(1+\mathfrak{m}_w)^2}.$$
Therefore, we have 
\begin{equation*}
    \begin{aligned}
        \frac{\mathfrak{m}_w'}{(1+\mathfrak{m}_w)^2} 
        = \frac{1}{\left\{\sqrt{(\lambda + \phi_w - 1)^2 + 4\lambda} + (\lambda + \phi_w + 1)\right\}^2/4 - \phi_w}.
    \end{aligned}
\end{equation*}
It follows that 
\begin{equation*}
    \begin{aligned}
         &\mathfrak{n}(\bm\Sigma = \mathbb\I_p, m) = \frac{\phi_w m}{p}\frac{1}{\left\{\sqrt{(\lambda + \phi_w - 1)^2 + 4\lambda} + (\lambda + \phi_w + 1)\right\}^2/4 - \phi_w}.
    \end{aligned}
\end{equation*}

\section{Proof for Section~\ref{subsubsec:ridge_new}}\label{sec_proof_ridge1}
In this section, we consider the genetically predicted value $\z^T\hat{\bm\beta}_{\text{R}}(\lambda)$  of the traditional ridge estimator given by $\hat{\bm\beta}_{\text{R}}(\lambda) = (\X^\top \X + n\lambda)^{-1}\X^\top \y$. Our previous argument using the leave-one-out strategy no longer works because it is difficult to obtain the second-order limit of quantities such as $\z^\top(\X^\top \X + n\lambda\mathbb{I}_p)^{-1}\bm\beta$. Instead of working with the generic distribution  $\mathcal{F}$, we use the asymptotic normal approximation of de-biased estimators under Gaussian assumptions on the data proposed by \cite{10.1214/22-AOS2243}.
It takes three steps to establish the quantitative CLT for $\z^T\hat{\bm\beta}_{\text{R}}(\lambda)$ under Gaussian assumptions. First, we adopt the framework of the asymptotic distribution of the de-biased estimator to quantify the quantitative CLT. Then, we provide the first-order concentration of two key items, $\hat{df}$ and $\hat{V}(\theta)$, whose explicit definitions are provided in later sections. For simplicity, we denote
$$\R(\X^\top \X, -n\lambda) = (\X^\top \X + n\lambda\mathbb\I_p)^{-1},$$
and we will write it as $\R$ in Sections~\ref{sec_proof_ridge1} and \ref{sec_proof_ridge2}.

\subsection{Asymptotic distribution of the de-biased ridge estimator}
To be consistent with \cite{10.1214/22-AOS2243}, we estimate confidence interval of fixed quantity $\theta = \|\bm\Sigma^{-1/2}\z\|_2 \langle \z, \bm\beta \rangle$, where the $\|\bm\Sigma^{-1/2}\z\|_2$ part is added for normalization. We first introduce $\hat{V}(\theta)$ to quantify the variation of $\langle \z, \hat{\bm\beta}\rangle$
$$\hat{V}(\theta) = \|\bm\Sigma^{-1/2}\z\|_2^2\|\y-\X\hat{\bm\beta}\|_2^2 + \Tr\left((\hat{\bm H} - \mathbb \I_p)^2\right)\left(\langle \z, \hat{\bm\beta}\rangle - \langle \z, \bm\beta\rangle\right)^2,$$
where $$\hat{\bm H}  = \X^\top(\X^\top\X + n\lambda\mathbb\I_p)^{-1}\X^\top.$$
Then under Assumptions~\ref{a:Sigmabound}-\ref{a:Sparsity} and S\ref{a:Gaussian_entries}, as $p\to\infty$, we have
\begin{equation*}
    \sup_{t\in\mathbb R}\left\{\left|\mathbb P\left(\frac{\langle \z, \hat{\bm\beta}^{(de-bias)}\rangle - \langle \z, \bm\beta\rangle}{\hat{V}(\theta)^{1/2}/(n-\hat{df})} < t\right) - \Phi(t)\right|\right\} \to 0,
\end{equation*}
where 
$$\hat{\bm\beta}^{(de-bias)} = \hat{\bm\beta} + (n-\hat{df})^{-1}\bm\Sigma^{-1}\X^\top(\y - \X\hat{\bm\beta})$$ and 
$$\hat{df} = \Tr\left\{\X(\X^\top \X + n\lambda\mathbb \I_p)^{-1}\X^\top\right\}.$$
{Here $\hat{df}$ is the effective degrees-of-freedom and $\hat{H}$ stands for the gradient of the map $\y \mapsto \X \hat{\bm\beta}$.}
\subsection{First-order concentration of \texorpdfstring{$\hat{df}$}{TEXT}}
Now we provide the limit of $\hat{df}$, and we can decompose this term into two parts
\begin{equation*}
    \begin{aligned}
        \hat{df} = \Tr\{\X^\top \X(\X^\top \X + n\lambda\mathbb \I_p)^{-1}\} = p - n\lambda\Tr\{(\X^\top \X + n\lambda\mathbb \I_p)^{-1}\}. 
    \end{aligned}
\end{equation*}
Note that 
$$\Tr\{(\X^\top \X + n\lambda\mathbb \I_p\}^{-1}) = \sum_{i=1}^p \e_i^\top(\X^\top \X + n\lambda \mathbb \I_p)^{-1}\e_i.$$ 
The following lemma helps us provide concentration for each summand. 
\begin{lem.s}
\label{lemma:1st_order_trace}
    Under Assumptions~\ref{a:Sigmabound}-\ref{a:Sparsity} and S\ref{a:Gaussian_entries}, for some large enough $D\in\mathbb R$, small $\vartheta > 0$ and $\eta \gg p^{-1}$, with probability of at least $1-O(p^{-D})$ over the randomness of $\X_0$, we have 
    $$\left|-\lambda \e_i^\top\left(n\lambda\mathbb\I_p + \bm\Sigma^{1/2} \X_0^\top \X_0 \bm\Sigma^{1/2}\right)^{-1}\e_i + n^{-1}\e_i^\top\left(\mathbb\I_p+\mathfrak{m}_n\bm\Sigma\right)^{-1}\e_i\right| \leq \psi(-\lambda)O_p(\frac{p^{\vartheta}}{n}),$$
    where $$\psi(\lambda) = \sqrt{\frac{Im(\mathfrak{m}_n)}{n\eta}} + \frac{1}{n\eta}.$$
\end{lem.s}

\begin{proof}
    Let $$\R_n(-\lambda) = \left(\frac{\X_0}{\sqrt{n}}\bm\Sigma\frac{\X_0^\top}{\sqrt{n}} + \lambda\mathbb\I_n\right)^{-1},\quad \R_p(-\lambda) = \left(\bm\Sigma^{1/2}\frac{\X_0^\top}{\sqrt{n}} \frac{\X_0}{\sqrt{n}}\bm\Sigma^{1/2} + \lambda\I_p\right)^{-1},$$ $$\G(-\lambda) = \left[\begin{array}{cc}
        -\bm\Sigma^{-1} & \frac{\X_0^\top}{\sqrt{n}} \\
        \frac{\X_0}{\sqrt{n}} & \lambda \mathbb\I_{n}
    \end{array} \right]^{-1}, \quad \bm \Pi(-\lambda) = \left[\begin{array}{cc}
        -\bm\Sigma(1+\mathfrak{m}_n\bm\Sigma)^{-1} &  0\\
        0 & \mathfrak{m}_n
    \end{array} \right], \quad \text{and} \quad \underline{\bm\Sigma} = \left[\begin{array}{cc}
         \bm\Sigma & 0 \\
            0   & \mathbb\I_{n}
    \end{array}\right].$$
    By Proposition~S\ref{prop:Schur_comp}, we further have
    \begin{equation*}
        \begin{aligned}
        &\underline{\bm\Sigma}^{-1} = \left[\begin{array}{cc}
           \bm\Sigma^{-1}  & 0 \\
            0           & \mathbb\I_{n}
        \end{array}\right]\quad \text{and}\\
        &G(-\lambda) = 
        \left[\begin{array}{cc}
           (-\lambda n)\bm\Sigma^{1/2}(\lambda n\mathbb\I_p + \bm\Sigma^{1/2}\X_0^\top \X_0\bm\Sigma^{1/2})^{-1}\bm\Sigma^{1/2}  & B \\
            C           & \frac{1}{\lambda}\mathbb\I_{n} - \frac{1}{n\lambda^2}\X_0(\bm\Sigma^{-1} + \frac{1}{n}\X_0^\top \X_0)^{-1}\X_0^\top
        \end{array}\right].
        \end{aligned}
    \end{equation*}
    Denote 
    $$\underline{\Vec{\bm v}} =  
    \left[
        \begin{array}{c}
            \bm\Sigma^{1/2}\Vec{\e_i} \\ \hdashline[2pt/2pt]
            0
        \end{array}
    \right]\quad \text{and}\quad \Vec{\underline{\bm w}} =
    \left[\begin{array}{c}
         \bm\Sigma^{1/2}\Vec{\e_i}  \\ \hdashline
         0 
    \end{array}\right],$$
    which are of shape $(n+p) \times 1$.
    From the Theorem 3.7 of \cite{anisotropic_local_law}, we have
    \begin{equation*}
        \begin{aligned}
            \left|\langle \vec{\underline{\bm v}}, \underline{\bm\Sigma}^{-1}(\G(-\lambda) - \bm \Pi(-\lambda))\underline{\bm\Sigma}^{-1}\vec{\underline{\bm w}}\rangle\right|  \prec \psi(-\lambda)O_p(1).
        \end{aligned}
    \end{equation*}
   It follows that 
    \begin{equation*}
        \begin{aligned}
            &\left|\e_i^\top\bm\Sigma^{1/2} \left\{(-\lambda n)\bm\Sigma^{-1/2}(\lambda n\mathbb\I_p + \bm\Sigma^{1/2}\X_0^\top \X_0\bm\Sigma^{1/2})^{-1}\bm\Sigma^{-1/2}+(\mathbb\I_p+\mathfrak{m}_n\bm\Sigma)^{-1}\bm\Sigma^{-1}\right\}\bm\Sigma^{1/2} \e_i\right| \\
            & \qquad\prec \psi(-\lambda)O_p(1).
        \end{aligned}
    \end{equation*}
    Rearranging terms, we have 
    \begin{equation*}
        \begin{aligned}
            &\left|-\e_i^\top(n\lambda\mathbb\I_p + \bm\Sigma^{1/2} \X_0^\top \X_0 \bm\Sigma^{1/2})^{-1} \e_i + \frac{1}{\lambda n}\e_i^\top(\mathbb\I_p+\mathfrak{m}_n\bm\Sigma)^{-1}\e_i\right| \leq \psi(-\lambda)O_p\left(\frac{p^{\vartheta}}{\lambda n}\right).
        \end{aligned}
    \end{equation*}
\end{proof}

From Lemma~S\ref{lemma:1st_order_trace}, we have 
\begin{equation}
\label{ineq:iso_first_order_form}
    \begin{aligned}
        &\left|\sum_{i = 1}^p \e_i^\top(n\lambda\mathbb I_p + \bm\Sigma^{1/2}\X_0^\top \X_0\bm\Sigma^{1/2})^{-1}\e_i - \frac{1}{\lambda n}\Tr\left((\mathbb I_p + \mathfrak{m}_n\bm\Sigma)^{-1}\right)\right| \leq \psi(-\lambda)O_p(\frac{p^{1+\vartheta}}{\lambda n}).
    \end{aligned}
\end{equation}
Now by choosing $D > 1$ and applying union bound, with probability of at least $1-O(p^{-D+1})$, we have 
\begin{equation*}
    \begin{aligned}
        &\left|\Tr\left((\X^\top \X+n\lambda\mathbb I_p)^{-1}\right) - \frac{1}{\lambda n}\Tr\left((\mathbb I_p + \mathfrak{m}_n\bm\Sigma)^{-1}\right)\right| \leq \psi(-\lambda)O_p(\frac{p^{1+\vartheta}}{\lambda n})
    \end{aligned}
\end{equation*}
Note that $$(\lambda n)^{-1}\Tr\left((\mathbb I_p + \mathfrak{m}_n\bm\Sigma)^{-1}\right) = O_p(1).$$ 
This means that $\Tr\left((\X^\top \X + n\lambda\mathbb\I_p)^{-1}\right)$ concentrates around $(\lambda n)^{-1}\Tr\left((\mathbb I_p + \mathfrak{m}_n\bm\Sigma)^{-1}\right)$ with high probability. It follows that
\begin{equation}
\label{concent:df_hat}
    \begin{aligned}
        \frac{\hat{df}}{p - \Tr\left((\mathbb I_p + \mathfrak{m}_n\bm\Sigma)^{-1}\right)} \overset{p}{\to} 1.
    \end{aligned}
\end{equation}

\subsection{First-order concentration of \texorpdfstring{$\hat{V}(\theta)$}{TEXT}}
To analyze the first-order concentration of $\hat{V}(\theta)$, we can break this down into finding the concentrations of the four different quantities below.\\~\\
\noindent\textbf{Task 1:}
We provide the first order concentration of $\Tr\left((\hat{\bm H} - \mathbb I_p)^2\right)$.\\
Notice that
\begin{equation*}
    \begin{aligned}
        \Tr(\hat{\bm H}^2)-2\Tr(\hat{\bm H}) + p &= \Tr\left\{\X^\top \X(\X^\top \X+n\lambda\mathbb I_p)^{-1}\X^\top \X(\X^\top \X + n\lambda \mathbb \I_p)^{-1}\right\} \\
        &- 2\Tr\left\{\X^\top \X(\X^\top \X + n\lambda\mathbb \I_p)^{-1}\right\} + p,
    \end{aligned}
\end{equation*}
which implies 
\begin{equation*}
    \begin{aligned}
       \Tr\left((\hat{\bm H} - \mathbb \I_p)^2\right) = & n^2\lambda^2\Tr\left((\X^\top \X + n\lambda\mathbb \I_p)^{-2}\right) = n^2\lambda^2\sum_{i = 1}^p \e_i^\top(\X^\top \X + n\lambda\mathbb \I_p)^{-2}\e_i.
    \end{aligned}
\end{equation*}
Therefore, our question boils down to find the concentration of 
$$\e_i^\top (\X^\top \X + n\lambda\mathbb I_p)^{-2}\e_i.$$
Denote
\begin{align*}
        &\G(-\lambda + \tau) = \left[\begin{array}{cc}
        -\bm\Sigma^{-1} & \frac{\X_0^\top}{\sqrt{n}} \\
        \frac{\X_0}{\sqrt{n}} & (\lambda-\tau) \mathbb\I_{n}
    \end{array} 
    \right]^{-1}, \\
    &\bm \Pi(-\lambda + \tau) = \left[\begin{array}{cc}
        -\bm\Sigma(\mathbb I_p+m_n(-\lambda + \tau)\bm\Sigma)^{-1} &  0\\
        0 & m_n(-\lambda + \tau)
    \end{array} \right], \quad\mbox{and}\\
    &\underline{\bm\Sigma} = \left[\begin{array}{cc}
         \bm\Sigma & 0 \\
            0   & \mathbb\I_{n}
    \end{array}\right].
\end{align*}
By applying Proposition~S\ref{prop:Schur_comp}, we can explicitly compute $\underline{\bm\Sigma}^{-1}$ and $\G(-\lambda + \tau)$ as follows 
    \begin{equation*}
        \begin{aligned}
        \underline{\bm\Sigma}^{-1} = \left[\begin{array}{cc}
           \bm\Sigma^{-1}  & 0 \\
            0           & \mathbb\I_{n}
        \end{array}\right] \quad \text{and} \quad
        \G(-\lambda + \tau) = 
        \left[\begin{array}{cc}
           \left(-\bm\Sigma^{-1} - \frac{\X_0^\top \X_0}{(\lambda-\tau) n}\right)^{-1}  & \B \\
            \C           &  \D
        \end{array}\right],
        \end{aligned}
    \end{equation*}
    {where $\B$, $\C$, and $\D$ are some matrices with shapes of $p\times n$, $n\times p$, and $n\times n$, respectively.}
Denote $$\underline{\Vec{\bm v}} =  
    \left[
        \begin{array}{c}
            \bm\Sigma^{1/2}\Vec{\e_i} \\ \hdashline[2pt/2pt]
            0
        \end{array}
    \right] \quad\text{and} \quad \Vec{\underline{\bm w}} =
    \left[\begin{array}{c}
         \bm\Sigma^{1/2}\Vec{\e_i}  \\ \hdashline
         0 
    \end{array}\right],$$
    both of which are of shape $(n+p) \times 1$.
    By the anisotropic local law \citep{anisotropic_local_law}, we have
    \begin{equation*}
        \begin{aligned}
            \left|\langle \vec{\underline{\bm v}}, \underline{\bm\Sigma}^{-1}(\G(-\lambda) - \bm \Pi(-\lambda))\underline{\bm\Sigma}^{-1}\vec{\underline{\bm w}}\rangle\right|  \prec \psi(-\lambda)O_p(1).
        \end{aligned}
    \end{equation*}
    After simplification, for some small $\vartheta > 0$, with probability of at least $1-O_p(p^{-D})$, we have 
    \begin{equation*}
        \begin{aligned}
        &\left|\e_i^\top\left(\mathbb I_p + m_n(-\lambda + \tau)\bm\Sigma\right)^{-1}\e_i - \e_i^\top \left(\mathbb I_p + \frac{\bm\Sigma^{1/2}\X_0^\top \X_0\bm\Sigma^{1/2}}{n(\lambda - \tau)}\right)^{-1}\e_i\right| \leq \psi(-\lambda)O_p(p^{\vartheta})\\
            &\Leftrightarrow \\
            &\left|\frac{1}{\lambda-\tau}\e_i^\top\left(\mathbb I_p + m_n(-\lambda+\tau)\bm\Sigma\right)^{-1}\e_i - \e_i^\top\left((\lambda - \tau)\mathbb I_p + \frac{\bm\Sigma^{1/2}\X_0^\top \X_0\bm\Sigma^{1/2}}{n}\right)^{-1}\e_i\right|\leq \psi(-\lambda)O_p(p^{\vartheta}).
        \end{aligned}
    \end{equation*}
    Notice that $m_n(-\lambda + \tau)$ satisfies
    \begin{align*}
        \frac{1}{m_{n}(-\lambda+\tau)} = \lambda - \tau + \phi_n\sum_{i=1}^p\frac{\pi_i}{1+m_n(-\lambda+\tau)\pi_i},
    \end{align*}
    from which we have 
    \begin{align*}
        \left.\frac{d(m_n(-\lambda+\tau))}{d\tau}\right|_{\tau = 0} = \mathfrak{m}_n'. 
    \end{align*}
    Taking the derivative with respect to $\tau$ and set $\tau = 0$, we have
    \begin{equation*}
        \begin{aligned}
            &\left|\frac{1}{\lambda^2}\left\{\e_i^\top\left(\mathbb I_p + \mathfrak{m}_n\bm\Sigma\right)^{-1} \e_i - \lambda\mathfrak{m}'_ne_i^\top\left(\mathbb I_p + \mathfrak{m}_n\bm\Sigma\right)^{-2}\bm\Sigma \e_i\right\} - e_i^\top\left(\lambda\mathbb I_p + \frac{\bm\Sigma^{1/2}\X_0^\top \X_0\bm\Sigma^{1/2}}{n}\right)^{-2}\e_i\right|\\
            &\qquad\leq \psi(-\lambda)O_p(p^\vartheta)\\
            &\Leftrightarrow&\\            
            &\left|\frac{1}{n^2\lambda^2}\left\{\e_i^\top\left(\mathbb I_p + \mathfrak{m}_n\bm\Sigma\right)^{-1} \e_i - \lambda\mathfrak{m}'_n\e_i^\top\left(\mathbb I_p + \mathfrak{m}_n\bm\Sigma\right)^{-2}\bm\Sigma \e_i\right\} - \e_i^\top\left(n\lambda\mathbb I_p + \X^\top \X\right)^{-2}\e_i\right|\\
            &\qquad\leq \psi(-\lambda)O_p\left(\frac{p^\vartheta}{n^2}\right).
        \end{aligned}
    \end{equation*}
    Since it holds for $\forall i \in [1, p]$ and our choice of $D$ is arbitrary, we can always choose $D > 1$. Then by union bound,  with probability of at least $1-O_p(p^{1-D})$, we have
    \begin{equation*}
        \begin{aligned}
        &\left|\frac{1}{n^2\lambda^2}\left\{\Tr\left(\left(\mathbb I_p + \mathfrak{m}_n\bm\Sigma\right)^{-1}\right) - \lambda\mathfrak{m}_n'\Tr\left(\left(\mathbb I_p + \mathfrak{m}_n\bm\Sigma\right)^{-2}\bm\Sigma\right)\right\} - \Tr\left(\left(\X^\top \X + \lambda n\mathbb I_p\right)^{-2}\right)\right| \\
        &\qquad\leq \psi(-\lambda)O_p(\frac{p^{1+\vartheta}}{n^2}).
        \end{aligned}
    \end{equation*}
   Therefore, we omit terms with smaller order since $\Tr\left(\left(\mathbb I_p + \mathfrak{m}_n\bm\Sigma\right)^{-2}\right) \succ O_p(p^\vartheta)$. \\~\\
    
    \noindent\textbf{Task 2:} We now provide the first-order concentration of $\left(\langle \z, \hat{\bm\beta}\rangle - \langle \z, \bm\beta\rangle\right)^2$. 
    From the definition, we have
    \begin{equation*}
        \begin{aligned}
            \left\{\z^\top(\X^\top \X + n\lambda\mathbb I_p)^{-1}\X^\top \y - \z^\top \bm\beta\right\}^2 = \left\{\z^\top(\X^\top \X + n\lambda\mathbb \I_p)^{-1}\X^\top\epsilon - n\lambda \z^\top(\X^\top \X + n\lambda\mathbb I_p)^{-1}\bm\beta\right\}^2.
        \end{aligned}
    \end{equation*}
    We have shown that, with probability of $1-O_p(p^{-D})$ for some large $D\in \mathbb R$, we have $$\left|-\lambda \z^\top(n\lambda\mathbb\I_p + \bm\Sigma^{1/2} \X_0^\top \X_0 \bm\Sigma^{1/2})^{-1}\bm\beta + \frac{1}{n}\z^\top(\mathbb\I_p+\mathfrak{m}_n\bm\Sigma)^{-1}\bm\beta\right| \leq \psi(-\lambda)O_p\left(\frac{p^{\vartheta + 1/2}}{n}\right).$$
    Now we show that $\z^\top(\X^\top \X + n\lambda\mathbb I_p)^{-1}\X^\top\epsilon \prec O_p(1)$ by applying Markov's inequality. Assume $\delta > 0$, we have 
    \begin{equation*}
        \begin{aligned}
            \mathbb P\left(\left|\z^\top\left(\X^\top \X + n\lambda\mathbb \I_p\right)^{-1}\X^\top\epsilon\right| \geq \varepsilon p^{-\delta}\right) \leq p^{2\delta}\mathbb E\left(\z^\top(\X^\top \X+n\lambda\mathbb I_p)^{-1}\X^\top \X(\X^\top \X + n\lambda\mathbb \I_p)^{-1}\z\right)/\varepsilon^2.
        \end{aligned}
    \end{equation*}
    Notice that
    \begin{equation*}
        \begin{aligned}
            &\mathbb E\left\{\z^\top(\X^\top \X+n\lambda\mathbb \I_p)^{-1}\X^\top \X(\X^\top \X + n\lambda\mathbb I_p)^{-1}\z\right\} \\
            & =\mathbb E\left\{\z^\top(\X^\top \X + n\lambda\mathbb I_p)^{-1}\z\right\}-n\lambda\mathbb E\left\{\z^\top(\X^\top \X + n\lambda\mathbb I_p)^{-2}\z\right\}.
        \end{aligned}
    \end{equation*}
    Similar to our previous argument, for some large enough $D$, with probability of $1-O_p(p^{-D})$, we have
    \begin{equation}
        \begin{aligned}
        \label{ineq:ztAz}
            &\left|\z^\top(\X^\top \X + n\lambda\mathbb\I_p)^{-1} \z - \frac{1}{n\lambda}\z^\top\left(\mathbb\I_p+m(-\lambda)\bm\Sigma\right)^{-1}\z\right| \leq \psi(-\lambda)O_p(\frac{p^{\vartheta}}{\lambda n})\quad \mbox{and} \\
            &\left|\z^\top\left(\X^\top \X + \lambda n\mathbb I_p\right)^{-2}\z - \frac{1}{n^2\lambda^2}\left\{\z^\top\left(\mathbb I_p + \mathfrak{m}_n\bm\Sigma\right)^{-1} \z - \lambda\mathfrak{m}'_n\z^\top\left(\mathbb I_p + \mathfrak{m}_n\bm\Sigma\right)^{-2}\bm\Sigma \z\right\}\right|\\
            &\qquad\leq \psi(-\lambda)O_p(\frac{p^{\vartheta}}{n^2}).
        \end{aligned}
    \end{equation}
    Therefore, with high probability, we have 
    $$\z^\top(\X^\top \X + n\lambda\mathbb I_p)^{-1}\z = O_p(p^{-1})$$ 
    and
    $$\z^\top(\X^\top \X + n\lambda\mathbb I_p)^{-2}\z = O_p(p^{-2}).$$ 
    Since $\z^\top(\X^\top \X + n\lambda\mathbb\I_p)^{-1}\z$ and $\z^\top(\X^\top \X + n\lambda\mathbb\I_p)^{-2}\z$ are both bounded, by bounded convergence theorem we have 
    $$p^{2\delta}\mathbb E\left\{\z^\top(\X^\top \X+n\lambda\mathbb I_p)^{-1}\X^\top \X(\X^\top \X + n\lambda\mathbb I_p)^{-1}\z\right\}/\varepsilon^2 = O_p(p^{2\delta-1}).$$
    By choosing $0 < \delta < 1/2$, with high probability, we have 
    $$\z^\top(\X^\top \X + n\lambda\mathbb I_p)^{-1}\X^\top\epsilon \prec O_p(1).$$ 
    It follows that $$\frac{\left(\langle \z, \hat{\bm\beta}\rangle - \langle \z, \bm\beta\rangle\right)^2}{\left\{\z^\top\left(\mathbb\I_p + \mathfrak{m}_n\bm\Sigma\right)^{-1}\bm\beta\right\}^2}\overset{p}{\to} 1.$$
    Combining concentrations for $\left(\langle \z, \hat{\bm\beta}\rangle - \langle \z, \bm\beta\rangle\right)^2$ and $\Tr\left((\hat{\bm H} - \mathbb\I_p)^2\right)$, we have 
    \begin{equation*}
        \begin{aligned}
            \frac{\Tr\left((\hat{\bm H} - \mathbb\I_p)^2\right)\left(\langle \z,\hat{\bm\beta}\rangle - \langle \z,\bm\beta\rangle\right)^2}{\left(\z^\top\left(\mathbb\I_p + \mathfrak{m}_n\bm\Sigma\right)^{-1}\bm\beta\right)^2n^{-2}\lambda^{-2}\left\{\Tr\left(\left(\mathbb I_p + \mathfrak{m}_n\bm\Sigma\right)^{-1}\right) - \lambda \mathfrak{m}_n'\Tr\left(\left(\mathbb I_p + \mathfrak{m}_n\bm\Sigma\right)^{-2}\bm\Sigma\right)\right\}} \overset{p}{\to} 1.
        \end{aligned}
    \end{equation*}\\~\\
    
    \noindent\textbf{Task 3:}
    We provide the first-order concentration of $\|\y-\X\hat{\bm\beta}\|_2^2$.
    First, we decompose $\|\y-\X\hat{\bm\beta}\|_2^2$ as follows 
    \begin{equation*}
        \begin{aligned}
            \|\y-\X\hat{\bm\beta}\|_2^2 &= \y^\top \y - 2\y^\top \X(\X^\top \X+n\lambda\mathbb\I_p)^{-1}\X^\top \y \\
            &+ \y^\top \X(\X^\top \X + n\lambda\mathbb\I_p)^{-1}\X^\top \X(\X^\top \X+n\lambda\mathbb\I_p)^{-1}\X^\top \y.
        \end{aligned}
    \end{equation*}
    The concentration of the first term can be obtained by applying Proposition~S\ref{prop:quad_general_limit}, where we have 
    $$\frac{\y^\top \y}{n\left(\bm\beta^\top\bm\Sigma\bm\beta+\sigma_{\epsilon}^2\right)} \overset{p}{\to} 1.$$
    The second term can be further decomposed into three simpler parts
    \begin{equation*}
        \begin{aligned}
            \y^\top \X(\X^\top \X+n\lambda\mathbb\I_p)^{-1}\X^\top \y &= \bm\beta^\top \X^\top \X(\X^\top \X+n\lambda\mathbb\I_p)^{-1}\X^\top \X\bm\beta \\
            &+ 2\bm \epsilon^\top \X(\X^\top \X+n\lambda\mathbb\I_p)^{-1}\X^\top \X \bm\beta
            +\bm  \epsilon^\top \X(\X^\top \X + n\lambda\mathbb\I_p)^{-1}\X^\top \bm \epsilon.
        \end{aligned}
    \end{equation*}
    Considering fixed $\X_0$, for $\forall \delta \in (0, 1/2)$, the concentration for $\bm  \epsilon^\top \X(\X^\top \X + n\lambda\mathbb\I_p)^{-1}\X^\top \bm \epsilon$ can be obtained by applying Proposition~S\ref{prop:quad_first_limit}
    \begin{equation*}
        \begin{aligned}
            &\mathbb E_{\epsilon}\left|\sigma_\epsilon^{-2}\bm \epsilon^\top \X \R\X^\top\bm \epsilon - \Tr\left(\X \R\X^\top\right)\right|^2 \leq C\left(\frac{\mathbb E|\epsilon|^4}{\sigma_\epsilon^4}\Tr\left(\X \R\X^\top \X \R\X^\top\right)\right)\\
            &\implies \mathbb P\left(\left|\bm \epsilon^\top \X \R\X^\top\bm \epsilon - \sigma_\epsilon^2\Tr\left( \X \R\X^\top\right)\right| < O_p(p^{1/2+\delta})\right) \\
            &\qquad\qquad\geq 1 - \frac{C\left(\mathbb E|\epsilon_1|^4\Tr\left(\X \R\X^\top \X \R\X^\top\right)\right)}{p^{1+2\delta}} 
            = 1-O_p(p^{-2\delta}). 
        \end{aligned}
    \end{equation*}
    Note that $\Tr(\X \R\X^\top)$ is exactly $\hat{df}$, which we have established  its concentration in \cref{concent:df_hat}. So every term with an order smaller than $\Theta_p(p)$ would be negligible. This leads to the negligibility of $\bm \epsilon^\top \X(\X^\top \X+n\lambda\mathbb\I_p)^{-1}\X^\top \X\bm\beta$, which can be seen from the following Markov's inequality
    \begin{equation*}
        \begin{aligned}
            &\mathbb P\left(\left|\bm \epsilon^\top \X(\X^\top \X+n\lambda\mathbb\I_p)^{-1}\X^\top \X\bm\beta\right| < O_p(p^{1/2+\delta})\right) \\
            &\qquad\geq 1-O_p\left(\frac{\sigma_\epsilon^2\bm\beta^\top \X^\top \X(\X^\top \X+n\lambda\mathbb\I_p)^{-1}\X^\top \X\bm\beta}{p^{1+2\delta}}\right)= 1-O_p(p^{-2\delta}).
        \end{aligned}
    \end{equation*}
    Therefore, with high probability over the randomness of $\bm \epsilon$, we can drop this interaction term.
    
    We now turn to the concentration of $\bm\beta^\top \X^\top \X(\X^\top \X+n\lambda\mathbb\I_p)^{-1}\X^\top \X\bm\beta$. This term can be further decomposed into the summation of three terms
    \begin{equation*}
        \begin{aligned}
            \bm\beta^\top \X^\top \X(\X^\top \X+n\lambda\mathbb\I_p)^{-1}\X^\top \X\bm\beta = \bm\beta^\top \X^\top \X\bm\beta -n\lambda \bm\beta^\top\bm\beta + n^2\lambda^2\bm\beta^\top(\X^\top \X+n\lambda\mathbb\I_p)^{-1}\bm\beta.
        \end{aligned}
    \end{equation*}
    By Lemma~S\ref{lemma:LLN_X^TX}, we have 
    \begin{equation*}
        \begin{aligned}
            &\mathbb P\left(\left|\frac{\bm\beta^\top \X^\top \X\bm\beta }{n} - \bm\beta^\top\bm\Sigma\bm\beta\right| \leq O_p(p^{-\delta})\right) \geq 1- O_p\left(\frac{\|\bm\beta\|_2^4}{np^{-2\delta}}\right) = 1-O_p\left(p^{-1+2\delta}\right),\\
            &\implies \mathbb P\left(\left|\bm\beta^\top \X^\top \X\bm\beta - n\bm\beta^\top\bm\Sigma\bm\beta\right| \leq O_p(p^{1-\delta})\right) \geq 1-O_p\left(p^{-1+2\delta}\right),
        \end{aligned}
    \end{equation*}
    leading to the concentration for $\bm\beta^\top \X^\top \X\bm\beta$. 
    
    Notice that $n\lambda \bm\beta^\top\bm\beta$ is deterministic and the concentration of $n^2\lambda^2\bm\beta^\top(\X^\top \X+n\lambda\mathbb\I_p)^{-1}\bm\beta$ follows from the same argument as in \cref{ineq:ztAz}. Explicitly, for some large $D\in \mathbb R$, with probability of at least $1-O_p(p^{-D})$ over the randomness of $\X_0$ , we have
    \begin{equation*}
        \begin{aligned}
            \left|\bm\beta^\top(\X^\top \X+n\lambda\mathbb\I_p)^{-1}\bm\beta - \frac{1}{\lambda n}\bm\beta(\mathbb\I_p + \mathfrak{m}_n\bm\Sigma)^{-1}\bm\beta\right| \leq \psi(-\lambda)O_p(\frac{p^{\vartheta}}{\lambda n}).
        \end{aligned}
    \end{equation*}
   It follows that 
    \begin{equation*}
        \begin{aligned}
            \frac{\bm\beta^\top \X^\top \X(\X^\top \X+n\lambda\mathbb\I_p)^{-1}\X^\top \X\bm\beta}{n\left(\bm\beta^\top\bm\Sigma\bm\beta - \lambda\bm\beta^\top\bm\beta + \lambda\bm\beta^\top(\mathbb\I_p + \mathfrak{m}_n\bm\Sigma)^{-1}\bm\beta\right)} \overset{p}{\to} 1.
        \end{aligned}
    \end{equation*}
    Combining all concentrations above, we have 
    \begin{equation*}
        \begin{aligned}
            \frac{\y^\top \X(\X^\top \X + n\lambda\mathbb\I_p)^{-1}\X^\top \y}{\sigma_{\epsilon}^2\left\{p-\Tr\left((\mathbb\I_p + \mathfrak{m}_n\bm\Sigma)^{-1}\right)\right\} + n\left\{\bm\beta^\top\bm\Sigma\bm\beta - \lambda\bm\beta^\top\bm\beta + \lambda\bm\beta^\top(\mathbb\I_p + \mathfrak{m}_n\bm\Sigma)^{-1}\bm\beta\right\}} \overset{p}{\to} 1.
        \end{aligned}
    \end{equation*}
    Lastly, we turn to the concentration for $\y^\top \X(\X^\top \X + n\lambda\mathbb\I_p)^{-1}\X^\top \X(\X^\top \X+n\lambda\mathbb\I_p)^{-1}\X^\top \y$.\\
    We consider the following decomposition
    \begin{equation*}
        \begin{aligned}
            \y^\top \X \R\X^\top \X \R\X^\top \y = \y^\top \X \R\X^\top \y - n\lambda \y^\top \X \R^{2}\X^\top \y.
        \end{aligned}
    \end{equation*}
    We have alright obtained the concentration for $\y^\top \X \R\X^\top \y$, so we only need to find the limit of $\y^\top \X \R^{2}\X^\top \y$, which is 
    \begin{equation*}
        \begin{aligned}
            \y^\top \X \R^{2}\X^\top \y = \bm \epsilon^\top \X \R^{2}\X^\top\bm \epsilon + 2\bm \epsilon^\top \X \R^{2}\X^\top \X\bm\beta + \bm\beta^\top \X^\top \X \R^{2}\X^\top \X\bm\beta.
        \end{aligned}
    \end{equation*}
    Considering the randomness of $\bm \epsilon$, it is easy to check that $2\bm \epsilon^\top \X \R^{2}\X^\top \X\bm\beta$ is negligible compared with $\bm \epsilon^\top \X \R^{2}\X^\top\bm \epsilon$ and $\bm\beta^\top \X^\top \X \R^{2}\X^\top \X\bm\beta$  by Markov's inequality. Furthermore, by Lemma~S\ref{prop:quad_first_limit}, we have 
    $$\frac{\epsilon^\top \X R^{2}\X^\top\epsilon}{\sigma_{\epsilon}^2\Tr(\X^\top \X R^{2})} \overset{p}{\to} 1.$$
    Notice that
    \begin{equation*}
        \begin{aligned}
            \Tr(\X^\top \X \R^2) = \Tr(\R) - n\lambda\Tr(\R^{2}). 
        \end{aligned}
    \end{equation*}
    By replacing $\z$ with $\epsilon/\|\epsilon\|_2$ in \cref{ineq:ztAz}, we have 
    \begin{equation*}
        \begin{aligned}
             \frac{\Tr(\R) - n\lambda\Tr(\R^2)}{\mathfrak{m}_n'\Tr\left((\mathbb\I_p + \mathfrak{m}_n\bm\Sigma)^{-2}\bm\Sigma\right)/n}\overset{p}{\to} 1. 
        \end{aligned}
    \end{equation*}
    It follows that 
    \begin{equation*}
        \begin{aligned}
            \frac{\bm \epsilon^\top \X \R\X^\top\bm \epsilon}{\sigma_\epsilon^2\mathfrak{m}_n'\Tr\left((\mathbb\I_p + \mathfrak{m}_n\bm\Sigma)^{-2}\bm\Sigma\right)/n} \overset{p}{\to} 1.
        \end{aligned}
    \end{equation*}
    Similarly, we simplify $\bm\beta^\top \X^\top \X R^{2}\X^\top \X\bm\beta$ by further decomposition
    \begin{equation*}
        \begin{aligned}
            \bm\beta^\top \X^\top \X \R^{2}\X^\top \X\bm\beta = \bm\beta^\top\bm\beta - 2n\lambda\bm\beta^\top \R\bm\beta + n^2\lambda^2\bm\beta^\top \R^{2}\bm\beta. 
        \end{aligned}
    \end{equation*}
    We have already obtained the concentration for each of the three terms. Therefore, we have 
    \begin{equation*}
        \begin{aligned}
            \frac{\bm\beta^\top \X^\top \X \R^{2}\X^\top \X\bm\beta}{\bm\beta^\top\bm\beta -\bm\beta^\top(\mathbb\I_p + \mathfrak{m}_n\bm\Sigma)^{-1}\bm\beta -\lambda\mathfrak{m}'_n\bm\beta^\top\left(\mathbb I_p + \mathfrak{m}_n\bm\Sigma\right)^{-2}\bm\Sigma \bm\beta} \overset{p}{\to} 1
        \end{aligned}
    \end{equation*}
   It follows that 
    \begin{equation*}
        \begin{aligned}
            \frac{\y^\top \X \R^{2}\X^\top \y}{\bm\beta^\top\bm\beta -\bm\beta^\top(\mathbb\I_p + \mathfrak{m}_n\bm\Sigma)^{-1}\bm\beta -\lambda\mathfrak{m}'_n\bm\beta^\top\left(\mathbb I_p + \mathfrak{m}_n\bm\Sigma\right)^{-2}\bm\Sigma \bm\beta + \sigma_\epsilon^2\mathfrak{m}_n'\Tr\left((\mathbb\I_p + \mathfrak{m}_n\bm\Sigma)^{-2}\bm\Sigma\right)/n}\overset{p}{\to} 1.
        \end{aligned}
    \end{equation*}
    Eventually, by combining all results above, we obtain the desired concentration for $\|\y - \X\hat{\bm\beta}\|_2^2$
    \begin{equation*}
        \begin{aligned}
            \frac{\|\y-\X\hat{\bm\beta}\|_2^2}{\sigma_\epsilon^2\left\{n-p + \Tr\left((\mathbb\I_p + \mathfrak{m}_n\bm\Sigma)^{-1}\right) - \lambda\mathfrak{m}_n'\Tr\left((\mathbb\I_p+\mathfrak{m}_n\bm\Sigma)^{-2}\bm\Sigma\right)\right\} + \mathfrak{m}_n'n\lambda^2\bm\beta^\top(\mathbb\I_p + \mathfrak{m}_n\bm\Sigma)^{-2}\bm\Sigma\bm\beta} \overset{p}{\to} 1.
        \end{aligned}
    \end{equation*}\\~\\
    
    \noindent\textbf{Task 4:} We consider the first-order concentration of the ``de-biased" part.
    We first explicitly write down the biasness of the ridge estimator as follows 
    \begin{equation*}
        \begin{aligned}
            \frac{\z^\top\bm\Sigma^{-1}\X^\top(\y - \X\hat{\bm\beta})}{n-\hat{df}} = &\frac{1}{{n-\hat{df}}}\left\{n\lambda \z^\top\bm\Sigma^{-1}(\X^\top \X + n\lambda\mathbb\I_p)^{-1}\X^\top \y\right\}\\
            = &\frac{n\lambda}{n-\hat{df}}\left(\z^\top\bm\Sigma^{-1}\bm\beta - n\lambda \z^\top\bm\Sigma^{-1}\R\bm\beta + \z^\top\bm\Sigma^{-1}\R\X^\top\bm\epsilon\right).
        \end{aligned}
    \end{equation*}
    It iseasy to check that $\z^\top\bm\Sigma^{-1}\R\X^\top\bm \epsilon$ is negligible, and with the same argument as Lemma~S\ref{lemma:1st_order_trace}, we obtain the following concentration
    \begin{equation*}
        \begin{aligned}
            \frac{\z^\top\bm\Sigma^{-1}\X^\top(\y - \X\hat{\bm\beta})}{n-\hat{df}} \overset{p}{\to} \frac{n\lambda}{n- p + \Tr\left((\mathbb I_p + \mathfrak{m}_n\bm\Sigma)^{-1}\right)}\left\langle \z\bm\Sigma^{-1}, \bm\beta - (\mathbb\I_p + \mathfrak{m}_n\bm\Sigma)^{-1}\bm\beta\right\rangle.
        \end{aligned}
    \end{equation*}
    Now, we are ready to provide the quantitative CLT for $\z^\top\hat{\bm\beta}_{\text{R}}(\lambda)$. Recall that we denote
    \begin{equation*}
        \begin{aligned}
            &\mathfrak{h} = \phi_n\left\{\Tr\left(\left(\mathbb I_p + \mathfrak{m}_n\bm\Sigma\right)^{-1}\right) - \lambda\mathfrak{m}_n'\Tr\left(\left(\mathbb I_p + \mathfrak{m}_n\bm\Sigma\right)^{-2}\bm\Sigma\right)\right\}/p,\\ 
            &\mathfrak{g} = 1-\phi_n\left\{1-\frac{\Tr\left((\mathbb\I_p + \mathfrak{m}_n\bm\Sigma)^{-1}\right)}{p}\right\},
        \end{aligned}
    \end{equation*}
    and $\sigma_{\text{R}}^2=(\sigma_{\text{R1}}^2+\sigma_{\text{R2}}^2)/\mathfrak{g}^2$, where 
    \begin{equation*}
        \begin{aligned}
        &\sigma_{\text{R1}}^2=\mathfrak{h}\left(z^\top(\mathbb I_p + \mathfrak{m}_n\bm\Sigma)^{-1}\bm\beta\right)^2 \quad\mbox{and} \\
        &\sigma_{\text{R2}}^2=\left\{\sigma_\epsilon^2\mathfrak{g} + \lambda\mathfrak{m}_n'\left\{\lambda\bm\beta^\top(\mathbb\I_p + \mathfrak{m}_n\bm\Sigma)^{-2}\bm\Sigma\bm\beta - \sigma_{\epsilon}^2/n\Tr\left((\mathbb\I_p + \mathfrak{m}_n\bm\Sigma)^{-2}\bm\Sigma\right)\right\}\right\}\|\bm\Sigma^{-1/2}\z\|_2^2.
        \end{aligned}
    \end{equation*}
    Then under Assumptions~\ref{a:Sigmabound}-S\ref{a:Gaussian_entries}, the following quantitative CLT holds
        \begin{equation*}
            \begin{aligned}
                \sup_{t\in\mathbb R}\left|\mathbb P\left(\sigma_{\text{R}}^{-1}\sqrt{n}\left\{\left\langle \z, \hat{\bm\beta}_{\text{R}}(\lambda)\right\rangle -\left\{\z^\top\bm\beta - \frac{\lambda}{\mathfrak{g}}\left\langle \z, \bm\Sigma^{-1}\left(\bm\beta - (\mathbb\I_p + \mathfrak{m}_n\bm\Sigma)^{-1}\bm\beta\right)\right\rangle\right\}\right\} < t\right) - \Phi(t)\right| \to 0.
            \end{aligned}
        \end{equation*}

\section{Proof for Section~\ref{subsubsec:ridge_A}}\label{sec_proof_ridge2}

To prove the quantitative CLT of $A(\hat{\bm\beta}_{\text{R}}(\lambda))$ in Theorem~S\ref{thm: CLT for ridge A^2}, we use eight steps, employing a leave-one-out technique. 
Briefly, we decompose the numerator of $A(\hat{\bm\beta}_{\text{R}}(\lambda))$ and provide the quantitative CLT regarding the randomness of $\bm \epsilon_z$ and $\Z_0$. Then, by adopting the CLT for the de-biased estimator, we consider the randomness of $\X_0$ and $\bm \epsilon$ to obtain the corresponding quantitative CLTs. Using the quantitative CLT of quadratic forms, we account for the randomness of $\bm\beta$. Finally, we obtain the first-order limit of the denominator, and the quantitative CLT of $A(\hat{\bm\beta}_{\text{R}}(\lambda))$ follows from Slutsky's theorem.

\subsection{Numerator decomposition}
The numerator of $A(\hat{\bm\beta}_{\text{R}}(\lambda))$ can be rewritten as
\begin{equation*}
    \begin{aligned}
        &\left(\Z\bm\beta + \bm \epsilon_z\right)^\top \Z(\X^\top \X + n\lambda\mathbb \I_p)^{-1}\X^\top \y \\
        &\qquad= \bm\beta \Z^\top \Z(\X^\top \X+n\lambda\mathbb \I_p)^{-1}\X^\top \y + \bm \epsilon_z^\top \Z(\X^\top \X + n\lambda\mathbb \I_p)^{-1}\X^\top \y.
    \end{aligned}
\end{equation*}
This decomposition enables us to conduct sequentially conditioning.
Adopting notations in Theorem~S\ref{thm: CLT for ridge A^2}, in this section we denote $$\sigma^2 = n_z\left(\tau_1 + 2\tau_2^2 + \tau_3 + \tau_4^2\right)$$
for simplicity.

\subsection{Berry-Esseen bounds with the randomness of testing error}
Conditional on $\Z_0$, $\X_0$, $\bm \epsilon$, and $\bm\beta$, and consider the randomness of $\bm \epsilon_z$. By Lemma~S\ref{lemma:non-asmptotic CLT 3}, we have the following Berry-Esseen bound
\begin{equation}
\label{ineq:BE_ridge_A2_eps_z}
    \begin{aligned}
        &\sup_{t\in\mathbb R}\left|\mathbb P\left(\frac{\langle \bm \epsilon_z, \Z(\X^\top \X + n\lambda\mathbb \I_p)^{-1}\X^\top \y\rangle}{\sigma_{\epsilon_z}\sqrt{\|\Z(\X^\top \X + n\lambda\mathbb \I_p)^{-1}\X^\top \y\|_2^2}} < t\right) - \Phi_{\epsilon_z}(t)\right| \\
        &\qquad\leq O_p\left(\sqrt{\frac{\sum_{i=1}^{n_z}\left\{\Z(\X^\top \X + n\lambda\mathbb \I_p)^{-1}\X^\top \y\right\}^4_i}{\|\Z(\X^\top \X + n\lambda\mathbb \I_p)^{-1}\X^\top \y\|_2^4}}\right).
    \end{aligned}
\end{equation}
For simplicity, we denote 
$$\sigma_1^2 = \sigma_{\epsilon_z}^2\|\Z(\X^\top \X + n\lambda\mathbb\I_p)^{-1}\X^\top \y\|_2^2.$$
Considering the following CDF
\begin{equation*}
    \begin{aligned}
        \bm H_{\text{R}}(t) &= \mathbb P\left(\frac{\bm\beta^\top \Z^\top \Z(\X^\top \X + n\lambda\mathbb \I_p)^{-1}\X^\top \y + \bm \epsilon_z^\top \Z(\X^\top \X + n\lambda\mathbb \I_p)^{-1}\X^\top \y - \tau_0}{\sigma} < t\right)\\
        &=\mathbb P\left(\frac{\bm \epsilon_z^\top \Z(\X^\top \X + n\lambda\mathbb \I_p)^{-1}\X^\top \y}{\sigma_1}\frac{\sigma_1}{\sigma} < t - \frac{\bm\beta^\top \Z^\top \Z(\X^\top \X + n\lambda\mathbb \I_p)^{-1}\X^\top \y - \tau_0}{\sigma}\right).
    \end{aligned}
\end{equation*}
Recall that we denote $\Lambda_{\epsilon_z}$ as the standard Gaussian random variable replies only on the randomness of $\bm \epsilon_z$. Then by applying \cref{ineq:BE_ridge_A2_eps_z}, we have 
\begin{equation}
\label{ineq:BE_ridge_A2_eps_z_app}
    \begin{aligned}
        &\sup_{t\in\mathbb R}\left|\bm H_{\text{R}}(t) - \mathbb P\left(\frac{\sigma_1}{\sigma}\Lambda_{\epsilon_z} < t - \frac{\bm\beta^\top \Z^\top \Z(\X^\top \X + n\lambda\mathbb \I_p)^{-1}\X^\top \y - \tau_0}{\sigma}\right)\right|\\ 
        &\qquad\leq O_p\left(\sqrt{\frac{\sum_{i=1}^{n_z}\left\{\Z(\X^\top \X + n\lambda\mathbb \I_p)^{-1}\X^\top \y\right\}^4_i}{\|\Z(\X^\top \X + n\lambda\mathbb \I_p)^{-1}\X^\top \y\|_2^4}}\right).
    \end{aligned}
\end{equation}

\subsection{Berry-Esseen bounds with the randomness of testing data matrix}
Conditional on $\X_0$, $\bm \epsilon$, $\bm\beta$, we consider the randomness of $\Z_0$ in this section. 
We first define the subset where quantities related to $\Z_0$ concentrate properly, and further show such concentration holds with high probability. For $\forall \varepsilon_1, \varepsilon_2 > 0$ and $0<\delta<1/2$, we denote the subset: 
\begin{align*}
    &\Gamma_1(\varepsilon_1, \varepsilon_2) \coloneqq\\
    &\quad\left\{\Z_0: \left\{\left|\|\Z(\X^\top \X + n\lambda\mathbb \I_p)^{-1}\X^\top \y\|_2^2 - n_z \|\bm\Sigma^{1/2}(\X^\top \X+ n\lambda\mathbb \I_p)^{-1}\X^\top \y\|_2^2\right| < p^{1/2 + \delta}\varepsilon_1\right\}\cap\right.\\
    &\left.\quad\left\{\left|\sum_{i=1}^{n_z}\left\{\Z(\X^\top \X + n\lambda\mathbb \I_p)^{-1}\X^\top \y\right\}^4_i - n_z\mathbb E\left|z_{0_i}^\top\bm\Sigma^{1/2}(\X^\top \X+n\lambda\mathbb\I_p)^{-1}\X^\top \y\right|^4\right| < p^{1/2+\delta}\varepsilon_2\right\}\right\}.
\end{align*}
By Lemma~S\ref{lemma:von bahr-Esseen bound}, we have the following inequalities
\begin{equation*}
    \begin{aligned}
        &\mathbb P\left(\left|\|\Z \R\X^\top \y\|_2^2 - n_z\|\bm\Sigma^{1/2}\R\X^\top \y\|_2^2\right| < p^{1/2+\delta}\varepsilon_1\right) \geq
        1-O_p\left(\frac{p\mathbb E\left|\Z_{0_i}^\top\bm\Sigma^{1/2}\R\X^\top \y\right|^4}{p^{1+2\delta}\varepsilon_1^2}\right)\quad \mbox{and} \\
        &\mathbb P\left(\left|\sum_{i=1}^{n_z}\left\{\Z \R\X^\top \y\right\}^4_i - n_z\mathbb E\left|\Z_{0_i}^\top\bm\Sigma^{1/2}\R\X^\top \y\right|^4\right| < p^{1/2+\delta}\varepsilon_2\right) \geq 1-O_p\left(\frac{p\mathbb E\left|\Z_{0_i}^\top\bm\Sigma^{1/2}\R\X^\top \y\right|^8}{p^{1+2\delta}\varepsilon_2^2}\right).
    \end{aligned}
\end{equation*}
Moreover, notice that $\mathbb E\left|\z_{0_i}^\top\bm\Sigma^{1/2}\R\X^\top \y\right|^4 = O_p(1)$ and $\mathbb E\left|\z_{0_i}^\top\bm\Sigma^{1/2}\R\X^\top \y\right|^8 = O_p(1)$. Therefore, we have $\mathbb P(\Z_0 \in \Gamma_1) \geq 1- O_p(p^{-2\delta})$. For $\Z_0 \in \Gamma_1$, we can simplify the Berry-Esseen upper bound in 
\cref{ineq:BE_ridge_A2_eps_z_app} as follows 
\begin{equation*}
    \begin{aligned}
        &O_p\left(\sqrt{\frac{\sum_{i=1}^{n_z}(\z_{0_{i}}^\top\bm\Sigma^{1/2}\R\X^\top \y)^4}{\|\Z \R\
        \X^\top \y\|_2^2}}\right)\\ 
        &= O_p\left(\sqrt{\frac{n_z\left\{\mathbb E\left(z_0^4-3\right)\sum_{i=1}^p(\bm\Sigma^{1/2}\R\X^\top \y)_i^4 + 3(\y^\top \X \R\bm\Sigma \R\X^\top \y)^2\right\} + O_p(p^{1/2 + \delta})}{n_z^2\|\bm\Sigma^{1/2}\R\X^\top \y\|_2^4 + O_p(p^{3/2+\delta})}}\right).
    \end{aligned}
\end{equation*}
Considering the randomness of $\Z_0$ and by Lemma~S\ref{lemma:non-asmptotic CLT 1}, we have the following inequality
\begin{equation}
\label{ineq:BE_ridge_A2_Z}
    \begin{aligned}
        &\sup_{t\in\mathbb R}\left|\mathbb P\left(\frac{\bm\beta^\top \Z^\top \Z(\X^\top \X+n\lambda\mathbb \I_p)^{-1}\X^\top \y - n_z\bm\beta^\top\bm\Sigma(\X^\top \X+n\lambda\mathbb\I_p)^{-1}\X^\top \y}{\sigma_2} < t\right) - \Phi_{\Z_0}(t)\right|\\
        &\qquad\leq O_p(p^{-1/2}),
    \end{aligned}
\end{equation}
where $$\sigma_2^2 = n_z\left\{\mathbb E(z_0^4 -3)\sum_{i=1}^p(\bm\Sigma^{1/2}\bm\beta)_i^2\left(\R\X^\top \y\right)_i^2 + 2\left(\bm\beta^\top\bm\Sigma \R\X^\top \y\right)^2 + \|\bm\Sigma^{1/2}\bm\beta\|_2^2\|\bm\Sigma^{1/2}\R\X^\top \y\|_2^2\right\}.$$
Rearranging the terms in second probability of \cref{ineq:BE_ridge_A2_eps_z_app} and applying \cref{ineq:BE_ridge_A2_Z}, we have 
\begin{equation}
\label{ineq:BE_ridge_A2_Z_app}
    \begin{aligned}
        &\sup_{t\in\mathbb R}\left|\mathbb P\left(\frac{\bm\beta^\top \Z^\top \Z \R\X^\top \y - n_z\bm\beta^\top\bm\Sigma \R\X^\top \y}{\sigma} < t-\frac{\sigma_1}{\sigma}\Lambda_{\epsilon_z} - \frac{n_z\bm\beta^\top \bm\Sigma \R\X^\top \y - \tau_0}{\sigma}\right) - \right.\\
        &\left.\qquad\mathbb P\left(\frac{\sigma_2}{\sigma}\Lambda_{\Z_0} < t - \frac{\sqrt{\sigma_{\epsilon_z}^2n_z\|\bm\Sigma^{1/2}\R\X^\top \y\|_2^2 + O_p(p^{1/2+\delta})}}{\sigma}\Lambda_{\epsilon_z} - \frac{n_z\bm\beta^\top \bm\Sigma \R\X^\top \y - \tau_0}{\sigma}\right)\right|\\
        &\qquad\leq O_p\left(\max(p^{-1/2}, p^{-2\delta})\right).
    \end{aligned}
\end{equation}

\subsection{First-order concentrations with the randomness of training error and training data matrix}
\label{subsec:var_lim}
Conditional on $\bm\beta$, we provide first-order concentrations for the following two quantities that are involved in the fluctuation of Gaussian random variables
\begin{enumerate}
    \item $\y^\top \X(\X^\top \X + n\lambda\mathbb\I_p)^{-1}\bm\Sigma(\X^\top \X + n\mathbb\I_p)^{-1}\X^\top \y$ and 
    \item $\bm\beta^\top\bm\Sigma(\X^\top \X + n\lambda\mathbb\I_p)^{-1}\X^\top \y.$
\end{enumerate}

\subsubsection{Concentration for the first quantity}
In this section, we provide concentration for $\y^\top \X(\X^\top \X + n\lambda\mathbb\I_p)^{-1}\bm\Sigma(\X^\top \X+n\lambda\mathbb\I_p)^{-1}\X^\top \y$ in three steps.\\~\\
\noindent\textbf{Step 1:} Decomposing $\y^\top \X(\X^\top \X + n\lambda\mathbb \I_p)^{-1}\bm\Sigma(\X^\top \X + n\lambda\mathbb \I_p)^{-1}\X^\top \y$, we have 
\begin{equation}
\label{decomp:ytXRSigmaRXty}
    \begin{aligned}
        &\y^\top \X \R\bm\Sigma \R\X^\top \y \\
        &\quad =\bm\beta^\top \X^\top \X \R\bm\Sigma \R\X^\top \X\bm\beta +\bm \epsilon^\top \X \R\bm\Sigma \R\X^\top\bm \epsilon + 2\bm \epsilon^\top \X \R\bm\Sigma \R\X^\top \X\bm\beta.
    \end{aligned}
\end{equation}\\
\noindent\textbf{Step 2:} 
We first show that the interaction term in \cref{decomp:ytXRSigmaRXty}, $\bm \epsilon^\top \X \R\bm\Sigma \R\X^\top \X\bm\beta$, is negligible with high probability, as it has order lower than $\Theta_p(1)$. By Markov's inequality, we have 
\begin{equation*}
    \begin{aligned}
        \mathbb P\left(\left|\bm \epsilon^\top \X \R\bm\Sigma \R\X^\top \X\bm\beta \right| > \varepsilon_3p^{- \delta} \right) \leq \frac{\mathbb E\left(\bm\beta^\top \X^\top \X \R\bm\Sigma \R\X^\top \X \R\bm\Sigma \R\X^\top \X\bm\beta\right)}{\varepsilon_3^2p^{-2\delta}}.
    \end{aligned}
\end{equation*}
Note that the numerator of the right-hand side is of order $O_p(p^{-1})$. Therefore, for $\forall \delta \in (0, 1/2)$, we have 
$$\bm \epsilon^\top \X(\X^\top \X+n\lambda\mathbb\I_p)^{-1}\bm\Sigma (\X^\top \X+n\lambda\mathbb\I_p)^{-1}\X^\top \X\bm\beta = O_p(p^{-\delta})$$ with probability of at least $1- O_p(p^{2\delta-1})$. 
Our question now boils down to finding the concentrations of $\bm\beta^\top \X^\top \X \R\bm\Sigma\X^\top \X\bm\beta$ and $\bm \epsilon^\top \X \R\bm\Sigma \R\X^\top\bm \epsilon$, the first two terms in \cref{decomp:ytXRSigmaRXty}.\\~\\
\noindent\textbf{Step 3:} We provide the first-order concentration for $\bm\beta^\top \X^\top \X \R\bm\Sigma\X^\top \X\bm\beta$.
Note that
\begin{equation*}
    \begin{aligned}
        &\bm\beta^\top \X^\top \X(\X^\top \X+n\lambda\mathbb \I_p)^{-1}\bm\Sigma(\X^\top \X+n\lambda\mathbb \I_p)^{-1}\X^\top \X\bm\beta \\
        &\quad= \bm\beta^\top\bm\Sigma\bm\beta - 2n\lambda\bm\beta^\top(\X^\top \X + n\lambda\mathbb\I_p)^{-1}\bm\Sigma\bm\beta
         + \lambda^2n^2\bm\beta^\top(\X^\top \X + n\lambda\mathbb\I_p)^{-1}\bm\Sigma(\X^\top \X+n\lambda\mathbb\I_p)^{-1}\bm\beta.
    \end{aligned}
\end{equation*}
Similar to our arguments in Lemma~S\ref{lemma:1st_order_trace} and \cref{concent:ASigA}, with probability of at least $1-O_p(p^{-D})$ for some large $D\in\mathbb R$ and any small $\vartheta > 0$, we have
\begin{equation*}
    \begin{aligned}
        &\left|\bm\beta^\top \R\bm\Sigma\bm\beta - \frac{1}{n\lambda}\bm\beta^\top(\mathbb \I_p + \mathfrak{m}_n\bm\Sigma)^{-1}\bm\Sigma\bm\beta\right| \leq O_p(p^{-3/2+\vartheta})\quad \mbox{and}\\
        &\left|\bm\beta^\top \R\bm\Sigma \R\bm\beta - \frac{\mathfrak{m}_n'}{n^2\lambda^2\mathfrak{m}_n^2}\bm\beta^\top(\mathbb\I_p + \mathfrak{m}_n\bm\Sigma)^{-2}\bm\Sigma\bm\beta\right| \leq O_p(p^{-5/2+\vartheta}).
    \end{aligned}
\end{equation*}
Therefore, with probability of at least $1-O_p(p^{-D})$, we have 
\begin{equation*}
    \begin{aligned}
        &\Bigg|\bm\beta^\top \X^\top \X \R\bm\Sigma \R\X^\top \X\bm\beta - \Big\{\bm\beta^\top\bm\Sigma\bm\beta - 2\bm\beta\left(\mathbb\I_p+\mathfrak{m}_n\bm\Sigma\right)^{-1}\bm\Sigma\bm\beta + \frac{\mathfrak{m}_n'}{\mathfrak{m}_n^2}\bm\beta^\top\left(\mathbb\I_p+\mathfrak{m}_n\bm\Sigma\right)^{-2}\bm\Sigma\bm\beta\Big\}\Bigg| \\
        &\qquad\leq O_p(p^{-1/2+\vartheta}).
    \end{aligned}
\end{equation*}
\noindent\textbf{Step 4:} We provide the first-order concentration for  $\bm \epsilon^\top \X \R\bm\Sigma \R\X^\top\bm \epsilon$.
We show that, by fixing $\X_0$, $\bm \epsilon^\top \X \R\bm\Sigma \R\X^\top\bm \epsilon$ concentrates regarding the randomness of $\bm \epsilon$. 
We first compute 
$$\mathbb E\left(\bm \epsilon^\top \X \R\bm\Sigma \R\X^\top\bm \epsilon\right) = \sigma_\epsilon^2\Tr\left(\X \R\bm\Sigma \R \X^\top\right).$$
Denote the fixed matrix
$$\M = \X \R\bm\Sigma \R\X^\top,$$
for $\forall \delta \in (0,1/2), \forall \varepsilon_4 > 0$, by Markov's inequality, we have
\begin{equation*}
    \begin{aligned}
        \mathbb P\left(\left|\bm \epsilon^\top \M\bm \epsilon - \sigma_\epsilon^2\Tr(\M)\right| > p^{-\delta}\varepsilon_4\right) \leq \frac{\mathbb E\left[\bm \epsilon^\top \M\bm \epsilon\bm \epsilon^\top \M\bm \epsilon\right] - \sigma_\epsilon^4\Tr(\M)^2}{\varepsilon_4^2p^{-2\delta}}.
    \end{aligned}
\end{equation*}
After careful computation, we have  
$$\mathbb E\left[\bm \epsilon^\top \M\bm \epsilon\bm \epsilon^\top \M\bm \epsilon\right] - \sigma_\epsilon^4\Tr(\M)^2 = \mathbb E\left(\epsilon^4 - 3\right)\sum_{i = 1}^n \M_{i,i}^2 + 2\sigma_\epsilon^4\Tr(\M^2).$$
Here $\sum_{i=1}^n \M_{i,i}^2 \preceq \Tr(\M^2)$, and note that $\Tr(\M^2) = O_p(p^{-1})$. Therefore, we have 
\begin{equation*}
    \begin{aligned}
        \mathbb P\left(\left|\epsilon^\top \M\epsilon - \sigma_\epsilon^2\Tr(\M)\right| > p^{-\delta}\varepsilon_4\right) \leq O_p(p^{2\delta - 1}).
    \end{aligned}
\end{equation*}
Furthermore, we can decompose $\Tr(\M)$ as follows 
$$\Tr\left(\X \R\bm\Sigma \R\X^\top\right) = \Tr\left(\bm\Sigma \R\right) - n\lambda \Tr\left(\R\bm\Sigma \R\right).$$
By the anisotropic local law \citep{anisotropic_local_law}, 
similar to our argument in \cref{ineq:iso_first_order_form}, 
with probability of at least $1-O_p(p^{-D+1})$ for some large enough $D > 1$ and some small enough $\vartheta > 0$, 
we have
\begin{equation*}
    \begin{aligned}
        \left|\Tr(\bm\Sigma \R) - \frac{1}{n\lambda}\Tr\left((\mathbb \I_p + \mathfrak{m}_n\bm\Sigma)^{-1}\bm\Sigma\right)\right| \leq \psi(-\lambda)O_p(p^{\vartheta}).
    \end{aligned}
\end{equation*}
Similar to our argument in \cref{concent:ASigA},  with probability of at least $1-O_p(p^{-D+1})$, we have 
\begin{equation*}
    \begin{aligned}
        \left|\Tr\left(\R\bm\Sigma \R\right) - \frac{\mathfrak{m}_n'}{n^2\lambda^2\mathfrak{m}_n^2}\Tr\left((\mathbb \I_p + \mathfrak{m}_n\bm\Sigma)^{-2}\bm\Sigma\right)\right| \leq \psi(-\lambda)O_p(p^{\vartheta - 1}).
    \end{aligned}
\end{equation*}
Combining these results above, for $\forall \delta \in (0, 1/2)$, with probability of at least $1-O_p(p^{2\delta-1})$, we have
\begin{equation*}
    \begin{aligned}
        \left|\bm \epsilon^\top \X \R\bm\Sigma \R\X^\top\bm \epsilon - \frac{\sigma_\epsilon^2}{n\lambda}\left\{\Tr\left((\mathbb \I_p + \mathfrak{m}_n\bm\Sigma)^{-1}\bm\Sigma\right) - \frac{\mathfrak{m}'_n}{\mathfrak{m}_n^2} \Tr\left((\mathbb \I_p + \mathfrak{m}_n\bm\Sigma)^{-2}\bm\Sigma\right)\right\}\right| \leq O_p(p^{-\delta}).
    \end{aligned}
\end{equation*}
Therefore, for $\forall \delta \in (0,1/2)$, with probability of at least $1-O_p(p^{2\delta - 1})$, the following concentration inequality holds
\begin{equation*}
    \begin{aligned}
        &\left|\y^\top \X \R\bm\Sigma \R\X^\top \y -  \left\{\bm\beta^\top\bm\Sigma\bm\beta - 2\bm\beta\left(\mathbb\I_p+\mathfrak{m}_n\bm\Sigma\right)^{-1}\bm\Sigma\bm\beta + \frac{\mathfrak{m}_n'}{\mathfrak{m}_n^2}\bm\beta^\top\left(\mathbb\I_p+\mathfrak{m}_n\bm\Sigma\right)^{-2}\bm\Sigma\bm\beta\right.\right.\\
        &\left.\left. \qquad+ \frac{\sigma_\epsilon^2}{n\lambda}\left\{\Tr\left((\mathbb \I_p + \mathfrak{m}_n\bm\Sigma)^{-1}\bm\Sigma\right) - \frac{\mathfrak{m}_n'}{\mathfrak{m}_n^2}\Tr\left((\mathbb \I_p + \mathfrak{m}_n\bm\Sigma)^{-2}\bm\Sigma\right)\right\}\right\}\right| \leq O_p(p^{-\delta}).
     \end{aligned}
\end{equation*}

\subsubsection{Concentration for the second quantity}
In this section, we provide concentration for $\bm\beta^\top\bm\Sigma(\X^\top \X + n\lambda\mathbb\I_p)^{-1}\X^\top \y$ in three steps.\\~\\
\noindent\textbf{Step 1:} Decomposing $\bm\beta^\top\bm\Sigma(\X^\top \X + n\lambda\mathbb\I_p)^{-1}\X^\top \y=\bm\beta^\top\bm\Sigma \R \X^\top \y$, we have 
\begin{equation*}
    \begin{aligned}
        \bm\beta^\top\bm\Sigma \R\X^\top \y &= \bm\beta^\top\bm\Sigma \R\X^\top \X\bm\beta + \bm\beta^\top\bm\Sigma \R\X^\top\bm \epsilon=\bm\beta^\top\bm\Sigma\bm\beta - n\lambda\bm\beta^\top\bm\Sigma \R\bm\beta + \bm\beta^\top\bm\Sigma \R\X^\top\bm \epsilon.
    \end{aligned}
\end{equation*}\\ 
\noindent\textbf{Step 2:} We first show that the interaction term $\bm\beta^\top\bm\Sigma \R\X^\top\bm \epsilon$ is negligible
with high probability. By Markov's inequality, we have
\begin{equation*}
    \begin{aligned}
        \mathbb P\left(\left|\bm\beta^\top\bm\Sigma \R\X^\top\bm \epsilon\right| > \varepsilon_5p^{-\delta}\right) \leq \frac{\sigma_\epsilon^2\mathbb E\left[\bm\beta^\top\bm\Sigma R\X^\top \X R\bm\Sigma\bm\beta\right]}{\varepsilon_5^2p^{-2\delta}}
    \end{aligned}
\end{equation*}
for $\forall \delta \in (0, 1/2)$ and $\forall \varepsilon_5 >0$. 
Note that $\mathbb E\left(\bm\beta^\top\bm\Sigma \R\X^\top \X \R\bm\Sigma\bm\beta\right) = O_p(1)$, with probability of at least $1-O_p(p^{-1+2\delta})$, we have 
$$\bm\beta^\top\bm\Sigma(\X^\top \X+n\lambda\mathbb\I_p)^{-1}\X^\top\bm \epsilon = O_p(p^{-\delta}).$$
Therefore, $\bm\beta^\top\bm\Sigma \R\X^\top\bm \epsilon$ is negligible with high probability.\\~\\ 
\noindent\textbf{Step 3:} 
Similar to our argument in \cref{ineq:iso_first_order_form}, for $\forall \vartheta > 0$ small enough, with probability of at least $1-O_p(p^{-D})$, we have 
\begin{equation*}
    \begin{aligned}
       &\left|n\lambda\bm\beta^\top\bm\Sigma \R\bm\beta - \bm\beta^\top(\mathbb\I_p + \mathfrak{m}_n\bm\Sigma)^{-1}\bm\Sigma\bm\beta\right| \leq O_p(p^{-1/2+\vartheta})\\
      & \implies \left|\bm\beta^\top\bm\Sigma \R\X^\top \y - \left\{\bm\beta^\top\bm\Sigma\bm\beta - \bm\beta^\top(\mathbb\I_p + \mathfrak{m}_n\bm\Sigma)^{-1}\bm\Sigma\bm\beta\right\}\right| \leq O_p(p^{-1/2+\vartheta}).
    \end{aligned}
\end{equation*}
Now define the subset where quantities related to $\bm \epsilon$ and $\X_0$ concentrates properly, which is 
\begin{equation*}
    \begin{aligned}
        &\Gamma_2 \coloneqq \\
        &\left\{\X_0, \epsilon: \left\{\left|\y^\top \X \R\bm\Sigma \R\X^\top \y -  \left\{\bm\beta^\top\bm\Sigma\bm\beta - 2\bm\beta\left(\mathbb\I_p+\mathfrak{m}_n\bm\Sigma\right)^{-1}\bm\Sigma\bm\beta + \frac{\mathfrak{m}_n'}{\mathfrak{m}_n^2}\bm\beta^\top\left(\mathbb\I_p+\mathfrak{m}_n\bm\Sigma\right)^{-2}\bm\Sigma\bm\beta\right.\right.\right.\right.\\
        &\left.\left.\left. \left.+ \frac{\sigma_\epsilon^2}{n\lambda}\left\{\Tr\left((\mathbb \I_p + \mathfrak{m}_n\bm\Sigma)^{-1}\bm\Sigma\right) - \frac{\mathfrak{m}_n'}{\mathfrak{m}_n^2}\Tr\left((\mathbb \I_p + \mathfrak{m}_n\bm\Sigma)^{-2}\bm\Sigma\right)\right\}\right\}\right| \leq p^{-\delta}\varepsilon_1\right\}\cap\right.\\
        &\left.\left\{\left|\bm\beta^\top\bm\Sigma \R\X^\top \y - \left(\bm\beta^\top\bm\Sigma\bm\beta - \bm\beta^\top(\mathbb\I_p + \mathfrak{m}_n\bm\Sigma)^{-1}\bm\Sigma\bm\beta\right)\right| \leq p^{-\delta}\varepsilon_2\right\}\right\}, 
    \end{aligned}
\end{equation*}
and we have shown that $\mathbb P\left(\{\X_0, \epsilon\} \in \Gamma_2\right) \geq 1-O_p(p^{-1+2\delta})$ for $\forall \delta \in (0,1/2)$.

\subsection{Berry-Esseen bounds with the randomness of training error and training data matrix}
Recall that we denote
\begin{equation*}
        \begin{aligned}
             \mathfrak{g} = 1-(1-\ell_1)\phi_n \quad \mbox{and} \quad \mathfrak{h} = \phi_n\left\{\ell_1 - \frac{\lambda\mathfrak{m}_n'}{\mathfrak{m}_n}\left(\ell_1 - \ell_2\right)\right\}.
        \end{aligned}
\end{equation*}
Considering fixed $\bm\beta$ and replacing $\z$ in Theorem~S\ref{thm: CLT for ridge new} by $\bm\Sigma\bm\beta$, as $p \to \infty$, we have the following quantitative CLT. 
\begin{lem.s}
\label{ineq:BE_ridge_A2_X_eps}
Under Assumptions~\ref{a:Sigmabound}-\ref{a:Sparsity} and S\ref{a:Gaussian_entries}, we have
\begin{align*}
            \sup_{t\in\mathbb R}\left|\mathbb P\left( \mathfrak{u}^{-1}\left\{\left\langle \bm\Sigma\bm\beta, \hat{\bm\beta}_{\text{R}}(\lambda)\right\rangle + \frac{\lambda}{\mathfrak{g}}\left\langle\bm\beta,\left\{\bm\beta - (\mathbb\I_p + \mathfrak{m}_n\bm\Sigma)^{-1}\bm\beta\right\}\right\rangle - \bm\beta^\top\bm\Sigma \bm\beta\right\} < t\right) - \Phi(t)\right| \to 0,
\end{align*}
where $\mathfrak{u}^2=(\mathfrak{u}_1+\mathfrak{u}_2)/(n\mathfrak{g}^2)$ with 
\begin{align*}
        &\mathfrak{u}_1=\mathfrak{h}\left\{\bm\beta^\top(\mathbb I_p + \mathfrak{m}_n\bm\Sigma)^{-1}\bm\Sigma\bm\beta\right\}^2 \quad \mbox{and} \\
         &\mathfrak{u}_2=\left[\sigma_\epsilon^2\mathfrak{g} + \lambda\mathfrak{m}_n'\left\{\lambda\bm\beta^\top(\mathbb\I_p + \mathfrak{m}_n\bm\Sigma)^{-2}\bm\Sigma\bm\beta - \frac{\sigma_{\epsilon}^2}{n}\Tr\left((\mathbb\I_p + \mathfrak{m}_n\bm\Sigma)^{-2}\bm\Sigma\right)\right\}\right]\|\bm\Sigma^{1/2}\bm\beta\|_2^2.
\end{align*}
\end{lem.s}
Moreover, following from our concentration results in Section~\ref{subsec:var_lim}, and combining \cref{ineq:BE_ridge_A2_eps_z_app}, \cref{ineq:BE_ridge_A2_Z_app}, and \cref{ineq:BE_ridge_A2_X_eps_concent}, we have the following Berry-Esseen inequality regarding the randomness of $\X_0$ and $\bm \epsilon$
\begin{equation}
\label{ineq:BE_ridge_A2_X_eps_concent}
    \begin{aligned}
        &\sup_{t\in\mathbb R}\left|\bm H_{\text{R}}(t) - 
        \mathbb P\left(\frac{\sqrt{n_z\left(\varkappa_1\|\bm\Sigma^{1/2}\bm\beta\|_2^2 + 2\varkappa_2^2\right) + O_p(p^{1-\delta})}}{\sigma}\Lambda_{\Z_0} < \right.\right.\\
        &\left.\left. \qquad t -\frac{\sqrt{n_z\sigma_{\epsilon_z}^2\varkappa_1 + O_p(p^{1/2+\delta}+p^{1-\delta})}}{\sigma}\Lambda_{\epsilon_z}- \frac{n_z\bm\beta^\top \bm\Sigma \R\X^\top \y - \tau_0}{\sigma}\right)\right| \\
        & \qquad \leq  O_p\left(\max\left(\sqrt{\frac{3n_z\varkappa_1^2 + O_p(p^{1/2 + \delta} + p^{1-2\delta})}{n_z^2\varkappa_1^2 + O_p(p^{3/2+\delta} + p^{2-2\delta})}}, p^{-1+2\delta}, p^{-1/2}, p^{-2\delta}\right)\right),
    \end{aligned}
\end{equation}
where
\begin{align*}
        &\varkappa_1 =  \bm\beta^\top\bm\Sigma\bm\beta - 2\bm\beta^\top\left(\mathbb\I_p+\mathfrak{m}_n\bm\Sigma\right)^{-1}\bm\Sigma\bm\beta + \frac{\mathfrak{m}_n'}{\mathfrak{m}_n^2}\bm\beta^\top\left(\mathbb\I_p+\mathfrak{m}_n\bm\Sigma\right)^{-2}\bm\Sigma\bm\beta + \frac{\sigma_\epsilon^2}{n\lambda}\left\{\Tr\left((\mathbb \I_p + \mathfrak{m}_n\bm\Sigma)^{-1}\bm\Sigma\right)\right.\\
        &\left. \qquad- \frac{\mathfrak{m}_n'}{\mathfrak{m}_n^2}\Tr\left((\mathbb \I_p + \mathfrak{m}_n\bm\Sigma)^{-2}\bm\Sigma\right)\right\}\quad \mbox{and} \\
        &\varkappa_2 = \bm\beta^\top\bm\Sigma\bm\beta - \bm\beta^\top(\mathbb\I_p + \mathfrak{m}_n\bm\Sigma)^{-1}\bm\Sigma\bm\beta.
\end{align*}

Note that we have dropped the first term in the original $\sigma_2^2$ by Lemma~S\ref{lemma:neg_high_order}, where we have shown that we have $(\bm\Sigma\bm\beta)_i^2 \leq O_p(p^{-1/2+\delta})$ with high probability for $\delta \in (0,1/2)$. 
Rearranging the second probability in \cref{ineq:BE_ridge_A2_X_eps_concent}, we have 
\begin{align*}
        &\mathbb P\left(\frac{n_z\mathfrak{u}}{\sigma}\frac{\bm\beta\bm\Sigma \R\X^\top \y  -\left(\bm\beta^\top\bm\Sigma\bm\beta - \lambda/\mathfrak{g}\left\langle\bm\beta,\left\{\bm\beta - (\mathbb\I_p + \mathfrak{m}_n\bm\Sigma)^{-1}\bm\beta\right\}\right\rangle\right)}{\mathfrak{u}} < t - \right.\\
        &\left. \qquad\frac{\sqrt{n_z\sigma_{\epsilon_z}^2\varkappa_1 + O_p(p^{1/2+\delta} + p^{1-\delta})}}{\sigma}\Lambda_{\epsilon_z}- \frac{\sqrt{n_z\left(\varkappa_1 + 2\varkappa_2^2\right) + O_p(p^{1-\delta})}}{\sigma}\Lambda_{\Z_0} - \right.\\
        &\left. \qquad\frac{n_z\left(\bm\beta^\top\bm\Sigma\bm\beta - \lambda/\mathfrak{g}\left\langle\bm\beta,\left\{\bm\beta - (\mathbb\I_p + \mathfrak{m}_n\bm\Sigma)^{-1}\bm\beta\right\}\right\rangle\right)-\tau_0}{\sigma}\right).
\end{align*}
Combining Lemma~S\ref{ineq:BE_ridge_A2_X_eps} with \cref{ineq:BE_ridge_A2_X_eps_concent}, we have
\begin{align*}
        &\sup_{t\in\mathbb R}\left|\bm H_{\text{R}}(t) - \mathbb P\left(\frac{n_z\mathfrak{u}}{\sigma}\Lambda_{X_0,\epsilon} < t - \frac{\sqrt{n_z\sigma_{\epsilon_z}^2\varkappa_1 + O_p(p^{1/2+\delta})}}{\sigma}\Lambda_{\epsilon_z} -\right.\right.\\ 
        &\left.\left.\qquad \frac{\sqrt{n_z\left(\|\bm\Sigma^{1/2}\bm\beta\|_2^2\varkappa_1 + 2\varkappa_2^2\right) + O_p(p^{1-\delta})}}{\sigma}\Lambda_{\Z_0} -\right.\right.\\
        & \left.\left.\qquad\frac{n_z\left(\bm\beta^\top\bm\Sigma\bm\beta - \lambda/\mathfrak{g}\left\langle\bm\beta,\left\{\bm\beta - (\mathbb\I_p + \mathfrak{m}_n\bm\Sigma)^{-1}\bm\beta\right\}\right\rangle\right)-\tau_0}{\sigma}\right)\right|\\
        &\qquad\leq \max\left(o_p(1), O_p\left(\max\left(\sqrt{\frac{3n_z\varkappa_1^2 + O_p(p^{1/2 + \delta} + p^{1-2\delta})}{n_z^2\varkappa_1^2 + O_p(p^{3/2+\delta} + p^{2-2\delta})}}, p^{-1+2\delta}, p^{-1/2}, p^{-2\delta}\right)\right)\right).
\end{align*}

Now, the only remaining source of randomness that we have not handled is from $\bm\beta$. In the next section, we will present the first-order concentration for quantities related to $\bm\beta$, as well as Berry-Esseen bounds for the quadratic form.

\subsection{Berry-Esseen bounds of quadratic quantities of the genetic effects}
By Lemma~S\ref{prop:quad_first_limit}, note that $\mathbb E(\bm\beta^4) = O_p(p^{-2})$, for $\forall\delta \in (0,1/2)$, we have 
\begin{align*}
        &\mathbb P\left(\left|\bm\beta^\top\bm\Sigma\bm\beta - \sigma_{\bm\beta}^2\Tr(\bm\Sigma\mathbb\I_m)/p \right| \leq O_p(p^{-\delta})\right) \geq 1-O_p(p^{2\delta-1}),\\
        &\mathbb P\left(\left|\bm\beta^\top\left(\mathbb\I_p+\mathfrak{m}_n\bm\Sigma\right)^{-1}\bm\Sigma\bm\beta - \sigma_{\bm\beta}^2\Tr\left((\mathbb\I_p + \mathfrak{m}_n\bm\Sigma)^{-1}\bm\Sigma\mathbb\I_m\right)/p\right| \leq O_p(p^{-\delta})\right) \geq 1-O_p(p^{2\delta - 1}),\\
        &\mathbb P\left(\left|\bm\beta^\top\left(\mathbb\I_p+\mathfrak{m}_n\bm\Sigma\right)^{-2}\bm\Sigma\bm\beta - \sigma_{\bm\beta}^2\Tr\left((\mathbb\I_p + \mathfrak{m}_n\bm\Sigma)^{-2}\bm\Sigma\mathbb\I_m\right)/p\right| \leq O_p(p^{-\delta})\right) \geq 1-O_p(p^{2\delta - 1}),\\
        &\mathbb P\left(\left|\bm\beta^\top\left(\mathbb\I_p+\mathfrak{m}_n\bm\Sigma\right)^{-1}\bm\beta - \sigma_{\bm\beta}^2\Tr\left((\mathbb\I_p + \mathfrak{m}_n\bm\Sigma)^{-1}\mathbb\I_m\right)/p\right| \leq O_p(p^{-\delta})\right) \geq 1-O_p(p^{2\delta - 1}),\quad \mbox{and} \\
        &\mathbb P\left(\left|\bm\beta^\top\left(\mathbb\I_p+\mathfrak{m}_n\bm\Sigma\right)^{-2}\bm\beta - \sigma_{\bm\beta}^2\Tr\left((\mathbb\I_p + \mathfrak{m}_n\bm\Sigma)^{-2}\mathbb\I_m\right)/p\right| \leq O_p(p^{-\delta})\right) \geq 1-O_p(p^{2\delta - 1}).
\end{align*}
Now for $\forall \varepsilon_1, \varepsilon_2,\cdots,\varepsilon_5 >0$, we denote the subset
\begin{align*}
        \Gamma_3(\varepsilon_1, \cdots, \varepsilon_5) \coloneqq& \left\{\bm\beta: \left\{\left|\bm\beta^\top\bm\Sigma\bm\beta - \sigma_{\bm\beta}^2\Tr(\bm\Sigma\mathbb\I_m)/p \right| \leq p^{-\delta}\varepsilon_1\right\}\cap\right.\\
        &\left.\left\{\left|\bm\beta^\top\left(\mathbb\I_p+\mathfrak{m}_n\bm\Sigma\right)^{-1}\bm\Sigma\bm\beta - \sigma_{\bm\beta}^2\Tr\left((\mathbb\I_p + \mathfrak{m}_n\bm\Sigma)^{-1}\bm\Sigma\mathbb\I_m\right)/p\right|\leq p^{-\delta}\varepsilon_2\right\}\cap\right.\\
        &\left.\left\{\left|\bm\beta^\top\left(\mathbb\I_p+\mathfrak{m}_n\bm\Sigma\right)^{-2}\bm\Sigma\bm\beta - \sigma_{\bm\beta}^2\Tr\left((\mathbb\I_p + \mathfrak{m}_n\bm\Sigma)^{-2}\bm\Sigma\mathbb\I_m\right)/p\right| \leq p^{-\delta}\varepsilon_3\right\}\cap\right.\\
        &\left.\left\{\left|\bm\beta^\top\left(\mathbb\I_p+\mathfrak{m}_n\bm\Sigma\right)^{-1}\bm\beta - \sigma_{\bm\beta}^2\Tr\left((\mathbb\I_p + \mathfrak{m}_n\bm\Sigma)^{-1}\mathbb\I_m\right)/p\right| \leq p^{-\delta}\varepsilon_4\right\}\cap\right.\\
        &\left.\left\{\left|\bm\beta^\top\left(\mathbb\I_p+\mathfrak{m}_n\bm\Sigma\right)^{-2}\bm\beta - \sigma_{\bm\beta}^2\Tr\left((\mathbb\I_p + \mathfrak{m}_n\bm\Sigma)^{-2}\mathbb\I_m\right)/p\right| \leq p^{-\delta}\varepsilon_5\right\}\right\}
\end{align*}
and we have shown that $\mathbb P\left(\bm\beta \in \Gamma_3\right) \geq 1-O_p(p^{2\delta-1})$. 
For $\bm\beta \in \Gamma_3$ we have $\varkappa_1 = \Theta_p(1)$ and $ \varkappa_2 = \Theta_p(1)$. 
Therefore, we have 
$$\max\left(o_p(1), O_p\left(\max\left(\sqrt{\frac{3n_z\varkappa_1^2 + O_p(p^{1/2 + \delta} + p^{1-2\delta})}{n_z^2\varkappa_1^2 + O_p(p^{3/2+\delta} + p^{2-2\delta})}}, p^{-1+2\delta}, p^{-1/2}, p^{-2\delta}\right)\right)\right) = o_p(1).$$

Moreover, we provide the concentration of $\varkappa_1, \varkappa_2$ and $\mathfrak{u}$ for $\bm\beta \in \Gamma_3$. 
Recall the following quantities we have defined in Theorem~S\ref{thm: CLT for ridge A^2}
\begin{align*}
            &\zeta_1 = \frac{\sigma_{\bm\beta}^2}{p}\left\{\Tr(\bm\Sigma\mathbb\I_m) - 2\Tr\left(\left(\mathbb\I_p+\mathfrak{m}_n\bm\Sigma\right)^{-1}\bm\Sigma\mathbb\I_m\right) + \frac{\mathfrak{m}_n'}{\mathfrak{m}_n^2}\Tr\left(\left(\mathbb\I_p+\mathfrak{m}_n\bm\Sigma\right)^{-2}\bm\Sigma\mathbb\I_m\right)\right\}  \\
            &\qquad+\frac{\sigma_\epsilon^2}{n\lambda}\left\{\Tr\left((\mathbb \I_p + \mathfrak{m}_n\bm\Sigma)^{-1}\bm\Sigma\right)- \frac{\mathfrak{m}_n'}{\mathfrak{m}_n^2}\Tr\left((\mathbb \I_p + \mathfrak{m}_n\bm\Sigma)^{-2}\bm\Sigma\right)\right\},\\
            &\tau_2 = \frac{\sigma_{\bm\beta}^2}{p}\left\{\Tr(\bm\Sigma\mathbb\I_m) - \Tr\left((\mathbb\I_p + \mathfrak{m}_n\bm\Sigma)^{-1}\bm\Sigma\mathbb\I_m\right)\right\}, \quad \mbox{and}\\
            &\tau_3 = n_z/(n\mathfrak{g}^2)\Bigg[\mathfrak{h}\left(\frac{\sigma_{\bm\beta}^2}{p}\Tr((\mathbb I_p + \mathfrak{m}_n\bm\Sigma)^{-1}\bm\Sigma\mathbb\I_m)\right)^2 \\
&\qquad+ \frac{\sigma_{\bm\beta}^2}{p}\Tr(\bm\Sigma\mathbb\I_m)\Bigg\{\sigma_\epsilon^2\mathfrak{g} + \lambda\mathfrak{m}_n'\Bigg(\frac{\lambda\sigma_{\bm\beta}^2}{p}\Tr((\mathbb\I_p + \mathfrak{m}_n\bm\Sigma)^{-2}\bm\Sigma\mathbb\I_m)- \frac{\sigma_{\epsilon}^2}{n}\Tr((\mathbb\I_p + \mathfrak{m}_n\bm\Sigma)^{-2}\bm\Sigma)\Bigg)\Bigg\}\Bigg].
\end{align*}
It follows that with probability of at least $1-O_p(p^{-2\delta})$ over the randomness of $\bm\beta$, we have
\begin{equation*}
    \begin{aligned}
        &\sup_{t\in\mathbb R}\left|\bm H_{\text{R}}(t) - \mathbb P\left(\frac{\sqrt{n_z\tau_3}}{\sigma}\Lambda_{\X_0,\epsilon} < t - \frac{\sqrt{n_z\sigma_{\epsilon_z}^2\zeta_1}}{\sigma}\Lambda_{\epsilon_z} -\frac{\sqrt{n_z\left(\gamma_1\zeta_1 + 2\tau_2^2\right)}}{\sigma}\Lambda_{Z_0} \right.\right.\\ 
        &\left.\left.\qquad-\frac{n_z\left\{\left\{\bm\beta^\top\bm\Sigma\bm\beta - \lambda/\mathfrak{g}\left\langle\bm\beta,\left\{\bm\beta - (\mathbb\I_p + \mathfrak{m}_n\bm\Sigma)^{-1}\bm\beta\right\}\right\rangle\right\} - \tau_0\right\}}{\sigma}\right)\right| \to 0.
    \end{aligned}
\end{equation*}
Moreover, by Theorem~\ref{thm:quad_form}, we have another Berry-Esseen bound for the randomness of $\bm\beta$ 
\begin{align*}
        &\sup_{t\in\mathbb R}\left|\mathbb P\left(\frac{\splitfrac{n_z\Big[\left[\bm\beta^\top\bm\Sigma\bm\beta - \lambda/\mathfrak{g}\left\langle\bm\beta,\left\{\bm\beta - (\mathbb\I_p + \mathfrak{m}_n\bm\Sigma)^{-1}\bm\beta\right\}\right\rangle\right]}{ - \sigma_{\bm\beta}^2/p\left[\Tr(\bm\Sigma\mathbb\I_m) - \lambda/\mathfrak{g}\left\{m - \Tr((\mathbb\I_p +\mathfrak{m}_n\bm\Sigma)^{-1}\mathbb\I_m)\right\}\right]\Big]}} {\sigma} \right.\right.\\
        &\left.\left. < t - \frac{\sqrt{n_z\sigma_{\epsilon_z}^2\zeta_1}}{\sigma}\Lambda_{\epsilon_z} - 
        \frac{\sqrt{n_z\left(\sigma_{\bm\beta}^2\gamma_1\zeta_1 + 2\tau_2^2\right)}}{\sigma}\Lambda_{\Z_0} - \frac{\sqrt{n_z\tau_3}}{\sigma}\Lambda_{\X_0,\epsilon}\right) - \mathbb P\left(\frac{\sqrt{n_z\tau_4^2}}{\sigma}\Lambda_{\bm\beta} < t -  \right.\right.\\
        &\left.\left.
        \frac{\sqrt{n_z\sigma_{\epsilon_z}^2\zeta_1}}{\sigma}\Lambda_{\epsilon_z} -\frac{\sqrt{n_z\left(\sigma_{\bm\beta}^2\gamma_1\zeta_1 + 2\tau_2^2\right)}}{\sigma}\Lambda_{\Z_0}- \frac{\sqrt{n_z\tau_3}}{\sigma}\Lambda_{\X_0,\epsilon}\right)\right|\leq O_p(m^{-1/5}),
\end{align*}
where
\begin{align*}
    \tau_4^2 = \frac{n_z\lambda^2}{\mathfrak{g}^2}\left\{\left(\mathbb E\left(\bm\beta^4\right) - \frac{3\sigma_{\bm\beta}^4}{p^2}\right)\sum_{i=1}^m\mathfrak{N}_{i,i}^2 + \frac{2\sigma_{\bm\beta}^4}{p^2}\Tr(\mathfrak{N}^2)\right\} \quad \mbox{and} \quad \mathfrak{N} = \left(\frac{\mathfrak{g}}{\lambda}\bm\Sigma - \mathbb\I_p + (\mathbb\I_p+\mathfrak{m}_n\bm\Sigma)^{-1}\right)\mathbb\I_m.
\end{align*}
The quantitative CLT for the numerator follows by applying the convolution formula for these independent Gaussian random variables. 
\begin{lem.s}
Under Assumptions~\ref{a:Sigmabound}-\ref{a:Sparsity} and S\ref{a:Gaussian_entries}, as $p \to \infty$, we have following Berry-Esseen bound
\label{lemma:BE_ineq_num}
\begin{equation*}
    \begin{aligned}
        \sup_{t\in\mathbb R}\left|\mathbb P\left(\frac{(\Z^\top\bm\beta+\bm \epsilon_z)^\top \Z(\X^\top \X+n\lambda\mathbb\I_p)^{-1}\X^\top \y - \tau_0}{\sigma} < t\right) - \Phi(t)\right| \to 0,
    \end{aligned}
\end{equation*}
where
\begin{align*}
    \tau_0 = \frac{n_z\sigma_{\bm\beta}^2}{p}\left[\Tr(\bm\Sigma\mathbb\I_m) - \frac{\lambda}{\mathfrak{g}}\left\{m - \Tr((\mathbb\I_p +\mathfrak{m}_n\bm\Sigma)^{-1}\mathbb\I_m)\right\}\right]\quad \mbox{and} \quad \sigma^2 = n_z\left(\tau_1 + 2\tau_2^2 + \tau_3 + \tau_4^2\right).
\end{align*}
\end{lem.s}

\subsection{Limits of the denominator}
We aim to obtain the first-order limits of the following two quantities
\begin{enumerate}
    \item $(\bm \epsilon_z^\top + \bm\beta^\top \Z)(\Z^\top\bm\beta + \bm \epsilon_z)$ and 
    \item $\y^\top \X(\X^\top \X+n\lambda\mathbb\I_p)^{-1}\Z^\top \Z(\X^\top \X+n\lambda\mathbb\I_p)^{-1}\X^\top \y$.
\end{enumerate}
For the first quantity $(\bm \epsilon_z^\top + \bm\beta^\top \Z)(\Z^\top\bm\beta + \bm \epsilon_z)$, recall in \cref{concent:marg_A2_denom_2} we have shown that
\begin{equation}
\label{concent:ridge_A2_denom_1}
    \begin{aligned}
        \mathbb P\left(\left|\|\Z\bm\beta + \bm \epsilon_z\|_2^2 - n_z(\sigma_{\bm\beta}^2\gamma_1 + \sigma_{\epsilon_z}^2)\right| \leq O_p(n_z^{1/2+\delta})\right)\geq 1-O_p(n_z^{-2\delta}).
    \end{aligned}
\end{equation}
For the second quantity $\y^\top \X(\X^\top \X+n\lambda\mathbb\I_p)^{-1}\Z^\top \Z(\X^\top \X+n\lambda\mathbb\I_p)^{-1}\X^\top \y$, we have shown that desirable concentration takes place on $\left\{\Z_0 \in \Gamma_1, \{\X_0, \epsilon\}\in \Gamma_3\right\}$. It follows that 
\begin{equation}
\label{concent:ridge_A2_denom_2}
    \begin{aligned}
        \mathbb P\left(\left|\y^\top \X(\X^\top \X+n\lambda\mathbb\I_p)^{-1}\Z^\top \Z(\X^\top \X+n\lambda\mathbb\I_p)^{-1}\X^\top \y - n_z\zeta_1\right| \leq O_p(p^{1-\delta})\right) \geq 1-O_p(p^{-1+2\delta}).
    \end{aligned}
\end{equation}
Now we are ready to present the quantitative CLT for $A(\hat{\bm\beta}_{\text{R}}(\lambda))$.

\subsection{Major quantitative CLT}
Combining  Lemma~S\ref{lemma:BE_ineq_num} and concentrations for the denominator in \cref{concent:ridge_A2_denom_1} and \cref{concent:ridge_A2_denom_2}, we conclude the following Berry-Esseen bound for $A(\hat{\bm\beta}_{\text{R}}(\lambda))$. 
\begin{thm.s}
Under Assumptions~\ref{a:Sigmabound}-\ref{a:Sparsity} and S\ref{a:Gaussian_entries}, 
the following quantitative CLT holds for $A(\hat{\bm\beta}_{\text{R}}(\lambda))$
\begin{equation*}
        \begin{aligned}
            \sup_{t\in\mathbb R}\left|\mathbb P\left(\sqrt{n_z\eta_{\text{R}}}\left\{\frac{(\Z^\top\bm\beta+\bm \epsilon_z)^\top \Z(\X^\top \X+n\lambda\mathbb\I_p)^{-1}\X^\top \y}{\|\Z^\top\bm\beta + \bm \epsilon_z\|_2\|\Z(\X^\top \X+n\lambda\mathbb\I_p)^{-1}\X^\top \y\|_2} - \frac{\tau_0}{\sqrt{\zeta_1\zeta_2}}\right\} < t\right) - \Phi(t)\right| \to 0,
        \end{aligned}
\end{equation*}
where 
\begin{align*}    
    &\tau_0 = \frac{\sigma_{\bm\beta}^2}{p}\left[\Tr(\bm\Sigma\mathbb\I_m) - \frac{\lambda}{\mathfrak{g}}\left\{m - \Tr\left(\left(\mathbb\I_p +\mathfrak{m}_n\bm\Sigma\right)^{-1}\mathbb\I_m\right)\right\}\right], \quad \eta_{\text{R}} = \frac{\zeta_1\zeta_2}{\tau_1 + 2\tau_2^2 + \tau_3+\tau_4^2},\\
    &\zeta_1 = \frac{\sigma_{\bm\beta}^2}{p}\left\{\Tr(\bm\Sigma\mathbb\I_m) - 2\Tr\left(\left(\mathbb\I_p+\mathfrak{m}_n\bm\Sigma\right)^{-1}\bm\Sigma\mathbb\I_m\right) + \frac{\mathfrak{m}_n'}{\mathfrak{m}_n^2}\Tr\left(\left(\mathbb\I_p+\mathfrak{m}_n\bm\Sigma\right)^{-2}\bm\Sigma\mathbb\I_m\right)\right\}  \\
    &\qquad +\frac{n_z\sigma_\epsilon^2}{n\lambda}\left\{\Tr\left((\mathbb \I_p + \mathfrak{m}_n\bm\Sigma)^{-1}\bm\Sigma\right) - \frac{\mathfrak{m}_n'}{\mathfrak{m}_n^2}\Tr\left(\left(\mathbb \I_p + \mathfrak{m}_n\bm\Sigma\right)^{-2}\bm\Sigma\right)\right\},\quad \mbox{and} \\
    &\zeta_2  =\sigma_{\bm\beta}^2\gamma_1 + \sigma_{\epsilon_z}^2.
\end{align*}
Recall that we have 
\begin{align*}
        &\tau_1 = \zeta_1\zeta_2,\\
        &\tau_2 = \frac{\sigma_{\bm\beta}^2}{p}\left\{\Tr(\bm\Sigma\mathbb\I_m) - \Tr\left((\mathbb\I_p + \mathfrak{m}_n\bm\Sigma)^{-1}\bm\Sigma\mathbb\I_m\right)\right\} + o_p(1),\\
        &\tau_3 = n_z/(n\mathfrak{g}^2)\Bigg[\mathfrak{h}\left(\frac{\sigma_{\bm\beta}^2}{p}\Tr((\mathbb I_p + \mathfrak{m}_n\bm\Sigma)^{-1}\bm\Sigma\mathbb\I_m)\right)^2 \\
        &\qquad+ \frac{\sigma_{\bm\beta}^2}{p}\Tr(\bm\Sigma\mathbb\I_m)\Bigg\{\sigma_\epsilon^2\mathfrak{g} + \lambda\mathfrak{m}_n'\Bigg(\frac{\lambda\sigma_{\bm\beta}^2}{p}\Tr((\mathbb\I_p + \mathfrak{m}_n\bm\Sigma)^{-2}\bm\Sigma\mathbb\I_m)- \frac{\sigma_{\epsilon}^2}{n}\Tr((\mathbb\I_p + \mathfrak{m}_n\bm\Sigma)^{-2}\bm\Sigma)\Bigg)\Bigg\}\Bigg],\\
        &\tau_4^2 = \frac{n_z\lambda^2}{\mathfrak{g}^2}\left\{\left(\mathbb E\left(\bm\beta^4\right) - 3\frac{\sigma_{\bm\beta}^4}{p^2}\right)\sum_{i=1}^m\mathfrak{N}_{i,i}^2 + 2\sigma_{\bm\beta}^4\Tr(\mathfrak{N}^2)\right\}, \quad \mbox{and} \\
        & \mathfrak{N} = \left(\frac{\mathfrak{g}}{\lambda}\bm\Sigma - \mathbb\I_p + (\mathbb\I_p+\mathfrak{m}_n\bm\Sigma)^{-1}\right)\mathbb\I_m.
\end{align*}  
\end{thm.s}
Here, we present our results in their original form. By simply replacing the quantities defined in  Definitions~\ref{def:heritability}, \ref{def:kwgamma}, and \ref{def:elleth}, we derive our results in Theorem~S\ref{thm: CLT for ridge A^2}.

\pagebreak
\begin{suppfigure}[!t] 
\includegraphics[page=1,width=1\linewidth]{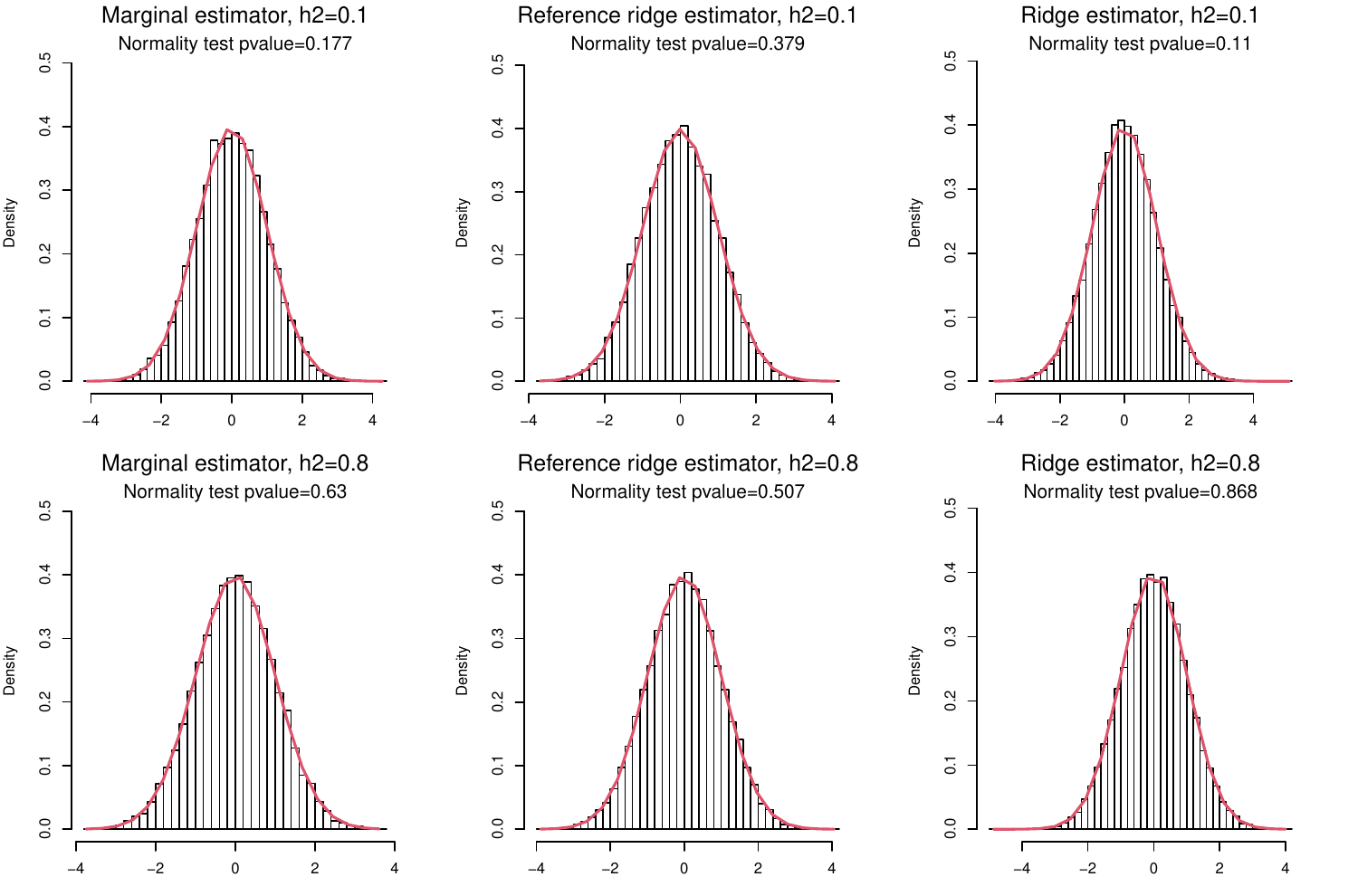}
\centering
\caption{
\textbf{Asymptotic normality of individual-level genetically predicted values.}
Based on the real genetic data from the UK Biobank study, we illustrate the empirical distribution of genetically predicted values $\z^\top\hat{\bm\beta}$ for $\hat{\bm\beta}_{\text{M}}$, $\hat{\bm\beta}_{\text{W}}(\lambda)$, and $\hat{\bm\beta}_{\text{R}}(\lambda)$ (from left to right).
We assess the asymptotic normality with the Shapiro-Wilk test \citep{shapiro1965analysis}. 
Here we set $p~=~$461,488, heritability $h_{\bm\beta}^2=h^2_{{\bm\beta}_z}=0.1$ (upper panels) or $0.8$ (lower panels), sparsity $m/p$ ranging from $0.001$ to $0.5$, and $n~=~50,000$.
}
\label{main_fig_s1}
\end{suppfigure}

\begin{suppfigure}[!t] 
\includegraphics[page=1,width=1\linewidth]{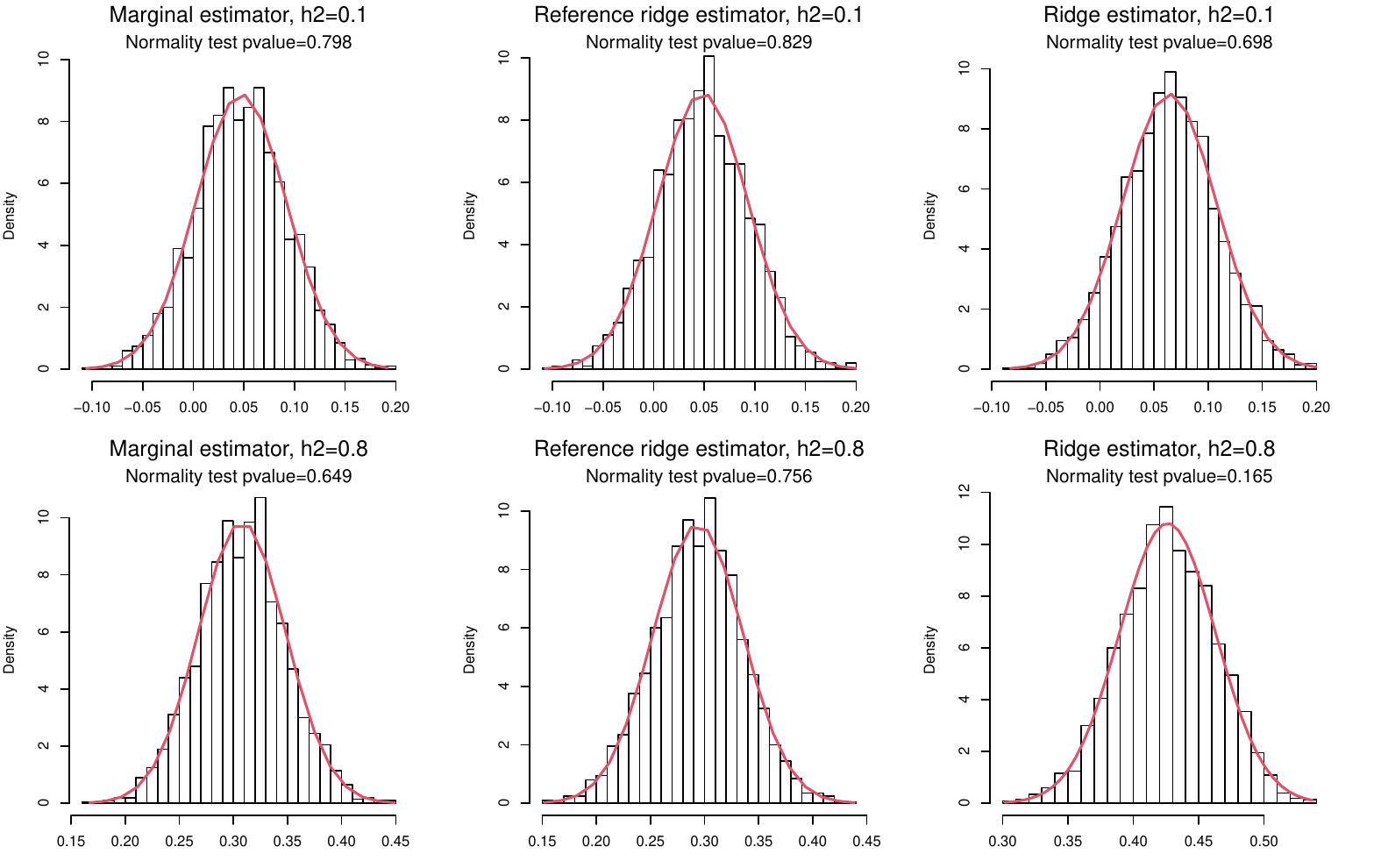}
\centering
\caption{
\textbf{Asymptotic normality of cohort-level prediction accuracy.}
Based on the real genetic data from the UK Biobank study, we illustrate the empirical distribution of cohort-level prediction accuracy $A (\hat{\bm\beta})$ for $\hat{\bm\beta}_{\text{M}}$, $\hat{\bm\beta}_{\text{W}}(\lambda)$, and $\hat{\bm\beta}_{\text{R}}(\lambda)$ (from left to right). 
We assess the asymptotic normality with the Shapiro-Wilk test \citep{shapiro1965analysis}. 
Here we set $p~=~$461,488, heritability $h_{\bm\beta}^2=h^2_{{\bm\beta}_z}=0.1$ (upper panels) or $0.8$ (lower panels), sparsity $m/p$ ranging from $0.001$ to $0.5$,  $n~=~50,000$, and $n_w~=~500$. 
}
\label{main_fig_s2}
\end{suppfigure}

\end{document}